\documentclass[11pt, letterpaper]{article}

\usepackage[margin=1in]{geometry}

\usepackage[table]{xcolor}
\definecolor{linkcol}{rgb}{0,0,0.38}
\definecolor{citecol}{rgb}{0.8,0,0}
\definecolor{urlcol}{rgb}{0.1,0.35,0}

\usepackage{graphicx} \usepackage{tikz}

\RequirePackage{silence}
\usepackage{wrapfig}
\usepackage{subcaption}
\usepackage{stackengine}

\usepackage{amsthm}
\usepackage{amsmath,amsfonts,amssymb}
\usepackage{mathtools}
\usepackage{bbm}
\usepackage{aligned-overset}
\usepackage{enumitem}
\usepackage[normalem]{ulem}

\usepackage{thm-restate}

\usepackage{multicol}
\usepackage{multirow}

\usepackage[
backend=biber,
style=alphabetic,
citestyle=alphabetic,
maxalphanames=4,
maxcitenames=99,
mincitenames=98,
maxbibnames=99,
giveninits=true,
url=false,
]{biblatex}

\usepackage[ruled,linesnumbered]{algorithm2e}
\DontPrintSemicolon

\usepackage[pdfpagelabels,plainpages=false,bookmarks=false,hyperfootnotes=false]{hyperref}
\hypersetup{colorlinks, linkcolor=linkcol, citecolor=citecol, urlcolor=urlcol}
\usepackage[nameinlink]{cleveref}

\newtheorem{theorem}{Theorem}[section]
\newtheorem{lemma}[theorem]{Lemma}
\newtheorem{observation}[theorem]{Observation}
\newtheorem{corollary}[theorem]{Corollary}

\newtheorem{invariant}[theorem]{Invariant}
\newtheorem{definition}[theorem]{Definition}
\newtheorem{proposition}[theorem]{Proposition}

\Crefname{algocf}{Algorithm}{Algorithms}
\Crefname{AlgoLine}{Line}{Lines}
\Crefname{observation}{Observation}{Observations}
\Crefname{assumption}{Assumption}{Assumptions}

\newcommand{\lr}[1]{\left(#1\right)}
\newcommand{\lre}[1]{\left[#1\right]}
\newcommand{\lrv}[1]{\left\vert#1\right\vert}
\newcommand{\lrb}[1]{\left\lbrace#1\right\rbrace}

\newcommand{\R}{\mathbb{R}}
\newcommand{\N}{\mathbb{N}}

\newcommand{\cB}{\mathcal{B}}

\newcommand{\cI}{\mathcal{I}}

\newcommand{\cL}{\mathcal{L}}
\newcommand{\cM}{\mathcal{M}}

\newcommand{\cS}{\mathcal{S}}

\newcommand{\cX}{\mathcal{X}}

\renewcommand{\epsilon}{\varepsilon}

\newcommand{\eback}{\reflectbox{$\vec{\reflectbox{$e$}}$}}

\newcommand{\atf}[2]{t^{#1\to#2}}
\newcommand{\dropi}{\textnormal{drop}}
\newcommand{\drop}[2][]{\textnormal{drop}_{#1}\lr{#2}}
\newcommand{\dropGp}[1]{\textnormal{drop}_{G'}\lr{#1}}
\newcommand{\dropG}[1]{\textnormal{drop}_{G}\lr{#1}}
\newcommand{\dist}[2][]{\textnormal{dist}_{#1}\lr{#2}}

\newcommand{\conn}[1]{\textnormal{cost}\lr{#1}}

\newcommand{\opt}{\textnormal{OPT}}
\newcommand{\tmst}{\textnormal{TMST}}
\newcommand{\bcr}{\textnormal{BCR}}
\newcommand{\dual}{\textnormal{dual}}

\newcommand{\hyp}{\textnormal{HYP}}
\newcommand{\mst}{\textnormal{MST}}
\newcommand{\merge}[2]{\textnormal{merge}_{#1}\lr{#2}}

\newcommand{\val}[1]{\textnormal{value}\lr{#1}}

\newcommand{\finalratio}{1.898}

\DeclareMathOperator{\cost}{cost}
\DeclareMathOperator{\depth}{depth}
\DeclareMathOperator{\rt}{root}

\newcommand{\defaulttag}{\stepcounter{equation}\tag{\theequation}}

\newcounter{safeeq}\newcommand{\beforecustomtag}[1]{
    \setcounter{safeeq}{\theequation}\setcounter{equation}{#1}}
\newcommand{\aftercustomtag}{
    \setcounter{equation}{\thesafeeq}}

  \usepackage{tikz}
\usetikzlibrary{calc, arrows.meta, decorations,nfold, decorations.pathreplacing, decorations.pathmorphing, decorations.markings, calligraphy, patterns, patterns.meta, shapes.geometric}

\makeatletter
\let\nfold@orig@tikz@finish\tikz@finish
\def\tikz@finish{\tikz@nfold@do\nfold@orig@tikz@finish}
\let\tikz@nfold@do\relax
\tikzset{offset/.code=\edef\tikz@temp{#1}\ifx\tikz@temp\tikz@nonetext \let\tikz@nfold@do\relax
  \else\def\tikz@nfold@do{\pgfgetpath\tikz@temp\pgfsetpath\pgfutil@empty\pgfoffsetpath\tikz@temp{#1}}\fi}
\makeatother

\tikzstyle{basevertex} =
[minimum size=5mm
, align=center
, text width=5mm
, inner sep=0mm
, outer sep=1mm
, fill
, fill=white
, draw=black
, line width = 0.7mm]
\tikzstyle{steiner} = [basevertex,circle, line width = 0.5mm]
\tikzstyle{smallsteiner} = [steiner, minimum size = 4mm, text width = 4mm]
\tikzstyle{tinysteiner} = [steiner, minimum size = 2mm, text width = 2mm]
\tikzstyle{terminal} = [basevertex,rectangle]
\tikzstyle{smallterminal} = [basevertex,rectangle, minimum size =4mm, text width =4mm]
\tikzstyle{root} = [terminal, fill=black, text=white]

\tikzstyle{activeset} = [line width = 1mm, blue]
\tikzstyle{activerootset} = [activeset, red]

\tikzstyle{directed} = [->,>={Stealth},
offset=-3]
\tikzstyle{edge} = [line width = 0.5mm, black]
\tikzstyle{diredge} = [edge, directed]
\tikzstyle{mstedge} = [line width = 1mm, black!30!green]
\tikzstyle{mstdiredge} = [mstedge,directed]
\tikzstyle{curly} = [decorate,decoration=snake]
\tikzstyle{curlyHalfA} = [decorate,decoration=snake,dashed,dash phase=0pt,dash pattern=on 14.6pt off 14.6pt]
\tikzstyle{curlyHalfB} = [curlyHalfA,dash phase=14.6pt]

\tikzstyle{edgeprogress} = [line width = 0.8mm, blue]
\tikzstyle{thickedgeprogress} = [edgeprogress,line width = 6pt]
\tikzstyle{edgedirprogress} = [line width = 0.8mm, blue, directed]
\tikzstyle{thickedgedirprogress} = [edgedirprogress, line width = 4pt]
\tikzstyle{greendirprogress} = [edgedirprogress, green!70!black]
\tikzstyle{edgerootprogress} = [edgedirprogress, red]
\tikzstyle{edgedirprogresshalfA} = [edgedirprogress,dashed,>={Stealth[black]}]
\tikzstyle{edgedirprogresshalfB} = [edgedirprogresshalfA,dash phase=3pt]

\tikzstyle{edgelabel}   = [fill=white,circle,inner sep=0pt,minimum size=4mm, text = black]
\tikzset{label distance=-3pt}

\tikzstyle{diagramAxis} =[line width = 2pt,->, >={Stealth}]

\title{The Bidirected Cut Relaxation for Steiner Tree:\\
Better Integrality Gap Bounds and the Limits of Moat Growing}

\author{
Paul Paschmanns\thanks{
Department of Computer Science, ETH Zurich, Zurich, Switzerland.
Email: \href{mailto:paul.paschmanns@inf.ethz.ch}{paul.paschmanns@inf.ethz.ch}.
}
\and
Vera Traub\thanks{Department of Computer Science, ETH Zurich, Zurich, Switzerland.
Email: \href{mailto:vtraub@ethz.ch}{vtraub@ethz.ch}.
}
}
\date{}

\begin{document}

\hypersetup{pageanchor=false}

\maketitle

\thispagestyle{empty}
\addtocounter{page}{-1}

\begin{abstract}
The Steiner Tree problem asks for the cheapest way of connecting a given subset of the vertices in an undirected graph.
One of the most prominent linear programming relaxations for Steiner Tree is the Bidirected Cut Relaxation (BCR).
Determining the integrality gap of this relaxation is a long-standing open question. 
For several decades, the best known upper bound was $2$, which is achievable by standard techniques.
Only very recently, Byrka, Grandoni, and Traub [FOCS 2024] showed that the integrality gap of BCR is strictly below $2$.

We prove that the integrality gap of BCR is at most $\finalratio$, improving significantly on the previous bound of $1.9988$.
For the important special case where a terminal minimum spanning tree is an optimal Steiner tree, we show that the integrality gap is at most $\frac{12}{7}$, by providing a tight analysis of the dual-growth procedure by Byrka et al. To obtain the general bound of $\finalratio$ on the integrality gap, we generalize their dual growth procedure to a broad class of moat-growing algorithms.
Moreover, we prove that no such moat-growing algorithm yields dual solutions certifying an integrality gap below $\frac{12}{7}$.

Finally, we observe an interesting connection to the Hypergraphic Relaxation.
\end{abstract}

\clearpage

\tableofcontents

\thispagestyle{empty}
\addtocounter{page}{-1}
\clearpage

\hypersetup{pageanchor=true}

\section{Introduction}\label{sec:intro}

The Steiner Tree problem is one of the most fundamental problems in Network Design and has a multitude of applications (see e.g. \cite{HRW92book, PS02book, BZ20book}). 
We are given an undirected graph $G=(V,E)$ with positive edge costs $c\colon E\to\R_{> 0}$ and a set $R\subseteq V$. The task is to find a tree $T$ in $G$ that connects all vertices from $R$ and has minimal cost $c\lr{T} = \sum_{e\in T} c(e)$. We refer to the vertices in $R$ as \emph{terminals} and to other vertices as \emph{Steiner vertices}.

The Minimum Spanning Tree problem, or MST problem for short, is a special case of Steiner Tree (where $R=V$). 
However, in contrast to the MST problem, approximating the Steiner Tree problem better than a factor of $\frac{96}{95}$ is NP-hard~\cite{CC08}.
Computing a spanning tree on the terminals (in the metric closure of the instance) yields a simple $2$-approximation for Steiner Tree \cite{MSTis2apxSteinerTree}.
The currently best-known approximation algorithm for Steiner Tree is due to Byrka, Grandoni, Rothvoss, and Sanit\'a \cite{ln4directedcomponent} and achieves an approximation factor of $\ln(4) +\epsilon$ (for any fixed $\epsilon > 0$).
See also \cite{ln4LocalSearchJournal} for a different algorithm achieving the same approximation factor.

The natural linear programming relaxation for Steiner Tree, which requires a fractional value of at least one for each cut separating two terminals, has an integrality gap of exactly $2$, even in the special case of minimum spanning trees.
Several stronger relaxations have been proposed, the most prominent ones being the \emph{Bidirected Cut Relaxation} and the \emph{Hypergraphic LP} (also known as the \emph{Directed Component LP}). 
Several other LPs turned out to be equivalent to one of these two relaxations, see e.g.~\cite{ListOfSteinerTreeFormulations, CKP10, BcR4over3onQuasiBipartite}.

The Hypergraphic LP (which we formally define in \Cref{sec:hypergraphic}) has been used in the approximation algorithm from \cite{ln4directedcomponent} and it is known to have an integrality gap of at most $\ln 4$ as shown in \cite{ln4Intgap} (see also \cite{ln4LocalSearchJournal}).
However, one important drawback of this relaxation is that it has an exponential number of variables (and constraints) and it is NP-hard to solve \cite{ln4Intgap}.
One can solve the Hypergraphic LP approximately up to a factor of $1+\epsilon$ in polynomial time for fixed $\epsilon > 0$, but the known algorithms for this
rely on enumerating all subsets of terminals of size up to $2^{1/\epsilon}$  and are thus very slow.

The Bidirected Cut Relaxation in contrast has only a polynomial number of variables. 
It is a natural cut-based relaxation in the bidirected graph. 
More precisely, we replace every undireced edge $\{u,w\}$ of the graph $G$ by two bidirected edges $(u,v)$ and $(v,u)$ of the same cost and we write $\vec{E}$ to denote the resulting directed edge set.
If instead of searching for a Steiner tree in the undirected graph $(V,E)$, we search for a Steiner tree in $(V, \vec{E})$ that is oriented towards some (arbitrarily chosen) root  $r\in R$, we naturally obtain the Bidirected Cut Relaxation:

\beforecustomtag{1234}
\begin{equation}\label{eq:bcr}
\tag{BCR}
\begin{aligned}
\min \sum_{e \in \vec{E}} c(e) \cdot x_e \\
\sum_{e\in \delta^+(S)} x_e \ \geq&\ 1 &\text{ for all } S\subseteq V\setminus\{r\} \text{ with } S\cap R \neq \emptyset\\
x_{e} \ \ge&\ 0 &\text{ for all }e\in \vec{E},
\end{aligned}
\end{equation}
\aftercustomtag
where $\delta^+(S) \coloneqq \{(v,w) \in \vec{E}: v\in S, w\notin S\}$ denotes the outgoing edges of $S$.
The value of \ref{eq:bcr} is independent of the choice of the root $r\in R$ (see e.g.~\cite{ListOfSteinerTreeFormulations}).
Work of Fulkerson and Edmonds on spanning arborescences \cite{F74, EdmondsBranching} implies that in the special case of minimum spanning trees, that is when $R=V$, the value of \ref{eq:bcr} is equal to the cost of a minimum spanning tree in $G$.
For general instances of Steiner Tree, the best known lower bound on the integrality gap of \ref{eq:bcr} is $\frac{6}{5}$ \cite{BcrLowerBound6over5} (improving on earlier bounds of $\frac{36}{31}$ from \cite{ln4directedcomponent} and $\frac{8}{7}$ by Goemans and Skutella; see~\cite{KPT11}).

Despite extensive work on the Bidirected Cut Relaxation (see the discussion in \Cref{sec:related_work}), until very recently it remained open whether the integrality gap of \ref{eq:bcr} is better than $2$ and thus better than the integrality gap of the natural undirected relaxation.
In a recent breakthrough, Byrka, Grandoni, and Traub \cite{BCRlessthan2} proved that the integrality gap of \ref{eq:bcr} is at most $1.9988$.
To obtain this result, they construct a solution to the dual of \ref{eq:bcr} via a moat-growing procedure (as it is used in standard primal-dual algorithms).
However, analyzing their construction in this bidirected setting is much more challenging than the analysis of standard primal-dual algorithms.
While \cite{BCRlessthan2} gave an analysis that shows that the integrality gap of \ref{eq:bcr} is below $2$, they did not provide a tight analysis of their construction and it remained open whether much better upper bounds on the integrality gap can be achieved via such a moat-growing approach.

\subsection{Our Contributions}

We provide a new and much stronger analysis of the construction from \cite{BCRlessthan2}, generalize their construction to a more general class of moat-growing algorithms, and provide a lower bound example for what is achievable via this approach.
Our main result is a significantly improved upper bound on the integrality gap of \ref{eq:bcr}.

\begin{theorem}\label{thm:main_gap}
The integrality gap of the Bidirected Cut Relaxation for the Steiner Tree problem is at most $\finalratio$.
\end{theorem}

As in~\cite{BCRlessthan2}, we bound the integrality gap by the ratio between the cost of some Steiner tree and the value of some solution to the dual of \ref{eq:bcr}.
By taking the metric closure, we may assume that the graph $G$ is complete and the cost function $c$ satisfies the triangle inequality.
Then one possible Steiner tree to consider is a \emph{terminal MST}, which is a minimum cost spanning tree in the complete graph $G[R]$ induced by the terminals.
For a fixed instance of the Steiner Tree problem, we denote the cost of a terminal MST by $\tmst$.
A well-known fact, which was already utilized by \cite{BCRlessthan2}, is that $\tmst$ is always at most twice the value of \ref{eq:bcr} (see \Cref{prop:simple_lb}).
An immediate consequence of this is that the integrality gap of \ref{eq:bcr} is less than $2$ whenever the cost $\opt$ of a cheapest Steiner tree is (significantly) less than $\tmst$.

For this reason, instances where a terminal MST is an optimum Steiner tree are the key case to consider when proving that \ref{eq:bcr} has an integrality gap of less than $2$.

\begin{definition}[MST-optimal instances]
An instance of the Steiner Tree problem is \emph{MST-optimal} if we have $\tmst =\opt$, that is, a terminal MST is a cheapest Steiner tree.
\end{definition}

We also remark that the best-known lower bounds on the integrality gap of \ref{eq:bcr} stem from MST-optimal instances \cite{BcrLowerBound6over5}. 
For this important special case, we prove a better upper bound on the integrality gap of \ref{eq:bcr}.
To this end, we consider the algorithm from \cite{BCRlessthan2} for constructing dual solutions and provide a tight analysis of it.
\begin{theorem}\label{thm:gap_mst_optimal}
 The integrality gap of the Bidirected Cut Relaxation for Steiner Tree restricted to MST-optimal instances is at most $\frac{12}{7}$.   
\end{theorem}
To prove \Cref{thm:gap_mst_optimal}, we crucially deviate from the proof strategy from \cite{BCRlessthan2}.
For an overview of the main ideas behind our new analysis, we refer to \Cref{sec:outline_analysis}.

In order to prove our upper bound of $\finalratio$ for the integrality gap of \ref{eq:bcr} for general Steiner tree instances, we generalize the construction from \cite{BCRlessthan2} to what we call \emph{merge plan based} dual growth procedures. 
This is a class of procedures for constructing dual solutions for \ref{eq:bcr} which captures algorithms that can be viewed as a natural directed analogue of the standard moat growing algorithms (as they are used in primal-dual algorithms).
In \Cref{sec:lower_bound}, we show that not only our analysis of the construction from \cite{BCRlessthan2} is tight for MST-optimal instances, but even no procedure from the much broader class of merge plan based dual growth procedures can lead to a better dual solution for \ref{eq:bcr}.

\begin{restatable}{theorem}{thmMergePlanAlgoLowerBound}
\label{thm: merge plan algos cannot be better than 7 over 12 mst}
For every $\varepsilon>0$, there is an MST-optimal instance such that no merge plan based dual solution has value more than $\lr{\frac{7}{12}+\epsilon}\tmst$.
\end{restatable}

For a concise definition of merge plan based dual solutions, we refer to \Cref{sec:lower_bound}.
\Cref{thm: merge plan algos cannot be better than 7 over 12 mst}~implies that significantly different techniques will be needed if one wants to prove that \ref{eq:bcr} has an integrality gap below $\frac{12}{7}$, even for the special case of MST-optimal instances.

Finally, we prove an interesting connection to the integrality gap of the Hypergraphic Relaxation for Steiner tree, whose value we denote by $\hyp$.
(For a definition of the Hypergraphic Relaxation, see \Cref{sec:hypergraphic}.)
We write $\opt(\cI)$ to denote the cost of a cheapest Steiner tree for instance $\cI$, and we will simply write $\opt$ if the instance is fixed and clear from the context.

\begin{restatable}{theorem}{HypRatiosSame}\label{thm:ratios_same}
    We have
    \begin{equation}\label{eq:ratios_same}
       \begin{aligned}
    &\ \sup \left\{ \frac{\hyp(\cI)}{\bcr(\cI)} \colon \cI \text{ instance of Steiner Tree} \right\} \\[2mm]
    \ = &\
    \sup \left\{ \frac{\opt(\cI)}{\bcr(\cI)} \colon \cI \text{ MST-optimal instance of Steiner Tree} \right\}.
    \end{aligned} 
    \end{equation}
\end{restatable}

Observe that the right-hand side of~\eqref{eq:ratios_same} is the integrality gap of the Bidirected Cut Relaxation for Steiner Tree restricted to MST-optimal instances, which by \Cref{thm:gap_mst_optimal} is at most $\frac{12}{7}$.

\subsection{Further Related Work}\label{sec:related_work}
The integrality gap of \ref{eq:bcr} has been studied for multiple decades, with \cite{BCRlessthan2} giving the first upper bound strictly below $2$ for the general case.

For the special case of quasi-bipartite instances (instances without any neighboring Steiner vertices), Rajagopalan and  Vazirani proved that the integrality gap is at most $\frac32$ \cite{BcR3over2onQuasiBipartite}, which was then improved to $\frac43$ by Chakrabarty, Devanur, and  Vazirani~\cite{BcR4over3onQuasiBipartite}. 
Later, \citeauthor{BcrClawFreeEquivToHyp} \cite{BcrClawFreeEquivToHyp} proved that for quasi-bipartite instances, \ref{eq:bcr} is equivalent to the Hypergraphic Relaxation, and the same applies even for the more general class of claw-free instances (instances that do not contain a Steiner vertex with three neighboring Steiner vertices).
Together with \cite{ln4Intgap}, this implies an integrality gap of at most $\frac{73}{60}$ on quasi-bipartite instances (for both \ref{eq:bcr} and HYP).
\citeauthor{BcrClawFree991over732} \cite{BcrClawFree991over732} proved an upper bound of $\frac{991}{732}$ on the integrality gap of HYP on claw-free instances, which by \cite{BcrClawFreeEquivToHyp} implies the same upper bound on the integrality of \ref{eq:bcr} on claw-free instances.

\citeauthor{BcrSteinerForest} \cite{BcrSteinerForest} considered a generalization of \ref{eq:bcr} to the Steiner Forest problem and showed that this relaxation has some promising properties. In particular, they proved that its integrality gap is at most $\frac{16}{9}$ on instances with half-integral optimal LP solutions.
While very recently approximation factors strictly below $2$ have been achieved for Steiner Forest \cite{SteinerForest2minusEpsilon, Steinerforest1point994}, no LP relaxation solvable in polynomial-time is known to have an integrality gap strictly below $2$ for Steiner Forest.

While the proof from \cite{BCRlessthan2} can be turned into a polynomial-time algorithm for Steiner Tree with an approximation factor of $1.9988$ and thus strictly better than $2$, significantly better approximation algorithms for Steiner Tree were known before.\footnote{Our proof implies that the relative greedy algorithm by \citeauthor{ImprovedRelGreedy} \cite{ImprovedRelGreedy} computes a Steiner tree of cost at most $\finalratio$ times the value of \ref{eq:bcr}.
(See \Cref{sec:existence_merge_plan}.)}
The first better-than-two approximation \cite{RelGreedy11over6} uses a relative approach. 
This has been improved multiple times (\cite{ImprovedRelGreedy,RelGreedy1point664,IteratedRelGreedyWithTerminalInsertion,RelGreedy5over3}) culminating in a $1.55$-approximation \cite{BestRelGreedy}. 
The currently best known approximation algorithms have an approximation factor of $\ln4+\varepsilon<1.4$. 
This was first achieved by \citeauthor{ln4directedcomponent} \cite{ln4directedcomponent} using an iterative rounding technique. 
Later, \citeauthor{ln4LocalSearchJournal} \cite{ln4LocalSearchJournal} obtained the same guarantee using a local search algorithm.

\subsection{Structure of the Paper}
The remainder of this paper is structured as follows. In \Cref{sec:overview}, we describe the techniques that we use to prove our upper bounds on the integrality gap (\Cref{thm:main_gap,thm:gap_mst_optimal}).
We start with a brief recap of the classical moat growing procedure in an undirected setting, then move on to a directed version that generalizes the dual-growth procedure from \cite{BCRlessthan2}, and end the section with an outline of our new analysis of this dual growth procedure.

Besides an instance of the Steiner tree problem, our dual growth algorithm additionally takes a so-called \emph{merge plan} as an input.
This object classifies which terminal sets are active at each point in time during the dual growth. 
In \Cref{sec:merge_plan}, we discuss several useful properties of such merge plans.

In \Cref{sec: assumptions}, we then introduce some technical assumptions about our instances and merge plans that we can make without loss of generality and that will simplify our proofs later on.
In \Cref{sec: analysis of moat growing}, we provide our new analysis of the dual growth procedure. 
That is, we show that (under the assumptions established in \Cref{sec: assumptions}) our dual growth procedure indeed produces a large dual solution if we provide it with a good merge plan.

In \Cref{sec:existence_merge_plan}, we then complete the proof of our new upper bound on the integrality gap of \ref{eq:bcr} (\Cref{thm:main_gap}).
To this end, we simultaneously construct a Steiner tree and a merge plan, such that the cost of the Steiner tree is at most by a factor of $\finalratio$ larger than the value of the dual solution to \ref{eq:bcr} that we obtain from the merge plan (using our dual-growth procedure analyzed in \Cref{sec: analysis of moat growing}).

Finally, \Cref{sec:lower_bound,sec:hypergraphic} are independent of our proof of the new upper bound on the integrality gap. 
In \Cref{sec:lower_bound}, we prove our lower bound on what is achievable using our moat-growing algorithm (\Cref{thm: merge plan algos cannot be better than 7 over 12 mst}). In \Cref{sec:hypergraphic}, we discuss a connection to the Hypergraphic Relaxation, that is we prove \Cref{thm:ratios_same}.

Most figures in this paper show (partial) results of an undirected or bidirected moat growing procedure. There are animated versions of these figures showing the progress of the algorithm over time available at \href{https://people.inf.ethz.ch/ppaschmanns/}{people.inf.ethz.ch/ppaschmanns/}. \section{Overview of our Techniques}\label{sec:overview}

In this section, we give an outline of how we prove \Cref{thm:main_gap,thm:gap_mst_optimal}.
First, we introduce a general class of moat-growing algorithms, which we call \emph{dual growth with merge plan} (\Cref{sec:outline_algo}).
Next, in \Cref{sec:outline_merge_plan}, we discuss how we choose the \emph{merge plans}, which are given as an input to the moat-growing algorithm.
Finally, in \Cref{sec:outline_analysis}, we provide an outline of our new analysis of the moat-growing algorithm for \ref{eq:bcr}.

\subsection{Dual Growth with Merge Plans}
\label{sec:outline_algo}

The key part of our proof is the construction of solutions to the dual of \ref{eq:bcr}.
This dual linear program is given by
\beforecustomtag{2345}
\begin{equation}\label{eq:dual_bcr}
\tag{Dual-BCR}
\begin{aligned}
\max \sum_{\substack{S\subseteq V\setminus\{r\}:\\ S\cap R \neq \emptyset}} y_S \\
\sum_{S : e\in \delta^+(S)} y_S \ \leq&\ c(e) &\text{ for all }e\in \vec{E}\\
y_S \ \ge&\ 0 &\text{ for all } S\subseteq V.
\end{aligned}
\end{equation}
\aftercustomtag
Note that we wrote the dual linear program such that it has variables for all subsets of $V$, but only those containing some terminal, but not the root, contribute to the objective function.
We can always set all other dual variables to zero in an optimal dual solution.

To construct a dual solution, we use a moat growing algorithm---similar to \cite{BCRlessthan2} and standard primal-dual methods.
Such moat growing algorithms start with the dual solution $y=0$. 
Then, for each terminal $s\in R$, they take the smallest set $U_s$ that contains $s$ and whose corresponding dual variable can be increased without violating the dual constraints.
This is 
\[
U_s \ =\ \{v\in V : v\text{ is reachable from $s$ by a path of tight edges}\},
\]
where an edge $e$ is \emph{tight} if the corresponding constraint in the dual LP is satisfied with equality.
Then the algorithm increases these dual variables (for all $s\in R$) at a unit speed, until any further improvement would lead to a violation of dual constraints.
At this point, the sets $U_s$ change, and one can iterate this process.
Moat growing algorithms do precisely this, except that at some times during the algorithm, some of the sets $U_s$ get merged into a single set.
Then, instead of growing the dual variable for a set $U_s$, we grow the dual variable for a set
\[
U_S \ =\ \{v\in V : v\text{ is reachable from the set $S$ by a path of tight edges}\}.
\]
where $S\subseteq R$ is a subset possibly containing several terminals.
For moat growing in undirected graphs, there is a natural point in time when two sets $U_{s_1}, U_{s_2}$ should be merged: this is when the corresponding terminals are connected by a path of tight edges.
See \Cref{fig:moat_growing} for an example.
\begin{figure}
    \centering

\resizebox{0.32\textwidth}{!}{
\begin{tikzpicture}[scale=0.65]
    \Huge
    \useasboundingbox (-6,-6) rectangle ++(24,18);
    \node[terminal, label=left:{$a$}] (A) at (  0,6) {};
    \node[steiner, label=above:{$b$}] (B) at (  4,3) {};
    \node[steiner, label=above:{$c$}] (C) at (  8,3) {};
    \node[terminal, label=right:{$d$}] (D) at ( 12,6) {};
    \node[terminal, label=left:{$e$}] (E) at (  0,0) {};
    \node[terminal, label=right:{$f$}] (F) at ( 12,0) {};
    \draw[edge] (A) -- (B);
    \draw[edge] (A) -- (D);
    \draw[edge] (A) -- (E);
    \draw[edge] (B) -- (C);
    \draw[edge] (B) -- (E);
    \draw[edge] (C) -- (D);
    \draw[edge] (C) -- (F);
    \draw[edge] (D) -- (F);
    \draw[edge] (E) -- (F);

    \draw[thickedgeprogress, red] (A) -- ($(A)!3/5!(B)$);
    \draw[thickedgeprogress, red] (A) -- ($(A)!1/4!(D)$);
    \draw[thickedgeprogress, red] (A) -- ($(A)!1/2!(E)$);
    \draw[thickedgeprogress] (D) -- ($(D)!1/4!(A)$);
    \draw[thickedgeprogress] (D) -- ($(D)!3/5!(C)$);
    \draw[thickedgeprogress] (D) -- ($(D)!1/2!(F)$);
    \draw[thickedgeprogress, green!70!black] (E) -- ($(E)!1/2!(A)$);
    \draw[thickedgeprogress, green!70!black] (E) -- ($(E)!3/5!(B)$);
    \draw[thickedgeprogress, green!70!black] (E) -- ($(E)!1/4!(F)$);
    \draw[thickedgeprogress, orange] (F) -- ($(F)!3/5!(C)$);
    \draw[thickedgeprogress, orange] (F) -- ($(F)!1/2!(D)$);
    \draw[thickedgeprogress, orange] (F) -- ($(F)!1/4!(E)$);   
    
    \draw[activeset, red] (A) ellipse [radius = 3];
    \draw[activeset, blue] (D) ellipse [radius = 3];
    \draw[activeset, green!70!black] (E) ellipse [radius = 3];
    \draw[activeset, orange] (F) ellipse [radius = 3];

\end{tikzpicture}
}
\hfill
\resizebox{0.32\textwidth}{!}{
\begin{tikzpicture}[scale=0.65]
    \Huge
    \useasboundingbox (-6,-6) rectangle ++(24,18);
    \node[terminal, label=left:{$a$}] (A) at (  0,6) {};
    \node[steiner, label=above right:{\!$b$}] (B) at (  4,3) {};
    \node[steiner, label=above left:{$c$\!\!}] (C) at (  8,3) {};
    \node[terminal, label=right:{$d$}] (D) at ( 12,6) {};
    \node[terminal, label=left:{$e$}] (E) at (  0,0) {};
    \node[terminal, label=right:{$f$}] (F) at ( 12,0) {};
    \draw[edge] (A) -- (B);
    \draw[edge] (A) -- (D);
    \draw[edge] (A) -- (E);
    \draw[edge] (B) -- (C);
    \draw[edge] (B) -- (E);
    \draw[edge] (C) -- (D);
    \draw[edge] (C) -- (F);
    \draw[edge] (D) -- (F);
    \draw[edge] (E) -- (F);

    \draw[thickedgeprogress, red] (A) -- ($(A)!3/5!(B)$);
    \draw[thickedgeprogress, red] (A) -- ($(A)!1/4!(D)$);
    \draw[thickedgeprogress, red] (A) -- ($(A)!1/2!(E)$);
    \draw[thickedgeprogress] (D) -- ($(D)!1/4!(A)$);
    \draw[thickedgeprogress] (D) -- ($(D)!3/5!(C)$);
    \draw[thickedgeprogress] (D) -- ($(D)!1/2!(F)$);
    \draw[thickedgeprogress, green!70!black] (E) -- ($(E)!1/2!(A)$);
    \draw[thickedgeprogress, green!70!black] (E) -- ($(E)!3/5!(B)$);
    \draw[thickedgeprogress, green!70!black] (E) -- ($(E)!1/4!(F)$);
    \draw[thickedgeprogress, orange] (F) -- ($(F)!3/5!(C)$);
    \draw[thickedgeprogress, orange] (F) -- ($(F)!1/2!(D)$);
    \draw[thickedgeprogress, orange] (F) -- ($(F)!1/4!(E)$);

    \draw[thickedgeprogress, brown] (B) -- ($(A)!3/5!(B)$);
    \draw[thickedgeprogress, brown] (B) -- ($(E)!3/5!(B)$);
    \draw[thickedgeprogress, brown] ($(A)!5/12!(D)$) -- ($(A)!1/4!(D)$);
    \draw[thickedgeprogress, brown] ($(E)!5/12!(F)$) -- ($(E)!1/4!(F)$);

    \draw[thickedgeprogress, purple] (C) -- ($(D)!3/5!(C)$);
    \draw[thickedgeprogress, purple] (C) -- ($(F)!3/5!(C)$);
    \draw[thickedgeprogress, purple] ($(A)!7/12!(D)$) -- ($(A)!3/4!(D)$);
    \draw[thickedgeprogress, purple] ($(E)!7/12!(F)$) -- ($(E)!3/4!(F)$);

    \draw[activeset, red] (A) ellipse [radius = 3];
    \draw[activeset, blue] (D) ellipse [radius = 3];
    \draw[activeset, green!70!black] (E) ellipse [radius = 3];
    \draw[activeset, orange] (F) ellipse [radius = 3];

\draw[activeset, brown] ($(A)-(5,0)$) arc (180:0:5);
    \draw[activeset, brown] ($(E)-(5,0)$) arc (-180:0:5);
    \draw[activeset, brown] ($(E)-(5,0)$) -- ($(A)-(5,0)$);
    \draw[activeset, brown] ($(E)+(5,0)$) edge[out=90,in=-90] ($(B)-(0.4,0)$);
    \draw[activeset, brown] ($(B)-(0.4,0)$) edge[out=90,in=-90] ($(A)+(5,0)$);

    \draw[activeset, purple] ($(D)-(5,0)$) arc (180:0:5);
    \draw[activeset, purple] ($(F)-(5,0)$) arc (-180:0:5);
    \draw[activeset, purple] ($(F)+(5,0)$) -- ($(D)+(5,0)$);
    \draw[activeset, purple] ($(F)-(5,0)$) edge[out=90,in=-90] ($(C)+(0.4,0)$);
    \draw[activeset, purple] ($(C)+(0.4,0)$) edge[out=90,in=-90] ($(D)-(5,0)$);
    
\end{tikzpicture}
}
\hfill
\resizebox{0.32\textwidth}{!}{
\begin{tikzpicture}[scale=0.65]
    \Huge
    \useasboundingbox (-6,-6) rectangle ++(24,18);
    \node[terminal, label=left:{$a$}] (A) at (  0,6) {};
    \node[steiner, label=above right:{\!$b$}] (B) at (  4,3) {};
    \node[steiner, label=above left:{$c$\!\!}] (C) at (  8,3) {};
    \node[terminal, label=right:{$d$}] (D) at ( 12,6) {};
    \node[terminal, label=left:{$e$}] (E) at (  0,0) {};
    \node[terminal, label=right:{$f$}] (F) at ( 12,0) {};
    \draw[edge] (A) -- (B);
    \draw[edge] (A) -- (D);
    \draw[edge] (A) -- (E);
    \draw[edge] (B) -- (C);
    \draw[edge] (B) -- (E);
    \draw[edge] (C) -- (D);
    \draw[edge] (C) -- (F);
    \draw[edge] (D) -- (F);
    \draw[edge] (E) -- (F);

    \draw[thickedgeprogress, red] (A) -- ($(A)!3/5!(B)$);
    \draw[thickedgeprogress, red] (A) -- ($(A)!1/4!(D)$);
    \draw[thickedgeprogress, red] (A) -- ($(A)!1/2!(E)$);
    \draw[thickedgeprogress] (D) -- ($(D)!1/4!(A)$);
    \draw[thickedgeprogress] (D) -- ($(D)!3/5!(C)$);
    \draw[thickedgeprogress] (D) -- ($(D)!1/2!(F)$);
    \draw[thickedgeprogress, green!70!black] (E) -- ($(E)!1/2!(A)$);
    \draw[thickedgeprogress, green!70!black] (E) -- ($(E)!3/5!(B)$);
    \draw[thickedgeprogress, green!70!black] (E) -- ($(E)!1/4!(F)$);
    \draw[thickedgeprogress, orange] (F) -- ($(F)!3/5!(C)$);
    \draw[thickedgeprogress, orange] (F) -- ($(F)!1/2!(D)$);
    \draw[thickedgeprogress, orange] (F) -- ($(F)!1/4!(E)$);

    \draw[thickedgeprogress, brown] (B) -- ($(A)!3/5!(B)$);
    \draw[thickedgeprogress, brown] (B) -- ($(E)!3/5!(B)$);
    \draw[thickedgeprogress, brown] ($(A)!1/2!(D)$) -- ($(A)!1/4!(D)$);
    \draw[thickedgeprogress, brown] ($(E)!1/2!(F)$) -- ($(E)!1/4!(F)$);
    \draw[thickedgeprogress, brown] ($(B)!1/4+1/12!(C)$) -- (B);

    \draw[thickedgeprogress, purple] (C) -- ($(D)!3/5!(C)$);
    \draw[thickedgeprogress, purple] (C) -- ($(F)!3/5!(C)$);
    \draw[thickedgeprogress, purple] ($(A)!1/2!(D)$) -- ($(A)!3/4!(D)$);
    \draw[thickedgeprogress, purple] ($(E)!1/2!(F)$) -- ($(E)!3/4!(F)$);
    \draw[thickedgeprogress, purple] ($(B)!3/4-1/12!(C)$) -- (C);

    \draw[activeset, red] (A) ellipse [radius = 3];
    \draw[activeset, blue] (D) ellipse [radius = 3];
    \draw[activeset, green!70!black] (E) ellipse [radius = 3];
    \draw[activeset, orange] (F) ellipse [radius = 3];

\draw[activeset, brown] ($(A)-(5,0)$) arc (180:0:5);
    \draw[activeset, brown] ($(E)-(5,0)$) arc (-180:0:5);
    \draw[activeset, brown] ($(E)-(5,0)$) -- ($(A)-(5,0)$);
    \draw[activeset, brown] ($(E)+(5,0)$) edge[out=90,in=-90] ($(B)-(0.4,0)$);
    \draw[activeset, brown] ($(B)-(0.4,0)$) edge[out=90,in=-90] ($(A)+(5,0)$);

    \draw[activeset, purple] ($(D)-(5,0)$) arc (180:0:5);
    \draw[activeset, purple] ($(F)-(5,0)$) arc (-180:0:5);
    \draw[activeset, purple] ($(F)+(5,0)$) -- ($(D)+(5,0)$);
    \draw[activeset, purple] ($(F)-(5,0)$) edge[out=90,in=-90] ($(C)+(0.4,0)$);
    \draw[activeset, purple] ($(C)+(0.4,0)$) edge[out=90,in=-90] ($(D)-(5,0)$);

\draw[activeset, brown] ($(A)-(6,0)$) arc (180:0:6);
    \draw[activeset, brown] ($(E)-(6,0)$) arc (-180:0:6);
    \draw[activeset, brown] ($(E)-(6,0)$) -- ($(A)-(6,0)$);
    \draw[activeset, brown] ($(E)+(6,0)$) edge[out=90,in=-90] ($(B)!1/4+1/12!(C)$);
    \draw[activeset, brown] ($(B)!1/4+1/12!(C)$) edge[out=90,in=-90] ($(A)+(6,0)$);

    \draw[activeset, purple] ($(D)-(6,0)$) arc (180:0:6);
    \draw[activeset, purple] ($(F)-(6,0)$) arc (-180:0:6);
    \draw[activeset, purple] ($(F)+(6,0)$) -- ($(D)+(6,0)$);
    \draw[activeset, purple] ($(F)-(6,0)$) edge[out=90,in=-90] ($(C)!1/4+1/12!(B)$);
    \draw[activeset, purple] ($(C)!1/4+1/12!(B)$) edge[out=90,in=-90] ($(D)-(6,0)$);
    
\end{tikzpicture}
}
     \caption{\label{fig:moat_growing}
    An example of the moat growing process for the undirected LP.
    The instance has four terminals $a,d,e,f$ (shown as squares), and initially, we grow a set around each of these terminals (left picture). 
    Once the sets $U_a$ (red) and $U_e$ (green) meet, we stop growing their corresponding dual variables and instead grow the dual variable $U_S$ (brown) for the set $S=\{a,e\}$. 
    Once, the vertex $b$ is reachable from $S=\{a,e\}$ by tight edges, it is included in the set $U_S$.
    The same applies symmetrically for the sets $U_d$ (blue) and $U_f$ (orange) and the vertex $c$.
    The colors on the edges indicate the contributions of the sets $S\subseteq R$ to the tightness of the edges. 
    An edge is tight once it is fully colored.
    }
\end{figure}
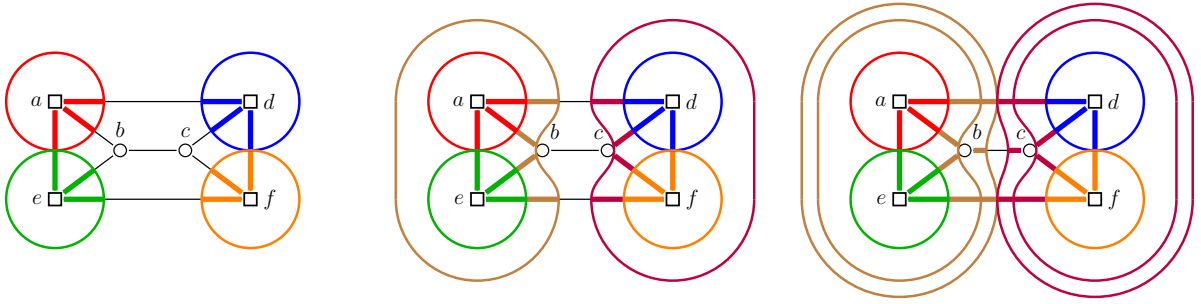

In a bidirected graph, however, it is not obvious at all when two sets should be merged into a single one.
The decision of when the sets ``growing around the terminals'' are merged is captured by the concept of a \emph{merge plan}.
For every time $t$ during the dual growth algorithm, a merge plan has a partition $\cS^t$ of the terminal set $R$, where the sets growing around terminals are already merged at time $t$ if and only if the terminals are part of the same part of the partition $\cS^t$.

\begin{definition}[Merge Plan]\label{def: merge plan}
    A \emph{merge plan} is a family $\cM=\lr{\cS^t}_{t\geq0}$ of partitions of $R$ such that
    for every two terminals $x, y\in R$ there exists a time $\merge{\cM}{x, y} \geq 0$ such that $x$ and $y$ are in different parts of the partition $\cS^t$ if and only if $0\leq t\leq\merge{\cM}{x, y}$.
    We call $\merge{\cM}{x, y}$ the \emph{merge time} of $x$ and $y$.
\end{definition}
 Note that the definition of a merge plan implies that for any times $0 \leq t_1<t_2$, the partition $\cS^{t_1}$ is a refinement of $\cS^{t_2}$.
 In particular, the partition changes only at a finite number of times.
 Moreover, there is some time $t \geq 0$ such that $|\cS^t| = 1 $, that is $\cS^t = \{R\}$.

The dual growth algorithm with a given merge plan is described in \Cref{algo:dual_growth_with_merge_plan}.
We provide a generic version that could be applied to both directed and undirected graphs. 
To obtain a dual solution to \ref{eq:bcr}, we apply the algorithm to the bidirected graph.
To obtain a dual solution to the natural undirected LP, one can apply it to the undirected graph.

\begin{algorithm}[ht]
\caption{Moat Growing with Merge Plan \label{algo:dual_growth_with_merge_plan}}
\KwIn{a (directed or undirected) graph $G=(V,E)$ with edge costs $c\colon E \to \mathbb{R}_{\geq 0}$,\newline a terminal set $R\subseteq V$, and a merge plan $\cM = \lr{\cS^t}_{t\geq0}$}
\KwOut{a solution $y$ to the dual linear program}
\vspace*{2mm}
Initialize $t= 0$ and $y_U=0$ for all $U\subseteq V$.\;
\While{ $\lrv{\cS^t}>1$}{
For $S\in \cS^t$, define
$U_S \ =\ \{v\in V : v\text{ is reachable from the set $S$ by a path of tight edges}\}$.\;
Increase the time $t$ and the dual variables $y_{U_S}$ for $S\in \cS^t$ at unit rate until
\newline the partition $\cS^t$ changes or a new edge becomes tight.
}
\Return{y}
\end{algorithm}

We will choose the merge plan $\cM$ so that we ensure the following.
Whenever we grow a dual variable corresponding to a vertex set $U_S$, then $U_S$ contains the root $r$ only if $r\in S$.
This guarantees that all dual variables we grow, except for one dual variable at a time, contribute to the objective value of \eqref{eq:dual_bcr}.
Then the value of the dual solution resulting from \Cref{algo:dual_growth_with_merge_plan} is equal to the \emph{value of the merge plan} $\cM$, which we define next.

\begin{definition}
    The value of a merge plan $\cM=\lr{\cS^t}_{t\geq0}$ is defined as
    $$
        \val{\cM}=\int_0^{\infty}\lrv{\cS^t}-1\,dt\,.
    $$
\end{definition}
Note that the value of a merge plan is finite because for every merge plan $\cM=\lr{\cS^t}_{t\geq0}$, there is some time $t^*$ such that $|\cS^t|=1$ for all $t> t^*$.

Both, the standard moat growing algorithm in undirected graphs and the construction of dual solutions from \cite{BCRlessthan2}, correspond to running \Cref{algo:dual_growth_with_merge_plan} with a particular choice of the merge plan $\cM$.
The dual growth in the classical primal-dual algorithm in undirected graphs corresponds to \Cref{algo:dual_growth_with_merge_plan} with the \emph{canonical merge plan}.

\begin{definition}[canonical merge plan]\label{def: canonical merge plan}
The \emph{canonical merge plan} for an instance $(G,c,R)$ of the Steiner Tree problem is the merge plan $\cM$ with largest value that satisfies 
\begin{align}\label{eq:def canonical merge plan}
    \merge{\cM}{x,y} \ \leq\ \frac{1}{2}\, \dist{x,y} \quad \forall x,y\in R,
\end{align}
where $\dist{x,y}$ is the length of a shortest $x$-$y$ path in $(G,c)$.
\end{definition}

The canonical merge plan has value $\frac{1}{2}\,\tmst$.
For a proof of this fact and a more detailed discussion of the canonical merge plan and its properties, we refer to \Cref{sec:merge_plan}.
The construction of dual solutions in \cite{BCRlessthan2} for MST-optimal instances essentially\footnote{
One minor difference to the construction of \cite{BCRlessthan2} is that their algorithm does not grow any dual variable corresponding to the set $S\in \cS^t$ containing the root. However, this difference is crucial in neither their nor our analysis.
Moreover, \cite{BCRlessthan2} first reduce to a particular class of instances (called $\gamma$-locally MST-optimal instances) and then apply this construction of dual solutions only to these instances.
}
corresponds to \Cref{algo:dual_growth_with_merge_plan} applied to the bidirected graph with a ``scaled'' version of the canonical merge plan: choose a suitable $\delta > 0$ and let $\cM$ be the merge plan with largest value satisfying
\begin{equation}\label{eq:scaled_merge_plan}
  \merge{\cM}{x,y} \ \leq\ \frac{1+\delta}{2}\, \dist{x,y} \quad \forall x,y\in R.  
\end{equation}

\subsection{Choosing the Merge Plan}
\label{sec:outline_merge_plan}

Next, we discuss which properties a good merge plan should satisfy.
On the one hand, we want the merge plan to have a large value. 
On the other hand, we want to ensure that the value of the dual solution resulting from \Cref{algo:dual_growth_with_merge_plan} is equal to the value of our merge plan, which is captured by the definition of a \emph{feasible} merge plan.

\begin{definition}[feasible merge plan]
A merge plan $\cM$ is \emph{feasible for an instance $\cI$} if \Cref{algo:dual_growth_with_merge_plan} applied to the bidirected graph and the merge plan $\cM$ satisfies:
whenever we grow a dual variable corresponding to a set $U_S$, 
the set $U_S$ contains the root $r$ only if $r\in S$.
\end{definition}

\begin{observation}\label{obs:feasible_implies_value}
If $\cM$ is a feasible merge plan for an instance $\cI$, then \Cref{algo:dual_growth_with_merge_plan} yields a solution to \ref{eq:dual_bcr} for $\cI$ whose value is $\val{\cM}$.
\end{observation}

For MST-optimal instances, we will work with a scaled version of the canonical merge plan.
More precisely, we choose $\cM$ to be the merge plan with largest value that satisfies
\begin{equation}\label{eq:merge_plan_mst_optimal}
\merge{\cM}{x,y} \ \leq\ \frac{7}{12}\, \dist{x,y} \quad \forall x,y\in R.
\end{equation}
This corresponds to the construction from \cite{BCRlessthan2} with $\delta=\frac{1}{6}$ and leads to a dual solution to \ref{eq:bcr} with a value of $\val{\cM} = \frac{7}{12}\,\tmst = \frac{7}{12} \, \opt$.

Naturally, if the cost of a cheapest Steiner tree is much cheaper than $\tmst$---in general, $\opt$ could be as small as $\frac{1}{2}\, \tmst$---then we cannot expect to get dual solutions of value $\frac{7}{12} \, \tmst$.
To address this, \cite{BCRlessthan2} provide a reduction to a certain class of instances.
Then they prove that for some sufficiently small $\delta > 0$ the merge plan with largest value satisfying \eqref{eq:scaled_merge_plan}, is a feasible merge plan for this restricted class of instances.

To explain what this class of instances is, we need the notion of \emph{components}.
\begin{definition}[components, full components]
   For a set $X\subseteq V$, we say that $K\subseteq E$ is a \emph{component} connecting $X$, if all vertices of $X$ belong to the same connected component of $(V,K)$.
   We say that $K\subseteq E$ is a \emph{full component} if it is the edge set of a tree whose leaves are terminals and whose other vertices are Steiner vertices.
\end{definition}
Observe that whenever there exists a set $X\subseteq R$ and a component $K$ connecting $X$, then we obtain a Steiner tree by taking the union of the component $K$ and a terminal MST in the graph $G/X$ arising from contraction of $X$ (and possibly discarding some edges).
This Steiner tree has cost (at most) $c(K) + \tmst_{G/X}$.
If 
\[
c(K) \ <\ \dropG{X} \ \coloneqq\ \tmst_G - \tmst_{G/X},
\]
then this Steiner tree is cheaper than a terminal MST in the original graph $G$.
In this case we say that $K$ is an \emph{improving component}.
If $c(K) < (1-\gamma) \cdot \dropG{X}$, we say that the component $K$ is $\gamma$-improving. 
An instance is MST-optimal if and only if no improving component exists.

\citeauthor{BCRlessthan2} \cite{BCRlessthan2} choose a sufficiently small number $\gamma > 0$ and show that the scaled version of the canonical merge plan is feasible for every instance that does not have any $\gamma$-improving component.
Moreover, they provide a reduction to instances without $\gamma$-improving component. 
To this end they repeatedly contract $\gamma$-improving components, as long as they exist. 
They then observe that either only a small part of the instance gets contracted---in which case it suffices to prove a good lower bound on $\bcr$ in the contracted instance---or $\opt$ must be significantly smaller than $\tmst$, which immediately implies a better upper bound on the integrality gap than $2$ because $\bcr \geq \frac{1}{2} \tmst$.

We strengthen the analysis from \cite{BCRlessthan2} in two ways.
First, we provide a significantly improved analysis for instances on which no $\gamma$-improving component exists. 
That is, we prove that for such instances there is a feasible merge plan of value $\frac{7-5\gamma}{12} \cdot \tmst$.
Moreover, we extend the analysis beyond the case where no $\gamma$-improving component exists.

Next, we explain the condition that we prove to be sufficient for a merge plan being feasible.
To see that our new condition is indeed less restrictive than the assumption from \cite{BCRlessthan2}, it will be useful to express their assumption in a different way.
Namely, we will express $\dropG{X}$ in terms of the canonical merge plan.
To this end we introduce the \emph{local value} of a merge plan. 
We remark that for $X=R$, the local value of $\cM$ for $X$ is identical to the value of $\cM$.
\begin{definition}[local value of a merge plan]\label{def: local value of merge plan}
    Given a merge plan $\cM=\lr{\cS^t}_{t\geq0}$ and a\\ nonempty set of terminals $X\subseteq R$, we define the local value of $\cM$ for $X$ as
    $$
        \val{\cM, X}\coloneqq\int_0^{\infty}\lrv{\lrb{S\in\cS^t\,:\,S\cap X\neq\emptyset}}-1\ dt\,.
    $$
\end{definition}
 If $\cM$ is the canonical merge plan, then we have $\dropG{X}=2\cdot\val{\cM,X}$ for all nonempty sets $X\subseteq R$, as we will show in \Cref{sec:merge_plan} (\Cref{lem: drop in terms of canonical merge plan}).
Hence, if $\cM_{\delta}$ is the $(1+\delta)$-scaled version of the canonical merge plan, i.e., $\cM_{\delta}$ is the merge plan with largest value satisfying \eqref{eq:scaled_merge_plan}, we have
\begin{equation}\label{eq:equaivalent_drop_outline}
\dropG{X}=\frac{2}{1+\delta}\cdot \val{\cM_{\delta},X}
\end{equation}
for all nonempty sets $X\subseteq R$.
We write $\cost(X)$ to denote the cost of a cheapest component $K$ connecting a set $X\subseteq R$.
Then a $\gamma$-improving component $K$ connecting a set $X\subseteq R$ exists if and only if
\[
\cost(X) \ <\ (1-\gamma) \cdot \frac{2}{1+\delta}\cdot \val{\cM_{\delta},X}.
\]
When choosing $\delta = \frac{1}{6} - \frac56\gamma$ (which is different from the choice of $\delta$ in \cite{BCRlessthan2}), this leads to the definition of \emph{$\gamma$-cheap sets}.

\begin{definition}[$\gamma$-cheap]\label{def: gamma cheap}
    Given $\gamma\in\lre{0,1}$ and a merge plan $\cM$, we call a terminal set $X\subseteq R$ $\gamma$-cheap with respect to $\cM$, if 
    $$
        \cost(X) <\lr{1-\gamma}\cdot \frac{12}{7-5\gamma} \cdot \val{\cM, X} .
    $$
    Otherwise we say that $X$ is $\gamma$-expensive.
\end{definition}

We will essentially prove that if for some $\gamma \in [0,\frac15]$ and a merge plan $\cM$, no $\gamma$-cheap set $X\subseteq R$ exists, then the merge plan $\cM_{\delta}$ is feasible (for $\delta = \frac{1}{6} - \frac56\gamma$).
This is already a strengthening of~\cite{BCRlessthan2}, where the key difference is that our analysis allows us to choose the value of $\delta$ larger than in \cite{BCRlessthan2}.

Nevertheless, to obtain the claimed upper bound of $\finalratio$ on the integrality gap of \ref{eq:bcr}, we prove an even stronger statement.
The key insight here is that we can replace the assumption that no $\gamma$-cheap set exists by a less restrictive condition: whenever there is a $\gamma$-cheap set $X\subseteq R$, there must be two elements of $X$ that are merged early.

\begin{definition}[$\gamma$-good]\label{def: gamma good}
Let $\gamma \in [0,\frac15]$.
    We say that a merge plan $\cM$ is $\gamma$-good, if for all $x,y\in R$
    \begin{align}
        \merge{\cM}{x, y}\leq\frac{7-5\gamma}{12}\,\dist{x, y},\label{ieq: general upper bound for gamma good}
    \end{align}
    and for every terminal set $X\subseteq R$ that is $\gamma$-cheap with respect to $\cM$, there are $x,y\in X$ with $x\neq y$ satisfying
    \begin{align}
        \merge{\cM}{x, y}\leq\frac{7-5\gamma}{18-14\gamma}\,\dist{x, y}\,.\label{ieq: stronger upper bound for gamma good}
    \end{align}
\end{definition}

Up to some minor technical details, which we discuss in \Cref{sec: assumptions}, we prove that any $\gamma$-good merge plan is feasible. For an outline of this proof, we refer to \Cref{sec:outline_analysis}. 
By \Cref{obs:feasible_implies_value}, this immediately implies that any $\gamma$-good merge plan yields lower bounds on the value of \eqref{eq:dual_bcr}, and hence on the value of \ref{eq:bcr}:

\begin{theorem}\label{cor: gamma good merge plan gives large dual}
    Let $\gamma\in [0,\frac15]$ and let $\cM$ be a $\gamma$-good merge plan. Then
    $$
        \bcr\geq\val{\cM}\,.
    $$
\end{theorem}

For MST-optimal instances, the merge plan with maximum value satisfying \eqref{eq:merge_plan_mst_optimal}, is a $0$-good merge plan with value $\frac{7}{12}\,\tmst$ (see \Cref{lem: merge plan mst optimal}).
Thus, \Cref{cor: gamma good merge plan gives large dual} implies \Cref{thm:gap_mst_optimal}.

For general instances, we show in \Cref{sec:existence_merge_plan} that $\gamma$-good merge plans with large value exist:

\begin{restatable}{theorem}{thmExistanceOfMergePlan}\label{thm: there is a gamma good merge plan with high value}
    For every Steiner tree instance and every $\gamma \in (0,\frac15]$, there exists a $\gamma$-good merge plan $\cM$ with
    $$
        \val{\cM}\ \geq\ \frac{7-5\gamma}{12}\lr{\tmst - \frac{\tmst - \opt}{\gamma} \cdot \frac{3-7\gamma}{9-7\gamma}}.
    $$
\end{restatable}

Observe that for MST-optimal instances we have $\tmst -\opt =0$ and thus \Cref{thm: there is a gamma good merge plan with high value} for arbitrarily small $\gamma >0$ yields a merge plan with value arbitrarily close to $\frac{7}{12} \, \tmst = \frac{7}{12} \, \opt$. 
In fact, in this case we can simply choose $\cM$ to be a suitably scaled version of the canonical merge plan and use that the value of the canonical merge plan is $\frac{1}{2}\,\tmst$ (\Cref{cor: canonical merge plan}).

For general instances, we observe that for $\gamma < \frac{1}{5}$, we have $\frac{7-5\gamma}{12} > \frac{1}{2}$. 
If $\tmst - \opt$ is sufficiently small, we then obtain a stronger lower bound on $\bcr$ than $\frac{1}{2}\, \tmst \geq \frac{1}{2} \, \opt$.
Otherwise, $\tmst$ is much larger than $\opt$ and thus $\bcr \geq \frac{1}{2}\, \tmst$ already implies a better upper bound on the integrality gap of \ref{eq:bcr} than $2$.
Choosing the value of $\gamma$ optimally, leads to an improved upper bound on the integrality gap of \ref{eq:bcr} compared to \cite{BCRlessthan2}.
To obtain the claimed upper bound of $\finalratio$ on the integrality gap, we use a stronger version of \Cref{thm: there is a gamma good merge plan with high value} and choose the value of $\gamma$ depending on the Steiner tree instance.
For details we refer to \Cref{sec:existence_merge_plan}.

\subsection{Analyzing the Dual Growth Algorithm}
\label{sec:outline_analysis}

In this section, we give an informal outline of our analysis of the dual growth procedure.
We will first focus on the MST-optimal case with a particular merge plan, and only at the end of this section we will comment briefly on the extension to general instances and general $\gamma$-good merge plans.

Consider an MST-optimal instance, let $\gamma \coloneqq 0$, and let $\cM$ be the merge plan with maximum value satisfying \eqref{eq:merge_plan_mst_optimal}.
The value of this merge plan is $\frac{7}{6}$ times the value of the canonical merge plan, and thus $\frac{7}{12} \, \tmst$ (see \Cref{cor: canonical merge plan}).
Moreover, because the merge plan $\cM$ corresponds to the canonical merge plan scaled by a factor $1+\delta = \frac{7}{6}$, by \eqref{eq:equaivalent_drop_outline} we have
\[
\dropG{X} \ =\ \frac{12}{7} \cdot \val{\cM,X}
\]
for all nonempty $X\subseteq R$ (see \Cref{lem: drop in terms of canonical merge plan}).
Thus, the MST-optimality of the instance implies that there is no $\gamma$-cheap component for $\gamma = 0$.

Recall that our goal is to show that if a set $S\in \cS^t$ (for some time $t\geq 0$) does not contain the root $r$, then also the set $U_S$ of vertices reachable from $S$ by tight edges does not contain the root.
By the choice of the merge plan, the vertices in $S$ have a large distance from $r$, namely at least $\frac{12}{7}\, t$ (otherwise, the vertices from $S$ would have been already merged with $r$ at time $t$).
A natural strategy, already employed in \cite{BCRlessthan2}, is to prove that the set $U_S$ cannot contain vertices that are too far away from $S$.

\paragraph{Examples: Interactions of at most two Terminals.} Let us start by considering some simple examples.
Consider a set $U_S$ growing around a single terminal, i.e. $S=\{s\}$ for some terminal $s$.
If there is no interaction with sets growing around other terminals, that is, if none of the vertices reachable from $S$ is reachable from other terminals (yet), then we have 
\begin{equation}\label{eq:no_interaction_case}
   U_S \ = \ \{ v\in V : \dist{s,v} \leq t\}. 
\end{equation}
In other words, at time $t$, the set $S$ can only reach vertices at distance at most $t$.
\begin{figure}
    \centering
    \begin{tikzpicture}[scale=0.8]

    \node[terminal, label=below:{$s_2$}] (A) at ( 10,0) {};
    \node[steiner, label=below:{$m$}]  (B) at ( 4,0) {};
    \node[terminal, label=below:{$s_1$}] (C) at ( -2,0) {};
    \node[steiner, label=left:{$v$}]  (D) at ( 4,3) {};

    \draw[diredge] (A) -- (B);
    \draw[diredge] (B) --node[below=6pt]{$\frac{6}{7}t$} (A);
    \draw[diredge] (B) -- (C);
    \draw[diredge] (C) --node[below=6pt]{$\frac{6}{7}t$} (B);
    \draw[diredge] (D) -- (B);

    \draw[edgedirprogress] (A) -- (B);
    \draw[edgedirprogress, green!70!black] (B) -- ($(B)!1/3.5!(A)$);
    \draw[edgedirprogress, green!70!black] (C) -- (B);
    \draw[edgedirprogress] (B) -- ($(B)!1/3.5!(C)$);
    
    \draw[edgedirprogresshalfA] (B) --node[black, right=6pt]{$\frac{2}{7}t$} (D);
    \draw[edgedirprogresshalfB,green!70!black] (B) -- (D);
    
\end{tikzpicture}     \caption{\label{fig:simple_interaction_example}
    An example of (a part of) an instance with two terminals $s_1$ and $s_2$, indicated as squares.
    The green and blue color indicate the contributions of the sets growings around $s_1$ (green) and $s_2$ (blue) to the tightness of the edges. To the edge $(m,v)$ both sets contribute equally.
    The edges between $s_i$ and $m$ have length $\frac{6}{7} \, t$ each, and the two terminals can reach the vertex $m$ at time $\frac{6}{7} \, t$. 
    However, they are merged only at time $t = \frac{7}{12} \cdot \dist{s_1,s_2}$.
    By this time they have reached the vertex $v$, which is at distance $\frac{2}{7} \, t$ from $m$.
    Thus, the edge from $m$ to $v$ effectively gets tight at speed~$2$. 
    }
\end{figure}
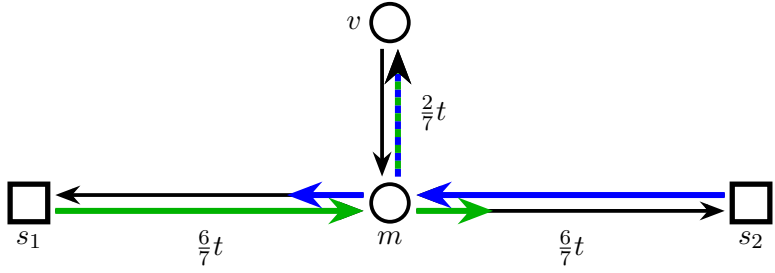
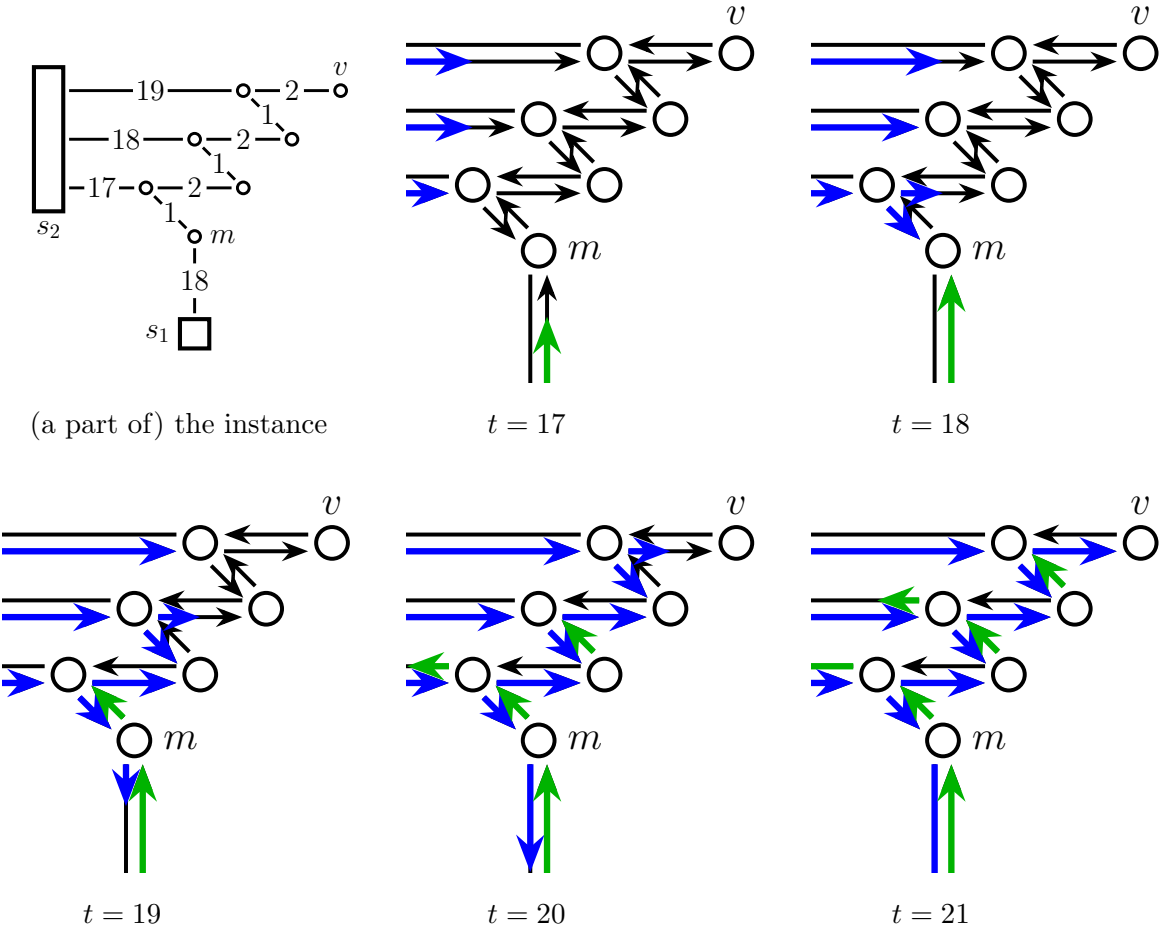
\begin{figure}[p]
    \centering

\newcommand{\drawAllVerticesSmallJumpII}{
\Large 
    \node[smallsteiner] (A1) at ( 0, 0) {};
    \node[smallsteiner] (A2) at ( 1, 1) {};
    \node[smallsteiner] (A3) at ( 2, 2) {};
    \node[smallsteiner, label=right:{$m$}] (B1) at ( 1,-1) {};
    \node[smallsteiner] (B2) at ( 2, 0) {};
    \node[smallsteiner] (B3) at ( 3, 1) {};
    \node[smallsteiner, label=above :{$v$}] (B4) at ( 4, 2) {};
    \node[terminal, transparent] (C1) at (-6, 0) {};
    \node[terminal, transparent] (C2) at (-6, 1) {};
    \node[terminal, transparent] (C3) at (-6, 2) {};
    \node[terminal, transparent] (C4) at (-6, 3) {};
    \node[terminal,minimum height=70pt, label=below:{}] (C) at (-6,1.5) {\LARGE$s_2$};
    \node[terminal, label=left:{\LARGE$s_1$}] (D) at (1,-8) {};
}

\newcommand{\drawAllEdgesSmallJumpII}{
    \draw[diredge] (C1) -- (A1);
    \draw[diredge] (C2) -- (A2);
    \draw[diredge] (C3) -- (A3);
    \draw[diredge] (A1) -- (C1);
    \draw[diredge] (A2) -- (C2);
    \draw[diredge] (A3) -- (C3);
    \draw[diredge] (A1) -- (B1);
    \draw[diredge] (A1) -- (B2);
    \draw[diredge] (A2) -- (B2);
    \draw[diredge] (A2) -- (B3);
    \draw[diredge] (A3) -- (B3);
    \draw[diredge] (A3) -- (B4);
    \draw[diredge] (B1) -- (A1);
    \draw[diredge] (B2) -- (A1);
    \draw[diredge] (B2) -- (A2);
    \draw[diredge] (B3) -- (A2);
    \draw[diredge] (B3) -- (A3);
    \draw[diredge] (B4) -- (A3);
    \draw[diredge] (D) -- (B1);
    \draw[diredge] (B1) -- (D);
}

\newcommand{\drawAllVerticesAndEdgesSmallJumpII}{
    \drawAllVerticesSmallJumpII
    \drawAllEdgesSmallJumpII
}

\newcommand{\WeirdBehaviourSmallJumpsIIA}{
\begin{tikzpicture}[scale=0.83]
\Large 
 \clip (-3,-4) rectangle (4.5,3.5);
    \node[tinysteiner] (A1) at ( 0, 0) {};
    \node[tinysteiner] (A2) at ( 1, 1) {};
    \node[tinysteiner] (A3) at ( 2, 2) {};
    \node[tinysteiner, label=right:{$m$}] (B1) at ( 1,-1) {};
    \node[tinysteiner] (B2) at ( 2, 0) {};
    \node[tinysteiner] (B3) at ( 3, 1) {};
    \node[tinysteiner, label=above:{$v$}] (B4) at ( 4, 2) {};
    \node[terminal, transparent] (C1) at (-2, 0) {};
    \node[terminal, transparent] (C2) at (-2, 1) {};
    \node[terminal, transparent] (C3) at (-2, 2) {};
    \node[terminal, transparent] (C4) at (-2, 3) {};
    \node[terminal,minimum height=70pt, label=below:{$s_2$}] (C) at (-2,1) {};
    \node[terminal, label=left:{$s_1$}] (D) at (1,-3) {};
    
    \draw[edge] (C1) -- node[edgelabel] {$17$} (A1);
    \draw[edge] (C2) -- node[edgelabel] {$18$} (A2);
    \draw[edge] (C3) -- node[edgelabel] {$19$} (A3);
    \draw[edge] (A1) -- node[edgelabel] {$1$} (B1);
    \draw[edge] (A1) -- node[edgelabel] {$2$} (B2);
    \draw[edge] (A2) -- node[edgelabel] {$1$} (B2);
    \draw[edge] (A2) -- node[edgelabel] {$2$} (B3);
    \draw[edge] (A3) -- node[edgelabel] {$1$} (B3);
    \draw[edge] (A3) -- node[edgelabel] {$2$} (B4);
    \draw[edge] (B1) -- node[edgelabel] {$18$} (D);
\end{tikzpicture}}

\newcommand{\WeirdBehaviourSmallJumpsIIB}{
\begin{tikzpicture}[scale=0.83]
\clip (-1,-3) rectangle (4.5,3.5);
    \drawAllVerticesAndEdgesSmallJumpII

    \draw[edgedirprogress] (C1) -- (A1);
    \draw[edgedirprogress] (C2) -- (A1|-A2);
    \draw[edgedirprogress] (C3) -- (A1|-A3);
    \draw[greendirprogress] (D) -- ($(D)!6/7!(B1)$);
\end{tikzpicture}}

\newcommand{\WeirdBehaviourSmallJumpsIIC}{
\begin{tikzpicture}[scale=0.83]
    \clip (-1,-3) rectangle (4.5,3.5);
    \drawAllVerticesAndEdgesSmallJumpII

    \draw[edgedirprogress] (C1) -- (A1);
    \draw[edgedirprogress] (C2) -- (A2);
    \draw[edgedirprogress] (C3) -- (A2|-A3);
    \draw[edgedirprogress] (A1) -- (B1);
    \draw[edgedirprogress] (A1) -- (B1|-B2);
    \draw[greendirprogress] (D) -- (B1);
\end{tikzpicture}}

\newcommand{\WeirdBehaviourSmallJumpsIID}{
\begin{tikzpicture}[scale=0.83]
    \clip (-1,-3) rectangle (4.5,3.5);
    \drawAllVerticesAndEdgesSmallJumpII

    \draw[edgedirprogress] (C1) -- (A1);
    \draw[edgedirprogress] (C2) -- (A2);
    \draw[edgedirprogress] (C3) -- (A3);
    \draw[edgedirprogress] (A1) -- (B1);
    \draw[edgedirprogress] (A1) -- (B2);
    \draw[edgedirprogress] (A2) -- (B2);
    \draw[edgedirprogress] (A2) -- (B2|-B3);
    \draw[edgedirprogress] (B1) -- ($(B1)!1/7!(D)$);
    \draw[greendirprogress] (D) -- (B1);
    \draw[greendirprogress] (B1) -- (A1);
\end{tikzpicture}}

\newcommand{\WeirdBehaviourSmallJumpsIIE}{
\begin{tikzpicture}[scale=0.83]
    \clip (-1,-3) rectangle (4.5,3.5);
    \drawAllVerticesAndEdgesSmallJumpII

    \draw[edgedirprogress] (C1) -- (A1);
    \draw[edgedirprogress] (C2) -- (A2);
    \draw[edgedirprogress] (C3) -- (A3);
    \draw[edgedirprogress] (A1) -- (B1);
    \draw[edgedirprogress] (A1) -- (B2);
    \draw[edgedirprogress] (A2) -- (B2);
    \draw[edgedirprogress] (A2) -- (B3);
    \draw[edgedirprogress] (A3) -- (B3);
    \draw[edgedirprogress] (A3) -- (B3|-B4);
    \draw[edgedirprogress] (B1) -- ($(B1)!2/7!(D)$);

    \draw[greendirprogress] (D) -- (B1);
    \draw[greendirprogress] (B1) -- (A1);
    \draw[greendirprogress] (B2) -- (A2);
    \draw[greendirprogress] (A1) -- ($(A1)!1/6!(C1)$);
\end{tikzpicture}}

\newcommand{\WeirdBehaviourSmallJumpsIIF}{
\begin{tikzpicture}[scale=0.83]
    \clip (-1,-3) rectangle (4.5,3.5);
    \drawAllVerticesAndEdgesSmallJumpII

    \draw[edgedirprogress] (C1) -- (A1);
    \draw[edgedirprogress] (C2) -- (A2);
    \draw[edgedirprogress] (C3) -- (A3);
    \draw[edgedirprogress] (A1) -- (B1);
    \draw[edgedirprogress] (A1) -- (B2);
    \draw[edgedirprogress] (A2) -- (B2);
    \draw[edgedirprogress] (A2) -- (B3);
    \draw[edgedirprogress] (A3) -- (B3);
    \draw[edgedirprogress] (A3) -- (B4);
    \draw[edgedirprogress] (B1) -- ($(B1)!3/7!(D)$);

    \draw[greendirprogress] (D) -- (B1);
    \draw[greendirprogress] (B1) -- (A1);
    \draw[greendirprogress] (B2) -- (A2);
    \draw[greendirprogress] (B3) -- (A3);
    \draw[greendirprogress] (A1) -- ($(A1)!2/6!(C1)$);
    \draw[greendirprogress] (A2) -- ($(A2)!1/7!(C2)$);
\end{tikzpicture}}

\begin{tabular}{ccc}
  \resizebox{0.3\textwidth}{!}{\WeirdBehaviourSmallJumpsIIA} 
  & \resizebox{0.3\textwidth}{!}{\WeirdBehaviourSmallJumpsIIB} 
  & \resizebox{0.3\textwidth}{!}{\WeirdBehaviourSmallJumpsIIC} 
  \\[2mm]
   (a part of) the instance  & $t=17$\hspace*{15mm} & $t= 18$\hspace*{15mm} \\
   \resizebox{0.3\textwidth}{!}{\WeirdBehaviourSmallJumpsIID} 
  & \resizebox{0.3\textwidth}{!}{\WeirdBehaviourSmallJumpsIIE} 
  & \resizebox{0.3\textwidth}{!}{\WeirdBehaviourSmallJumpsIIF} 
  \\[2mm]
  $t=19$\hspace*{15mm}  & $t=20$\hspace*{15mm} & $t= 21$\hspace*{15mm} 
\end{tabular}
     \caption{\label{fig:zig_zag_example}
    An example of two terminals interacting. The top left picture shows (a part of) an instance with two terminals $s_1$ and $s_2$, where the numbers indicate the cost/length of the edges.
    The other pictures indicate the contributions of the sets growing around $s_1$ (green) and $s_2$ (blue) to the tightness of the edges at different times $t$.
    The terminals $s_1$ and $s_2$ have distance $36$ and have thus a merge time of (at most) $21$.
    Observe that at time $t=18$ the two sets meet at the vertex $m$, which is at distance $18$ from $s_1$ (and from $s_2$).
    At time $t=21$, the vertex $v$ is reachable from $s_1$ by a path of tight edges. 
    The vertex $v$ has distance $27$ from $s_1$ and thus between the meeting time $t=18$ until the merge time $t=21$, the maximum distance of vertices reachable from $s_1$ grows effectively at speed $3$.
    (This maximum distance is not continuous and there are ``jumps'', but on average between the meeting time and the merge time it increases at speed $3$.)
    }
\end{figure}
If two sets growing around two different terminals interact, the situation can get more complicated.
\Cref{fig:simple_interaction_example} shows that the sets can reach vertices at distance $\frac{8}{7}\, t$.
However, this is not a worst-case example. 
\Cref{fig:zig_zag_example} shows that it is possible that a set reaches vertices at distance $\frac{9}{7}\, t$ and this turns out to be a worst-case example.\footnote{This follows from \eqref{ieq: t close dist} in \Cref{def: t close} and \Cref{inv: the invariant} for $\gamma = 0$.
}

We remark that the interaction of two terminals has already implicitly been analyzed tightly in \cite{BCRlessthan2}.
From now on, however, we will follow a different proof strategy than \cite{BCRlessthan2}.

\paragraph{Safe Edges.}
If more than two disjoint sets $S$ contribute to the same edge, then one could generalize the example from \Cref{fig:simple_interaction_example} to a situation, where a set $S$ can reach vertices very far away.
(See the left part of \Cref{fig:k_interaction_example}.) 
Even worse, already the interaction of three sets can lead to a situation where $S$ can reach vertices very far away.
(See the right part of \Cref{fig:k_interaction_example}.)
However, a key observation is that both of the instances in this figure are not MST-optimal.

\begin{figure}[t]
    \centering
    \resizebox{0.47\textwidth}{!}{
    \begin{tikzpicture}[scale=0.5, rotate =-90]
    \node[smallterminal] (A3) at ({4*cos(150)},{4*sin(150)}) {};
    \node[smallterminal] (A2) at ({4*cos(190)},{4*sin(190)}) {};
    \node[smallterminal] (A1) at ({4*cos(230)},{4*sin(230)}) {};
    \node[smallterminal] (A6) at ({4*cos(-50)},{4*sin(-50)}) {};
    \node[smallterminal] (A5) at ({4*cos(-10)},{4*sin(-10)}) {};
    \node[smallterminal] (A4) at ({4*cos( 30)},{4*sin( 30)}) {};
    \node[smallsteiner, label=left:{\Large$m$}]  (B) at ( 0,0) {};
    \node[smallsteiner, label=right:{\Large$v$}]  (D) at ( 0,7) {};
    \node[] at (0,-3) {$\dots$};

    \draw[diredge] (D) -- (B);
    \draw[diredge] ($(B)!11/24!(D)$) -- (D);

    \draw[edgedirprogress,red]        (A1) -- (B);
    \draw[edgedirprogress,orange] (A2) -- (B);
    \draw[edgedirprogress,yellow!90!black] (A3) -- (B);
    \draw[edgedirprogress,green!70!black] (A4) -- (B);
    \draw[edgedirprogress,blue] (A5) -- (B);
    \draw[edgedirprogress,purple]         (A6) -- (B);

    \tikzstyle{rainbow} = [edgedirprogress,dashed,dash pattern=on 3pt off 24pt]
    \tikzstyle{dots} = [line width = 1pt, dash pattern=on 1pt off 2pt, dash phase=2pt, offset = -3, shorten >=3pt]
    \draw[dots] (B) -- ($(B)!1/2!(D)$);
    \draw[rainbow,dash phase =-6pt, red] (B) -- ($(B)!1/2!(D)$);
    \draw[rainbow,dash phase =-3pt, orange] (B) -- ($(B)!1/2!(D)$);
    \draw[rainbow,dash phase = 0pt, yellow!90!black] (B) -- ($(B)!1/2!(D)$);
    \draw[rainbow,dash phase = 3pt, green!70!black] (B) -- ($(B)!1/2!(D)$);
    \draw[rainbow,dash phase = 6pt, blue] (B) -- ($(B)!1/2!(D)$);
    \draw[rainbow,dash phase = 9pt, purple,>={Stealth[black]}] (B) -- ($(B)!1/2!(D)$);

    \tikzstyle{rainbow} = [edgedirprogress,dashed,dash pattern=on 3pt off 21pt]
    \draw[dots] (B) -- (A1);
    \draw[rainbow,dash phase =-3pt, orange] (B) -- (A1);
    \draw[rainbow,dash phase = 0pt, yellow!90!black] (B) -- (A1);
    \draw[rainbow,dash phase = 3pt, green!70!black] (B) -- (A1);
    \draw[rainbow,dash phase = 6pt, blue] (B) -- (A1);
    \draw[rainbow,dash phase = 9pt, purple,>={Stealth[black]}] (B) -- (A1);

    \draw[dots] (B) -- (A2);
    \draw[rainbow,dash phase =-3pt, red] (B) -- (A2);
    \draw[rainbow,dash phase = 0pt, yellow!90!black] (B) -- (A2);
    \draw[rainbow,dash phase = 3pt, green!70!black] (B) -- (A2);
    \draw[rainbow,dash phase = 6pt, blue] (B) -- (A2);
    \draw[rainbow,dash phase = 9pt, purple,>={Stealth[black]}] (B) -- (A2);

    \draw[dots] (B) -- (A3);
    \draw[rainbow,dash phase =-3pt, red] (B) -- (A3);
    \draw[rainbow,dash phase = 0pt, orange] (B) -- (A3);
    \draw[rainbow,dash phase = 3pt, green!70!black] (B) -- (A3);
    \draw[rainbow,dash phase = 6pt, blue] (B) -- (A3);
    \draw[rainbow,dash phase = 9pt, purple,>={Stealth[black]}] (B) -- (A3);

    \draw[dots] (B) -- (A4);
    \draw[rainbow,dash phase =-6pt, red] (B) -- (A4);
    \draw[rainbow,dash phase =-3pt, orange] (B) -- (A4);
    \draw[rainbow,dash phase = 0pt, yellow!90!black] (B) -- (A4);
    \draw[rainbow,dash phase = 3pt, blue] (B) -- (A4);
    \draw[rainbow,dash phase = 6pt, purple,>={Stealth[black]}] (B) -- (A4);

    \draw[dots] (B) -- (A5);
    \draw[rainbow,dash phase =-6pt, red] (B) -- (A5);
    \draw[rainbow,dash phase =-3pt, orange] (B) -- (A5);
    \draw[rainbow,dash phase = 0pt, yellow!90!black] (B) -- (A5);
    \draw[rainbow,dash phase = 3pt, green!70!black] (B) -- (A5);
    \draw[rainbow,dash phase = 6pt, purple,>={Stealth[black]}] (B) -- (A5);

    \draw[dots] (B) -- (A6);
    \draw[rainbow,dash phase =-6pt, red] (B) -- (A6);
    \draw[rainbow,dash phase =-3pt, orange] (B) -- (A6);
    \draw[rainbow,dash phase = 0pt, yellow!90!black] (B) -- (A6);
    \draw[rainbow,dash phase = 3pt, green!70!black] (B) -- (A6);
    \draw[rainbow,dash phase = 6pt, blue,>={Stealth[black]}] (B) -- (A6);

\end{tikzpicture}     }
    \resizebox{0.48\textwidth}{!}{
    
\newcommand{\drawAllVerticesLongJump}{
    \LARGE
    \node[smallsteiner] (A0) at (-2,-1) {};
    \node[smallsteiner] (A1) at (-2, 1) {};
    \node[smallsteiner] (A2) at (-2, 3) {};
    \node[smallsteiner] (A3) at (-2, 5) {};
    \node[smallsteiner] (A4) at (-2, 7) {};
    \node[smallsteiner, label=below:{$v$}] (B1) at ( 1, 0) {};
    \node[smallsteiner] (B2) at ( 1, 2) {};
    \node[smallsteiner] (B3) at ( 1, 4) {};
    \node[smallsteiner] (B4) at ( 1, 6) {};
    \node[smallsteiner, label=right:{$w$}] (B5) at ( 1, 8) {};
    \node[terminal, transparent] (C1) at (-6, 1) {};
    \node[terminal, transparent] (C2) at (-6, 3) {};
    \node[terminal, transparent] (C3) at (-6, 5) {};
    \node[terminal, transparent] (C4) at (-6, 7) {};
    \node[terminal,minimum height=120pt, label=above:{$s_2$}] (C) at (-6, 4) {};
    \node[terminal, transparent] (D1) at ( 5, 0) {};
    \node[terminal, transparent] (D2) at ( 5, 2) {};
    \node[terminal, transparent] (D3) at ( 5, 4) {};
    \node[terminal, transparent] (D4) at ( 5, 6) {};
    \node[terminal,minimum height=120pt, label=above:{$s_3$}] (D) at ( 5, 3) {};
    \node[terminal, label=below:{$s_1$}] (E) at (-6, -1) {};
    
}

\newcommand{\drawAllEdgesLongJump}{
    \draw[diredge] (C1) -- (A1);
    \draw[diredge] (C2) -- (A2);
    \draw[diredge] (C3) -- (A3);
    \draw[diredge] (C4) -- (A4);
    \draw[diredge] (A1) -- (C1);
    \draw[diredge] (A2) -- (C2);
    \draw[diredge] (A3) -- (C3);
    \draw[diredge] (A4) -- (C4);
    \draw[diredge] (D1) -- (B1);
    \draw[diredge] (D2) -- (B2);
    \draw[diredge] (D3) -- (B3);
    \draw[diredge] (D4) -- (B4);
    \draw[diredge] (B1) -- (D1);
    \draw[diredge] (B2) -- (D2);
    \draw[diredge] (B3) -- (D3);
    \draw[diredge] (B4) -- (D4);
    
    \draw[diredge] (A0) -- (B1);
    \draw[diredge] (A1) -- (B1);
    \draw[diredge] (A1) -- (B2);
    \draw[diredge] (A2) -- (B2);
    \draw[diredge] (A2) -- (B3);
    \draw[diredge] (A3) -- (B3);
    \draw[diredge] (A3) -- (B4);
    \draw[diredge] (A4) -- (B4);
    \draw[diredge] (A4) -- (B5);
    \draw[diredge] (B1) -- (A0);
    \draw[diredge] (B1) -- (A1);
    \draw[diredge] (B2) -- (A1);
    \draw[diredge] (B2) -- (A2);
    \draw[diredge] (B3) -- (A2);
    \draw[diredge] (B3) -- (A3);
    \draw[diredge] (B4) -- (A3);
    \draw[diredge] (B4) -- (A4);
    \draw[diredge] (B5) -- (A4);

    \draw[diredge] (E) -- (A0);
    \draw[diredge] (A0) -- (E);
}

\newcommand{\drawAllVerticesAndEdgesLongJump}{
    \drawAllVerticesLongJump
    \drawAllEdgesLongJump
}

\begin{tikzpicture}[scale=0.6]
    \drawAllVerticesAndEdgesLongJump
    
    \draw[edgedirprogress] (C1) -- (A1);
    \draw[edgedirprogress] (C2) -- (A2);
    \draw[edgedirprogress] (C3) -- (A3);
    \draw[edgedirprogress] (C4) -- (A4);
    \draw[edgedirprogress, red] (D1) -- (B1);
    \draw[edgedirprogress, red] (D2) -- (B2);
    \draw[edgedirprogress, red] (D3) -- (B3);
    \draw[edgedirprogress, red] (D4) -- (B4);
    \draw[edgedirprogress, green!70!black] (E) -- (A0);

    \draw[edgedirprogress] (A1) -- (B1);
    \draw[edgedirprogress] (A1) -- (B2);
    \draw[edgedirprogress] (A2) -- (B2);
    \draw[edgedirprogress] (A2) -- (B3);
    \draw[edgedirprogress] (A3) -- (B3);
    \draw[edgedirprogress] (A3) -- (B4);
    \draw[edgedirprogress] (A4) -- (B4);
    \draw[edgedirprogress] (A4) -- (B5);
    
    \draw[edgedirprogress, red] (B1) -- (A0);
    \draw[edgedirprogress, red] (B1) -- (A1);
    \draw[edgedirprogress, red] (B2) -- (A1);
    \draw[edgedirprogress, red] (B2) -- (A2);
    \draw[edgedirprogress, red] (B3) -- (A2);
    \draw[edgedirprogress, red] (B3) -- (A3);
    \draw[edgedirprogress, red] (B4) -- (A3);
    \draw[edgedirprogress, red] (B4) -- (A4);

    \draw[edgedirprogress, green!70!black] (A0) -- (B1);

\end{tikzpicture}
     }
    \caption{\label{fig:k_interaction_example}
    The left part of the figure shows (a part of) an instance with $k$ terminals interacting. 
    If they meet at the vertex $m$ at time $t=1$, then the edge $(m,v)$ afterwards becomes tight at speed~$k$.
    However, this instance is not MST-optimal, because the star centered at $m$ connecting the $k$ terminals would be an improving component (for $k>2$).
    (If the drop of this terminal set was smaller than $k$, then the terminals would have been merged much earlier and would not continue growing each on their own.)\newline
    The example in the right part of the figure shows that even the interaction of just three terminals can lead to reaching far away vertices.
    In this example, the edges incident to the terminals are long, e.g. they could have length $1$ each, whereas the other edges are short, e.g. they have length $\epsilon$ each for some small $\epsilon > 0$.
    The figure shows the contributions of sets growing around the three terminals at time $t=1+\epsilon$.
    The vertex $w$ is reachable from $s_1$ by a path of tight edges, but this path could be arbitrarily long.
    (The only bound on the distance of $s_1$ and $w$ that we can give in this example is $3+3\epsilon$ by considering a non-tight path via one of the terminals $s_2$ or $s_3$. However, this bound is too weak to be useful.)
    Again, this instance is not MST-optimal, where a star centered at $v$ connecting $s_1, s_2, s_3$ is an improving component.
    }
\end{figure}
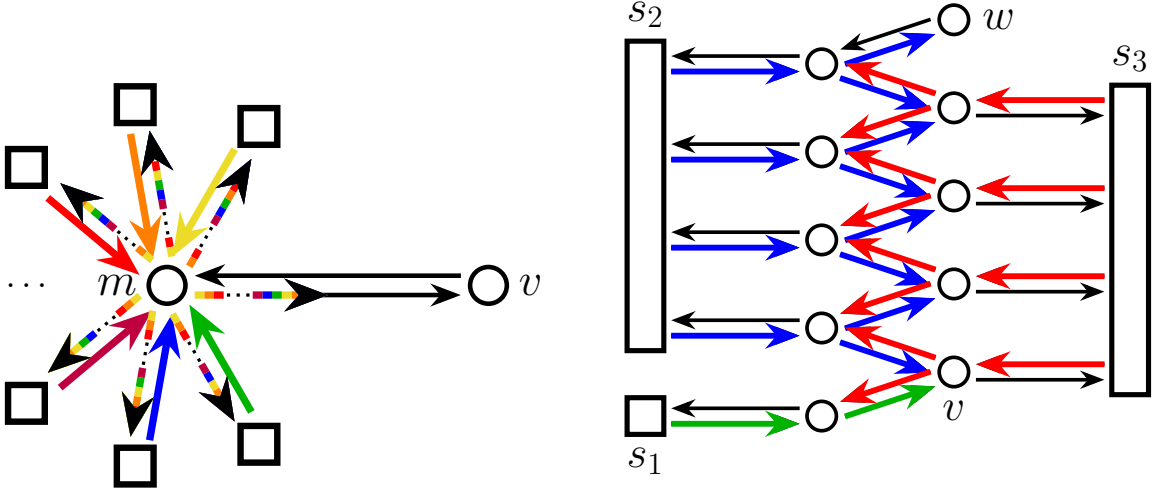

If a set $S$ ``uses'' an edge to which already two other disjoint sets contributed, we will say that this edge is not \emph{safe} for $S$. (We refer to \Cref{def: safe path} for a precise definition.)
We will first consider the case that whenever a set $S\in \cS^t$ reaches a vertex $v$ at time $t$ via a path of tight edges, all edges of this path are safe (for $S$).
Later, we will prove that this is indeed always true.
To do so, we will crucially use that the instance is MST-optimal.

\paragraph{The Component Invariant.}\label{par:component invariant}
Next, we state the key invariant that we prove in our analysis:
Suppose a set $S\in \cS^t$ reaches a vertex $v$ at time $t$ without having interacted with any sets disjoint from $S$ before. 
Then there exist a subset $X\subseteq S$ and a component $K$ connecting $X\cup \{v\}$ such that
\[
c(K) \ \leq\ \dropG{X} + t.
\]
We will also prove that a very similar (but more technical) invariant applies when $S$ is interacting with some other set, but for the purpose of this outline we will focus on the invariant stated above.

Before expanding on how we prove our component invariant, let us first explain why it implies that the set $U_S$ does not contain the root (unless $r\in S$).
Suppose for the sake of deriving a contradiction that $r\in U_S$, i.e., the root $r$ is reachable from $S$ by a path of tight edges at a time~$t$ with $S\in \cS^t$. 
Then the invariant implies that there exist a set $X\subseteq S$ and a component $K$ connecting $X\cup \{r\}$ such that $c(K) \leq \dropG{X} + t$.
Using that $S \in \cS^t$ does not contain $r$ and $X\subseteq S$, we get $\val{\cM,X\cup \{r\}} \geq \val{\cM,X} + t$ by the definition of the local value (see also \Cref{lem: drop of X union X bar}). 
Because $\cM$ is the default merge plan scaled by a factor of $\frac{7}{6}$, using \eqref{eq:equaivalent_drop_outline} for $\delta =\frac{1}{6}$, this implies
\[
\dropG{X\cup \{r\}}\ = \ \frac{12}{7}\,\val{\cM,X\cup \{r\}}
\ \geq\ \frac{12}{7}\,\big(\val{\cM,X} + t\big) \ >\ \dropG{X} + t.
\]
Therefore, the cost $c(K) \leq \dropG{X} + t$ of the component $K$ connecting $X\cup \{r\}$ is strictly less than $\dropG{X\cup \{r\}}$, contradicting the MST-optimality of our instance.
This way we can use our invariant to show that our merge plan is feasible. 

We remark that the only property of the root $r$ that we used is that $r$ is a terminal not contained in $S$.
Thus, the component invariant also implies that the only terminals contained in $U_S$ are the terminals in $S$.
Next, we explain how we prove our component invariant.
     
\paragraph{Proving the Component Invariant.}
\label{par: proving invariant}
First, observe that the invariant is satisfied if $S$ contains only a single terminal $s$.
In this case $X=S =\{s\}$ and $\dropG{X} = 0$.
The component $K$ consists of a shortest $s$-$v$ path, which has length at most $t$ by \eqref{eq:no_interaction_case}.

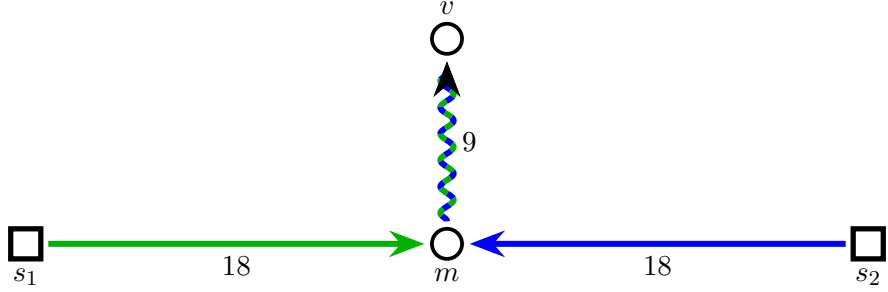
\begin{figure}[t]
    \centering
    \begin{tikzpicture}[scale=0.8]
\node[smallterminal, label=below:{$s_1$}] (S1) at (-3,0) {};
    \node[smallterminal, label=below:{$s_2$}] (S2) at (11,0) {};
    \node[smallsteiner, label=below:{$m$}] (M1) at (4,0) {};
    \node[smallsteiner, label=above:{$v$}] (M2) at (4,3.41) {};

    \draw[edgedirprogress, green!70!black, offset=0pt] (S1) -- node[below, text=black] {$18$} (M1);
    \draw[edgedirprogress, offset=0pt] (S2) -- node[below, text=black] {$18$} (M1);

    \draw[edgedirprogresshalfA, blue, offset=0pt, decorate, decoration={snake}] (M1) -- (M2);
    \draw[edgedirprogresshalfB, offset=0pt, green!70!black, decorate, decoration={snake}] (M1) -- node[right=1pt, text=black] {$9$} (M2);
\end{tikzpicture}     \caption{\label{fig:component_invariant_two_sets}
    An abstract drawing of the component $K$ in our invariant for $X=\{s_1,s_2\}$. 
    The component shown here corresponds to the example from \Cref{fig:zig_zag_example}.
    Both $s_1$ and $s_2$ reach the vertex $m$ at time $t = 18 = \frac{1}{2}\, \dist{s_1,s_2} \geq \frac{1}{2}\, \dropG{\{s_1,s_2\}}$. 
    At time $t=21=\frac{7}{12}\, \dropG{\{s_1,s_2\}}$ the two terminals are merged and become part of the same set $S$, which at this time has not interacted with any sets disjoint from $S$. 
    From time $t=18$ until time $t=21$ the $m$-$v$ path becomes tight effectively at speed $3$. Thus, if $v$ is reached at time $21$, it can have distance $9$ from $m$.
    }
\end{figure}

In this case, the two terminals $s_1$ and $s_2$ have distance at least $\dropG{\{s_1,s_2\}}$ from each other. 
Recall that terminals not interacting with other terminals can at time $t$ only reach vertices at distance at most $t$.
Thus, $s_1$ and $s_2$ cannot reach the same vertex and start interacting before time $\frac{1}{2}\,\dropG{\{s_1,s_2\}}$.
Starting from the vertex $m$ where $s_1$ and $s_2$ ``meet'' (until $s_1$ and $s_2$ are merged) the length of the path via which the terminals reach the vertex $v$ can grow effectively at speed $3$ (see \Cref{lem: length from m_p to v}), where the upper bound is attained by the example in \Cref{fig:zig_zag_example}.
Using that $s_1$ and $s_2$ are not merged ``too late'' (by the choice of the merge plan), we can prove that they do not travel at speed $3$ for too long.
We then choose $X=\{s_1,s_2\}$ and choose the component $K$ as the component consisting of the paths from $s_1$ and $s_2$, respectively, to the vertex $m$ together with the path from $m$ to~$v$.
See \Cref{fig:component_invariant_two_sets}.
This is our strategy for analyzing the interaction of sets growing around two single terminals that then get merged to a new set $S$.

Let us next explain, how we proceed in the general case.
Here we use induction over time, which is possible because the set of reachable vertices only changes at discrete points in time.
If $S$ reaches $v$ via a path $P$ of tight edges, there might be two disjoint subsets $S_1$ and $S_2$ of $S$ that interacted with each other along this path $P$. Because of our assumption that our path $P$ consists of safe edges, it will be sufficient to analyze the interaction of two disjoint terminal sets.
Let $S_1 \subseteq S$ and $S_2 \subseteq S$ be such interacting sets.
(Note that each of $S_1$ and $S_2$ could each have again disjoint interacting subsets, which again could have interacting subsets and so on.)
Similar to the case of two interacting terminals, we identify a vertex $m$ where the two sets ``meet''. 
In particular, the vertex $m$ will have the property that none of the two sets $S_1$ and $S_2$ contributed\footnote{
If a set $S'$ contributes to $S_i$ reaching $v$, this means that some corresponding set $U_{S'}$ contributed to the tightness of the edges on the $S_i$-$v$ path via which the set $S_i$ reached the vertex $v$.}  to the other set reaching $m$.
We then apply the induction hypothesis to $S_1$ and $m$, as well as to $S_2$ and $m$.
We obtain sets $X_1 \subseteq S_1$ and $X_2\subseteq S_2$ and components $K_1$ and $K_2$ connecting $X_1\cup\{m\}$ and $X_2\cup \{m\}$, respectively.
See \Cref{fig:induction_step_outline} for an illustration.
\begin{figure}[t]
    \centering
    \begin{tikzpicture}[scale=0.6]
\useasboundingbox (-3,-2) rectangle (16,11);
\node[smallterminal] (S1) at (-2, 3) {};
    \node[smallterminal] (S2) at ( 2, 1) {};
    \node[smallterminal] (S3) at ( 7, 0) {};
    \node[smallterminal] (S4) at (10, 1) {};
    \node[smallterminal] (S5) at (14, 4) {};
    \node[smallterminal] (S6) at (12, 3) {};
    \node[smallterminal] (S7) at (13, 2) {};

    \draw[activeset, green!50!black] (S1) ellipse [radius = 0.7];
    \draw[activeset, green!90!black] (S2) ellipse [radius = 0.7];
    \draw[activeset, blue] (S3) ellipse [radius = 0.7];
    \draw[activeset, blue!40!white] (S4) ellipse [radius = 0.7];
    \draw[activeset, blue!40!black] (S5) ellipse [radius = 0.7];
    \coordinate (S12) at ($(S1)!1/2!(S2)$);
    \draw[activeset, green!70!black,rotate around={-26.6:(S12)}](S12) ellipse [x radius=3.2, y radius=1.5];
    \coordinate (S34) at ($(S3)!1/2!(S4)$);
    \draw[activeset, blue!70!white,rotate around={18.4:(S34)}](S34) ellipse [x radius=3, y radius=1.2];
    \coordinate (S345) at ($(S34)!1/3!(S5)$);
    \draw[activeset, blue!70!black,rotate around={28:(S345)}](S345) ellipse [x radius=6, y radius=2];

    \coordinate (A1) at (1,4);
    \coordinate (A2) at (2,6);
    \node[smallsteiner, label={[label distance=-2.5mm]100:{$m$}}] (A3) at (5,9) {};
    \coordinate (A4) at (8,2);
    \coordinate (A5) at (8,4);
    \coordinate (A6) at (9,6);
    \coordinate (A7) at (7,7);
    \coordinate (A8) at ( 9,11);

    \draw[mstedge, green!50!black] (S1) -- (A1);
    \draw[mstedge, green!90!black] (S2) -- (A1);
    \draw[mstedge, curlyHalfA, green!50!black] (A1) -- (A2);
    \draw[mstedge, curlyHalfB, green!90!black] (A1) -- (A2);
    \draw[mstedge] (A2) -- (A3);
    
    \draw[mstedge, blue, blue] (S3) -- (A4);
    \draw[mstedge, blue, blue!40!white] (S4) -- (A4);
    \draw[mstedge, curlyHalfA, blue] (A4) -- (A5);
    \draw[mstedge, curlyHalfB, blue!40!white] (A4) -- (A5);
    \draw[mstedge, blue!70!white] (A5) -- (A6);
    \draw[mstedge, blue!40!black] (S5) -- (A6);
    \draw[mstedge, curlyHalfA, blue!70!white] (A6) -- (A7);
    \draw[mstedge, curlyHalfB, blue!40!black] (A6) -- (A7);
    \draw[mstedge, blue!70!black] (A7) -- (A3);
    
    \draw[mstedge, curlyHalfA, blue!70!black] (A3) -- (A8);
    \draw[mstedge, curlyHalfB, green!70!black] (A3) -- (A8);

    \node[edgelabel, text = green!70!black] (K1) at (1.5, 7) { $K_1$};
    \node[edgelabel, text = blue!70!black] (K2) at (10, 7) { $K_2$};
    \node[edgelabel, text = green!50!black] (S1L) at (-1, 1) { $S_1$};
    \node[edgelabel, text = blue!70!black] (S2L) at (13, 1) { $S_2$};

\end{tikzpicture}     \caption{\label{fig:induction_step_outline}
    Two components $K_1$ (green) and $K_2$ (blue) meeting at a vertex $m$. 
    Once they met, they interact analogously to the interaction of two sets growing around single terminals.
    Each of the components $K_i$ connects a subset $X_i \subset S_i$ to $m$.
    Once $S_1$ and $S_2$ get merged in $\cS^t$, we can combine them to a new component $K$ for the new set $S$ arising from the merge of $S_1$ and $S_2$ (and possibly further sets).
    The components $K_i$ arise by induction and hence there are several subsets of $S_i$ that contributed to the tightness of the edges of $K_i$, which are the subsets of $S_i$ shown in the picture.
    }
\end{figure}
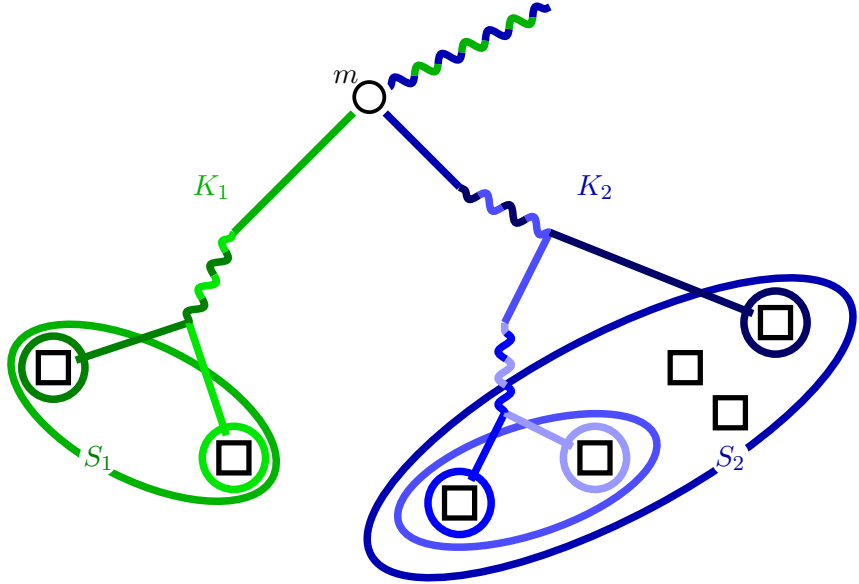

The union of $K_1$ and $K_2$ is a component connecting $X_1\cup X_2$.
If $t'$ is the time where the sets $S_1$ and $S_2$ reach $m$, then the
induction hypothesis implies
\[
c(K_1\cup K_2) \ \leq\ \dropG{X_1} + \dropG{X_2} + 2 t'.
\]
Using that the instance is MST-optimal and hence $c(K_1\cup K_2) \geq \dropG{X_1\cup X_2}$, we obtain a lower bound on when $S_1$ and $S_2$ ``meet'' at the vertex $m$.
Then, analogously to the special case of two interacting terminals we considered before, the path from $m$ to $v$ is effectively growing at speed $3$, and the two sets stop interacting once they are merged into a single set.
Once the two sets are merged, we can construct a new component $K$ for the merged set $S$ by taking the union of $K_1$, $K_2$, and the $m$-$v$ path.

\paragraph{Bifurcation Vertices.}\label{par: bifurcation}
In the proof of our component invariant, the bifurcation vertex $m$ where two set $S_1$ and $S_2$ ``meet'' plays a crucial role.
So far, we did not give a concise definition of this bifurcation vertex and indeed it is not obvious how to choose this vertex $m$.
A key property of the bifurcation vertex is that each of the sets $S_1$ and $S_2$ reaches $m$ without any interaction with the other set.
However, not every path of tight edges from $S_i$ to $v$ contains such a point, as shown in \Cref{fig:bifurcation_example}.

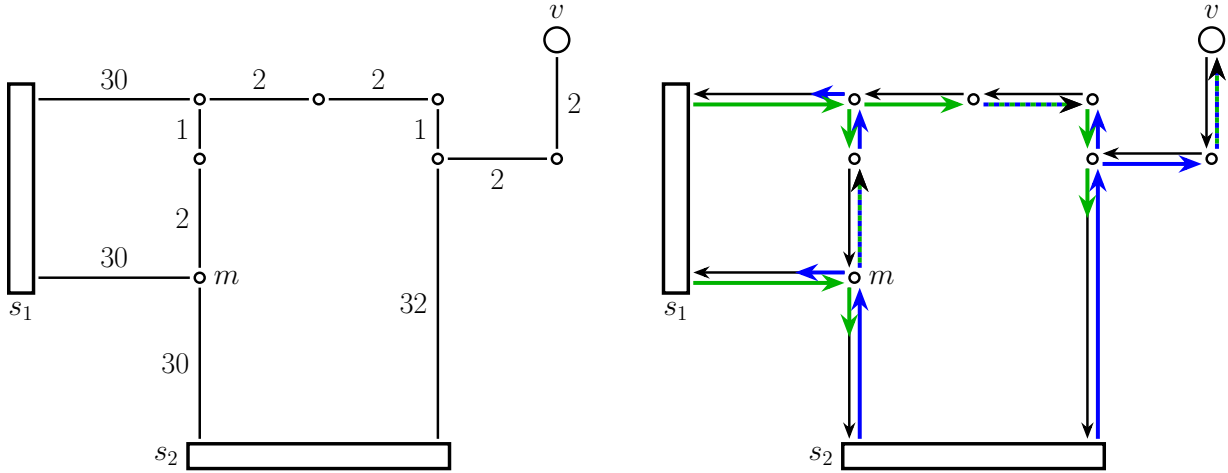
\begin{figure}[ht]
    \centering
    \resizebox{\textwidth}{!}{\newcommand{\VerticesNeedMeetingPoint}{
    \LARGE
    \node[terminal, transparent] (S11) at (0,6) {};
    \node[terminal, transparent] (S12) at (0,3) {};
    \node[terminal, minimum height = 120pt, label=below:{$s_1$}] (S1) at (0,4.5) {};
    \node[terminal, transparent] (S21) at (3,0) {};
    \node[terminal, transparent] (S22) at (7,0) {};
    \node[terminal, minimum width = 150pt, label=left:{$s_2$}] (S2) at (5,0) {};

    \node[tinysteiner, label=right:{$m$}] (V1) at (3,3) {};
    \node[tinysteiner] (V2) at (3,5) {};
    \node[tinysteiner] (V3) at (3,6) {};
    \node[tinysteiner] (V4) at (5,6) {};
    \node[tinysteiner] (V5) at (7,6) {};
    \node[tinysteiner] (V6) at (7,5) {};
    \node[tinysteiner] (V7) at (9,5) {};
    \node[steiner, label=above:{$v$}] (V8) at (9,7) {};
}

\begin{tikzpicture}[scale=1.2]
    \VerticesNeedMeetingPoint

    \draw[edge] (V1) -- node[left] {$30$} (S21);
    \draw[edge] (V1) -- node[above] {$30$} (S12);
    \draw[edge] (V2) -- node[left] {$2$} (V1);
    \draw[edge] (V2) -- node[left] {$1$} (V3);
    \draw[edge] (V3) -- node[above] {$30$} (S11);
    \draw[edge] (V4) -- node[above] {$2$} (V3);
    \draw[edge] (V5) -- node[above] {$2$} (V4);
    \draw[edge] (V6) -- node[left] {$32$} (S22);
    \draw[edge] (V7) -- node[below] {$2$} (V6);
    \draw[edge] (V5) -- node[left] {$1$} (V6);
    \draw[edge] (V8) -- node[right] {$2$} (V7);

\begin{scope}[shift={(11,0)}]
    \VerticesNeedMeetingPoint

    \draw[diredge] (V1) -- (S21);
    \draw[diredge] (V1) -- (S12);
    \draw[diredge] (V2) -- (V1);
    \draw[diredge] (V3) -- (S11);
    \draw[diredge] (V4) -- (V3);
    \draw[diredge] (V5) -- (V4);
    \draw[diredge] (V6) -- (S22);
    \draw[diredge] (V7) -- (V6);
    \draw[diredge] (V8) -- (V7);

    \draw[edgedirprogress] (S21) -- (V1);
    \draw[edgedirprogresshalfA] (V1) -- (V2);
    \draw[edgedirprogress] (V1) -- ($(V1)!1/3!(S12)$);
    \draw[edgedirprogress] (V2) -- (V3);
    \draw[edgedirprogress] (V3) -- ($(V3)!1/4!(S11)$);
    \draw[edgedirprogresshalfA] (V4) -- (V5);
    \draw[edgedirprogress] (V6) -- (V5);
    \draw[edgedirprogress] (S22) -- (V6);
    \draw[edgedirprogress] (V6) -- (V7);
    \draw[edgedirprogresshalfA] (V7) -- (V8);

    \draw[edgedirprogress, green!70!black] (S12) -- (V1);
    \draw[edgedirprogress, green!70!black] (V1) -- ($(V1)!1/3!(S21)$);
    \draw[edgedirprogresshalfB, green!70!black] (V1) -- (V2);
    \draw[edgedirprogress, green!70!black] (S11) -- (V3);
    \draw[edgedirprogress, green!70!black] (V3) -- (V2);
    \draw[edgedirprogress, green!70!black] (V3) -- (V4);
    \draw[edgedirprogresshalfB, green!70!black] (V4) -- (V5);
    \draw[edgedirprogress, green!70!black] (V5) -- (V6);
    \draw[edgedirprogress, green!70!black] (V6) -- ($(V6)!1/5!(S22)$);
    \draw[edgedirprogresshalfB, green!70!black] (V7) -- (V8);
\end{scope}
\end{tikzpicture} }
    \caption{\label{fig:bifurcation_example}
    An example illustrating the difficulty in the choice of the bifurcation vertex $m$.
    The left part of the figure shows (a part of) an instance with the lengths of the edges.
    The right part of the figure indicates the contributions of the two sets growing around the two terminals $s_1$ (green) and $s_2$ (blue) at time $t=35$.
    The tight $s_1$-$v$ path using the upper outgoing edge of $s_1$ does not contain any point that each of the terminals can reach via a path to which the other terminal did not contribute.
    The same applies for the tight $s_2$-$v$ path using the right outgoing edge of $s_2$.
    In our proof, we will choose the indicated vertex $m$ as the bifurcation vertex and analyse the component consisting of the edges $\{s_1,m\}$ and $\{s_2,m\}$ as well as the $m$-$v$ path of tight edges.
    }
\end{figure}

The concept of bifurcation vertices and the proof of their existence is a key technical contributions in our new analysis that enables our recursive proof strategy, which differs substantially from the analysis from \cite{BCRlessthan2}.\footnote{
The analysis of \cite{BCRlessthan2} always started from an arbitrary $S$-$v$ path $P$ with a particular property (a so-called $S$-tight path; see \Cref{def: S tight path}) and constructed only components containing this path $P$. 
The example in \Cref{fig:bifurcation_example} shows that it might be that no such paths contains any bifurcation vertex.
Hence, for our proof strategy to work, we crucially need to deviate from including an arbitrary $S$-tight path to $v$ in the component $K$.
}

\paragraph{Everything is Safe.}
Let us briefly comment on how we prove that our assumption regarding safe edges is indeed always satisfied.
Suppose some set would encounter a non-safe edge.
Then we consider a first such non-safe edge $e$. 
Let $S_1$ be the set that encounters an edge that is non-safe (for $S_1$) and let $S_2$ and $S_3$ be the other disjoint sets that contribute to $e$.
We will apply our component invariant to obtain three components $K_1,K_2,K_3$ connecting sets $X_1\cup\{v\}, X_2\cup \{v\}, X_3\cup \{v\}$, respectively, and combine them to an improving component $K$ connecting $X_1\cup X_2 \cup X_3$.
Then, we use this component $K$ to derive a contradiction to the MST-optimality and can thus conclude that we never encounter non-safe edges.
Here, we will make use of the (slightly more technical) version of our invariant that also applies to sets $S$ that interacted with one other disjoint set.
This will be needed for example to handle the case where some of the sets $S_1, S_2, S_3$ already had pairwise interactions with each other.
See \Cref{fig:accident} for an example.

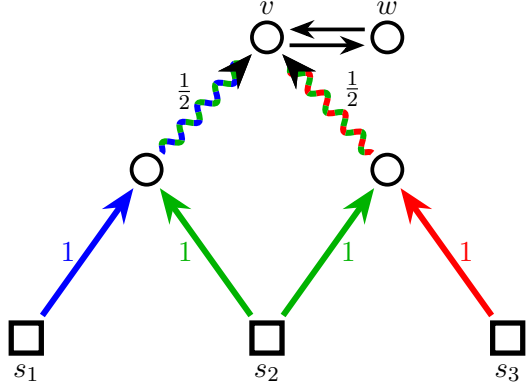
\begin{figure}
    \centering
    \begin{tikzpicture}[scale=0.8]
    \node[smallterminal] (S1) at ( 0, 1) {};
    \node[smallterminal] (S2) at ( 4, 1) {};
    \node[smallterminal] (S3) at ( 8, 1) {};
    \node[smallsteiner] (V1) at ( 2, 3.8) {};
    \node[smallsteiner] (V2) at ( 6, 3.8) {};
    \node[smallsteiner] (V3) at ( 4, 6) {};
    \node[above=5pt] at (V3) {$v$};
    \node[smallsteiner] (w) at ( 6, 6) {};
    \node[above=5pt] at (w) {$w$};
    \node[below=6pt] at (S1) {$s_1$};
    \node[below=6pt] at (S2) {$s_2$};
    \node[below=6pt] at (S3) {$s_3$};

    \draw[edgedirprogress, offset=0pt] (S1) -- node[left] {$1$} (V1);
    \draw[edgedirprogress, green!70!black, offset=0pt] (S2) -- node[left] {$1$} (V1);
    \draw[edgedirprogress, green!70!black, offset=0pt] (S2) -- node[right] {$1$} (V2);
    \draw[edgedirprogress, red, offset=0pt] (S3) -- node[right] {$1$} (V2);

    \draw[edgedirprogresshalfA, offset=0pt, decorate, decoration={snake}] (V1) -- (V3);
    \draw[edgedirprogresshalfB, offset=0pt, green!70!black, decorate, decoration={snake}] (V1) -- node[above left=-3pt, pos=0.4, black] {$\frac{1}{2}$} (V3);

    \draw[edgedirprogresshalfA, red, offset=0pt, decorate, decoration={snake}] (V2) -- (V3);
    \draw[edgedirprogresshalfB, offset=0pt, green!70!black, decorate, decoration={snake}] (V2) -- node[above right=-2pt, pos=0.4, black] {$\frac{1}{2}$} (V3);

    \draw[diredge] (V3) -- (w);
    \draw[diredge] (w) -- (V3);
\end{tikzpicture}     \caption{\label{fig:accident}
    An example with three terminals $s_1, s_2, s_3$ reaching the same vertex $v$ at time $t=\frac{7}{6}$.
    If the terminals are not merged at this time, they would all contribute to the edge $\lr{v, w}$, causing this edge to be non-safe.
    If a terminal MST contains the edges $\{s_1, s_2\}$ and $\{s_2, s_3\}$, each having length $2$, then we have $\dropG{\{s_1, s_2, s_3\}}= 4$, and the three terminals will be merged exactly at time $t=\frac{7}{6}$.
    \newline
    If one modifies the example by making the edges entering $v$ just a little bit shorter, $\lr{v, w}$ would indeed not be safe. 
    However, we could then obtain an improving component by connecting $s_1$, $s_2$ and $s_3$ to $v$ (omitting one of the two green edges).
    This example shows that in the MST-optimal case, we cannot merge terminals any later if we want edges to be safe. 
    As we will show in \Cref{sec:lower_bound}, this is not an artifact of our analysis strategy relying on safe edges, but there is indeed no merge plan that leads to any better dual solution.
    }
\end{figure}

\paragraph{Beyond the MST-Optimal Case.}
To generalize the analysis beyond the MST-optimal case, we make two changes to the outline above. 
First, we strengthen the upper bound on the cost of the component $K$ in our invariant and prove
\[
c(K) \ \leq\ \frac{8-8\gamma}{7-5\gamma}\,\val{\cM, X}+t.
\]
For $\gamma = 0$ and $\cM$ being the merge plan we used in the MST-optimal case, this is equivalent to $c(K) \leq \frac{2}{3} \, \dropG{X}+t$, where we used \eqref{eq:equaivalent_drop_outline} for $\delta =\frac{1}{6}$ to express the drop in terms of the local value of the merge plan. 
Hence, up to the factor of $\frac{2}{3}$ this indeed corresponds exactly to the invariant we stated before for the MST-optimal case.

Second, we additionally prove a bound on the distance of reachable vertices. More precisely, if $X\subseteq S$ is the set of terminals from $S$ that $K$ connects to the vertex $v$ (in our component invariant), then we prove an upper bound on the $x$-$v$ distance for all $x\in X$.
(See \Cref{def: t close} for a precise description of the bounds.)
We make use of these distance bounds in two ways.
First, they will help us to obtain a good dependence of our bounds on $\gamma$.

Second, they are crucial for us to go beyond the case where all components are $\gamma$-expensive.
If a $\gamma$-cheap component connecting a set $Y$ exists, then by the definition of a $\gamma$-good merge plan (\Cref{def: gamma good}), there are some vertices $x,y\in Y$ that have large distance from each other (compared to their merge time).
Using this, the distance bounds imply that if $x$ and $y$ are not merged at time $t$, the two sets in $\cS^t$ that each contain one of $x$ and $y$ cannot reach the same vertex $v$.
For this reason, these two sets cannot interact in any way, and we will never consider a component connecting $Y$ in our analysis.
Thus, it does not matter for our analysis whether such components are $\gamma$-expensive. 

\section{Merge Plans}\label{sec:merge_plan}
In this section, we establish a connection between merge plans and minimum spanning trees that will serve two purposes: first, it provides a more convenient way of constructing a merge plan, without defining each partition explicitly; second, for a set $X\subseteq R$, it allows us to express $ \dropG{X} \coloneqq \tmst_G - \tmst_{G/X} $ in terms of local value $\val{\cM, X}$ of $X$ with respect to the canonical merge plan $\cM$ of $G$.

\begin{definition}
    Given a terminal set $R$ and an upper bound $u\colon \binom{R}{2}\to\R_{\geq 0}$ we say that a merge plan $\cM$ is \emph{upper bounded by $u$} if
    \begin{align*}
        \merge{\cM}{x, y}\leq u\lr{\lrb{x, y}}\qquad\text{for all }x, y\in R\,.
    \end{align*}
\end{definition}

A natural question to ask is: given an upper bound $u\colon \binom{R}{2}\to\R_{\geq 0}$, how large can the value $\val{\cM}$ be if the merge plan $\cM$ is upper bounded by $u$? Before we answer this question in \Cref{lem: max merge plan has value mst and there are defining paths}, we observe an upper bound on the merge time $\merge{\cM}{x, y}$ that holds whenever $\cM$ is upper bounded by $u$:

\begin{observation}\label{obs: min max bound on merge times}
    Given a terminal set $R$ and an upper bound $u \colon \binom{R}{2}\to\R_{\geq 0}$, every merge plan $\cM$ upper bounded by $u$ satisfies
    \begin{align}
        \merge{\cM}{x, y}\leq\ \min_{\substack{P\,:\,x\text{-}y\,\text{path}\\\text{in }(R,\binom{R}{2})}}\ \max_{e\in P}\ u\lr{e}\label{ieq: min max bound on merge times}
    \end{align}
    for all $x,y\in R$.
\end{observation}
\begin{proof}
    Let $\cM=\lr{\cS^t}_{t\geq0}$ be a merge plan that is upper bounded by $u$, and let $x, y\in R$. Further, let $P$ be an $x$-$y$ path in $\big(R,\binom{R}{2}\big)$. Then, for every time $t>\max_{e\in P}\ u\lr{e}$, for each edge $e$ of $P$ both endpoints of $e$ have to be in the same part of the partition $\cS^t$. Thus, all vertices of the path, in particular $x$ and $y$, are in the same part of $\cS^t$, which implies \eqref{ieq: min max bound on merge times}.
\end{proof}

Next we show that the upper bound \eqref{ieq: min max bound on merge times} on $\merge{\cM}{x, y}$ is tight as there exists a merge plan that is upper bounded by $u$ and satisfies \eqref{ieq: min max bound on merge times} with equality for all $x, y\in R$ simultaneously.
For an upper bound $u \colon \binom{R}{2}\to\R_{\geq 0}$, we write $\mst_u$ to denote the cost of a minimum spanning tree in $(R,\binom{R}{2})$ with respect to the edge costs $u$.

\begin{lemma}\label{lem: max merge plan has value mst and there are defining paths}
    Given a terminal set $R$ and an upper bound $u$, there is a unique merge plan $\cM_u$ that fulfills \eqref{ieq: min max bound on merge times} with equality for all $x,y\in R$. This merge plan is also the unique merge plan with maximal value among all merge plans upper bounded by $u$. Its value is $\val{\cM_u}=\mst_u$.
\end{lemma}
\begin{proof}
    We start by defining a merge plan $\cM=\lr{\cS^t}_{\geq0}$ for which we then show that it fulfills \eqref{ieq: min max bound on merge times} with equality and satisfies $\val{\cM}=\mst_u$.

    Let $T$ be a minimum spanning tree in the weighted graph $\big(R,\binom{R}{2}, u\big)$. For every $t\geq0$, let $T^t$ be the subgraph of $\lr{R, T}$ that contains an edge $\lrb{x, y}\in T$ if and only if $u\lr{\lrb{x, y}}<t$. We define $\cS^t$ to be the partition of $R$ according to the connected components of $T^t$.

    For every tree edge $\lrb{x, y}\in T$, this implies that $\merge{\cM}{x, y}=u\lr{\lrb{x, y}}$. For general $x, y\in R$, consider the unique $x$-$y$ path $P$ in $T$. By definition of $\cM$, the terminals $x$ and $y$ are merged at time $\max_{e\in P}u\lr{e}$. Because $T$ is an MST, none of the edges in $P$ can be more expensive than the edge $\lrb{x, y}$. More precisely, $u\lr{e}\leq u\lr{\lrb{x, y}}$ for all $e\in P$. Thus,
    \begin{align*}
        \merge{\cM}{x, y}=\max_{e\in P}u\lr{e}\leq u\lr{\lrb{x, y}}\,,
    \end{align*}
    which implies not only that $\cM$ fulfills the upper bound $u$ and thus \eqref{ieq: min max bound on merge times}, but also that $\cM$ obtains equality in \eqref{ieq: min max bound on merge times}.

    Note that there can only be one merge plan that fulfills \eqref{ieq: min max bound on merge times} with equality, because a merge plan is uniquely defined by the vector of merge times $\lr{\merge{\cM}{x, y}}_{x, y\in R}$. 
Additionally, if there was a different merge plan $\cM'\neq\cM$ that is upper bounded by $u$ with $\val{\cM'}\geq\val{\cM}$, there would be terminals $x, y\in R$ with $\merge{\cM'}{x, y}>\merge{\cM}{x, y}$ contradicting \Cref{obs: min max bound on merge times}.

    It remains to determine $\val{\cM}$. To this end, we observe that for $t\geq0$, the edges $e\in T$ with $u\lr{e}\geq t$ form a spanning tree in the graph that arises from $\big(R, \binom{R}{2}\big)$ by contracting each part of the partition $\cS^t$ (by construction of $\cS^t$).
    Thus, the number of such edges is
    \begin{align*}
        \big|\lrb{e\in T \mid u\lr{e}\geq t}\big|\ =\ \lrv{\cS^t} - 1\,.
    \end{align*}
    This allows us to express the value of $\cM$ as
    \begin{align*}
        \val{\cM}\ =\ \int_0^{\infty}\lrv{\cS^t}-1\,dt\ =\ \int_0^{\infty}\big|\lrb{e\in T\ :\ u\lr{e}\geq t}\big|\,dt\ =\ \sum_{e\in T}u\lr{e}\ =\ \mst_u\,.
    \end{align*}
\end{proof}

A direct consequence of \Cref{lem: max merge plan has value mst and there are defining paths} is that if we choose $u\lr{\lrb{x, y}}$ proportional to the distance $\dist{x, y}$ in the graph $G$ for all $x,y\in R$, then the value of $\cM_u$ is proportional to the cost of a terminal minimum spanning tree. In particular, for $u(x,y) = \frac{1}{2}\,\dist{x, y}$, this determines the value of the canonical merge plan (\Cref{def: canonical merge plan}):

\begin{corollary}\label{cor: canonical merge plan}
  The canonical merge plan has value $\frac{1}{2}\,\tmst$.
\end{corollary}

Another consequence of \Cref{lem: max merge plan has value mst and there are defining paths} is the following.

\begin{observation}\label{obs: MSTs have same cost wrt u and to merge times}
    Given an upper bound $u$, the minimum spanning trees in the complete graph on $R$ with respect to the cost function $u$ and with respect to the cost function $\lrb{x, y}\mapsto\merge{\cM_u}{x, y}$ all have the same cost, namely $\val{\cM_u}$.
\end{observation}
\begin{proof}
  Let $\tilde{u}$ denote the upper bound given by $\tilde{u}(x,y) \coloneqq \merge{\cM_u}{x, y}$.
  Then, because the merge plan $\cM_u$ is upper bounded by $u$, we have $\mst_u \leq \mst_{\tilde{u}}$.
  Moreover, the merge plan $\cM_u$ is also upper bounded by $\tilde{u}$. 
  Thus, we have $\val{\cM_{\tilde{u}}} \geq \val{\cM_{u}}$, where we used that $\cM_{\tilde{u}}$ is the merge plan with maximum value that is upper bounded by $\tilde{u}$.
  Finally, by \Cref{lem: max merge plan has value mst and there are defining paths}, we have $\mst_u = \val{\cM_u}$ and $\mst_{\tilde{u}} = \val{\cM_{\tilde{u}}}$, implying
  \[
     \mst_u = \val{\cM_u} \leq \val{\cM_{\tilde{u}}} = \mst_{\tilde{u}} \leq \mst_u\,
  \]
  and hence we have equality throughout.
\end{proof}

As mentioned earlier, we can express $ \dropG{X} \coloneqq \tmst_G - \tmst_{G/X} $ in terms of the local value of $X$ with respect to the canonical merge plan $\cM$ of $G$.
\begin{lemma}\label{lem: drop in terms of canonical merge plan}
    Let $\cM$ be the canonical merge plan of $G$ and $X\subseteq R$, then
    \begin{align*}
        \dropG{X}=2\cdot\val{\cM, X}\,.
    \end{align*}
\end{lemma}

To prove \Cref{lem: drop in terms of canonical merge plan}, we show that the canonical merge plan (and more generally $\cM_u$ for any upper bound $u$) behaves nicely under contraction of terminal sets.
Given a terminal set $R$, a merge plan $\cM=\lr{\cS^t}_{t\geq0}$ on $R$, and a set $X\subseteq R$, we define the contracted merge plan $\cM/X=\lr{\cS_X^t}_{t\geq0}$ to be the merge plan on $R/X$ with the following property: For every time $t\geq0$ the partition $\cS_X^t$ arises from $\cS^t$ by combining all sets that intersect $X$ into a single set. 
Moreover, we define an upper bound $u_X$ arising from $u$ by contraction of $X$ as follows:  for $a,b\in R\setminus X$, let $u_X\lr{\lrb{a, b}}\coloneqq u\lr{\lrb{a, b}}$ and $u_X\lr{\lrb{a, X}}\coloneqq\min_{x\in X}u\lr{\lrb{a, x}}$.

Using this notation, we can show that $\cM_u$ is stable under contraction. More precisely:
\begin{lemma}\label{lem: Mu is preserved under contraction}
    Given a terminal set $R$, an upper bound $u$, and a subset $X\subseteq R$ we have
    \begin{align*}
        \cM_u/X=\cM_{u_X}\,.
    \end{align*}
\end{lemma}
\begin{proof}
    Because for all $a, b\in R$, the merge time cannot increase when contracting $X$, that is $\merge{\cM_u/X}{a, b}\leq\merge{\cM_u}{a, b}$, the merge plan $\cM_u/X$ satisfies the upper bound $u_X$.
    We will show that for all $a, b\in R/X$ with $a\neq b$, there exists an $a$-$b$ path $P$ in the complete graph on $R/X$  with
    \begin{align}
        \merge{\cM_u/X}{a, b}\ \geq\ \max_{e\in P} \ u_X\lr{e}.\label{ieq: min max bound in contracted graph}
    \end{align}
    Then \Cref{lem: max merge plan has value mst and there are defining paths} implies that $\cM_u/X=\cM_{u_X}$. It remains to show \eqref{ieq: min max bound in contracted graph}.

    Consider the time $t\coloneqq\merge{\cM_u/X}{a, b}$. By definition of $\cM_u/X$ one of the following it true: 
    (i) in $\cM_u$ the terminals $a$ and $b$ are merged exactly at time $t$ , or 
    (ii) in $\cM_u$ the terminal $a$ is merged with some terminal $x_a\in X$ and $b$ is merged with some terminal $x_b\in X$ both no later than time $t$. 
    Thus, by \Cref{lem: max merge plan has value mst and there are defining paths} the complete graph on $R$ contains (i) an $a$-$b$ path $P$ with $u\lr{e}\leq t$ for all $e\in P$, or (ii) an $a$-$x_a$ path $P_a$ and an $b$-$x_b$ path $P_b$ with $u\lr{e}\leq t$ for all $e\in P_a\cup P_b$. In both cases there is an $a$-$b$ path $P_{ab}$ in the complete graph on $R/X$ with $u_X\lr{e}\leq t$ for all $e\in P_{ab}$ (in case (i) choose $P_{ab}=P$ and for (ii) choose $P_{ab}=P_a\cup P_b$). This path $P_{ab}$ fulfills \eqref{ieq: min max bound in contracted graph}, which concludes the proof.
\end{proof}

Using \Cref{lem: Mu is preserved under contraction} and \Cref{cor: canonical merge plan}, we can prove \Cref{lem: drop in terms of canonical merge plan}:

\begin{proof}[Proof of \Cref{lem: drop in terms of canonical merge plan}]
    Let $\cM$ be the canonical merge plan of $G$. Then by \Cref{lem: Mu is preserved under contraction}, $\cM/X$ is the canonical merge plan of $G/X$.
    By \Cref{cor: canonical merge plan}, the values of the terminal spanning tree in $G$ and $G/X$ equal twice the values of $\cM$ and $\cM/X$, respectively. Thus
    \begin{align*}
        \frac12 \,\dropG{X}
        \ =\ \frac12\,\mst_G-\frac12\,\mst_{G/X}
        \ =\ \val{\cM}-\val{\cM/X}\,.
    \end{align*}
    It remains to show that that the right hand-side equals $\val{\cM, X}$. To this end, we fix a time $t\geq 0$ and consider the corresponding partitions $\cS^t$ of $\cM$ and $\cS^t_X$ of $\cM/X$. By definition of $\cM/X$, the partition $\cS^t_X$ arises from $\cS^t$ by combining all parts of $\cS^t$ that intersect $X$. 
    Thus,
    \begin{align*}
        \lr{\lrv{\cS^t}-1}-\lr{\lrv{\cS^t_X}-1}
        \ =\ \lrv{\cS^t}-\lrv{\cS^t_X}
        \ =\ \lrv{\lrb{S\in\cS^t\,:\,S\cap X\neq\emptyset}}-1\,.
    \end{align*}
    Integrating both sides from $0$ to $\infty$ yields
    \begin{align*}
        \val{\cM}-\val{\cM/X}\ =\ \val{\cM, X}.
    \end{align*}
\end{proof}

A direct consequence of \Cref{lem: drop in terms of canonical merge plan} is the following.

\begin{lemma}\label{lem: merge plan mst optimal}
    For MST-optimal instances of Steiner tree, there exists a merge plan $\cM$ that is $0$-good and satisfies $\val{\cM} \geq \frac{7}{12} \cdot \tmst$.
\end{lemma}
\begin{proof}
    For an MST-optimal instance, there are no improving components.
    That is, we have $\conn{X}\geq\dropG{X}$ for all $X\subseteq R$.
    Let $\cM_G$ be the canonical merge plan, and consider a ``scaled'' version $\cM$ of $\cM_G$, namely the merge plan with maximal value subject to satisfying 
    \begin{equation*}
    \merge{\cM}{x,y} \ \leq\ \frac{7}{12}\, \dist{x,y} \quad \forall x,y\in R\,.
    \end{equation*}
    Then, by \Cref{lem: drop in terms of canonical merge plan}, we obtain for all $X\subseteq R$,
    \begin{align*}
        \conn{X}\ \geq\ \dropG{X}\ =\ 2\cdot\val{\cM_G, X}\ =\ \frac{12}{7}\cdot \val{\cM, X}\,,
    \end{align*}
    where the last equation follows from the fact that $\cM$ is a suitably scaled version of $\cM_G$. 
Hence, there are no $0$-cheap components with respect to $\cM$ (see \Cref{def: gamma cheap}) and the merge plan $\cM$ is $0$-good (see \Cref{def: gamma good}).
    The value of $\cM$ equals $\frac{7}{6}$ times that value of $\cM_G$.
Together with \Cref{cor: canonical merge plan}, this implies $\val{\cM}=\frac{7}{6}\, \val{\cM_G}=\frac{7}{12}\,\tmst$.
\end{proof}

Next, we discuss how the local value of sets behaves under taking unions.
For general $X,\bar X\subseteq R$, the local value $\val{\cM, X\cup\bar X}$ might be as small as
\begin{align*}
    \max\lrb{\val{\cM, X},\,\val{\cM, \bar X}}\,.
\end{align*}
This is the case, for example, if for essentially all times $t$ the sets $X$ and $\bar X$ intersect the same parts of the partition $\cS^t$.
However, this can not be the case if the partition $\cS^t$ separates $X$ and $\bar X$ until one of them is entirely contained in a single part of $\cS^t$, as we prove next.

\begin{lemma}\label{lem: drop of X union X bar}
    Given a fixed merge plan $\cM=\lr{\cS^t}_{t\geq0}$, let $S\in\cS^{t^*}$ be a set that is active at time $t^*$ and consider two terminal sets $X\subseteq S$ and $\bar X\subseteq R\setminus S$. Then
    $$
        \val{\cM, X\cup\bar X}\ \geq\ \val{\cM, X}+\val{\cM, \bar X}+t^*\,.
    $$
\end{lemma}
\begin{proof}
    For $Y\in\lrb{X, \bar X, X\cup\bar X}$ let $f_Y\lr{t}\coloneqq\lrv{\lrb{S'\in\cS^{t}\,:\,S'\cap Y\neq\emptyset}}$.
    For every time $t<t^*$, any active set $S'\in\cS^{t}$ can intersect at most one of $X$ and $\bar X$ and thus $ f_{X\cup\bar X}\lr{t}=f_X\lr{t}+f_{\bar X}\lr{t}$, implying
    \begin{align}
       f_{X\cup\bar X}\lr{t}-1\ =\ \lr{f_X\lr{t}-1}+\lr{f_{\bar X}\lr{t}-1}+1\,.\label{ieq: num intersected active sets before t}
    \end{align}
    For $t\geq t^*$ exactly one set $S'\in\cS^{t}$ will intersect $X$ and thus 
    \begin{align}
        f_{X\cup\bar X}\lr{t}-1\ \geq\ f_{\bar X}\lr{t}-1\ =\ \lr{f_{X}\lr{t}-1}+\lr{f_{\bar X}\lr{t}-1}\,.\label{ieq: num intersected active sets after t}
    \end{align}
    Together \eqref{ieq: num intersected active sets before t} and \eqref{ieq: num intersected active sets after t} imply
    \begin{align*}
        \val{X\cup\bar X}
        &=\int_0^{\infty}f_{X\cup\bar X}\lr{t}-1\,dt\\
        &\geq\int_0^{\infty}\lr{f_{X}\lr{t}-1}+\lr{f_{\bar X}\lr{t}-1}\,dt\ + t^*\\
        &=\val{X}+\val{\bar X}+t^*\,.
    \end{align*}
\end{proof}

Finally, we prove a well-known fact about $\dropG{X}$, namely that it can only decrease, when we contract terminals:

\begin{observation}\label{obs:drop reduces under contraction}
    For a graph $G$ and two terminal sets $X, Y\subseteq R$, we have
    \begin{align*}
        \drop[G/Y]{X}\leq\dropG{X}\,.
    \end{align*}
\end{observation}
\begin{proof}
    Due to \Cref{lem: drop in terms of canonical merge plan,lem: Mu is preserved under contraction}, it suffices to prove that for any merge plan $\cM$ we have $\val{\cM/Y, X}\leq\val{\cM, X}$.
    Let $\cM=\lr{\cS^t}_{t\geq0}$ and $\cM/Y=\lr{\cS_Y^t}_{t\geq0}$. For each set of the partition $\cS_Y^t$ that $X$ intersects, it has to intersect at least one set of the partition $\cS^t$. Hence,
    \begin{align*}
        \val{\cM/Y, X}&=\int_0^{\infty}\lrv{\lrb{S\in\cS_Y^t\,:\,S\cap X\neq\emptyset}}-1\ dt\\
        &\leq\ \int_0^{\infty}\lrv{\lrb{S\in\cS^t\,:\,S\cap X\neq\emptyset}}-1\ dt\ =\ \val{\cM, X}.
    \end{align*}
\end{proof}

 \section{Preliminary Remarks on the Moat Growing Algorithm}\label{sec: assumptions}

The goal of this section is to describe the precise setting for which we analyze \Cref{algo:dual_growth_with_merge_plan}. To avoid tedious case distinctions we do two things:
First, we will not work with $\gamma$-good merge plans, but with \emph{strictly} $\gamma$-good merge plans (\Cref{def: strictly gamma good}).
Second, we will not run \Cref{algo:dual_growth_with_merge_plan} on the original instance. Instead, if we are given an instance $\cI$ and a strictly $\gamma$-good merge plan $\cM$ for $\cI$, we show that $\cM$ is feasible for an instance $\tilde\cI$ that arises from $\cI$ by subdivision of edges. Note that subdividing edges changes neither the value of \ref{eq:bcr} nor the cost of an optimum Steiner tree. 
Thus, the instance arising from subdivision is equivalent and, in particular, has the same integrality gap.

Before we introduce the precise theorem that we prove, we want to clarify the following: as we are interested in dual solutions of \ref{eq:bcr}, we always run \Cref{algo:dual_growth_with_merge_plan} on the bidirected graph. 
That is, we replace each undirected edge $\lrb{v, w}$ of the given undirected graph by two directed edges $\lr{v, w}$ and $\lr{w, v}$ to which we assign the same cost $c\lr{\lr{v, w}}=c\lr{\lr{w, v}}=c\lr{\lrb{v, w}}$. The step of bidirecting edges happens after we subdivided the undirected edges.

The only difference between \emph{strictly} $\gamma$-good merge plans and $\gamma$-good merge plans is that we want strict inequalities. Hence, we say that a set $X\subset R$ is \emph{strictly $\gamma$-expensive} with respect to $\cM$ if $\conn{X}>\frac{12-12\gamma}{7-5\gamma}\,\val{\cM, X}$, and define \emph{strictly} $\gamma$-good merge plans as follows:\begin{definition}[strictly $\gamma$-good]\label{def: strictly gamma good}
Let $\gamma \in [0,\frac15]$.
    We say that a merge plan $\cM$ is \emph{strictly} $\gamma$-good, if for all $x,y\in R$ with $x\neq y$
    \begin{align}
        \merge{\cM}{x, y}<\frac{7-5\gamma}{12}\,\dist{x, y},\label{ieq: strict general upper bound for gamma good}
    \end{align}
    and for every terminal set $X\subseteq R$ that is not strictly $\gamma$-expensive with respect to $\cM$, there are $x,y\in X$ with $x\neq y$ satisfying
    \begin{align}
        \merge{\cM}{x, y}<\frac{7-5\gamma}{18-14\gamma}\,\dist{x, y}\,.\label{ieq: strict stronger upper bound for gamma good}
    \end{align}
\end{definition}

We can approximate (the value of) a $\gamma$-good merge plan using strictly $\gamma$-good merge plans as follows:
Whenever we have a $\gamma$-good merge plan $\cM$, we can obtain a strictly $\gamma$-good merge plan $\tilde\cM$, by scaling down all merge times by a factor of $\lr{1-\varepsilon}$:

\begin{observation}\label{obs: gamma good to strictly gamma good}
    Let $\epsilon \in (0,1)$ and let $\cM$ be a $\gamma$-good merge plan. 
    Then we obtain a strictly $\gamma$-good merge plan $\tilde\cM=\big(\tilde\cS^t\big)_{t\geq0}$ with $\textnormal{value}\big(\tilde\cM\big)=\lr{1-\epsilon}\val{\cM}$ by defining $\tilde\cS^t\coloneqq\cS^{t/\lr{1-\epsilon}}$ for all $t\geq 0$.
\end{observation}
\begin{proof}
    Consider two terminals $x, y\in R$ with $x\neq y$. By definition of the merge plan $\tilde\cM$, we have $\merge{\tilde\cM}{x, y}=\lr{1-\epsilon}\merge{\cM}{x, y}$.
    For every $X\subseteq R$ the local value of $X$ with respect to $\tilde\cM$ is $\textnormal{value}\big(\tilde\cM, X\big)=\lr{1-\varepsilon}\val{\cM, X}$. In particular, $\textnormal{value}\big(\tilde\cM\big)=\lr{1-\epsilon}\,\val{\cM}$.
    
    Because our edge costs are positive, we have for all $x\neq y\in R$ that $\dist{x, y}>0$. Thus, if $\merge{\cM}{x, y}=0$, then also $\merge{\tilde\cM}{x, y}=0<\frac{7-5\gamma}{12}\,\dist{x, y}$. Otherwise, $\merge{\cM}{x, y}>0$ and we get
    \begin{align*}
        \merge{\tilde\cM}{x, y}\ =\ \lr{1-\epsilon}\,\merge{\cM}{x, y}\ <\ \merge{\cM}{x, y}\ \leq\ \frac{7-5\gamma}{12}\,\dist{x, y}\,,
    \end{align*}
    which shows the first property \eqref{ieq: strict general upper bound for gamma good} of \Cref{def: strictly gamma good}.
    
    For the second property we can argue in a similar way. Because $\cM$ is $\gamma$-good (\Cref{def: gamma good}), for each set $X\subseteq R$ one of the following holds: (i) $X$ is not $\gamma$-cheap with respect to $\cM$, or (ii) there are $x,y\in X$ with $x\neq y$ that where merged very early. More precisely, we have (i)
    \begin{align}
        \frac{12-12\gamma}{7-5\gamma}\,\textnormal{value}\big(\tilde\cM, X\big)\ <\ \frac{12-12\gamma}{7-5\gamma}\,\val{\cM, X}\ \leq\ \conn{X}\,,\label{ieq:not cheap}
    \end{align}
    or (ii)
    \begin{align}
        \merge{\tilde\cM}{x, y}\ <\ \merge{\cM}{x, y}\ \leq\ \frac{7-5\gamma}{18-14\gamma}\,\dist{x, y}\,.\label{ieq:early merge}
    \end{align}
    Hence, every set $X\subseteq R$ is strictly $\gamma$-expensive, or there are $x,y\in X$ that are merged strictly before time $\frac{7-5\gamma}{18-14\gamma}\,\dist{x, y}$.
    (Note that, in both cases, if $\val{\cM, X}$ respectively $\merge{\cM}{x, y}$ is zero, the first inequality in \eqref{ieq:not cheap} respectively \eqref{ieq:early merge} is not strict, but then the second inequality is strict instead.)
\end{proof}

The reason why we run \Cref{algo:dual_growth_with_merge_plan} on a subdivided instance $\tilde\cI$, and not on $\cI$ itself, is that it allows us to split edges into segments with the same properties and to easily refer to points ``in the middle of some edge'' by adding new vertices at these points.
To formally define the properties obtained by subdivision, we introduce some notation.

\begin{definition}
We say that a directed edge $e$ is \emph{tight} if $\ \sum_{U\subseteq V:e\in\delta^+\lr{U}}y_U=c\lr{e}$.    
\end{definition}

\begin{definition}
A set $X\subseteq R$ \emph{reaches} a vertex $v$ at time $t$, if there is an $X$-$v$ path $P$ such that all edges of $P$ are tight at time $t$. By $\atf{X}{v}$ we denote the first time $t$ for which $v$ is reachable from $X$.
\end{definition}

\begin{definition}\label{def: contribution}
We say that a set $S$ \emph{contributes} to an edge $e=\lr{v, w}$ at time $t$, if $S\in\cS^t$ and 
$\atf{S}{v} < t \leq \atf{S}{w}$.
\end{definition}
Note that if $S\in\cS^t$, we increase the dual variable $y_{U_S}$ at time $t$.
Moreover, for times $t$ between $\atf{S}{v}$ and $\atf{S}{w}$, we have $(v,w)\in \delta^+\lr{U_S}$.

There are two reasons, why we subdivide edges. 
First, we want to ensure that sets don't ``meet in the middle of an edge'' (see \ref{itm: reachability of both endpoints} of \Cref{def:well-subdivided} below). 
In particular, we want that every bifurcation vertex (as mentioned in \Cref{par: bifurcation}) is indeed a vertex and not ``some point in the middle of an edge''.
Second, it will be convenient to split edges into segments with the same properties (see \ref{itm: uniform edges} of \Cref{def:well-subdivided} below). 

Both of this is captured by the definition of a well-subdivided instance:

\begin{definition}[well-subdivided]\label{def:well-subdivided}
    We say that a Steiner tree instance $(G=(V,E),c,R)$ is \emph{well-subdivided} for a merge plan $\cM$, if \Cref{algo:dual_growth_with_merge_plan} with merge plan $\cM$ fulfills the following two properties:
 \begin{enumerate}[label=(\roman*)]
     \item\label{itm: reachability of both endpoints} 
     If at time $t$ both endpoints of an edge $e$ are reachable (by some terminals), then at least one orientation of $e$ is tight at time $t$.
     \item\label{itm: uniform edges} For every edge $\vec{e}$ of the bidirected graph with $\sum_{U\,:\,\vec{e}\in\delta^+\lr{U}}y_U>0$, there is a time interval $(t_1,t_2]$ and a set family $\cB \subseteq \cS^{t_2}$ such that 
   \begin{itemize}
       \item until (including) time $t_1$ no set $S\in \cS^t$ contributes to $\vec{e}$, and
       \item from time $t_1$ until (including) $t_2$, the sets contributing to $\vec{e}$ are exactly the sets in $\cB$, and
       \item at time $t_2$ the edge $\vec{e}$ is tight.
   \end{itemize}
 \end{enumerate}
\end{definition}

A natural strategy to obtain a well-subdivided instance would be to establish the properties for one edge after the other by a suitable subdivision of this edge.
For this approach to work, we need to ensure that a subdivision does not destroy the properties on already considered edges.
Thus, it would be desirable that \Cref{algo:dual_growth_with_merge_plan} behaves equivalently on both the original and the subdivided instance. Unfortunately this is not the case as the example in \Cref{fig:not_invariant_under_subdivision} shows.

However, it is possible to subdivide the instance such that any further subdivision of edges does not change the behavior of \Cref{algo:dual_growth_with_merge_plan}.
For this reason, we will subdivide our instance in two steps. In the first step, we will ensure that \ref{itm: reachability of both endpoints} of \Cref{def:well-subdivided} is fulfilled and that further subdivisions of our instance won't change the behavior of \Cref{algo:dual_growth_with_merge_plan}. The result of this first step is captured by the following lemma.

\begin{restatable}{lemma}{InvarianceSubdivision}
\label{lem:achieving_ivanriance_under_subdivision}
 Let $\cI =(G=(V, E), c, R)$ be an instance of Steiner Tree and let $\cM$ be a merge plan for $\cI$.
 Then there exists a graph $G'=\lr{V', E'}$ that arises from $G$ by subdividing each edge at most once, such that instance $\cI'=(G', c, R)$ satisfies \ref{itm: reachability of both endpoints} of \Cref{def:well-subdivided} and any further subdivisions of the edges of $G'$ leave $\atf{X}{v}$ unchanged for all $X\subseteq R$ and $v\in V'$.
\end{restatable}

\begin{figure}
    \centering
    
\newcommand{\UnsubdividedVertices}{
    \useasboundingbox (-2,-1) rectangle (7, 4);
\LARGE
    \node[terminal, label=below:{$s_1$}] (S1) at (-1, 0) {};
    \node[terminal, transparent] (S2a) at ( 2, 0) {};
    \node[terminal,minimum width=130pt, label=below:{$s_2$}] (A3) at ( 4, 0) {};
    \node[terminal, transparent] (S2b) at ( 6, 0) {};
    \node[steiner, label=above:{$v$}] (V) at (-1, 3) {};
    \node[steiner, label=above:{$x$}] (X) at ( 2, 3) {};
    \node[steiner, label=above:{$z$}] (Z) at ( 6, 3) {};
}

\newcommand{\SubdividedVertices}{
    \UnsubdividedVertices
    \node[steiner, label=above:{$w$}] (W) at (0.5, 3) {};
    \node[steiner, label=above:{$y$}] (Y) at ( 4, 3) {};
}

\newcommand{\UnsubdividedEdges}{
    \draw[diredge] (S1) -- (V);
    \draw[diredge] (S2a) -- (X);
    \draw[diredge] (S2b) -- (Z);
    \draw[diredge] (V) -- (X);
    \draw[diredge] (X) -- (Z);
    
    \draw[diredge] (V) -- (S1);
    \draw[diredge] (X) -- (S2a);
    \draw[diredge] (Z) -- (S2b);
    \draw[diredge] (X) -- (V);
    \draw[diredge] (Z) -- (X);
}

\newcommand{\SubdividedEdges}{
    \draw[diredge] (S1) -- (V);
    \draw[diredge] (S2a) -- (X);
    \draw[diredge] (S2b) -- (Z);
    \draw[diredge] (V) -- (W);
    \draw[diredge] (W) -- (X);
    \draw[diredge] (X) -- (Y);
    \draw[diredge] (Y) -- (Z);
    
    \draw[diredge] (V) -- (S1);
    \draw[diredge] (X) -- (S2a);
    \draw[diredge] (Z) -- (S2b);
    \draw[diredge] (X) -- (W);
    \draw[diredge] (W) -- (V);
    \draw[diredge] (Y) -- (X);
    \draw[diredge] (Z) -- (Y);
}

\newcommand{\UnsubdividedA}{
\begin{tikzpicture}[scale=1]
    \UnsubdividedVertices
    
    \draw[edge] (S1) -- node[left] {$18$} (V);
    \draw[edge] (S2a) -- node[left] {$18$} (X);
    \draw[edge] (S2b) -- node[right] {$18$} (Z);
    \draw[edge] (V) -- node[above] {$2$} (X);
    \draw[edge] (X) -- node[above] {$4$} (Z);
\end{tikzpicture}}

\newcommand{\SubdividedA}{
\begin{tikzpicture}[scale=1]
    \SubdividedVertices
    
    \draw[edge] (S1) -- node[left] {$18$} (V);
    \draw[edge] (S2a) -- node[left] {$18$} (X);
    \draw[edge] (S2b) -- node[right] {$18$} (Z);
    \draw[edge] (V) -- node[above] {$1$} (W);
    \draw[edge] (W) -- node[above] {$1$} (X);
    \draw[edge] (X) -- node[above] {$2$} (Y);
    \draw[edge] (Y) -- node[above] {$2$} (Z);
\end{tikzpicture}}

\newcommand{\UnsubdividedB}{
\begin{tikzpicture}[scale=1]
    \UnsubdividedVertices
    \UnsubdividedEdges
    \draw[edgedirprogress] (S2a) -- (X);
    \draw[edgedirprogress] (S2b) -- (Z);
    \draw[edgedirprogress] (X) -- (V);
    \draw[edgedirprogress, green!70!black] (S1) -- (V);
    \draw[edgedirprogress, green!70!black] (V) -- (X);
\end{tikzpicture}}

\newcommand{\SubdividedB}{
\begin{tikzpicture}[scale=1]
    \SubdividedVertices
    \SubdividedEdges
    \draw[edgedirprogress] (S2a) -- (X);
    \draw[edgedirprogress] (S2b) -- (Z);
    \draw[edgedirprogress] (X) -- (W);
    \draw[edgedirprogress] (W) -- (V);
    \draw[edgedirprogress] (X) -- (Y);
    \draw[edgedirprogress] (Z) -- (Y);
    \draw[edgedirprogress, green!70!black] (S1) -- (V);
    \draw[edgedirprogress, green!70!black] (V) -- (W);
    \draw[edgedirprogress, green!70!black] (W) -- (X);
\end{tikzpicture}}

\newcommand{\UnsubdividedC}{
\begin{tikzpicture}[scale=1]
    \UnsubdividedVertices
    \UnsubdividedEdges
    \draw[edgedirprogress] (S2a) -- (X);
    \draw[edgedirprogress] (S2b) -- (Z);
    \draw[edgedirprogress] (X) -- (V);
    \draw[edgedirprogress] (V) -- ($(V)!1/3.5!(S1)$);
    \draw[edgedirprogress, green!70!black] (S1) -- (V);
    \draw[edgedirprogress, green!70!black] (V) -- (X);
    \draw[edgedirprogress, green!70!black] (X) -- ($(X)!1/2!(Z)$);
    \draw[edgedirprogress, green!70!black] (X) -- ($(X)!1/3.5!(S2a)$);
\end{tikzpicture}}

\newcommand{\SubdividedC}{
\begin{tikzpicture}[scale=1]
    \SubdividedVertices
    \SubdividedEdges
    \draw[edgedirprogress] (S2a) -- (X);
    \draw[edgedirprogress] (S2b) -- (Z);
    \draw[edgedirprogress] (X) -- (W);
    \draw[edgedirprogress] (W) -- (V);
    \draw[edgedirprogress] (X) -- (Y);
    \draw[edgedirprogress] (Z) -- (Y);
    \draw[edgedirprogress] (V) -- ($(V)!1/3.5!(S1)$);
    \draw[edgedirprogress, green!70!black] (S1) -- (V);
    \draw[edgedirprogress, green!70!black] (V) -- (W);
    \draw[edgedirprogress, green!70!black] (W) -- (X);
    \draw[edgedirprogress, green!70!black] (Y) -- (Z);
    \draw[edgedirprogress, green!70!black] (X) -- ($(X)!1/3.5!(S2a)$);
\end{tikzpicture}}

\begin{tabular}{ccc}
  \resizebox{0.3\textwidth}{!}{\UnsubdividedA} 
  & \resizebox{0.3\textwidth}{!}{\UnsubdividedB} 
  & \resizebox{0.3\textwidth}{!}{\UnsubdividedC} 
  \\[0mm]
  original instance  & $t=20$ & $t=22$ \\[4mm]
   \resizebox{0.3\textwidth}{!}{\SubdividedA} 
  & \resizebox{0.3\textwidth}{!}{\SubdividedB} 
  & \resizebox{0.3\textwidth}{!}{\SubdividedC} 
  \\
  subdivided instance  & $t=20$ & $t=22$
\end{tabular}     \caption{An example showing that \Cref{algo:dual_growth_with_merge_plan} is in general not invariant under subdivision. In the instance shown in the first row, the terminal $s_2$ contributes to neither $\lr{x, z}$ nor to $\lr{z, x}$ because it reaches $x$ and $z$ at the same time ($t=18$). 
    In the second row, we see a subdivided version of the instance from the first row. That is, the edge $\lrb{x, z}$ of length two was replaced by two edges $\lrb{x, y}$ and $\lrb{y, z}$, each of length two. 
    Similarly, the edge $\lrb{v, x}$ is subdivided. 
    Because $s_2$ does not reach the newly introduced vertex $y$ at time $18$, it will contribute to both $\lr{x, y}$ and $\lr{z, y}$. 
    While this does not affect when $s_2$ can reach the original vertices, it has an impact on $s_1$. 
    In the subdivided instance, $s_1$ reaches $z$ at time $t=22$, which is not the case in the original instance.
    The subdivided instance satisfies \ref{itm: reachability of both endpoints} of \Cref{def:well-subdivided}. 
    In fact, $w$ and $y$ are precisely the vertices we would introduce to prove \Cref{lem:achieving_ivanriance_under_subdivision}. 
    The second property \ref{itm: uniform edges} of \Cref{def:well-subdivided} is not satisfied (in both instances) because to both edges $\lr{v, s_1}$ and $\lr{x, s_2}$ some terminal contributed, but the edges are not tight in the end.
    \label{fig:not_invariant_under_subdivision}}
\end{figure}
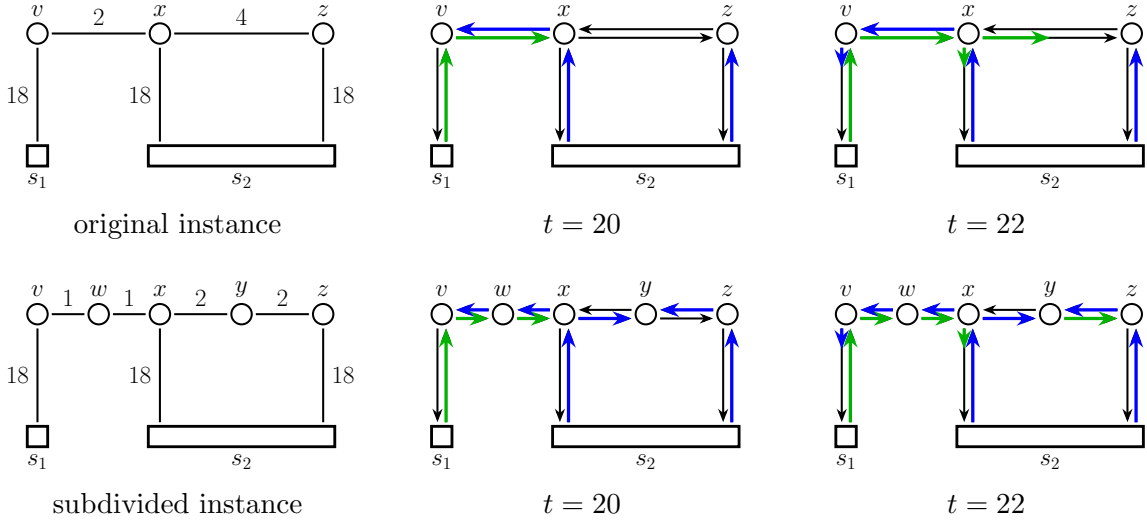

The idea to prove \Cref{lem:achieving_ivanriance_under_subdivision} is the following:
We run a \emph{continuous} version of \Cref{algo:dual_growth_with_merge_plan}, which corresponds to the limit of what we obtain when subdividing edges into shorter and shorter pieces. This continuous version is, by construction, invariant under subdivision of edges.
Afterwards, we detect for every original edge $\lrb{v, w}$, the point where the sets contributing to $\lr{v, w}$ and the sets contributing to $\lr{w, v}$ meet on $\lrb{v, w}$. This point is the only one where we have to subdivide the edge. 
In the example shown in \Cref{fig:not_invariant_under_subdivision}, the continuous version of \Cref{algo:dual_growth_with_merge_plan} behaves on the original instance as \Cref{algo:dual_growth_with_merge_plan} behaves on the subdivided instance. In particular, from time $t=18$ to $t=20$, the terminal $s_2$ would contribute to both $\lr{x, z}$ and $\lr{z, x}$. 
At time $t=20$, the terminal $s_2$ would meet itself at the midpoint of $\lrb{x, z}$. Thus, we subdivide $\lrb{x, z}$ by inserting $y$. Similarly, we subdivide $\lrb{v, x}$ at the point where $s_1$ and $s_2$ meet at time $t=19$, namely we insert the vertex $w$. The three other edges do not need to be subdivided to obtain \ref{itm: reachability of both endpoints} of \Cref{def:well-subdivided}.
For a detailed proof of \Cref{lem:achieving_ivanriance_under_subdivision}, we refer to \Cref{sec:invariance_under_subdivision}.

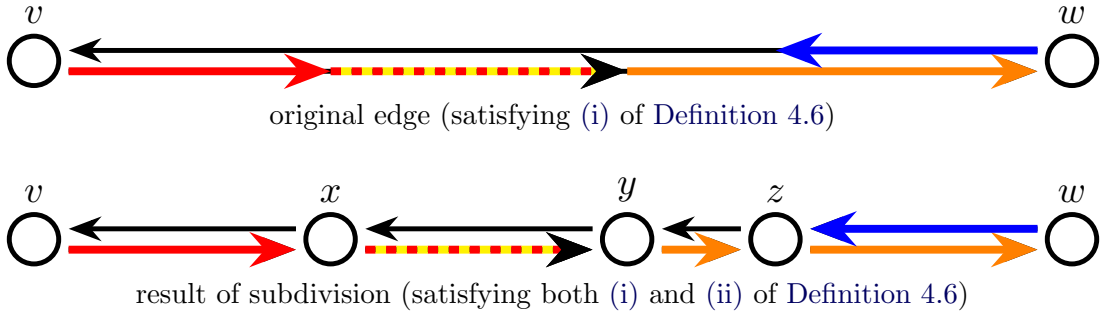
\begin{figure}
    \centering
    \newcommand{\OriginalEdge}{
\begin{tikzpicture}[scale=1.5]
    \node[steiner, label=above:{$v$}] (V) at ( 0, 0) {};
    \node[steiner, label=above:{$w$}] (W) at ( 7, 0) {};
    
    \draw[diredge, black] (V) -- (W);
    \draw[diredge, black] (W) -- (V);
    \draw[edgedirprogress, red] (V) -- ($(V)!2/7!(W)$);
    \draw[edgedirprogresshalfA, red] ($(V)!2/7!(W)$) -- ($(V)!4/7!(W)$);
    \draw[edgedirprogresshalfB, yellow] ($(V)!2/7!(W)$) -- ($(V)!4/7!(W)$);
    \draw[edgedirprogress, orange] ($(V)!4/7!(W)$) -- (W);
    
    \draw[edgedirprogress, blue] (W) -- ($(W)!2/7!(V)$);
\end{tikzpicture}
}

\newcommand{\SubdividedEdge}{
\begin{tikzpicture}[scale=1.5]
    \node[steiner, label=above:{$v$}] (V) at ( 0, 0) {};
    \node[steiner, label=above:{$x$}] (X) at ( 2, 0) {};
    \node[steiner, label=above:{$y$}] (Y) at ( 4, 0) {};
    \node[steiner, label=above:{$z$}] (Z) at ( 5, 0) {};
    \node[steiner, label=above:{$w$}] (W) at ( 7, 0) {};
    
    \draw[diredge, black] (V) -- (X);
    \draw[diredge, black] (X) -- (Y);
    \draw[diredge, black] (Y) -- (Z);
    \draw[diredge, black] (Z) -- (W);
    \draw[diredge, black] (X) -- (V);
    \draw[diredge, black] (Y) -- (X);
    \draw[diredge, black] (Z) -- (Y);
    \draw[diredge, black] (W) -- (Z);
    
    \draw[edgedirprogress, red] (V) -- (X);
    \draw[edgedirprogresshalfA, red] (X) -- (Y);
    \draw[edgedirprogresshalfB, yellow] (X) -- (Y);
    \draw[edgedirprogress, orange] (Y) -- (Z);
    \draw[edgedirprogress, orange] (Z) -- (W);
    \draw[edgedirprogress, blue] (W) -- (Z);
\end{tikzpicture}
}

\begin{tabular}{c}
    \resizebox{0.9\textwidth}{!}{\OriginalEdge} 
    \\[0mm]
    original edge (satisfying \ref{itm: reachability of both endpoints} of \Cref{def:well-subdivided}) \\[4mm]
    \resizebox{0.9\textwidth}{!}{\SubdividedEdge}
    \\
    result of subdivision (satisfying both \ref{itm: reachability of both endpoints} and \ref{itm: uniform edges} of \Cref{def:well-subdivided})
\end{tabular}     \caption{An edge $\{v,w\}$ on which, for different points in time, different families of sets contributed.
    The first time a set contributes to an orientation of $\lrb{v, w}$ is when the red terminal set reaches~$v$. While the red set still contributes to $\lr{v, w}$, eventually also the yellow set reaches $v$ and thus also contributes to $\lr{v, w}$. A bit later (still before $\lr{v, w}$ got tight), the red set and the yellow set are merged into an orange set, which then is the only set contributing to $\lr{v, w}$ until the edge becomes tight. Later, the blue terminal set reaches $w$ and starts contributing to the edge $\lr{w, v}$. Soon afterwards, the algorithm terminates before $\lr{w, v}$ got tight. 
    Below the original edge, it is shown how we subdivide this edge. If the edge $\{v,w\}$ is part of an instance that is invariant under subdivision (as we obtain it from \Cref{lem:achieving_ivanriance_under_subdivision}), then subdividing at each position where the contributing sets change will result in an instance satisfying \ref{itm: uniform edges} of \Cref{def:well-subdivided}.
    \label{fig:uniform_edge_assumption}
    }
\end{figure}

Given an instance $\cI'$ as in \Cref{lem:achieving_ivanriance_under_subdivision} (in particular, an instance for which further subdivisions leave $\atf{X}{v}$ unchanged), we can easily obtain a well-subdivided instance $\tilde\cI$ by further subdividing the edges as we explain next.

By \Cref{lem:achieving_ivanriance_under_subdivision}, the times $\atf{S}{v}$ at which active sets $S$ can reach vertices $v$ don't change, if we further subdivide edges. 
Hence, also the sets $U_S$ of vertices that are reachable from $S$ remain unchanged (except that it might contain newly introduced vertices). 
Thus, if we subdivide an edge~$e$, we cannot destroy properties \ref{itm: reachability of both endpoints} and \ref{itm: uniform edges} of \Cref{def:well-subdivided} on other edges where they were already established.

Therefore, it suffices to show that we can subdivide a single edge such that the edges arising from the subdivision satisfy both properties of \Cref{def:well-subdivided}.
This can be achieved by subdividing the edge at each position, where (for one of the orientations) the family of sets contributing to this edge changes. See \Cref{fig:uniform_edge_assumption} for an example.

Hence, for every instance $\cI$ of the Steiner tree problem and merge plan $\cM$, there is an instance $\tilde\cI$ that arises from $\cI$ by subdividing edges, such that $\tilde\cI$ is well-subdivided for merge plan $\cM$.

In \Cref{sec: analysis of moat growing} we will prove the following:

\begin{restatable}{theorem}{thmMergePlanFeasible}\label{thm: gamma good merge plan feasible}
    Let $\gamma\in [0,\frac15]$.  Let $\cI$ a Steiner tree instance and let $\cM$ be a strictly $\gamma$-good merge plan for $\cI$ such that $\cI$ is well-subdivided for $\cM$. Then $\cM$ is feasible for $\cI$.
\end{restatable}

Before proving \Cref{thm: gamma good merge plan feasible}, we show that it indeed implies \Cref{cor: gamma good merge plan gives large dual}:

\begin{proof}[Proof of \Cref{cor: gamma good merge plan gives large dual}]
    Let $\cI$ be a Steiner tree instance and let $\cM$ be a $\gamma$-good merge plan.
    We show that for every $\varepsilon>0$ we have $\bcr\lr\cI\geq\lr{1-\varepsilon}\val{\cM}$.
    The merge plan $\tilde\cM$ with $\merge{\tilde\cM}{x, y}=\lr{1-\varepsilon}\merge{\cM}{x, y}$ for all $x,y\in R$ is strictly $\gamma$-good for $\cI$ (by \Cref{obs: gamma good to strictly gamma good}). 
    By subdividing the edges of $\cI$ as described above, we obtain an instance $\tilde\cI$ that is well-subdivided for $\tilde\cM$. 
    Subdividing edges changes neither distances between terminals, nor the cost of connecting a given terminal set.
    Hence, $\tilde\cM$ is strictly $\gamma$-good not only for $\cI$ but also for $\tilde\cI$. 
    Thus, by \Cref{thm: gamma good merge plan feasible}, $\tilde\cM$ is feasible for $\tilde\cI$. 
    Because subdividing edges leaves the value of $\bcr$ unchanged, we conclude
    \begin{align*}
        \bcr\lr{\cI}\ =\ \bcr\big(\tilde\cI\big)\ \geq\ \textnormal{value}\big(\tilde\cM\big)\ =\ \lr{1-\varepsilon}\,\val{\cM}\,.
    \end{align*}
    Hence, we have $\bcr\lr{\cI} \geq \lr{1-\varepsilon}\,\val{\cM}$ for all $\epsilon > 0$ and thus $\bcr\lr{\cI}\geq\val{\cM}$.
\end{proof} 

\section{Feasibility of Strictly $\gamma$-Good Merge Plans}\label{sec: analysis of moat growing}

The goal of this section is to prove \Cref{thm: gamma good merge plan feasible}.
Thus, throughout this section, we fix an instance $\cI$, a merge plan $\cM$, and $\gamma\in\lre{0, \frac15}$ such that $\cI$ is well-subdivided and $\cM$ is strictly $\gamma$-good. Our goal is to prove that $\cM$ is feasible for $\cI$, namely that whenever we increase a dual variable $y_{U_S}$ during \Cref{algo:dual_growth_with_merge_plan}, the set $U_S$ contains the root $r$ only if $r\in S$.

To simplify notation, we abbreviate $\merge{}{x, y}\coloneqq\merge{\cM}{x, y}$ for $x, y\in R$ and $\val{X}\coloneqq\val{\cM, X}$ for $X\subseteq R$.

This section is structured as follows.
In \Cref{sec: tight active safe paths} we define increasingly restrictive classes of paths, namely tight paths, active paths, and safe paths. Safe paths will play a crucial role in our analysis: every path and every component we consider in the analysis will be composed of (parts of) safe paths. In \Cref{sec: meeting points and length of safe path}, we formalize the notion of a set ``meeting'' another set. This will not immediately give the bifurcation vertices mentioned in \Cref{par: bifurcation}, but it will play the main role in (implicitly) finding such bifurcations. 
Furthermore, we bound the length of (parts of) safe paths.
In particular, we show that the example from \Cref{fig:zig_zag_example}  is indeed a worst-case example.
\Cref{sec: introducing invariant} formally introduces our main invariant, and outlines how we use this invariant to prove that $\cM$ is a feasible merge plan.
The latter is done in two main steps: First, we prove that if at time $t$ a set $S\in\cS$ reaches a vertex $v$ via a safe path and strictly before time $t$ all active paths were safe, then the invariant holds (\Cref{sec: invariant holds}).
Second, we show that indeed all active paths are safe, implying that the invariant is satisfied throughout the entire algorithm (\Cref{sec: all active paths are safe}).

\subsection{Tight, Active, and Safe Paths}\label{sec: tight active safe paths}

In this section, we introduce several basic definitions, including in particular the notion of $S$-tight, $S$-active, and $S$-safe paths. The definition of $S$-tight paths (\Cref{def: S tight path}) and a proof of their existence (\Cref{lem: tight paths always exist}) were already given in \cite{BCRlessthan2}.

We say that a set $S$ is \emph{active at time $t$} if $S\in \cS^t$.
By $\cS=\bigcup_{t\geq0}\cS^t$ we denote all sets that are active at some time during the algorithm.

\begin{definition}[activation time, deactivation time]
    For a set $S\in\cS$, we denote its activation time and its deactivation time by
    \begin{align}
        a^S=\inf\lrb{t\geq0\,:\,S\in\cS^t}\qquad\text{and}\qquad d^S=\max\lrb{t\geq0\,:\,S\in\cS^t}\,,\label{ieq: def of aS and dS}
    \end{align}
    respectively.
\end{definition}
Note that by \Cref{def: merge plan}, the infimum in \eqref{ieq: def of aS and dS} is a minimum if and only if $S$ is a singleton. Also by \Cref{def: merge plan}, the partitions $\cS^t$ are refinements of each other. Thus $\cS$ is a laminar family. That is, for any two sets $S, S'\in\cS$, either $S$ and $S'$ are disjoint, or one is a subset of the other.

Recall that a set $X\subseteq R$ \emph{reaches} a vertex $v$ at time $t$, if there is an $X$-$v$ path for which all of its edges are tight at time $t$. 
Further, recall that $\atf{X}{v}$ denotes the first time at which $X$ reaches~$v$. 
For a set $S\in\cS$, we are mainly interested in the vertices it can reach before it is deactivated, which motivates the following definition:
\begin{definition}[actively reaching]
    We say that a set $S\in\cS$ \emph{actively reaches} a vertex $v\in V$ if $S$ is active at time $\atf{S}{v}$.
\end{definition}

If a set $S$ (actively) reaches a vertex $v$, there might be multiple $S$-$v$ paths consisting only of tight edges. 
For the analysis, we are interested in paths that represent how $S$ built its way towards~$v$. 
This is captured by the concept of $S$-tight paths, which was introduced in \cite{BCRlessthan2}.
\begin{definition}[$S$-tight, $S$-active]\label{def: S tight path}
    Given a set $S\in\cS$, we call an edge $e=\lr{v, w}$ \emph{$S$-tight}, if $\atf{S}{v}\leq\atf{S}{w}$ and $e$ is tight at time $\atf{S}{w}$.
    Similarly, a path is called \emph{$S$-tight} if it starts in $S$ and all of its edges are $S$-tight.
    An $S$-tight $S$-$v$ path is called \emph{$S$-active}
    if $S$ actively reaches $v$.
\end{definition}

For a path $P$ visiting vertices $a$ and $b$ in this order, we write $P_{[a,b]}$ to denote the $a$-$b$ subpath of~$P$.
For an $S$-tight path $P$ visiting a vertex $v$, we also write $P_{[S,v]}$ to denote the prefix of $P$ until the vertex $v$.
In other words, if $s\in S$ is the first vertex of $P$, then $P_{[S,v]}=P_{[s,v]}$.

A direct consequence of \Cref{def: S tight path} is the following observation.

\begin{observation}\label{obs: edges of S tight path are tight early}
    Let $S\in\cS$, $v\in V$, $P$ an $S$-tight path to $v$, and $w$ a vertex on $P$. 
    Then all edges of $P_{\lre{S, w}}$ are tight at time $\atf{S}{w}$.
\end{observation}
\begin{proof}
    Let $v_1, \dots, v_k=w$ be the vertices of $P_{\lre{S, w}}$. By the $S$-tightness of the path $P$, we have
    \begin{align*}
        \atf{S}{v_1}\leq\atf{S}{v_2}\leq\dots\leq\atf{S}{w}\,.
    \end{align*}
    Moreover, by \Cref{def: S tight path}, each edge $\lr{v_i, v_{i+1}}$ of $P_{[S,w]}$ is tight at time $\atf{S}{v_{i+1}}$. 
    Hence, all edges of $P_{\lre{S, w}}$ are tight at time $\atf{S}{w}$.
\end{proof}

Another useful fact about $S$-tight paths, which was already proven by \citeauthor{BCRlessthan2}, is that they always exist if a set $S$ reaches a vertex $v$.
\begin{lemma}[\cite{BCRlessthan2}]
\label{lem: tight paths always exist}
    If a set $S\in\cS$ reaches a vertex $v\in V$, there exists an $S$-tight $S$-$v$ path.
\end{lemma}
\begin{proof}
    We prove the lemma by induction over time $t$.
    Note that it suffices to consider a finite number of points in time, namely, those at which an edge becomes tight.
    For $t=0$, the set $S$ can only reach vertices contained in $S$. 
    Thus, in this case, the path containing no edges is an $S$-tight path. 
    For $t>0$, consider any $S$-$v$ path $P$ whose edges are tight at time $t$. 
    Let $w$ be the last vertex of $P$ for which $\atf{S}{w}<t$. 
    Because $P$ is $S$-tight and by the choice of the vertex $w$, the suffix $P_{\lre{w, v}}$ consists of one edge along which $\atf{S}{\,\cdot}$ is increasing, followed by (zero or more) edges along which $\atf{S}{\,\cdot}=t$ stays constant. 
    Thus, all edges of $P_{\lre{w, v}}$ are $S$-tight. 
    By the induction hypothesis, there is an $S$-tight $S$-$w$ path $Q$. 
    Combining $Q$ with $P_{\lre{w, v}}$, yields an $S$-tight $S$-$v$ path.
\end{proof}

By \Cref{def: S tight path}, every prefix of an $S$-tight path is again $S$-tight. 
However, a prefix $P_{\lre{S, w}}$ of an $S$-active path $P$ might not be $S$-active, because $S$ might only be activated after it already reached~$w$. 
Nevertheless, $P_{\lre{S, w}}$ is still an active path, namely an $S'$-active path for a suitable subset $S'\subseteq S$.
Before we prove this, we need a small observation regarding the first time a set $X\subseteq R$ can reach a vertex $v$. 

\begin{observation}\label{obs: larger sets arrive no later}
    Let $v\in V$ and $\emptyset\neq X'\subseteq X\subseteq R$ then $\atf{X}{v}\leq\atf{X'}{v}$.
\end{observation}
\begin{proof}
    Whenever a set $X\subseteq R$ can reach a vertex $v$, then by definition of reachability also all super sets of $X$ can reach $v$. 
\end{proof}

\begin{lemma}\label{lem: active reaching subset}
    For a set $S\in\cS$, let $v\in V$ be a vertex that is reachable from $S$ at time $\atf{S}{v}\leq d^S$, and let $P$ be an $S$-tight path to $v$, then there exists a unique subset $S'\in\cS$ of $S$ such that $P$ is $S'$-active. This subset satisfies $\atf{S'}{v}=\atf{S}{v}$.
\end{lemma}
\begin{proof}
    Let $s\in S$ be the start vertex of $P$. Let $S' \in \cS$ be the unique set that is active at time $\atf{S}{v}$ and contains $s$ (note that this is the only possible candidate for $S'$).
    Because $\atf{S}{v}\leq d^S$ and $S\cap S' \neq \emptyset$, we have $S'\subseteq S$.

    Consider a vertex $w$ visited by $P$.
    Then, by \Cref{obs: edges of S tight path are tight early}, all edges of the path $P_{[s,w]}$ are tight at time  $\atf{S}{w}$, implying $\atf{\{s\}}{w}\leq\atf{S}{w}$. 
    Moreover, by \Cref{obs: larger sets arrive no later}, we have $\atf{S}{w}\leq\atf{S'}{w}\leq\atf{\{s\}}{w}$. 
    We conclude
    $\atf{S}{w}=\atf{S'}{w}=\atf{\{s\}}{w}$,
    for every vertex $w$ visited by $P$.
    This implies, that $P$ is $S'$-tight and $S'$ actively reaches $v$ (because $P$ is $S$-tight and $S'$ is active at time $\atf{S}{v}$). 
    Hence, $P$ is an $S'$-active path.
\end{proof}

Even though \citeauthor{BCRlessthan2} did not define $S$-active paths, they proved that for their particular choice of the merge plan $\cM$, every $S$-active path is short, which then implies that no set can actively reach the root (see \cite[Lemma 17]{BCRlessthan2}). 
We will follow a different proof strategy and won't bound the length of every $S$-active path.

Instead, we will bound the distance from a set $S$ to a vertex $v$ that it (actively) reaches as follows.
If $S$ reaches $v$ without interacting with any disjoint set along the way, then a simple inductive argument will be sufficient. 
If $S$ did not reach $v$ on its own, but interacted with a set $\bar S$ disjoint from $S$ along the way, we will identify a bifurcation vertex $m$ such that both $S$ and $\bar S$ can reach $m$ without the help of other terminals (that is without interacting with disjoint sets along the way). 
(Recall the discussion in \Cref{par: proving invariant} and \Cref{fig:component_invariant_two_sets}.)
By induction, we obtain a bound on the distance from $S$ to $m$. 
Furthermore, we will bound the distance from the bifurcation vertex $m$ to the vertex~$v$. 

The example in \Cref{fig:bifurcation_example} shows that the path from $m$ to $v$ might not be contained in a single $\tilde{S}$-active path for any terminal set $\tilde{S}$. (In this example, there is an $m$-$v$ path consisting of tight edges, but neither $\atf{s_1}{\,\cdot}$ nor $\atf{s_2}{\,\cdot}$ nor $\atf{\lrb{s_1, s_2}}{\,\cdot}$ is monotone along this $m$-$v$ path, which implies that it is not contained in any $\tilde{S}$-tight path for any terminal set $\tilde{S}$.) 
Nevertheless, the $m$-$v$ paths we consider will always be composed of certain suffixes of $\tilde{S}$-tight paths for some terminal sets $\tilde{S}$. 
Hence, it will be sufficient to analyze the length of such suffixes.

As the examples in \Cref{fig:k_interaction_example} show, (suffixes of) tight paths can be arbitrarily long in general. 
What both examples have in common is that more than two disjoint sets interact with each other on the problematic path. 
Thus, we will restrict our analysis to \emph{safe} paths, which we define next. Intuitively safe paths are those on which at most two disjoint sets interact with each other.
Eventually, we will show that (for strictly $\gamma$-good merge plans) every active path is safe (see \Cref{sec: all active paths are safe}), but for now, we define safe paths as the following (possibly more restrictive) class of paths.

\begin{definition}[safe paths]\label{def: safe path}
    Consider a directed edge $e=\lr{v, w}$ and let $t$ be the time at which $e$ became tight. For a set $S\in\cS$, we say that $e$ is $S$-safe, if until time $t$ at most two sets contributed to $e$ or $\eback=\lr{w, v}$ and at most one of these sets is a subset of $R\setminus S$.
    We call an $S$-active path \emph{$S$-safe} if all of its edges are safe. 
\end{definition}

The definition of safe paths is chosen such that the problematic paths in the examples from \Cref{fig:k_interaction_example} are not safe.
In \Cref{lem: length from m_p to v}, we will prove that along an $S$-safe path $P$ to the vertex~$v$, the set $S$ travels (amortized) at most with speed $3$. More precisely, we will show that if $m_P$ is the first vertex on $P$ where $S$ ``met'' another set $\bar S\subseteq R\setminus S$ (see \Cref{def: meeting point} below), then the length of $P_{\lre{m_P, v}}$ is at most $3\lr{\atf{S}{v}-\atf{S}{m_P}}$.
In particular, this implies that \Cref{fig:zig_zag_example} is indeed a worst-case example.

\subsection{Meeting Points and the Length of Safe Paths}\label{sec: meeting points and length of safe path}

Before we analyze the length of (parts of) safe paths, we define the meeting point of an $S$-active path $P$, that is the first vertex on $P$ where $R\setminus S$ arrives no later than $S$ does:

\begin{definition}[Meeting point]\label{def: meeting point}
    For an $S$-active path $P$ the \emph{meeting point} $m_P$ is the first vertex of $P$ with $\atf{R\setminus S}{m_P}\leq \atf{S}{m_P}$.
    If no such vertex exists, we say that $P$ has no meeting point.
    If $m_P$ is not the last vertex of $P$, we call it a \emph{proper meeting point}.
\end{definition}

\Cref{fig: examples meeting point} shows the meeting points of two different paths.
While the definition itself allows that $R\setminus S$ arrives at $m_P$ even \emph{strictly} earlier than $S$ does, this will not be the case for well-subdivided instances (see \Cref{lem: proper meeting point is earlier} below).

We are mainly interested in proper meeting points, because only if $P$ has a proper meeting point, the set $S$ might have gotten help from disjoint terminal sets to reach $v$, that is, such sets disjoint from $S$ could have contributed to edges of $P$.
If the path $P$ has a proper meeting point, then $S$ can reach the proper meeting point $m_P$ strictly earlier than $v$. 
Before we prove this fact (\Cref{lem: proper meeting point is earlier}), we need an auxiliary statement that captures a useful property of well-subdivided instances:

Consider a directed edge $(w,v)$ and a terminal set $X$ (not necessarily contained in $\cS$) such that no other terminal reaches the vertex $v$ strictly earlier than $X$ reaches $v$.
We show that if the edge $(w,v)$ is not tight at the time when $X$ reaches $v$, then $X$ reaches $w$ strictly later than it reaches $v$, and moreover, no other terminal reaches $w$ strictly earlier than $X$ reaches~$w$.

\begin{lemma}\label{lem:first_implies_first_at_neighbor_or_tight_incoming_edge}
    Let $\lr{w, v}$ be a directed edge, and $X\subseteq R$ a set with $\atf{X}{v}=\atf{R}{v}$. If $\lr{w, v}$ is not tight at time $\atf{X}{v}$, then $\atf{X}{w}=\atf{R}{w}>\atf{R}{v}=\atf{X}{v}$ (or the algorithm terminates at time $\atf{R}{v}$).
\end{lemma}
\begin{proof}
    Let $\lr{w, v}$ be an edge that is not tight at time $\atf{R}{v}$, and let $X\subseteq R$ be a terminal set reaching $v$ at time $\atf{R}{v}$.

    Before time $\atf{R}{v}$, no set contributed to $\lr{v, w}$, implying that $\lr{v, w}$ is not tight at that time. 
    Recall that also $\lr{w, v}$ is not tight at time $\atf{R}{v}$. 
    Hence, by \ref{itm: reachability of both endpoints} of \Cref{def:well-subdivided}, the vertex $w$ is not reachable at time $\atf{R}{v}$ and will not be reachable until one of the edges $(v,w)$ and $(w,v)$ is tight.
    Because the edge $(w,v)$ can only become tight strictly after some terminal reached $w$, the vertex $w$ will not be reachable from any terminal until the edge $\lr{v, w}$ becomes tight. 

    Thus, the time at which the edge $\lr{v, w}$ becomes tight is the first time at which any terminal can reach $w$. 
    The set $X$ can then reach $w$ via $\lr{v, w}$. 
    Hence,
    \begin{align*}
        \atf{X}{w}=\atf{R}{w}>\atf{R}{v}=\atf{X}{v}\,,
    \end{align*}
    unless the edge $\lr{v, w}$ never becomes tight.
    We claim that in the latter case, the algorithm terminates as time $\atf{R}{v}$.
    Indeed, if the algorithm does not terminate at time $\atf{R}{v}$, there is at least one set $S\in\cS$ that is active immediately after time $\atf{R}{v}$, and that can reach $v$. 
    Because this set cannot yet reach $w$ (no terminal can reach $w$ at time $\atf{R}{v}$), this set $S$ will contribute to $\lr{v, w}$. Hence, \ref{itm: uniform edges} of \Cref{def:well-subdivided} implies that $\lr{v, w}$ will eventually become tight.
\end{proof}

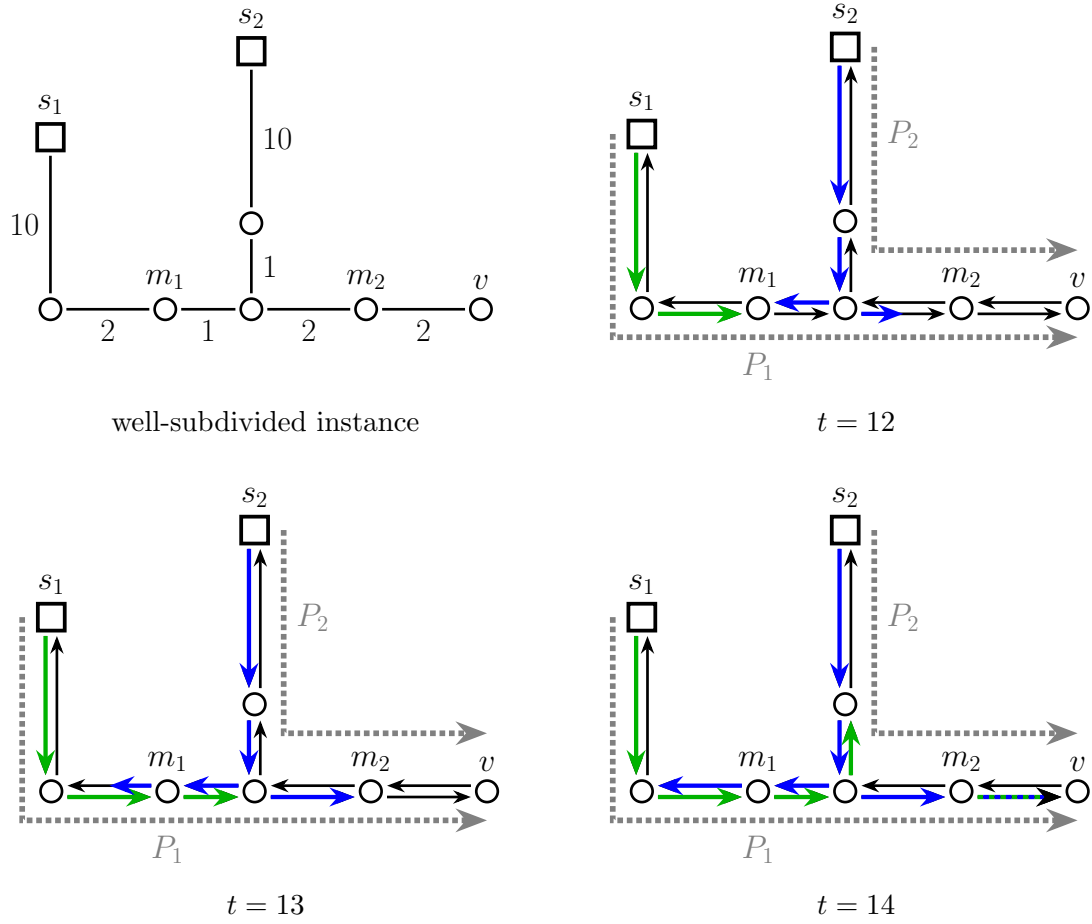
\begin{figure}
    \centering
    \newcommand{\VerticesMeetingPoint}
{
    \LARGE
    \useasboundingbox (-4.5,-1.5) rectangle (5,5.5);
    \node[terminal, label=above:{$s_1$}] (S1) at (-3.5, 3) {};
    \node[terminal, label=above:{$s_2$}] (S2) at ( 0, 4.5) {};
    \node[smallsteiner] (A0) at (-3.5, 0) {};
    \node[smallsteiner] (A1) at ( 0, 0) {};
    \node[smallsteiner, label=above:{$m_1$}] (M1) at (-1.5, 0) {};
    \node[smallsteiner, label=above:{$v$}] (A2) at ( 4, 0) {};
    \node[smallsteiner, label=above:{$m_2$}] (M2) at ( 2, 0) {};
    \node[smallsteiner] (A4) at ( 0,1.5) {};
}

\newcommand{\DiredgesMeetingPoint}
{
    \draw[diredge] (S1) -- (A0);
    \draw[diredge] (A0)  -- (S1);
    \draw[diredge] (A0) -- (M1);
    \draw[diredge] (M1)  -- (A0);
    \draw[diredge] (A1) -- (M1);
    \draw[diredge] (M1)  -- (A1);
    \draw[diredge] (S2) -- (A4);
    \draw[diredge] (A4) -- (S2);
    \draw[diredge] (A4) -- (A1);
    \draw[diredge] (A1) -- (A4);
    \draw[diredge] (A1) -- (M2);
    \draw[diredge] (M2) -- (A1);
    \draw[diredge] (M2) -- (A2);
    \draw[diredge] (A2) -- (M2);
}

\newcommand{\MeetingPointA}{
\begin{tikzpicture}[scale=1.08]
    \VerticesMeetingPoint
    \draw[edge] (S1) -- node[left] {$10$} (A0);
    \draw[edge] (A0) -- node[below] {$2$} (M1);
    \draw[edge] (M1) -- node[below] {$1$} (A1);
    \draw[edge] (S2) -- node[right] {$10$} (A4);
    \draw[edge] (A4) -- node[right] {$1$} (A1);
    \draw[edge] (A1) -- node[below] {$2$} (M2);
    \draw[edge] (M2) -- node[below] {$2$} (A2);
\end{tikzpicture}
}

\newcommand{\drawPaths}{

    \draw[line width = 3pt, gray, dotted, ->, >={Stealth}] ($(S1)-(.5,0)$) -- ($(A0)-(.5,.5)$) -- ($(A2)-(0,.5)$);
    \node[text = gray] at (-1.5,-1) {$P_1$};
    
    \draw[line width = 3pt, gray, dotted, ->, >={Stealth}] ($(S2)+(.5,0)$) -- ($(A1)+(.5,1)$) -- ($(A2)+(0,1)$);
    \node[text = gray] at ( 1, 3) {$P_2$};
}

\newcommand{\MeetingPointB}{
\begin{tikzpicture}[scale=1.08]
    \VerticesMeetingPoint
    \DiredgesMeetingPoint

    \draw[edgedirprogress, green!70!black] (S1) -- (A0);
    \draw[edgedirprogress, green!70!black] (A0) -- (M1);
    \draw[edgedirprogress] (S2) -- (A4);
    \draw[edgedirprogress] (A4) -- (A1);
    \draw[edgedirprogress] (A1) -- (M1);
    \draw[edgedirprogress] (A1) -- ($(A1)!1/2!(M2)$);

    \drawPaths
\end{tikzpicture}}

\newcommand{\MeetingPointC}{
\begin{tikzpicture}[scale=1.08]
    \VerticesMeetingPoint
    \DiredgesMeetingPoint
    
    \draw[edgedirprogress, green!70!black] (S1) -- (A0);
    \draw[edgedirprogress, green!70!black] (A0) -- (M1);
    \draw[edgedirprogress, green!70!black] (M1) -- (A1);
    \draw[edgedirprogress] (S2) -- (A4);
    \draw[edgedirprogress] (A4) -- (A1);
    \draw[edgedirprogress] (A1) -- (M1);
    \draw[edgedirprogress] (M1) -- ($(M1)!1/2!(A0)$);
    \draw[edgedirprogress] (A1) -- (M2);
    
    \drawPaths
\end{tikzpicture}}

\newcommand{\MeetingPointD}{
\begin{tikzpicture}[scale=1.08]
    \VerticesMeetingPoint
    \DiredgesMeetingPoint
    
    \draw[edgedirprogress, green!70!black] (S1) -- (A0);
    \draw[edgedirprogress, green!70!black] (A0) -- (M1);
    \draw[edgedirprogress, green!70!black] (M1) -- (A1);
    \draw[edgedirprogress, green!70!black] (A1) -- (A4);
    \draw[edgedirprogresshalfA, green!70!black] (M2) -- (A2);
    \draw[edgedirprogresshalfB] (M2) -- (A2);
    \draw[edgedirprogress] (S2) -- (A4);
    \draw[edgedirprogress] (A4) -- (A1);
    \draw[edgedirprogress] (A1) -- (M1);
    \draw[edgedirprogress] (M1) -- (A0);
    \draw[edgedirprogress] (A1) -- (M2);
    
    \drawPaths
\end{tikzpicture}}

\begin{tabular}{cc}
  \resizebox{0.45\textwidth}{!}{\MeetingPointA} 
  & \resizebox{0.45\textwidth}{!}{\MeetingPointB}
  \\[0mm]
well-subdivided instance  & $t=12$ \\[4mm]
   \resizebox{0.45\textwidth}{!}{\MeetingPointC} 
  & \resizebox{0.45\textwidth}{!}{\MeetingPointD} 
  \\
   $t=13$ & $t=14$
\end{tabular}     \caption{A well-subdivided instance in which both terminals $S_1$ and $s_2$ can reach $v$ at time $14$ via the $s_1$-active path $P_1$ and $s_2$-active path $P_2$, respectively. The first vertex of $P_1$ at which $s_2$ arrives no later than $s_1$ does is the meeting point of $m_1$ of $P_1$. The meeting point of $P_2$ is $m_2$. Because the instance is well-subdivided, $s_1$ and $s_2$ arrive at $m_1$ simultaneously at time $t=12$ (and both arrive at $m_2$ simultaneously at time $t=13$). Note that while both $s_1$ and $s_2$ can reach $m_1$ without the help of the other terminal, the same does not hold for $s_1$ and $m_2$.\label{fig: examples meeting point}}
\end{figure}

\begin{lemma}\label{lem: proper meeting point is earlier}
    Let $S\in\cS$, $v\in V$ and consider an $S$-active path $P$ to $v$. If $P$ has a proper meeting point $m_P$, then
    $$
        \atf{S}{m_P}=\atf{R}{m_P}<\atf{S}{v}\,.
    $$
\end{lemma}
\begin{proof}
    We start by proving $\atf{R}{m_P}<\atf{S}{v}$. By assumption $m_P$ is a proper meeting point, that is $m_P \neq v$. Hence, $P$ contains an outgoing edge of $m_P$. By \Cref{obs: edges of S tight path are tight early}, this edge is tight at time $\atf{S}{v}$. 
    Thus, any set contributing to it must have reached the meeting point $m_P$ strictly before time $\atf{S}{v}$, which implies $\atf{R}{m_P}<\atf{S}{v}$.

    Next, we show $\atf{S}{m_P}=\atf{R}{m_P}$. By the definition of the meeting point (\Cref{def: meeting point}), we have $\atf{R\setminus S}{m_P}\leq\atf{S}{m_P}$. 
    Thus, assume for the sake of deriving a contradiction that $\atf{R\setminus S}{m_P}<\atf{S}{m_P}$. 
    Then $m_P$ cannot be the start vertex of $P$ (because the start vertex is contained in $S$ and is thus reached by $S$ at time $0$). 
    Consider the predecessor $w$ of $m_P$ on $P$ and let $e=\lr{w, m_P}$.
If the edge $e$ is tight at time $\atf{R\setminus S}{m_P}<\atf{S}{m_P}$, then no subset of $S$ contributed to $e$. 
    In particular, this implies that $R\setminus S$ reaches $w$ before $S$ does.
    Otherwise, $e$ is not tight at time $\atf{R\setminus S}{m_P}$ and
    then by \Cref{lem:first_implies_first_at_neighbor_or_tight_incoming_edge} applied for $X=R\setminus S$, the set $R\setminus S$ can reach $w$ at time $\atf{R}{w}$. 
    In particular, $R\setminus S$ can reach $w$ no later than $S$ can reach $w$. 
    Hence, in both cases we would have rather chosen $w$ instead of $m_P$ as the meeting point (according to \Cref{def: meeting point}), which leads to the desired contradiction.
\end{proof}

By \Cref{def: meeting point}, the set $R\setminus S$ can reach the meeting point $m_P$ whenever $S$ can reach it. In fact, this not only holds for $m_P$, but also for every vertex of $P$ that comes after the meeting point.

\begin{observation}\label{obs: atf ieq behind meeting point}
    Let $S\in\cS$ and $v\in V$, and  let $P$ be an $S$-active path to $v$ with meeting point $m_P$.
    Then for every vertex $w$ of $P_{\lre{m_P, v}}$, we have $\atf{R\setminus S}{w}\leq\atf{S}{w}$.
\end{observation}
\begin{proof}
    If $P$ is an $S$-active path with meeting point $m_P$, then by \Cref{def: meeting point}, the set $R\setminus S$ can reach $m_P$ at time $\atf{S}{m_P}$. 
    Consider a vertex $w$ that is visited after $m_P$ by $P$, then the $S$-tightness of $P$ implies that $\atf{S}{m_P}\leq\atf{S}{w}$, and  all edges of $P_{\lre{S, w}}$ are tight at time $\atf{S}{w}$ (by \Cref{obs: edges of S tight path are tight early}). 
    Thus, at time $\atf{S}{w}$ the set $R\setminus S$ can reach $w$.
\end{proof}

By \Cref{def: meeting point}, we know that $R\setminus S$ can reach the meeting point $m_P$ at time $\atf{R}{m_P}$. Often reachability alone will not be sufficient for us. Instead, we want to find a set in $\cS$ that is disjoint from $S$ and \emph{actively} reaches $m_P$ at time $\atf{R}{m_P}$. By the following lemma this is always possible.

\begin{lemma}\label{lem: there is another set reaching mP}
    Let $S\in\cS$ and $v\in V$, and let $P$ be an $S$-active path to $v$. If $m_P$ is a proper meeting point of $P$, then there is a subset $\bar S\in\cS$ of $R\setminus S$ that actively reaches $m_P$ at time $\atf{R}{m_P}$.
\end{lemma}
\begin{proof}
    By \Cref{def: meeting point}, the set $R\setminus S$ can reach $m_P$ at time $\atf{R}{m_P}$. 
    Thus, there exists a terminal $\bar s\in R\setminus S$ such that $\bar s$ reaches $m_P$ at time $\atf{R}{m_P}$. Let $\bar S\in\cS$ be the set containing $\bar s$ that is active at time $\atf{R}{m_P}$. Then $\bar S$ actively reaches $m_P$. If $\bar S=\lrb{\bar s}$, it is clearly disjoint from $S$. 
    Otherwise, we have $a^{\bar S} < \atf{R}{m_P}$.
    By \Cref{lem: proper meeting point is earlier}, this implies 
    \[
    a^{\bar S} \ <\ \atf{R}{m_P} \ <\ \atf{S}{v}\ \leq\ d^S,
    \]
    where the last inequality follows from the fact that $S$ actively reaches $v$.
    Thus, $\bar S$ is no super-set of $S$.
    Because $\bar S$ contains the vertex $\bar s\in R\setminus S$, the set $\bar S$ its also no subset of $S$, implying that $S$ and $\bar S$ are disjoint (because both $S$ and $\bar S$ belong to the laminar family $\cS$).
\end{proof}

To bound the length of the path $P_{\lre{m_P, v}}$, we need to understand which sets contribute to an edge and how much a set $S$ contributes to an edge $\vec{e}$. The latter, we define as the length of the time interval during which $S$ contributes to $e$. An edge $\vec{e}$ is tight if and only if the total amount that all sets contributed to $\vec{e}$ adds up to $c\lr{\vec{e}}$. We first observe, that two sets that contribute to (orientations of) the same edge have to be disjoint: 

\begin{lemma}\label{lem: contributing to same undir edge implies disjoint}
    Let $e$ be an undirected edge and $S,\bar S\in\cS$ be two different sets that each contribute to some direction of $e$. Then $S$ and $\bar S$ are disjoint.
\end{lemma}
\begin{proof}
    Suppose that $S$ and $\bar S$ are not disjoint. Then, by laminarity of $\cS$, we can assume without loss of generality that $\bar S\subsetneq S$. Let $\lr{v, w}$ be the direction of $e$ to which $\bar S$ contributed.
    Then \Cref{obs: larger sets arrive no later}, the fact that $\bar S$ contributed to $\lr{v, w}$, and $\bar S\subsetneq S$ imply
    \begin{align*}
        \atf{S}{v}\leq\atf{\bar S}{v}<d^{\bar S}\leq a^S\,.
    \end{align*}
    Hence, when $S$ is activated, it can already reach $v$. Thus, $S$ cannot contribute to the direction $\lr{w, v}$. 
    Using that our instance is well-subdivided, we have \ref{itm: uniform edges} of \Cref{def:well-subdivided} and thus the set $S$ cannot contribute to $\lr{v, w}$ either (because it is not active at the same time as $\bar S$). 
    Together, this yields a contradiction to $S$ contributing to some direction of $e$. 
\end{proof}

By \Cref{def: contribution}, a set $S$ contributes to an edge $\lr{v, w}$ at times $t\in(\atf{S}{v},\atf{S}{w}]\cap(a^S,d^S]$. However, using \ref{itm: uniform edges} of the definition of well-subdivided instances (\Cref{def:well-subdivided}), we can give a simpler characterization.

\begin{observation}\label{obs: contribution and atfs}
    If a set $S\in\cS$ contributes to an edge $\vec{e}=\lr{v, w}$ at some time $t$, then it does so exactly at the times $t\in(\atf{S}{v}, \atf{S}{w}]$.
\end{observation}
\begin{proof}
    By \Cref{def: contribution}, the set $S$ contributes to $\vec{e}$ at each time $t$ with $t\in(\atf{S}{v},\atf{S}{w}]\cap(a^S,d^S]\eqqcolon(t_1,t_2]$. 
    Thus, our goal is to show that if $(t_1,t_2]$ is non-empty, then
    \begin{enumerate}[label=(\alph*), itemsep=0pt]
        \item \label{item:lower_bound_interval}
        $t_1=\atf{S}{v}$, and
        \item \label{item:upper_bound_interval}
        $t_2=\atf{S}{w}$.
    \end{enumerate}
    By \ref{itm: uniform edges} of the definition of well-subdivided instances (\Cref{def:well-subdivided}), all other sets that contribute to $\vec{e}$ do so exactly during the interval $(t_1, t_2]$ and $\vec{e}$ becomes tight at time $t_2$.
    In particular, $S$ can reach $w$ at time $t_2$ (via the edge $\vec{e}$).
    This implies \ref{item:upper_bound_interval} as follows:
    \begin{align*}
        \atf{S}{w}\leq t_2\ =\ \min\lr{d^S,\ \atf{S}{w}}\ \leq\ \atf{S}{w}\,.
    \end{align*}
    To prove \ref{item:lower_bound_interval}, consider a terminal $s\in S$ with $\atf{\lrb{s}}{v}=\atf{S}{v}$.
    Let $\tilde S\in\cS$ be the subset of $S$ containing~$s$ that is active immediately after time $\atf{S}{v}$ (which exists because the interval $(t_1,t_2]$ is non-empty and thus $\atf{S}{v}<d^S$).
    Then, using \Cref{obs: larger sets arrive no later} and the fact that $S$ contributed to the edge $(v,w)$, we get
    \begin{align*}
        \atf{\tilde S}{v}\ \leq\ \atf{\lrb{s}}{v}\ =\ \atf{S}{v}\ <\ \atf{S}{w}\ \leq\ \atf{\tilde S}{w}\,.
    \end{align*}
    Because $\tilde S$ is active immediately after time $\atf{S}{v}$, this implies that $\tilde S$ starts contributing to $\vec{e}$ immediately after time $\atf{S}{v}$.
    Thus, by \Cref{lem: contributing to same undir edge implies disjoint}, we have $S=\tilde S$. 
    We conclude that the set $S$ is activated no later than at time $\atf{S}{v}$, which yields \ref{item:lower_bound_interval}.
\end{proof}

One consequence of \Cref{obs: contribution and atfs} is the following.

\begin{lemma}\label{lem: contribution of subset of S}
    For a set $S\in\cS$, let $\vec{e}=\lr{v, w}$ be an $S$-tight edge. Each subset of $S$ contributes at most $\atf{S}{w}-\atf{S}{v}$ to $\vec{e}$.
\end{lemma}
\begin{proof}
    Let $S'\in\cS$ be a subset of $S$ that contributes to $\vec{e}$. By \Cref{obs: contribution and atfs}, the set $S'$ contributes exactly $\atf{S'}{w}-\atf{S'}{v}$ to $\vec{e}$. 
    Hence, it remains to show that $\atf{S'}{w}-\atf{S'}{v}\leq\atf{S}{w}-\atf{S}{v}$. 
    By \Cref{obs: larger sets arrive no later}, we get $\atf{S'}{v}\geq\atf{S}{v}$. 
    Further, by the $S$-tightness of $\vec{e}$, the edge $\vec{e}$ is tight at time $\atf{S}{w}$. 
    In particular, we have $\atf{S'}{v} \leq \atf{S}{w}$ (because $S'$ contributes to $\vec{e}$ at some time) and
    thus $S'$ can reach $w$ via $\vec{e}$ at time $\atf{S}{w}$, implying $\atf{S'}{w}\leq\atf{S}{w}$. 
    All together, we thus obtain $\atf{S'}{w}-\atf{S'}{v}\leq\atf{S}{w}-\atf{S}{v}$ and hence $S'$ contributes at most $\atf{S}{w}-\atf{S}{v}$ to $\vec{e}$.
\end{proof}

The previous statements in this section all applied to any $S$-active path. 
The final lemma we will prove in this section will apply only to $S$-safe paths.
We will prove that \Cref{fig:zig_zag_example} is indeed a worst-case example.
Namely, from the meeting point $m_P$ to the vertex $v$, the set $S$ (on average) travels with a speed of at most $3$. 
More precisely, as $S$ reaches $m_P$ at time $\atf{S}{m_P}=\atf{R}{m_P}$ and $v$ at time $\atf{S}{v}$, traveling with average speed at most $3$, means that the length of $P_{\lre{m_P,v}}$ is at most $3\lr{\atf{S}{v}-\atf{S}{m_P}}$. We will actually show a slightly stronger statement: our upper bound on the length will be even smaller if $v$ is reachable by some other terminal before time $\atf{S}{v}$.

\begin{lemma}\label{lem: length from m_p to v}
    Let $S\in\cS$, $v\in V$ and consider an $S$-safe path $P$ to $v$. If $m_P$ is a proper meeting point of $P$ then the length of $P_{\lre{m_P,v}}$ is at most 
    $$
        2\atf{S}{v}+\atf{R}{v}-3\atf{S}{m_P}\,.$$
\end{lemma}
\begin{proof}
    We start by showing that for every edge $e=\lr{a,b}$ of $P_{\lre{m_P, v}}$ there is a unique set $\bar S\subset R\setminus S$ that contributed to $e$ or $\eback$ before $e$ got tight and this set fulfills 
    \begin{align}\label{ieq: S bar is earliest on both sides}
        \atf{\bar S}{a}=\atf{R}{a}\qquad\text{and}\qquad\atf{\bar S}{b}=\atf{R}{b}.
    \end{align}

    Consider an edge $\lr{a,b}$ of $P_{\lre{m_P, v}}$, and let $x\in\lrb{a, b}$ be the vertex of $a$ and $b$ that is first reachable by some terminal, and let $y$ be the other one of these two vertices, that is $\atf{R}{x}\leq\atf{R}{y}$.
    This in particular implies that $\lr{y, x}$ is not tight at time $\atf{R}{x}$.
    By \Cref{obs: atf ieq behind meeting point}, there is a terminal $\bar s\in R\setminus S$, that reaches $x$ at time $\atf{R}{x}$. 
    Let $\bar S\in\cS$ be the set containing $\bar s$ that is active immediately after time $\atf{R}{x}$, then \Cref{lem:first_implies_first_at_neighbor_or_tight_incoming_edge} implies that $\atf{\bar S}{y}=\atf{R}{y}>\atf{R}{x}=\atf{\bar S}{x}$. Thus $\bar S$ contributes to the edge $\lr{x, y}$ and satisfies \eqref{ieq: S bar is earliest on both sides}.
    
    The fact that $P$ is $S$-safe (\Cref{def: safe path}) has two useful implications:
    First, $\bar S$ is the only subset of $R\setminus S$ that contributes to $\lr{a, b}$ or $\lr{b, a}$ before $\lr{a, b}$ is tight. In particular, $\bar S$ is the only subset of $R\setminus S$ can contribute to $\lr{a, b}$.
    Second, at most one subset of $S$ can contribute to $\lr{a, b}$ or $\lr{b, a}$ before $\lr{a, b}$ is tight. In particular, at most one subset of $S$ can possibly contribute to $\lr{a, b}$.

    Using this, we can bound the length of $P_{\lre{m_P, v}}$ as follows.
    Consider a directed edge $\vec{e}=\lr{a,b}\in P_{\lre{m_P, v}}$ and the corresponding set $\bar S$ as above.
    By \Cref{obs: contribution and atfs}, the set $\bar S$ contributes exactly $\max\lrb{\atf{\bar S}{b}-\atf{\bar S}{a},\ 0}$ to $\vec{e}$. Similarly, if there is a subset of $S$ that contributes to $\vec{e}$, then by \Cref{lem: contribution of subset of S}, it contributes at most $\atf{S}{b}-\atf{S}{a}$. The edge $\vec{e}$ eventually becomes tight (as it is tight at time $\atf{S}{v}$). Hence, its length equals the total contribution of $\bar S$ and the corresponding subset of $S$, which gives the following bound:
    \begin{align}\label{ieq: exact cost of e}
        c\lr{\vec{e}}\ \leq\ \atf{S}{b}-\atf{S}{a}\ +\ \max\lrb{\atf{\bar S}{b}-\atf{\bar S}{a},\ 0}\,.
    \end{align}

    The set $\bar S$ cannot contribute more than $c\lr{\vec{e}}$ to either of $\vec{e}$ and $\eback=\lr{b, a}$.
    Hence, again using \Cref{obs: contribution and atfs},
    \begin{align}\label{ieq: lower bound on cost of e}
        c\lr{\vec{e}}\ \geq\ \lrv{\atf{\bar S}{b}-\atf{\bar S}{a}}\,.
    \end{align}
    Subtracting \eqref{ieq: lower bound on cost of e} from two times \eqref{ieq: exact cost of e}, we get
    \begin{align*}
        c\lr{\vec{e}}\ 
        &\leq\ 2\atf{S}{b}-2\atf{S}{a}+2\cdot\max\lrb{\atf{\bar S}{b}-\atf{\bar S}{a},\  0}-\lrv{\atf{\bar S}{b}-\atf{\bar S}{a}}\\
        &=\ 2\atf{S}{b}-2\atf{S}{a}+\atf{\bar S}{b}-\atf{\bar S}{a}\\
        &=\ 2\atf{S}{b}-2\atf{S}{a}+\atf{R}{b}-\atf{R}{a},\defaulttag\label{ieq: upper bound on cost of e}
    \end{align*}
    where we used \eqref{ieq: S bar is earliest on both sides} in the last equality.
    Summing \eqref{ieq: upper bound on cost of e} over all edges on $P_{\lre{m_P, v}}$, we obtain a telescopic sum that collapses to
    \begin{align*}
        c\lr{P_{\lre{m_P, v}}}\ \leq\ 2\atf{S}{v}-2\atf{S}{m_P}+\atf{R}{v}-\atf{R}{m_P}\ =\ 2\atf{S}{v}+\atf{R}{v}-3\atf{R}{m_P},
    \end{align*}
    where we used \Cref{lem: proper meeting point is earlier} in the last step.
\end{proof}

We want to point out, that the proof of \Cref{lem: length from m_p to v} is the key part of our analysis where we use that the paths we consider are safe.
For non-safe paths, there are multiple reasons why the proof of \Cref{lem: length from m_p to v} might fail: If there are two or more subsets of $S$ or two or more subsets of $R\setminus S$ contributing to $\vec{e}$, then \eqref{ieq: exact cost of e} does not hold in general (no matter which set we choose as $\bar S$). 
For the example on the right-hand side of \Cref{fig:k_interaction_example}, we could for each directed edge of the $v$-$w$ path identify one subset $\bar S$ of $R\setminus S$ such that \eqref{ieq: exact cost of e} is fulfilled (by taking $\bar S$ as the subset of $R\setminus S$ contributing to this directed edge, with is unique in this example)\footnote{
The instance on the right-hand side of \Cref{fig:k_interaction_example} is not well-subdivided. However, subdividing each diagonal edge exactly at its midpoint would result in a well-subdivided instance with the same issue.}.
However, $\bar S$ will not satisfy \eqref{ieq: S bar is earliest on both sides}. 
Consequently, the last step of \eqref{ieq: upper bound on cost of e} does not work, and we eventually don't obtain a telescopic sum (because we cannot choose $\bar S$ uniformly on all edges).

\subsection{The Main Invariant}\label{sec: introducing invariant}
In this section we introduce the invariant (\Cref{inv: the invariant}) which we later prove to be satisfied throughout the moat growing procedure. 
As mentioned in \Cref{sec:outline_analysis}, we want to prove that whenever $S$ actively reaches a vertex $v$, there is a non-empty subset $X\subseteq S$ that can be cheaply connected to $v$. 
To get a good analysis also for instances that are not MST-optimal, we additionally prove that all terminals in $X$ have a short distance to $v$. 
Both of this is captured by the notion of \emph{$t$-closeness}, which we introduce next. 
Recall that $\conn{X}$ denotes the minimum cost of a component that connects the terminals in $X$. 
Similarly we define $\conn{X, v}$ to be the minimum cost of a component that connects terminals $X$ and the vertex~$v$.

\begin{definition}[$t$-close]\label{def: t close}
    A terminal set $X\subseteq R$ is $t$-close to a vertex $v$, if every vertex $x\in X$ reaches $v$ at time $\atf{x}{v}\leq t$ and satisfies
    \begin{align}
        \dist{x, v}&\ \leq\ \min\lrb{t+\frac{2-6\gamma}{7-5\gamma}\,\val{X},\ \frac{9-7\gamma}{7-5\gamma}t}\label{ieq: t close dist}\,,\quad\text{and}\\
        \conn{X, v}&\ \leq\ t+\frac{8-8\gamma}{7-5\gamma}\,\val{X}\label{ieq: t close connect}\,.
    \end{align}
\end{definition}

In the MST-optimal case it would suffice to require \eqref{ieq: t close connect}.
The first part of the upper bound on $\dist{x, v}$ in \eqref{ieq: t close dist} will be used to obtain a better dependence on $\gamma$. The second part allows us to identify strictly $\gamma$-expensive set and will be used as follows.

If we want to use our upper bound \eqref{ieq: t close connect} on $\conn{X, v}$ and $\conn{Y,v}$ to obtain a lower bound on the time when two sets $X$ and $Y$ ``meet'' at a vertex $v$, then we will need a lower bound on $\conn{X, v} + \conn{Y,v} \geq \conn{X\cup Y}$.
More precisely, we want $X\cup Y$ to be strictly $\gamma$-expensive. 
For an MST-optimal instance and a strictly $0$-good merge plan, this is the case for every set $X \cup Y\subseteq R$. 
In general, a terminal set $X$ is strictly $\gamma$-expensive if no two terminals $x,y\in X\cup Y$ satisfy \eqref{ieq: strict stronger upper bound for gamma good} (because our merge plan $\cM$ is strictly $\gamma$-good). 
This yields a sufficient condition for a set to be strictly $\gamma$-expensive.

\begin{definition}[late-merged]\label{def: late-merged}
    We say that a terminal set $X\subseteq R$ is late-merged if for all $x, y\in X$ with $x\neq y$:
    \begin{align}
        \merge{}{x, y}\ \geq\ \frac{7-5\gamma}{18-14\gamma}\,\dist{x, y}\label{ieq: late-merged condition}\,.
    \end{align}
\end{definition}

\begin{observation}\label{obs: connecting late-merged set is expensive}
    Every late-merged set $X\subseteq R$ is strictly $\gamma$-expensive, that is, we have
    \begin{align*}
        \conn{X}\ >\ \frac{12-12\gamma}{7-5\gamma}\,\val{X}\,.
    \end{align*}
\end{observation}

Using $\dist{x,y} \leq \dist{x,v} +\dist{y,v}$, will allow us to use \eqref{ieq: t close dist} to prove that we have \eqref{ieq: late-merged condition} for any $x\in X$ and $y\in Y$, which we will then use to prove that $X\cup Y$ is indeed late-merged (see \Cref{lem: inv and m not v implies that X union X bar is late-merged}) and thus strictly $\gamma$-expensive.

To summarize, if $S$ actively reaches $v$ at time $t$, we would ideally want to find a non-empty subset $X\subseteq S$ that is late-merged and $t$-close to $v$. 
This will only be the case if $S$ reached $v$ on its own, meaning without the help of terminals in $R\setminus S$. 
If this is not the case, we will additionally identify a set $\bar X\subseteq R\setminus S$ that helped $S$ to reach $v$. 
Further, we will identify a vertex $m$ that serves as a bifurcation when connecting $X\cup\bar X$ to $v$. 
More precisely, we want both $X$ and $\bar X$ to be close to $m$, and we want that $m$ is not too far away from $v$.

The intuition is that both $X$ and $\bar X$ can reach the bifurcation vertex $m$ on their own, but they work together to get from $m$ to $v$.
All of this is captured by the following invariant, which we will prove to be satisfied throughout the dual growth procedure. 

\begin{invariant}\label{inv: the invariant}
    Let $v\in V$ be a vertex and $S\in\cS$ a set that actively reaches $v$ at time $\atf{S}{v}$ via an $S$-safe path.
    Then one of the following is true:
    \begin{enumerate}[label=(\alph*)]
        \item\label{item: no help case}
        there is a non-empty late-merged subset $X\subseteq S$ that is $\atf{S}{v}$ close to $v$, or
        \item \label{item: help case}
        there is a vertex $m\in V$ and non-empty late-merged sets $X\subseteq S$, and $\bar X\subseteq R\setminus S$ such that both $X$ and $\bar X$ are $\atf{R}{m}$-close to $m$ with $\atf{R}{m}<\atf{S}{v}$, and there is an $m$-$v$ path $Q$, all of whose edges are tight at time $\atf{S}{v}$, with
        \begin{align}
            c\lr{Q}\ \leq\ 2\atf{S}{v}+\atf{R}{v}-3\atf{R}{m}\label{ieq: inv dist m v}\,.
        \end{align}
    \end{enumerate}
\end{invariant}

This is illustrated in \Cref{fig: invariant variable names}. 
For $S$, $v$, $Q$, $m$, $X$, and $\bar X$ as in \Cref{inv: the invariant}, we will also say that the tuple $\lr{S, v, Q, m, X, \bar X}$ satisfies \Cref{inv: the invariant}. 
For the first case of the invariant, we say that $\lr{S, v, \emptyset, v, X, \emptyset}$ satisfies \Cref{inv: the invariant}.

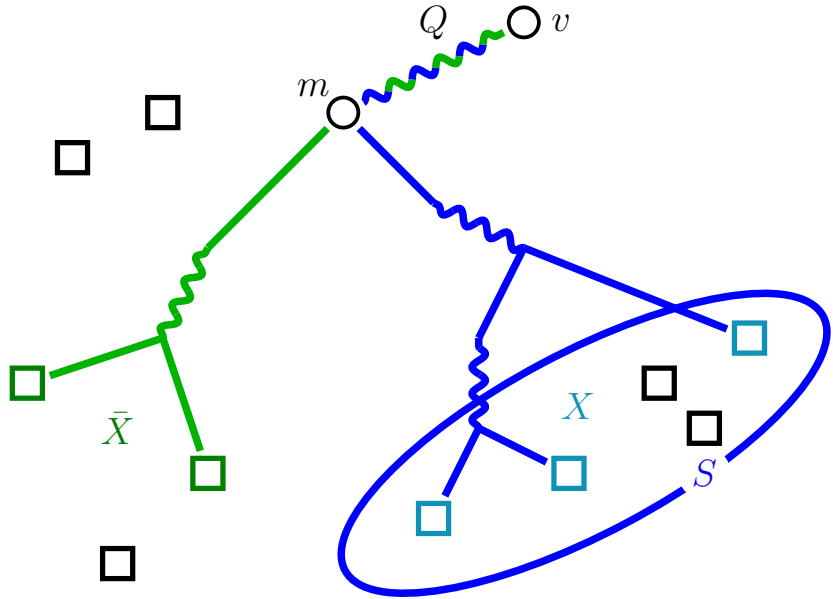
\begin{figure}[ht]
    \centering
    \begin{tikzpicture}[scale=0.6]
\useasboundingbox (-3,-2) rectangle (16,12);
    \Large
    \node[smallterminal, draw = green!50!black] (S1) at (-2, 3) {};
    \node[smallterminal, draw = green!50!black] (S2) at ( 2, 1) {};
    \node[smallterminal, draw = cyan!70!black] (S3) at ( 7, 0) {};
    \node[smallterminal, draw = cyan!70!black] (S4) at (10, 1) {};
    \node[smallterminal, draw = cyan!70!black] (S5) at (14, 4) {};
    \node[smallterminal] (S6) at (12, 3) {};
    \node[smallterminal] (S7) at (13, 2) {};
    
    \node[smallterminal] (S7) at ( 0,-1) {};
    \node[smallterminal] (S7) at (-1, 8) {};
    \node[smallterminal] (S7) at ( 1, 9) {};

    \coordinate (S34) at ($(S3)!1/2!(S4)$);
    \coordinate (S345) at ($(S34)!1/3!(S5)$);
    \draw[activeset, blue,rotate around={28:(S345)}](S345) ellipse [x radius=6, y radius=2];

    \coordinate (A1) at (1,4);
    \coordinate (A2) at (2,6);
    \node[smallsteiner, label={[label distance=-2.5mm]100:{$m$}}] (A3) at (5,9) {};
    \coordinate (A4) at (8,2);
    \coordinate (A5) at (8,4);
    \coordinate (A6) at (9,6);
    \coordinate (A7) at (7,7);
    \node[smallsteiner, label=right:{$v$}] (A8) at (9,11) {};

    \draw[mstedge, green!70!black] (S1) -- (A1);
    \draw[mstedge, green!70!black] (S2) -- (A1);
    \draw[mstedge, curlyHalfA, green!70!black] (A1) -- (A2);
    \draw[mstedge, curlyHalfB, green!70!black] (A1) -- (A2);
    \draw[mstedge] (A2) -- (A3);
    
    \draw[mstedge, blue, blue] (S3) -- (A4);
    \draw[mstedge, blue, blue] (S4) -- (A4);
    \draw[mstedge, curlyHalfA, blue] (A4) -- (A5);
    \draw[mstedge, curlyHalfB, blue] (A4) -- (A5);
    \draw[mstedge, blue] (A5) -- (A6);
    \draw[mstedge, blue] (S5) -- (A6);
    \draw[mstedge, curlyHalfA, blue] (A6) -- (A7);
    \draw[mstedge, curlyHalfB, blue] (A6) -- (A7);
    \draw[mstedge, blue] (A7) -- (A3);
    
    \draw[mstedge, curlyHalfA, blue] (A3) -- (A8);
    \draw[mstedge, curlyHalfB, green!70!black] (A3) -- (A8);

    \node[edgelabel, minimum size=7mm, text = blue] (S2L) at (13, 1) {$S$};

    \node[edgelabel, text = cyan!70!black] (X) at (10.2,2.5) {$X$};
    \node[edgelabel, text = green!50!black] (Xbar) at (0,2) {$\bar X$};
    \node[edgelabel] (Q) at (7,11) {$Q$};

\end{tikzpicture}     \caption{\label{fig: invariant variable names}
    An example of a tuple $\lr{S, v, Q, m, X, \bar X}$ satisfying case \ref{item: help case} of \Cref{inv: the invariant}. The set $S$ actively reaches the vertex $v$. The path $Q$ from $m$ to $v$ has length at most $2\atf{S}{v}+\atf{R}{v}-3\atf{R}{m}$ and all of its edges are tight when $S$ first reaches $v$. The light-blue terminals form the set $X$ that is $\atf{R}{m}$-close to $m$. Similarly, the green terminals form the set $\bar X$ that is also $\atf{R}{m}$-close to $m$. All of the blue terminals are contained in $S$, while none of the green terminals is contained in $S$.
    }
\end{figure}

If $S$ had help from external terminals $\bar X$, there is, in general, no $X\subseteq S$ that is $\atf{S}{v}$-close to $v$. 
In this case, we have to find a suitable bifurcation vertex $m$. 
A natural candidate for $m$ would be the meeting point of the $S$-active path to $v$. 
Then, the bound \eqref{ieq: inv dist m v} on the length of $Q$ would be satisfied by \Cref{lem: length from m_p to v}.
However, as the example from \Cref{fig:bifurcation_example} shows\footnote{Note that in this example, the bifurcation vertex $m$ lies neither on the only $\{s_1\}$-tight $s_1$-$v$ path, nor on the only  $\{s_2\}$-tight $s_2$-$v$ path.}, the meeting point $m_P$ is not always a valid choice because the terminals in $R\setminus S$ reaching it might again had help from (subsets of) $S$. If this is the case, we will recursively (by some inductive argument) choose the bifurcation vertex by considering the bifurcation vertex that we would choose for a path from $\bar S \subseteq R\setminus S$ to $m_P$.

To prove that \Cref{inv: the invariant} holds, we use an inductive argument. To this end, we will not only need that $S$ reaches $v$ via a safe path, but we will also need the following: if some other set $S'\in\cS$ actively reaches a vertex $w$ strictly before time $\atf{S}{v}$, then it does so via a safe path. 
For this reason, we first prove the following lemma in \Cref{sec: invariant holds}.

\begin{restatable}{lemma}{lemInvUntilFirstAccident}\label{lem: inv holds until first accident}
    \Cref{inv: the invariant} is satisfied until (including) the first time at which there exists an $S$-active path (for some set $S\in \cS$) that is not $S$-safe.
\end{restatable}

To prove \Cref{lem: inv holds until first accident} by induction, we will also show the following lemma, which is a key part of the induction step.

\begin{restatable}{lemma}{lemInductionStep}
    \label{lem: X and X bar atfRm close to m implies X union X bar atfSv close to v}
    Let $\lr{S, v, Q, m, X, \bar X}$ fulfill \Cref{inv: the invariant}. Then $X\cup\bar X$ is $\atf{S}{v}$-close to $v$.
\end{restatable}

Using \Cref{lem: inv holds until first accident}, we then show that indeed all $S$-active paths are $S$-safe (\Cref{sec: all active paths are safe}).

\begin{restatable}{lemma}{lemAllActivePathsAreSafe}\label{lem: all active paths are safe}
    Every $S$-active path (for a set $S\in \cS$) is $S$-safe.
\end{restatable}

Given \Cref{lem: inv holds until first accident,lem: X and X bar atfRm close to m implies X union X bar atfSv close to v,lem: all active paths are safe}, we can complete the proof of \Cref{thm: gamma good merge plan feasible} as follows:

\begin{proof}[Proof of \Cref{thm: gamma good merge plan feasible}]
    Suppose, for the sake of deriving a contradiction, that the merge plan $\cM$ is not feasible. 
    Then at some time $t$ during the algorithm we grow the dual variable corresponding to a set $U_S$ containing $r$ for a set $S\in\cS$ with $r\notin S$. 
    This implies that $S$ reaches $r$ no later than time~$t$. Hence, by \Cref{lem: active reaching subset}, there is a subset $S'\in\cS$ of $S$ that actively reaches $r$.
    By \Cref{lem: all active paths are safe}, it does so via an $S'$-safe path and \Cref{lem: inv holds until first accident} applies.
    Together with \Cref{lem: X and X bar atfRm close to m implies X union X bar atfSv close to v}, this implies that there is a vertex $s\in S'\subseteq S$ such that $\dist{s, r}\leq\frac{9-7\gamma}{7-5\gamma}t$. 
    However, $s$ and $r$ are not merged before time $t$ (because we grow the dual variable $U_S$ at time $t$ and $s\in S$, but $r\notin S$).
    Thus, \eqref{ieq: strict general upper bound for gamma good} from the definition of strictly $\gamma$-good merge plans implies $\frac{12}{7-5\gamma}t<\dist{s, r}$, which leads to a contradiction:
    \begin{align*}
        \dist{s, r}\ \leq\ \frac{9-7\gamma}{7-5\gamma}\,t\ <\ \frac{12}{7-5\gamma}\, t\ <\ \dist{s, r}\,.
    \end{align*}
\end{proof}

We remark that, as sketched in \Cref{par:component invariant}, it is also possible to prove \Cref{thm: gamma good merge plan feasible} by (lower- and upper-) bounding the cost $\conn{X\cup\lrb{r}}=\conn{X, r}$ of connecting a suitable subset $X\subseteq S'$ of terminals to $r$. 
This would require a few more computations, but allows avoiding bounds on distances (namely \eqref{ieq: t close dist}) for the special case of MST-optimal instances.

\subsection{The Invariant Holds as long as all Active Paths are Safe}\label{sec: invariant holds}

The goal of this section is to prove \Cref{lem: inv holds until first accident}, that is, that \Cref{inv: the invariant} holds until (including) the first time that there is some $S$-active path that is not $S$-safe.
Before we can do this, we need a few auxiliary statements.

\begin{lemma}\label{lem: from close to w to close to v}
    Let $t\geq t'\geq0$, $v, w\in V$, and $X\subseteq R$.
    If $X$ is $t'$-close to $w$ and there is a $w$-$v$ path $Q$ with $c\lr{Q}\leq t-t'$ such that all edges of $Q$ are tight at time $t$, then $X$ is $t$-close to $v$.
\end{lemma}
\begin{proof}
    Both \eqref{ieq: t close dist} and \eqref{ieq: t close connect} are preserved when adding the length of $Q$ (where we used that $\frac{9-7\gamma}{7-5\gamma} \geq 1$ because $\gamma \leq \frac{1}{5}$). 
    Moreover, all $x\in X$ can reach $w$ at time $t'\leq t$, and the edges of $Q$ are tight at time $t$. 
    Thus, all $x\in X$ can reach $v$ at time $t$.
\end{proof}

Our second auxiliary statement will allow us to ensure that we only ever encounter late-merged sets. 
We remark that the only reason why we care about sets being late-merged is that late-merged sets are strictly $\gamma$-expensive, and thus, this statement is not needed for the MST-optimal case.

\begin{lemma}\label{lem: inv and m not v implies that X union X bar is late-merged}
    If $\lr{S, v, Q, m, X, \bar X}$ fulfills \Cref{inv: the invariant}, then $X\cup\bar X$ is late-merged.
\end{lemma}
\begin{proof}
    If $\bar X$ is empty, then \ref{item: no help case} of \Cref{inv: the invariant} applies and and in particular, $X\cup \bar X = X$ is late-merged. 
    Otherwise, by \ref{item: help case} of \Cref{inv: the invariant}, both $X$ and $\bar X$ are late-merged. 
    Hence, it suffices to prove \eqref{ieq: late-merged condition} for $x\in X$ and $y\in\bar X$.
    By the $\atf{R}{m}$-closeness of $X$ and $\bar X$ to $v$, we can upper bound the distance from $x$ to $y$ by
    \begin{align}
    \label{ieq: dist x to y}
        \dist{x, y}\ \leq\ \dist{x, v}+\dist{y, v}
        \ \leq\ 2\cdot\frac{9-7\gamma}{7-5\gamma}\cdot \atf{R}{m}
        \ =\ \frac{18-14\gamma}{7-5\gamma}\cdot \atf{R}{m}\,.
    \end{align}
    The terminals $x$ and $y$ are separated by the set $S$ that is active at time $\atf{S}{v}>\atf{R}{m}$. Thus, their merge time is at least $\atf{R}{m}$. Together with \eqref{ieq: dist x to y}, this gives
    \begin{align*}
        \merge{}{x, y}\ \geq\ \atf{R}{m}\ \geq\ \frac{7-5\gamma}{18-14\gamma}\,\dist{x, y}\,,
    \end{align*}
    which implies that $X\cup\bar X$ is late-merged.
\end{proof}

\paragraph{How early can two sets meet?}
A crucial part of \Cref{inv: the invariant} is \eqref{ieq: inv dist m v}, the bound on the length of $Q$. However, this bound would be very weak if $\atf{R}{m}$ could be arbitrarily small. 
Recalling \Cref{fig:induction_step_outline}, this means that we want to lower bound the time that the green component ($K_1$) and the blue component ($K_2$) can meet. In \Cref{fig:simple_interaction_example} we have seen that two \emph{terminals} cannot meet earlier than $\frac67$ times the time they will be merged. By arguing about the cost of $K_1 \cup K_2$, we will be able to derive a similar bound on the time that two \emph{terminal sets} meet (\Cref{lem: connect bound}). Note that for the MST-optimal case ($\gamma=0$) this generalizes the bound from \Cref{fig:simple_interaction_example}. However, in general (for larger values of $\gamma$) we only get a weaker bound. To compensate this, we will additionally argue about the \emph{distance} between the two sets that meet (\Cref{lem: dist bound}).

\begin{lemma}\label{lem: connect bound}
    If $\lr{S, v, Q, m, X, \bar X}$ fulfills \Cref{inv: the invariant} and $\bar X$ is non-empty, then
    \begin{align}
        \atf{R}{m}&\ >\ \frac{6-6\gamma}{7-5\gamma}\cdot\atf{S}{v}+\frac{2-2\gamma}{7-5\gamma}\cdot \lr{\val{X}+\val{\bar X}}\,.\label{ieq: connect bound}
    \end{align}
\end{lemma}
\begin{proof}
    By \Cref{lem: inv and m not v implies that X union X bar is late-merged}, $X\cup\bar X$ is late-merged. 
    Thus, by \Cref{obs: connecting late-merged set is expensive}, $X\cup\bar X$ is strictly $\gamma$-expensive, that is
    \begin{align}
        \conn{X\cup\bar X}>\frac{12-12\gamma}{7-5\gamma}\, \val{X\cup\bar X}\label{ieq: drop leq connect for X union X bar}\,.
    \end{align}
    Using \Cref{lem: drop of X union X bar}, we can lower bound the right hand-side of \eqref{ieq: drop leq connect for X union X bar} by
    $$
        \frac{12-12\gamma}{7-5\gamma}\lr{\val{X}+\val{\bar X}+\atf{S}{v}}\,.
    $$
    The left hand-side of \eqref{ieq: drop leq connect for X union X bar} can be upper bounded by $\conn{X, m}+\conn{\bar X, m}$. Together with \eqref{ieq: t close connect} this yields
    \begin{align*}
        \frac{8-8\gamma}{7-5\gamma}\lr{\val{X}+\val{\bar X}}+2\atf{R}{m}
        \ >\ \frac{12-12\gamma}{7-5\gamma}\lr{\val{X}+\val{\bar X}+\atf{S}{v}}
        \,,
    \end{align*}
    which is equivalent to \eqref{ieq: connect bound}.
\end{proof}

\begin{lemma}\label{lem: dist bound}
    If $\lr{S, v, Q, m, X, \bar X}$ fulfils \Cref{inv: the invariant} and $\bar X$ is non-empty, then
    \begin{align}        \label{ieq: dist bound}
        \atf{R}{m}\ >\ \frac{6}{7-5\gamma}\cdot \atf{S}{v}-\frac{1-3\gamma}{7-5\gamma}\cdot\lr{\val{X}+\val{\bar X}}.
    \end{align}
\end{lemma}
\begin{proof}
    We consider the distance between $X$ and $\bar X$. Using that both $X$ and $\bar X$ are $\atf{R}{m}$-close to the bifurcation vertex $m$, we can upper bound their distance by
    \begin{align*}
        \dist{X, \bar X}\ \leq\ \dist{X, m}+\dist{\bar X, m}
        \ \leq\ 2 \atf{R}{m}+\frac{2-6\gamma}{7-5\gamma}\lr{\val{X}+\val{\bar X}}.
    \end{align*}
    However, $X$ and $\bar X$ are not merged before time $\atf{S}{v}$, that is, for times $t< \atf{S}{v}$, no vertex from $X$ belongs to the same part of $\cS^t$ as a vertex from $\bar X$. 
    Thus, the distance of $X$ and $\bar X$ is strictly larger than $\frac{12}{7-5\gamma}\cdot\atf{S}{v}$ (because the merge plan $\cM$ is strictly $\gamma$-good), implying
    \begin{align*}
         2\atf{R}{m}+\frac{2-6\gamma}{7-5\gamma}\lr{\val{X}+\val{\bar X}}\ >\ \frac{12}{7-5\gamma}\atf{S}{v},
    \end{align*}
     which is equivalent to \eqref{ieq: dist bound}.
\end{proof}

In the MST-optimal case (where we can choose $\gamma=0$), the bound \eqref{ieq: connect bound} dominates \eqref{ieq: dist bound} and is thus sufficient. 
In general, neither of the bounds \eqref{ieq: connect bound} and \eqref{ieq: dist bound} dominates the other, which is why we often use linear combinations of \eqref{ieq: connect bound} and \eqref{ieq: dist bound}. 
To simplify later computations, we provide the relevant combinations in \Cref{cor: lin comb of dist and connect bound}.

\begin{corollary}\label{cor: lin comb of dist and connect bound}
    If $\lr{S, v, Q, m, X, \bar X}$ fulfils \Cref{inv: the invariant} and $\bar X$ is non-empty, then
    \begin{align*}
        \atf{S}{v}-\atf{R}{m}&<\frac{1-3\gamma}{7-5\gamma}\cdot \atf{S}{v}\defaulttag\label{ieq: gamma over three bound}\,,\quad\textnormal{and}\\
        2\lr{\atf{S}{v}-\atf{R}{m}}&<\frac{2-2\gamma}{7-5\gamma}\cdot \atf{S}{v}-\frac{2-2\gamma}{7-5\gamma}\cdot \lr{\val{X}+\val{\bar X}}\,.\defaulttag\label{ieq: bound with minus drop}
    \end{align*}
\end{corollary}
\begin{proof}
    Subtracting $\frac{2}{3}$ of \eqref{ieq: dist bound} and $\frac{1}{3}$ of \eqref{ieq: connect bound} from $\atf{S}{v}$ implies \eqref{ieq: gamma over three bound}:
    \begin{align*}
        \atf{S}{v}-\atf{R}{m}
        \ <&\ \atf{S}{v}-\frac23\lr{\frac{6}{7-5\gamma}\cdot\atf{S}{v}-\frac{1-3\gamma}{7-5\gamma}\lr{\val{X}+\val{\bar X}}}\\
        &\ -\frac13 \lr{\frac{6-6\gamma}{7-5\gamma}\cdot\atf{S}{v}+\frac{2-2\gamma}{7-5\gamma}\lr{\val{X}+\val{\bar X}}}\\[2mm]
        \leq&\ \frac{1-3\gamma}{7-5\gamma}\cdot \atf{S}{v}.
    \end{align*}
    Inequality \eqref{ieq: connect bound} is equivalent to
    \begin{align*}
        \atf{S}{v}-\atf{R}{m}<\frac{1+\gamma}{7-5\gamma}\cdot\atf{S}{v}-\frac{2-2\gamma}{7-5\gamma}\lr{\val{X}+\val{\bar X}}\,.
    \end{align*}
    Adding this to \eqref{ieq: gamma over three bound} gives \eqref{ieq: bound with minus drop}.
\end{proof}

Using \Cref{cor: lin comb of dist and connect bound}, we can prove \Cref{lem: X and X bar atfRm close to m implies X union X bar atfSv close to v}, which we restate here.
At this point, the proof boils down to a few calculations. 
However, we want to emphasize that this is the key lemma that allows us to control the behavior of \Cref{algo:dual_growth_with_merge_plan} with merge plans for which $\cS$ is not a collection of singletons. 

\lemInductionStep*
\begin{proof}
    If $\bar X= \emptyset$, then we are in case \ref{item: no help case} of \Cref{inv: the invariant} and thus the statement trivially holds. 
    
    Therefore, we now assume that $\bar X$ is non-empty.
    In particular, we can apply \Cref{cor: lin comb of dist and connect bound}.
    We start by proving the desired upper bound on $\cost(X\cup \bar X, v)$, that is \eqref{ieq: t close connect}. 
    Due to $X$ and $\bar X$ being $\atf{R}{m}$-close to $m$, and \eqref{ieq: inv dist m v} from \Cref{inv: the invariant}, we have
    \begin{align*}
        &\conn{X\cup\bar X, v}\\
        &\leq\ \conn{X, m}+\conn{\bar X, m}+\dist{m, v}\\
        &\leq\ \frac{8-8\gamma}{7-5\gamma}\,\val{X}+\atf{R}{m}\ +\ \frac{8-8\gamma}{7-5\gamma}\,\val{\bar X}+\atf{R}{m}\ +\ 2\atf{S}{v}+\atf{R}{v}-3\atf{R}{m}\\
        &\leq\ \frac{8-8\gamma}{7-5\gamma}\lr{\val{X}+\val{\bar X}}+2\atf{S}{v}+\atf{S}{v}-\atf{R}{m},
    \end{align*}
    where the last inequality follows from \Cref{obs: larger sets arrive no later}.
    Together with \eqref{ieq: gamma over three bound} from \Cref{cor: lin comb of dist and connect bound}, this implies
    \begin{align*}
        \conn{X\cup\bar X, v}
        &\leq\ \frac{8-8\gamma}{7-5\gamma}\lr{\val{X}+\val{\bar X}}+\lr{2+\frac{1-3\gamma}{7-5\gamma}}\atf{S}{v}\\
        &=\ \frac{8-8\gamma}{7-5\gamma}\lr{\val{X}+\val{\bar X}+\atf{S}{v}}+\atf{S}{v}\\
        &\leq\ \frac{8-8\gamma}{7-5\gamma}\,\val{X\cup\bar X}+\atf{S}{v}
        \,,
    \end{align*}
    where we used \Cref{lem: drop of X union X bar} in the last step. This proves \eqref{ieq: t close connect}.

    To prove the desired distance bound (see \eqref{ieq: t close dist}) for all $x\in X\cup\bar X$, fix such a vertex $x$ and assume without loss of generality that $x\in X$. 
    (The argument for $x\in \bar X$ is the same, except that we replace $X$ everywhere by $\bar X$.)
    Using again \eqref{ieq: inv dist m v} from \Cref{inv: the invariant}, we have $\dist{m, v} \leq 2\atf{S}{v}+\atf{R}{v}-3\atf{R}{m} \leq 3\lr{\atf{S}{v}-\atf{R}{m}}$, where the second inequality follows from \Cref{obs: larger sets arrive no later}.
    By the $\atf{R}{m}$-closeness of $X$ to $m$, this implies
    \begin{align*}
        \dist{x, v}&\leq\dist{x, m}+\dist{m, v}\\
        &\leq\atf{R}{m}+\frac{2-6\gamma}{7-5\gamma}\val{X}+3\lr{\atf{S}{v}-\atf{R}{m}}\\
        &=\atf{S}{v}+\frac{2-6\gamma}{7-5\gamma}\val{X}+2\lr{\atf{S}{v}-\atf{R}{m}}\,.\defaulttag\label{ieq: dist x v general form}
    \end{align*}
    
    Combining \eqref{ieq: dist x v general form} with each of the bounds \eqref{ieq: gamma over three bound} and \eqref{ieq: bound with minus drop} from \Cref{cor: lin comb of dist and connect bound}, yields
    \begin{align*}
        \,\dist{x, v}
        \leq&\,\min\lrb{\atf{S}{v}+\frac{2-6\gamma}{7-5\gamma}\lr{\val{X}+\atf{S}{v}},\ \atf{S}{v}+\frac{2-2\gamma}{7-5\gamma}\atf{S}{v}}\,.
    \end{align*}
    Using \Cref{lem: drop of X union X bar}, this implies the desired distance bound (see \eqref{ieq: t close dist}) for  $x\in X \cup \bar X$.
    
    Finally, we have to show that every vertex $x\in X\cup\bar X$ can reach $v$ at time $\atf{S}{v}$. 
    This is indeed the case because each $x\in X\cup\bar X$ can reach $m$ at time $\atf{R}{m}\leq\atf{S}{v}$, and by \Cref{inv: the invariant} there is an $m$-$v$ path consisting only of edges that are tight at time $\atf{S}{v}$.
    Thus, $X\cup\bar X$ is $\atf{S}{v}$-close to~$v$.
\end{proof}

Before we can prove the main result of this section (\Cref{lem: inv holds until first accident}), we need to establish one technical detail about $\bar X$. Namely, that at time $\atf{R}{v}$ the terminal set $\bar X$ is contained in a single active set.

\begin{lemma}\label{lem: wlog merge X bar is less than tRv}
    If $\lr{S, v, Q, m, X, \bar X}$ fulfills \Cref{inv: the invariant} and $\bar X$ is non-empty, then all vertices of $\bar X$ are contained in the same active set at time $\atf{R}{v}$ (and thus also at any later time).
\end{lemma}
\begin{proof}
    Suppose that at time $\atf{R}{v}$ the set $\bar X$ is not contained in a single active set. 
    Then in particular $\val{\bar X}\geq\atf{R}{v}$ (by \Cref{lem: drop of X union X bar} or directly by the definition of value). 
    Thus, by \Cref{lem: connect bound}, we have 
    \begin{align*}
        2\atf{R}{m}
        \ >\ \frac{12-12\gamma}{7-5\gamma}\,\atf{S}{v}+\frac{4-4\gamma}{7-5\gamma}\lr{\val{X}+\val{\bar X}}
        \ \geq\ \frac{12-12\gamma}{7-5\gamma}\,\atf{S}{v}+\frac{4-4\gamma}{7-5\gamma}\,\atf{R}{v}.
    \end{align*}
    Using $\atf{S}{v}\geq\atf{R}{v}$ (by \Cref{obs: larger sets arrive no later}), this yields
    \begin{align}
        2\atf{R}{m}\ >\ \frac{9-11\gamma}{7-5\gamma}\,\atf{S}{v}+\atf{R}{v}\ \geq\ \atf{S}{v}+\atf{R}{v}\,,
    \end{align}
    where we used $\gamma\leq\frac15$. 
    Thus, the length of the $m$-$v$ path $Q$ can be bounded as follows:
    \begin{align}
        \label{ieq: too good upper bound on length of Q}
        c\lr{Q}
        \ \leq\ 2\atf{S}{v}+\atf{R}{v}-3\atf{R}{m}
        \ <\ \atf{S}{v}-\atf{R}{m}
        \ =\ \atf{S}{v}-\atf{S}{m},
    \end{align}
    where in the last step we used that $X\subseteq S$ is $\atf{R}{m}$-close to $m$, which in particular implies that $X$ (and thus also $S$) can reach $m$ at time $\atf{R}{m}$.
    
    To derive a contradiction, we show that \eqref{ieq: too good upper bound on length of Q} cannot be true. Namely, we show that $c\lr{Q}\geq\atf{S}{v}-\atf{S}{m}$.
    Indeed, for every time $t\in\left[\atf{S}{m}, \atf{S}{v}\right)$, there is at least one edge $e=\lr{a, b}$ on the directed $m$-$v$ path $Q$ for which $a$ is reachable by $S$ but $b$ is not. Hence, at every time $t\in\left(\atf{S}{m}, \atf{S}{v}\right]$ at least one subset of $S$ contributes to an edge of $Q$.
    Thus, in total, the subsets of $S$ contribute at least $\max\lrb{0,\  \atf{S}{v}-\atf{S}{m}}$ to the edges of the directed path $Q$, 
    and hence 
    \begin{align}
        \atf{S}{v}-\atf{S}{m}
        \ \leq\ \max\lrb{0,\ \atf{S}{v}-\atf{S}{m}}
        \ \leq\ c\lr{Q}\,.
    \end{align}
    This contradicts \eqref{ieq: too good upper bound on length of Q}.
\end{proof}

Finally, we are ready to prove the main result of this section.

\lemInvUntilFirstAccident*
\begin{proof}
    We do a proof by induction over time. 
    Note that it suffices to consider the finite set of points in time at which an edge becomes tight, or the partition of active sets changes.
    We work with a slightly stronger induction hypothesis and additionally show that if a set $S\in\cS$ reaches a vertex $v\in V$ via a safe path $P$, then we can choose $\bar X$ (see \Cref{inv: the invariant}) such that every vertex $x\in\bar X$ can reach some inner vertex $w$ of the path $P$ at time $\atf{R}{w}$.

    At time $t=0$, each terminal set $S\in \cS^t$ can only reach itself.
    For every $v\in S$, the tuple $\lr{S, v, \emptyset, v, \lrb{v}, \emptyset}$ satisfies \Cref{inv: the invariant}.

    Now let $v\in V$ be a vertex and $S\in\cS$ a set that actively reaches $v$ at time $t\coloneqq\atf{S}{v}>0$ via a safe path $P$.
    Further, assume that all active paths that existed strictly before time $\atf{S}{v}$ are safe.
    
    We distinguish two cases, based on whether or not the path $P$ has a proper meeting point.

    \paragraph{Case 1: $\mathbf{P}$ has no proper meeting point.}
    Let $e=\lr{w,v}$ be the last edge of $P$. 
    Then $e$ is tight at time $\atf{S}{v}$ and we have $\atf{S}{w}\leq \atf{S}{v}$.
    Because $P$ has no proper meeting point, we have $\atf{R\setminus S}{w}>\atf{S}{w}$.
    This implies that $\atf{S}{v}$ is strictly larger than $\atf{S}{w}$, as otherwise no set could have contributed to $e$ before it became tight.
    Hence, $S$ or some subset of $S$ contributed to $e$ throughout the time interval $(\atf{S}{w},\atf{S}{v}]$.
    Because $S$ (and thus non of its strict subsets) is active at time $\atf{S}{v}$ and because $\atf{R\setminus S}{w}>\atf{S}{w}$, property \ref{itm: uniform edges} of well-subdivided instances (\Cref{def:well-subdivided}) implies that $S$ is the only set that ever contributes to $e$.
    Hence, $c\lr{e}=\atf{S}{v}-\atf{S}{w}$ and $t'\coloneqq\atf{S}{w}<\atf{S}{v}=t$.
    
    By \Cref{lem: active reaching subset}, there is a subset $S'\in\cS$ of $S$ that actively reaches $w$ at time $\atf{S}{w}$ and for which the path $P'$ resulting from $P$ by removing $e$ is an $S'$-active path. 
    All edges of $P'$ were already tight at time $t'<t$. Hence, this path is $S'$-safe by assumption. 
    By induction, we get a tuple $\lr{S', w, Q', m', X', \bar X'}$ that fulfills \Cref{inv: the invariant}.

    The path $P$ has no proper meeting point, but every vertex $x\in\bar X'$ can reach $w$ at time $\atf{S'}{w}=\atf{S}{w}$. Thus, the set $\bar X'$ has to be disjoint from $R\setminus S$, implying $X'\cup\bar X'\subseteq S$. 
    By \Cref{lem: X and X bar atfRm close to m implies X union X bar atfSv close to v}, the set $X'\cup\bar X'$ is $\atf{S}{w}$-close to $w$. By \Cref{lem: from close to w to close to v} it is also $\atf{S}{v}$-close to $v$. 
    Further, $X'\cup\bar X'$ is late-merged by \Cref{lem: inv and m not v implies that X union X bar is late-merged}. 
    All together this implies that $\lr{S, v, \emptyset, v, X'\cup\bar X', \emptyset}$ satisfies \Cref{inv: the invariant}.

    \paragraph{Case 2: $\mathbf{P}$ has a proper meeting point.}
    By \Cref{lem: there is another set reaching mP}, there is a set $\bar S\subseteq R\setminus S$ that actively reaches $m_P$ at time $\atf{R}{m_P}$, which by \Cref{lem: proper meeting point is earlier} is strictly earlier than time $\atf{S}{v}$. 
    Thus, we can use the induction hypothesis for $\bar S$ and $m_P$, to obtain a tuple $\lr{\bar S, m_P, \bar Q, \bar m, \bar X, \hat X}$ satisfying \Cref{inv: the invariant}.
    By \Cref{lem: wlog merge X bar is less than tRv}, $\hat X$ is either a non-empty subset of $S$, or disjoint from $S$. 
    
    If $\hat X$ is a non-empty subset of $S$, then $\lr{S, v, \bar Q + P_{\lre{m_P,v}}, \bar m, \hat X, \bar X}$ fulfills \Cref{inv: the invariant}. To verify this, note that all edges of $\bar Q + P_{\lre{m_P,v}}$ are tight at time $\atf{S}{v}$.
    Moreover, by \Cref{lem: proper meeting point is earlier}, we have $\atf{\bar S}{m_P} = \atf{R}{m_P}$, and thus by \Cref{lem: length from m_p to v}, 
    \begin{align*}
        c\lr{\bar Q + P_{\lre{m_P,v}}}
        \ &=\ c\lr{\bar Q} + c\lr{P_{\lre{m_P,v}}}\\
        &\leq\ 2\atf{\bar S}{m_P}+\atf{R}{m_P}-3\atf{R}{\bar m}+2\atf{S}{v}+\atf{R}{v}-3\atf{R}{m_P}\\
        &=\ 2\atf{S}{v}+\atf{R}{v}-3\atf{R}{\bar m}.
    \end{align*}
    Furthermore, the additional condition in our strengthened induction hypothesis is satisfied because each vertex $x\in\bar X$ can reach the inner vertex $m_P$ of $P$ at time $\atf{R}{m_P}$.

    Otherwise, $\hat X$ and $S$ are disjoint.
    Then $\bar X\cup\hat X$ is a subset of $R\setminus S$ and, by \Cref{lem: X and X bar atfRm close to m implies X union X bar atfSv close to v}, is $\atf{R}{m_P}$-close to $m_P$.
    By \Cref{lem: active reaching subset}, there is a subset $S'\subseteq S$ for which $P_{\lre{S,m_P}}$ is an $S'$-active path. 
    Applying the induction hypothesis to $S'$ and $m_P$, we obtain a tuple $\lr{S',m_P, Q', m', X', \bar X'}$ fulfilling \Cref{inv: the invariant} and our extra condition, namely that all $x\in\bar X'$ can reach some inner vertex $w$ of $P_{\lre{S, m_P}}$ at time $\atf{R}{w}$. By the choice of $m_P$, the latter implies that $\bar X'\subseteq S$ (otherwise we would have chosen the meeting point $m_P$ closer to $S$ on $P$; see \Cref{def: meeting point}). 
    By \Cref{lem: X and X bar atfRm close to m implies X union X bar atfSv close to v}, $X'\cup \bar X'$ is $\atf{S'}{m_P}$-close to $m_P$ and hence it is also $\atf{R}{m_P}$-close to $m_P$ (because $\atf{S'}{m_P} = \atf{S}{m_P} = \atf{R}{m_P}$ by \Cref{lem: active reaching subset,lem: proper meeting point is earlier}). 
    We conclude that the tuple $\lr{S, v, P_{\lre{m_P, v}}, m_P, X'\cup\bar X', \hat X\cup\bar X}$ fulfills \Cref{inv: the invariant}:
    Both sets $X'\cup\bar X'$ and $\hat X\cup\bar X$ are $\atf{R}{m_P}$-close to $m_P$, and $\atf{R}{m_P}<\atf{S}{v}$ (\Cref{lem: proper meeting point is earlier}). Furthermore, all edges of the path $P_{\lre{m_P, v}}$ are tight at time $\atf{S}{v}$ and by \Cref{lem: length from m_p to v} its length is at most $2\atf{S}{v}+\atf{R}{v}-3\atf{R}{m_p}$. By the $\atf{R}{m_P}$-closeness of $\hat X\cup\bar X$ to the inner vertex $m_P$ of $P$, the additional condition in our strengthened induction hypothesis is satisfied as well.
\end{proof}

\subsection{All Active Paths are Safe}\label{sec: all active paths are safe}

In this section, we will complete the proof of \Cref{thm: gamma good merge plan feasible} by showing that there is never an active path that fails to be safe:

\lemAllActivePathsAreSafe*

The overall strategy will be a proof by contradiction.
We consider the first time when there is a non-$S$-safe edge on an $S$-active path (for some set $S\in \cS$).
Then, we will construct a component $K$ from three parts, each of whose cost we will upper bound using that \Cref{inv: the invariant} is still satisfied at this time (by \Cref{lem: inv holds until first accident}).
Additionally, we will show that $K$ connects a \emph{late-merged} (and thus strictly $\gamma$-expensive) set of terminals, which yields a lower bound on the cost of $K$. 
The lower bound we obtain will be strictly greater than our upper bound on the cost of $K$, leading to the desired contradiction.

Suppose that \Cref{lem: all active paths are safe} is not true and let $t$ be the first time at which there is an $S$-active path (for some set $S\in \cS$) that is not $S$-safe. 
(Note that such a first time exists because it suffices to consider the finite set of times at which some edge becomes tight.)
Further, let $P+e$ be an inclusion-wise minimal such path that exists at time $t$, and let $S_1\in \cS$ such that $P+e$ is $S_1$-active, but not $S_1$-safe. Let $e=\lr{v,w}$.
Then, by   \Cref{lem: active reaching subset}, there is a set $S_1'\subseteq S_1$ such that $P$ is $S_1'$-active and $\atf{S_1'}{v}=\atf{S_1}{v}$. By minimality of the path $P+e$, the path $P$ is not only $S_1'$-active, but also $S_1'$-safe.
By the definition of safe edges (\Cref{def: safe path}) and because $S_1' \subseteq S_1$, every $S_1'$-safe edge is also $S_1$-safe. 
This implies that all edges of $P$ are $S_1$-safe. 
Hence, the edge $e=\lr{v, w}$ is not $S_1$-safe.

According to the definition of safe edges (\Cref{def: safe path}), there are two possible reasons why $e$ could be non-$S_1$-safe: 
(i) there are three sets that contributed to $e$ or $\lr{w, v}$ before $e$ got tight, or 
(ii) there are two subsets of $R\setminus S_1$ that contributed to $e=(v,w)$ or $\lr{w, v}$ before $e$ became tight. 
We next show that (ii) is always the case.

\begin{lemma}
    There exist two subsets of $R\setminus S_1$ that each contributed to one of the directed edges $e=(v,w)$ or $\lr{w, v}$ before the time at which $e$ became tight.
\end{lemma}
\begin{proof}
   We first show that no proper subset of $S_1$ can contribute to $e$ or $\lr{w, v}$ before $e$ became tight.
   Suppose that some proper subset $S'\subsetneq S_1$ contributed to $e$, then $S'$ actively reaches $w$ (by \ref{itm: uniform edges} of the definition of well-subdivided instances; see \Cref{def:well-subdivided}).
   This implies that $S_1$ cannot actively reach $w$, and thus contradicts $P+e$ being $S_1$-active.
   
    Now suppose instead that some proper subset $S'\subsetneq S_1$ contributed to $\lr{w, v}$ before $e=(v,w)$ became tight.
    Then $e$ is not yet tight at time $\atf{S'}{w}$ and thus also not tight at time $\atf{S_1}{w}\leq \atf{S'}{w}$ (where we used \Cref{obs: larger sets arrive no later}).
    This contradicts the fact that $P+e$ is an $S_1$-tight path.

    Hence, there is at most one subset of $S_1$ (namely $S_1$ itself) that could have contributed to $e$ or $\lr{w, v}$ before $e$ got tight. Thus, if there are three or more sets that contributed to $e$ or $\lr{w, v}$ before $e$ got tight, at least two of them are subsets of $R\setminus S$ (by laminarity of $\cS$).
\end{proof}

Let $S_2, S_3\in\cS$ be two subsets of $R\setminus S_1$ that contributed to $e$ or $\lr{w, v}$ before $e$ became tight.
By \Cref{lem: active reaching subset}, there are subsets $S_2', S_3'\in\cS$ of $S_2$ and $S_3$, respectively, that actively reach $v$ at the same time as $S_2$ and $S_3$, respectively. In particular, we have 
\begin{align}
    \atf{S_1'}{v}=\atf{S_1}{v},\qquad
    \atf{S_2'}{v}=\atf{S_2}{v},\qquad
    \atf{S_3'}{v}=\atf{S_3}{v}\,.\label{eq: atfs are identical for S and Sprime}
\end{align}

By \Cref{lem: contributing to same undir edge implies disjoint}, the sets $S_1, S_2, S_3$ are pairwise disjoint. 
Thus, also their subsets $S_1', S_2', S_3'$ are pairwise disjoint.
The reason that we use $S_i'$ instead of $S_i$ is a technical detail only needed to ensure that \Cref{inv: the invariant} applies (as $S_i$ might not \emph{actively} reach $v$).
For intuition, one can think of $S_i$ and $S_i'$ as being equal (and this is indeed the case if $S_i$ actively reaches $v$).

\begin{lemma}\label{claim:active_reachability_of_non_safe_edge}
Each set $S_i'$ (for $i=1,2,3$) actively reaches the vertex $v$ via an $S_i'$-safe path. 
Moreover, at time $\atf{S_1}{v}$, each set $S_i$ (with $i=1,2,3$) can reach the vertex $v$.
\end{lemma}
\begin{proof}
For $S_1'$, this is the case because $P$ is an $S_1'$-safe path to $v$. 
Let now $i\in \{2,3\}$.
We show that $S_i$ (and thus by \eqref{eq: atfs are identical for S and Sprime} also $S_i'$) can reach $v$ strictly before time $t$, which by our choice of $t$ implies that the $S_i'$-active path via which $S_i'$ reaches $v$ is $S_i'$-safe. 

If the set $S_i$ contributed to $e$, it reached $v$ before $e$ became tight, and thus strictly before time~$t$ and no later than time $\atf{S_1}{v}$.
Otherwise, $S_i$ contributed to $\lr{w, v}$ before $e$ became tight.
 Then in particular, the edge $e$ is not tight at time $\atf{R}{w}\leq\atf{S_i}{w}$. 
 Consider a set $S\in\cS$ that actively reaches $w$ at time $\atf{R}{w}$. 
 By \Cref{lem:first_implies_first_at_neighbor_or_tight_incoming_edge}, we have $\atf{S}{v}=\atf{R}{v}>\atf{R}{w}=\atf{S}{w}$. Hence, $S$ or a super-set of $S$ contributed to $\lr{w, v}$. 
Recall that the set $S_i$ also contributed to $\lr{w, v}$. 
Thus, by \ref{itm: uniform edges} of \Cref{def:well-subdivided} (well-subdivided instances), the sets $S_i$ and $S$ reach $v$ at the same time, namely at time $\atf{S}{v} = \atf{R}{v}$. 
Before time $\atf{R}{v}$ no set can contribute to $e$ and thus $\atf{R}{v} < t$.
We conclude that the set $S_i$ indeed reaches $v$ strictly before time $t$ and no later than time $\atf{S_1}{v}$ (because $\atf{S_1}{v} \geq \atf{R}{v}$).
\end{proof}

By \Cref{claim:active_reachability_of_non_safe_edge}, we may assume without loss of generality that $\atf{S_1}{v}\geq\atf{S_2}{v}\geq\atf{S_3}{v}$.
Moreover, by \Cref{claim:active_reachability_of_non_safe_edge,lem: inv holds until first accident}, for all $i\in \{1, 2, 3\}$ we can apply \Cref{inv: the invariant} for $S_i'$ and $v$, that is, there are tuples $\lr{S_i', v, Q_i, m_i, X_i, \bar X_i}$ fulfilling \Cref{inv: the invariant}.\\
Let $i, j\in\lrb{1,2,3}$. By \Cref{lem: wlog merge X bar is less than tRv}, for any time $t'\geq\atf{R}{v}$, all vertices of $\bar X_i$ are contained in the same active set (that is the same set from the partition $\cS^{t'}$). 
Because $S_j'$ is active at time $\atf{S_j'}{v}\geq\atf{R}{v}$, the set $\bar X_i$ is either a non-empty subset of $S_j'$ or it is disjoint from $S_j'$.
\Cref{fig: accident example case 1,fig: accident example case 2} provide some examples and illustrate the tuples $\lr{S_i', v, Q_i, m_i, X_i, \bar X_i}$ for two different scenarios.

To derive the desired contradiction, we will to choose three (non-empty) sets out of $X_1$, $\bar X_1$, $X_2$, $\bar X_2$, $X_3,$ $\bar X_3$. We then do two things:
First, we upper bound the cost (for connecting) the union of these three sets, and second, we show that the union of these three sets is late-merged and thus strictly $\gamma$-expensive (\Cref{obs: connecting late-merged set is expensive}).

By the definition of late-merged sets (\Cref{def: late-merged}), the union of the three sets is late-merged if each union of two of the sets is late-merged:
\begin{observation}\label{obs: pairwise late-merged implies late-merged}
    Let $A, B, C\subseteq R$ be three sets, such that $A\cup B$, $B\cup C$, and $C\cup A$ are all late-merged, then also $A\cup B\cup C$ is late-merged.
\end{observation}

We will always choose at least two of $X_1$, $X_2$, and $X_3$. Hence, we in particular need that the union of two of these sets is late-merged:

\begin{observation}\label{obs: Xa union Xb late-merged}
    For $a,b\in\lrb{1, 2, 3}$ the set $X_a\cup X_b$ is late-merged.
\end{observation}
\begin{proof}
    Assume without loss of generality that $\atf{S_a'}{v}\leq\atf{S_b'}{v}$. 
    By \Cref{lem: X and X bar atfRm close to m implies X union X bar atfSv close to v}, both $X_a\cup\bar X_a$ and $X_b\cup\bar X_b$ are $\atf{S_b'}{v}$-close to $v$. 
    Hence, \eqref{ieq: t close dist} from the definition of $t$-closeness implies that for any $x\in X_a$ and $y\in X_b$,
    \begin{align*}
        \dist{x, y}\ \leq\ \dist{x, v}+\dist{y, v}
        \ \leq\ 2\cdot\frac{9-7\gamma}{7-5\gamma}\cdot\atf{S_b'}{v}
        \ \leq\ \frac{18-14\gamma}{7-5\gamma}\cdot\merge{}{x, y}.
    \end{align*}
\end{proof}

Which three sets we choose depends on which of $\bar X_1, \bar X_2, \bar X_3$ are empty. If all of them are empty, our only option is to choose $X_1, X_2, X_3$. However, in general, it could be better to choose, for example, $X_1, X_2, \bar X_2$ because we can connect them without using the path $Q_3$, which leads to a better upper bound on the cost. (See the example in \Cref{fig: accident example case 2}.)
We distinguish the following two cases:

\paragraph{Case 1: Both $\mathbf{\bar X_2}$ and $\mathbf{\bar X_3}$ are empty. }\label{par: Xbar2 and Xbar3 are empty}
\hfill

\begin{figure}
    \centering
    \resizebox{0.8\textwidth}{!}{\begin{tikzpicture}
    \Large
    \useasboundingbox (-9,-9) rectangle (9,6.5);
    \node[steiner, label=above:{$v$}] (V) at ( 0,-1) {};
    \node[steiner, label=below:{$m_1$}] (M1) at ( 3, 1) {};

    \node[smallterminal, draw = green!60!black] (X11) at (-2, 5) {};
    \node[smallterminal, draw = green!60!black] (X12) at (-3, 3) {};
    \node[smallterminal, draw = green!60!black] (X13) at (-3, 1) {};
    
    \node[smallterminal, gray] (Xb1) at ( 5, 5) {};
    \node[smallterminal, gray] (Xb2) at ( 7, 5) {};
    \node[smallterminal, gray] (Xb3) at ( 9, 3) {};
    \node[smallterminal, gray] (Xb4) at ( 9, 1) {};
    
    \node[smallterminal, draw = red] (X21) at ( 7,-5) {};
    \node[smallterminal, draw = red] (X24) at ( 5,-7) {};
    
    \node[smallterminal, draw = blue] (X31) at (-7,-3) {};
    \node[smallterminal, draw = blue] (X32) at (-7,-5) {};
    \node[smallterminal, draw = blue] (X33) at (-5,-7) {};

    \draw[activeset, blue]           (-5.5,-5) ellipse [x radius=3, y radius=3];
    \draw[activeset, red]            ( 5.5,-5.5) ellipse [x radius=3, y radius=3];
    \draw[activeset, green!60!black] (-2, 3) ellipse [x radius=3, y radius=3];

    \draw[mstedge, green!60!black] (X11) -- ( 0, 3) -- (-2, 2) -- (X12) -- (-2, 2) -- (X13) -- (-2, 2) -- ( 0, 3) -- (M1);
    \draw[mstedge, gray] (Xb1) -- ( 6, 4) -- (Xb2) -- ( 6, 4) -- ( 6, 2) -- ( 8, 2) -- (Xb3) -- ( 8, 2) -- (Xb4) -- ( 8, 2) -- ( 6, 2) -- (M1);
    \draw[mstedge, red] (X21) -- ( 4,-4) -- (X24) -- ( 4,-4) -- (V);
    \draw[mstedge, blue] (X31) -- (-6,-4) -- (X32) -- (-6,-4) -- (-4,-4) -- (X33) -- (-4,-4) -- (V);
    
    \draw[mstedge, curlyHalfA, green!60!black] (M1) -- (V);
    \draw[mstedge, curlyHalfB, red!50!white] (M1) -- (V);
    
    \node[edgelabel, text = green!60!black] (S1) at (0.3, 5) {$S_1=S_1'$};
    \node[edgelabel, text = red] (S2) at ( 6,-8.4) {$\ \ S_2=S_2'\ \ $};
    \node[edgelabel, text = blue] (S3) at (-7.3,-7.3) {$S_3=S_3'$};
    
    \node[edgelabel] (Q1) at (  2,-0.3) {$Q_1$};
    
    \node[text = green!60!black] (X1) at (-3.5, 4) {$X_1$};
    \node[text = gray] (Xbar1) at (8,4) {$\bar X_1$};
    \node[text = red] (X2) at (6,-6) {$X_2$};
    \node[text = blue] (X3) at (-5.8,-5.8) {$X_3$};
 
\end{tikzpicture} }
    \caption{An abstract drawing of three disjoint sets $S_1$, $S_2$, and $S_3$ that all (actively) reach the same vertex $v$ as such that \Cref{inv: the invariant} applies. 
    The shown configuration belongs to \hyperref[par: Xbar2 and Xbar3 are empty]{Case 1} of our proof. The tuple $\lr{S_1,v,Q_1,m_1,X_1,\bar X_1}$ satisfies case \ref{item: help case} of \Cref{inv: the invariant} and for $i=2,3$, the tuple $\lr{S_i, v, \emptyset, v, X_i,\emptyset}$ satisfies case \ref{item: no help case} of \Cref{inv: the invariant}. That is, $X_2$ and $X_3$ are close to $v$, the set $X_1$ is close to $m_1$ and we can upper bound the length of the $m_1$-$v$ path $Q_1$.
    Using that $X_1$, $X_2$, and $X_3$ are separated by the active sets $S_1$, $S_2$, and $S_3$, we can lower bound $\val{X_1\cup X_2\cup X_3}$ and show that $X_1\cup X_2\cup X_3$ is late-merged (and thus strictly $\gamma$-expensive).
    To derive a contradiction, we prove that $X_1\cup X_2\cup X_3$ can be connected in a cheap way (using the path $Q_1$ and all shown components except for the gray one).
    The set $\bar X_2$ could also be contained in $S_2$, or in $S_3$. 
    However, for our proof in \hyperref[par: Xbar2 and Xbar3 are empty]{Case 1} this is not relevant, as we will anyways not connect the terminals in $\bar X_1$.}
    \label{fig: accident example case 1}
\end{figure}
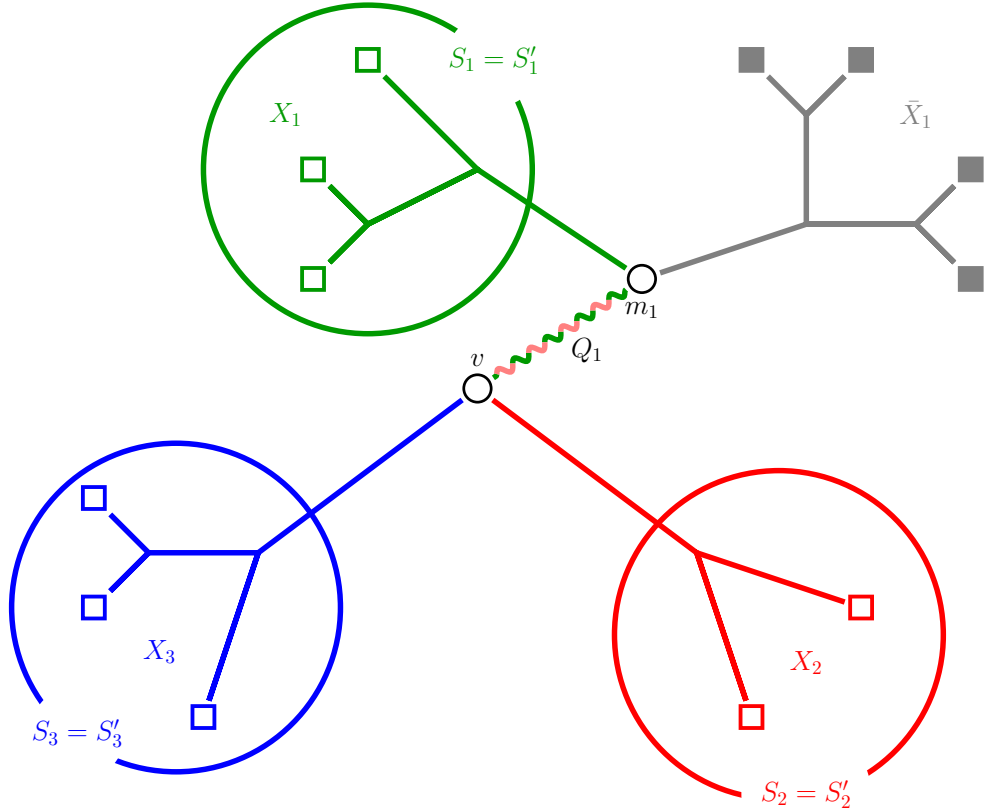

An example is shown in \Cref{fig: accident example case 1}.
In this case, \ref{item: no help case} of \Cref{inv: the invariant} applies, and thus both $X_2$ and $X_3$ are $\atf{S_2}{v}$-close respectively $\atf{S_3}{v}$-close to $v$, where we used \eqref{eq: atfs are identical for S and Sprime}.
The set $X_1$ might not be $\atf{S_1}{v}$-close to $v$, but we can still bound to cost of connecting it to $v$ by using $Q_1$ and the closeness to $m_1$. 
More precisely, if $\bar X_1 \neq \emptyset$, then \ref{item: help case} of \Cref{inv: the invariant} applies and yields
\begin{align*}
   \conn{X_1, v} \ \leq&\  \atf{R}{m_1} + \frac{8-8\gamma}{7-5\gamma}\, \val{X_1}  + c(Q_1) \\
   \leq&\  \atf{R}{m_1} + \frac{8-8\gamma}{7-5\gamma}\, \val{X_1} + 2\atf{S_1}{v}+\atf{R}{v}-3\atf{R}{m_1} \\
   \leq&\ \atf{R}{m_1} + \frac{8-8\gamma}{7-5\gamma}\, \val{X_1} + 3\atf{S_1}{v}-3\atf{R}{m_1} \\
   =&\ \atf{S_1}{v} + \frac{8-8\gamma}{7-5\gamma}\, \val{X_1} + 2\lr{\atf{S_1}{v}-\atf{R}{m_1}}.
\end{align*}
By \Cref{cor: lin comb of dist and connect bound}, this implies
\begin{align*}
   \conn{X_1, v} \ \leq&\ \atf{S_1}{v} + \frac{8-8\gamma}{7-5\gamma}\, \val{X_1} + \frac{2-6\gamma}{7-5\gamma}\, \atf{S_1}{v}.
\end{align*}
If $\bar X_1 = \emptyset$, then \ref{item: no help case} of \Cref{inv: the invariant} applies, which also yields
\begin{align*}
   \conn{X_1, v} \ \leq&\  \atf{S_1}{v} + \frac{8-8\gamma}{7-5\gamma}\, \val{X_1} \
   \leq\ \atf{S_1}{v} + \frac{8-8\gamma}{7-5\gamma}\, \val{X_1} +  \frac{2-6\gamma}{7-5\gamma}\, \atf{S_1}{v}.
\end{align*}
We conclude that no matter whether $\bar X_1$ is empty or not, we have
\begin{align*}
    \conn{X_1\cup X_2\cup X_3}
    \leq&\ \conn{X_1, v}+\conn{X_2, v}+\conn{X_3, v} \\[2mm]
    \leq&\ \frac{8-8\gamma}{7-5\gamma}\lr{\val{X_1}+\val{X_2}+\val{X_3}}\\
    &\ +\atf{S_1}{v}+\atf{S_2}{v}+\atf{S_3}{v} 
       +\frac{2-6\gamma}{7-5\gamma}\, \atf{S_1}{v}.   
\end{align*}
Because $\atf{S_1}{v}\geq\atf{S_2}{v}\geq\atf{S_3}{v}$, and  using $1+\frac{2-6\gamma}{7-5\gamma}<\frac{12-12\gamma}{7-5\gamma}$ and $3+\frac{2-6\gamma}{7-5\gamma}<\frac{24-24\gamma}{7-5\gamma}$, we can simplify our upper bound to
\begin{align*}
    \conn{X_1\cup X_2\cup X_3}
    &<\frac{12-12\gamma}{7-5\gamma}\, \lr{\val{X_1}+\val{X_2}+\val{X_3}+\atf{S_1}{v}+\atf{S_2}{v}}\\
    &\leq\frac{12-12\gamma}{7-5\gamma}\ \val{X_1\cup X_2\cup X_3}\,,\defaulttag
    \label{ieq: final bound first case}
\end{align*}
where the second inequality follows from \Cref{lem: drop of X union X bar} (applied twice).

By \Cref{obs: Xa union Xb late-merged,obs: pairwise late-merged implies late-merged}, the set $X_1\cup X_2\cup X_3$ is late-merged, contradicting \eqref{ieq: final bound first case}.

\paragraph{Case 2: At least one of $\mathbf{\bar X_2}$ and $\mathbf{\bar X_3}$ is non-empty.}\label{par: one of Xbar2 and Xbar3 is non-empty}
\hfill

An example of this case is shown in \Cref{fig: accident example case 2}.

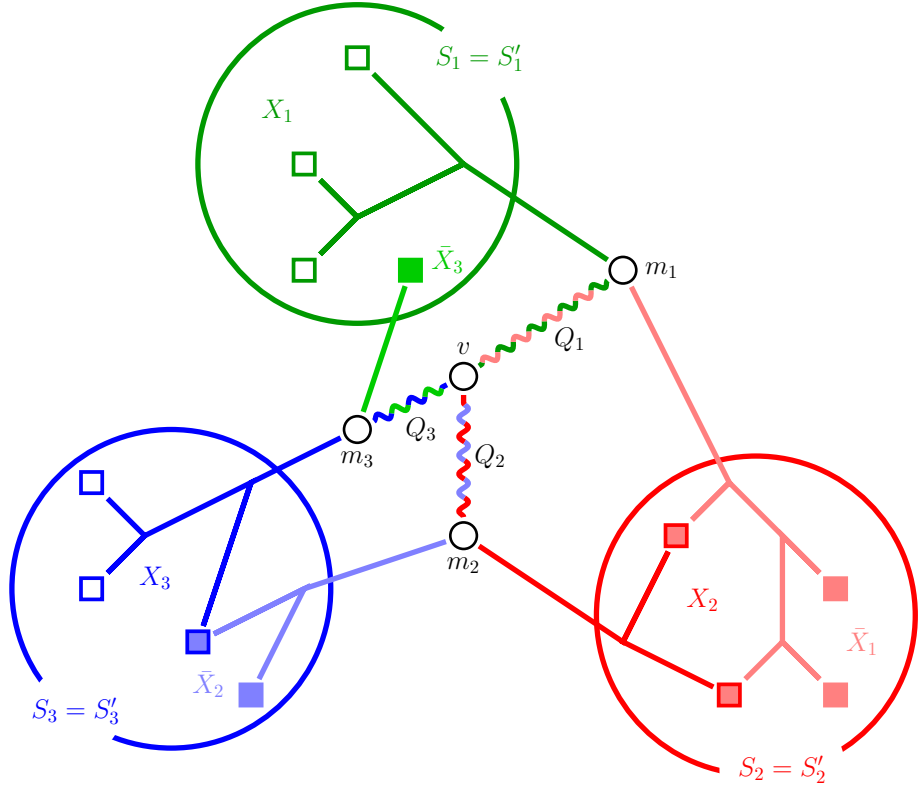
\begin{figure}
    \centering
    \resizebox{0.75\textwidth}{!}{\begin{tikzpicture}
    \Large
    \useasboundingbox (-8.7,-8.7) rectangle (8.7,6.2);
    \node[steiner, label=above:{$v$}] (V) at ( 0,-1) {};
    \node[steiner, label=right:{$m_1$}] (M1) at ( 3, 1) {};
    \node[steiner, label=below:{$m_2$}] (M2) at ( 0, -4) {};
    \node[steiner, label=below:{$m_3$}] (M3) at (-2, -2) {};

    \node[smallterminal, draw = red, fill = red!50!white] (X21) at ( 5,-7) {};
    \node[smallterminal,                    red!50!white] (X22) at ( 7,-7) {};
    \node[smallterminal,                    red!50!white] (X23) at ( 7,-5) {};
    \node[smallterminal, draw = red, fill = red!50!white] (X24) at ( 4,-4) {};
    
    \node[smallterminal, draw = blue                      ] (X31) at (-7,-3) {};
    \node[smallterminal, draw = blue                      ] (X32) at (-7,-5) {};
    \node[smallterminal, draw = blue, fill = blue!50!white] (X33) at (-5,-6) {};
    \node[smallterminal,                     blue!50!white] (X34) at (-4,-7) {};

    \node[smallterminal, draw = green!60!black] (X11) at (-2, 5) {};
    \node[smallterminal, draw = green!60!black] (X12) at (-3, 3) {};
    \node[smallterminal, draw = green!60!black] (X13) at (-3, 1) {};
    \node[smallterminal,        green!80!black] (X14) at (-1, 1) {};
    
    \draw[activeset, blue]           (-5.5,-5) ellipse [x radius=3, y radius=3];
    \draw[activeset, red]            ( 5.5,-5.5) ellipse [x radius=3, y radius=3];
    \draw[activeset, green!60!black] (-2, 3) ellipse [x radius=3, y radius=3];

    \draw[mstedge, green!60!black] (X11) -- ( 0, 3) -- (-2, 2) -- (X12) -- (-2, 2) -- (X13) -- (-2, 2) -- ( 0, 3) -- (M1);
    \draw[mstedge, red!50!white] (X21) -- ( 6,-6) -- (X22) -- ( 6,-6) -- ( 6,-4) -- (X23) -- ( 6,-4) -- ( 5,-3) -- (X24) -- ( 5,-3) -- (M1);
    \draw[mstedge, red] (X21) -- ( 3,-6) -- (X24) -- ( 3,-6) -- (M2);
    \draw[mstedge, blue!50!white] (X34) -- (-3,-5) -- (X33) -- (-3,-5) -- (M2);
    \draw[mstedge, blue] (X31) -- (-6,-4) -- (X32) -- (-6,-4) -- (-4,-3) -- (X33) -- (-4,-3) -- (M3);
    \draw[mstedge, green!80!black] (X14) -- (M3);

    \draw[mstedge, curlyHalfA, green!60!black] (M1) -- (V);
    \draw[mstedge, curlyHalfB, red!50!white] (M1) -- (V);

    \draw[mstedge, curlyHalfA, red] (M2) -- (V);
    \draw[mstedge, curlyHalfB, blue!50!white] (M2) -- (V);
    
    \draw[mstedge, curlyHalfA, blue] (M3) -- (V);
    \draw[mstedge, curlyHalfB, green!80!black] (M3) -- (V);

    \node[edgelabel, text = green!60!black] (S1) at (0.3, 5) {$S_1=S_1'$};
    \node[edgelabel, text = red] (S2) at ( 6,-8.4) {$\ \ S_2=S_2'\ \ $};
    \node[edgelabel, text = blue] (S3) at (-7.3,-7.3) {$S_3=S_3'$};
    
    \node[edgelabel] (Q1) at (  2,-0.3) {$Q_1$};
    \node[edgelabel] (Q2) at (0.5,-2.5) {$Q_2$};
    \node[edgelabel] (Q3) at (-0.8,-2) {$Q_3$};
    
    \node[text = green!60!black] (X1) at (-3.5, 4) {$X_1$};
    \node[text = red!50!white] (Xbar1) at (7.5,-6) {$\bar X_1$};
    \node[text = red] (X2) at (4.5,-5.2) {$X_2$};
    \node[text = blue!50!white] (Xbar2) at (-4.8,-6.8) {$\bar X_2$};
    \node[text = blue] (X3) at (-5.8,-4.8) {$X_3$};
    \node[text = green!80!black] (Xbar3) at (-0.3, 1.2) {$\bar X_3$};

\end{tikzpicture} }
    \caption{An abstract drawing of three disjoint sets $S_1$, $S_2$, and $S_3$ that all (actively) reach the same vertex $v$ as such that \Cref{inv: the invariant} holds. 
    The shown configuration belongs to \hyperref[par: one of Xbar2 and Xbar3 is non-empty]{Case 2} of our proof. 
    Here, for $i\in\lrb{1,2,3}$, the tuple $\lr{S_i,v,Q_i,m_i,X_i,\bar X_i}$ satisfies case \ref{item: help case} of \Cref{inv: the invariant}. 
    In particular, both sets $X_i$ and $\bar X_i$ are close to $m_i$, and we can upper bound the length of the $m_i$-$v$ path~$Q_i$.
    \newline
    To derive a contradiction, we choose three sets among $X_1,\bar X_1,X_2,\bar X_2,X_3,\bar X_3$ whose union is strictly $\gamma$-expensive, and we show that connecting all terminals of these three sets to $v$ is too cheap. 
    There are two aspects to consider for our choice: we want to choose three sets for which we (i) know how to cheaply connect them (using the components and paths shown in the figure) and (ii) know that their union has large value (which implies connecting them must be expensive). 
To argue that the union of the three sets has large value, we will choose them such that they are separated by (two of) the sets $S_1,S_2,S_3$.
    (Then \Cref{lem: drop of X union X bar} will yield a good lower bound on the value of the union of our three sets.)
    In this example, this means that we want to choose one of the green sets $X_1$ and $\bar X_3$, one of the red sets $X_2$ and $\bar X_1$, and one of the blue sets $X_3$ and $\bar X_2$.
    \newline
    Choosing $X_1$, $X_2$, and $X_3$ would be the easiest choice to ensure (ii). However, the cheapest way of connecting them might include all three path $Q_1$, $Q_2$, and $Q_3$ and the tree components connecting $X_i$ to $m_i$. In this example, we could instead choose $X_1$, $X_2$ and $\bar X_2$. Connecting them requires three components (two connecting $X_2$ and $\bar X_2$ to $m_2$ and one connecting $X_1$ to $m_1$) but only two of the three paths~$Q_i$.
    \newline
    In general, in \hyperref[par: one of Xbar2 and Xbar3 is non-empty]{Case 2} of our proof, the set $\bar X_2$ could also be contained in $S_1$ or it could be disjoint from all $S_1$, $S_2$, and $S_3$. But, by \Cref{lem: wlog merge X bar is less than tRv}, we know that $\bar X_2$ cannot contain elements of both $S_1$ and $S_3$ simultaneously. Hence, choosing $\bar X_1$ cannot simultaneously prevent us from choosing $X_1$ and from choosing $X_3$.
    (It is also possible that $\bar X_2$ is empty, in which case we use $X_3$ and $\bar X_3$ instead of $X_2$ and $\bar X_2$.)
    }
    \label{fig: accident example case 2}
\end{figure}

Let $j\in\lrb{2, 3}$ be an index such that $\bar X_j\neq\emptyset$. Our strategy will be to choose $X_j$ and $\bar X_j$ as two of the three sets. 
As the third set we will choose one of the sets $X_1, X_2, X_3$ (but not $X_j$ again).
We write $X_k$ to denote this third set, where $k\in \{1,2,3\} \setminus \{j\}$. 
Let $k\neq j$ such that exactly one of $X_k$ and $\bar X_j$ is contained in $S_1$
(which is possible because $X_1 \subseteq S_1$, but $X_2$ and $X_3$ are disjoint from $S_1$).
Recall that (by \Cref{lem: wlog merge X bar is less than tRv}) $\bar X_j$ is either contained in $S_1$ or disjoint from $S_1$.
Thus, our choice of $k$ guarantees that the terminals from $X_k$ and $\bar X_j$ are not merged before time $d^{S_1}\geq\atf{S_1}{v}$ (where the inequality follows from $P+e$ being an $S_1$-active path with $e=(v,w)$).

We will now use this property to prove that $X_j\cup\bar X_j\cup X_k$ is late-merged. 
We already know that $X_j\cup\bar X_j$ is late-merged (\Cref{lem: inv and m not v implies that X union X bar is late-merged}) and that $X_j\cup X_k$ is late-merged (\Cref{obs: Xa union Xb late-merged}).
Thus, by \Cref{obs: pairwise late-merged implies late-merged}, it suffices to show that $X_k\cup\bar X_j$ is late merged.

\begin{observation}\label{obs: Xjbar union Xk late-merged}
    $\bar X_j\cup X_k$ is late-merged.
\end{observation}
\begin{proof}
    By our choice of $k$, exactly one of $X_k$ and $\bar X_j$ is contained in $S_1$ and the other set is disjoint from $S_1$. 
    Further, by \Cref{claim:active_reachability_of_non_safe_edge}, $\atf{S_1}{v}\geq\max\lrb{\atf{S_2}{v},\atf{S_3}{v}}$, and hence, $\atf{S_1'}{v}\geq\max\lrb{\atf{S_j'}{v}, \atf{S_k'}{v}}$ by \eqref{eq: atfs are identical for S and Sprime}. 
    Thus, by \Cref{lem: X and X bar atfRm close to m implies X union X bar atfSv close to v}, both $X_j\cup\bar X_j$ and $X_k\cup\bar X_k$ are $\atf{S_1'}{v}$-close to~$v$. 
    Hence, using \eqref{ieq: t close dist} from the definition of $\atf{S_1'}{v}$-closeness, we get for all $x\in\bar X_j$ and $y\in X_k$,
    \begin{align*}
        \dist{x, y}\ \leq\ \dist{x, v}+\dist{y, v}
        \ \leq\ 2\cdot\frac{9-7\gamma}{7-5\gamma}\cdot\atf{S_1'}{v}
        \ \leq\ \frac{18-14\gamma}{7-5\gamma}\cdot\merge{}{x, y}.
    \end{align*}
\end{proof}

We conclude that $X_j\cup\bar X_j\cup X_k$ is late-merged and thus strictly $\gamma$-expensive (\Cref{obs: connecting late-merged set is expensive}).
To derive the desired contradiction, we will now establish an upper bound 
on the cost of connecting $X_j\cup\bar X_j\cup X_k$.

We know that both $X_j$ and $\bar X_j$ are $\atf{R}{m_j}$-close to $m_j$ and $X_k$ is $\atf{R}{m_k}$-close to $m_k$. 
Thus, using the $m_j$-$v$ path $Q_j$ and the $m_k$-$v$ path $Q_k$, we get
\begin{align*}
    &\conn{X_j\cup\bar X_j\cup X_k}\\
    \leq&\ \conn{X_j, m_j}+\conn{\bar X_j, m_j}+\conn{X_k, m_k}+c\lr{Q_j}+c\lr{Q_k}\\
    \leq&\ \frac{8-8\gamma}{7-5\gamma}\lr{\val{X_j}+\val{\bar X_j}+\val{X_k}}+2\atf{R}{m_j}+\atf{R}{m_k}+c\lr{Q_j}+c\lr{Q_k}\defaulttag\label{ieq: bound on Xj barXj Xk}\,.
\end{align*}
Using \Cref{cor: lin comb of dist and connect bound}, we can bound the length of the path $Q_j$ as follows:
\begin{align}
    c\lr{Q_j}\ 
    &\leq\ 2\atf{S_j}{v}+\atf{R}{v}-3\atf{R}{m_j}
    \ <\ \lr{1+\frac{1-3\gamma}{7-5\gamma}}\,\atf{S_j}{v}+\atf{R}{v}-2\atf{R}{m_j}\,.
    \label{ieq: bound on Qj}
\end{align}
If $\bar X_k$ is non-empty, we can similarly bound $c\lr{Q_k}$, using again \Cref{cor: lin comb of dist and connect bound}:
\begin{align}
    c\lr{Q_k}
    \ &\leq\ 2\atf{S_k}{v}+\atf{R}{v}-3\atf{R}{m_k}
    \ \leq\ \frac{2-6\gamma}{7-5\gamma}\ \atf{S_k}{v}+\atf{R}{v}-\atf{R}{m_k}\,.
    \label{ieq: bound on Qk}
\end{align}
If $\bar X_k$ is empty, then $v=m_k$ and $c\lr{Q_k}=0$, and thus \eqref{ieq: bound on Qk} is trivially satisfied.

Combining \eqref{ieq: bound on Xj barXj Xk}, \eqref{ieq: bound on Qj}, and \eqref{ieq: bound on Qk} yields
\begin{align*}
    \conn{X_j\cup\bar X_j\cup X_k}
    &<\frac{8-8\gamma}{7-5\gamma}\lr{\val{X_j}+\val{\bar X_j}+\val{X_k}}\\
    &\qquad+\frac{8-8\gamma}{7-5\gamma}\atf{S_j}{v}+\frac{2-6\gamma}{7-5\gamma}\atf{S_k}{v}+2\atf{R}{v}.
\end{align*}
Using $2\atf{R}{v}=\frac{14-10\gamma}{7-5\gamma}\,\atf{R}{v}\leq\frac{4-4\gamma}{7-5\gamma}\,\atf{S_j}{v}+\frac{10-6\gamma}{7-5\gamma}\,\atf{S_k}{v}$ (by \Cref{obs: larger sets arrive no later}), this implies
\begin{align*}
    \conn{X_j\cup\bar X_j\cup X_k}\
       <&\ \frac{12-12\gamma}{7-5\gamma}\lr{\val{X_j}+\val{\bar X_j}+\val{X_k}+\atf{S_j}{v}+\atf{S_k}{v}}\\
    \leq&\ \frac{12-12\gamma}{7-5\gamma}\lr{\val{X_j\cup\bar X_j}+\val{X_k}+\atf{S_k}{v}}\\
    \leq&\ \frac{12-12\gamma}{7-5\gamma}\val{X_j\cup\bar X_j\cup X_k},
\end{align*}
where the last two inequalities follow from \Cref{lem: drop of X union X bar}:
first, applied to the set $S_j'$, which is active at time $\atf{S_j}{v}$ and separates $X_j$ and $\bar X_j$, and second, applied to the set $S_k'$, which is active at time $\atf{S_k}{v}$ and separates $X_j\cup\bar X_j$ and $X_k$.
This contradicts the fact that $X_j\cup\bar X_j\cup X_k$ is strictly $\gamma$-expensive.
 \section{Existence of Good Merge Plans}
\label{sec:existence_merge_plan}

The goal of this section is to prove the existence of $\gamma$-good merge plans whose value is large whenever the cost $\opt$ of a cheapest Steiner tree is large.
In particular, we will prove \Cref{thm: there is a gamma good merge plan with high value} and show how our upper bound on the integrality gap of \ref{eq:bcr} (\Cref{thm:main_gap}) follows from \Cref{cor: gamma good merge plan gives large dual}.

To this end, we will consider the Steiner tree resulting from Zelikovsky's relative greedy algorithm~\cite{ImprovedRelGreedy}, whose cost yields an upper bound on $\opt$ (see \Cref{sec: bounding opt}).
Moreover, for every $\gamma \in [0,\frac{1}{5}]$, we will construct a $\gamma$-good merge plan  $\cM_{\gamma}$. 
This will lead to the following theorem.

\begin{theorem}\label{thm: there is a good merge plan with rho}
    For every instance $\cI$ of the Steiner tree problem, there is a non-increasing function $\rho_{\cI}:\lre{0, \frac12}\to\lre{0, 1}$ such that for every $\gamma \in [0,\frac{1}{5}]$, there exists a $\gamma$-good merge plan $\cM_{\gamma}$ with
    \begin{equation}\label{eq:lower_bound_value_using_rho}
        \frac{\val{M_{\gamma}}}{\tmst} \ \geq\ \frac{7-5\gamma}{12}\lr{1-\rho_{\cI}\lr\gamma\,\frac{3-7\gamma}{9-7\gamma}}
    \end{equation}
    and
    \begin{equation}\label{eq:upper_bound_opt_using_rho}
     \frac{\opt}{\tmst}\ \leq\ 1-\int_0^{\frac12}\rho_{\cI}\lr\gamma \ d\gamma.
    \end{equation}
\end{theorem}

Observe that the lower bound \eqref{eq:lower_bound_value_using_rho} on the value of our merge plan gets weaker, if $\rho_{\cI}(\gamma)$ is large.
However, if this is the case for all (or many) values of $\gamma$, then \eqref{eq:upper_bound_opt_using_rho} is a strong upper bound on $\opt$.
In fact, \Cref{thm: there is a good merge plan with rho} is a strengthening of \Cref{thm: there is a gamma good merge plan with high value}.
Indeed, by \eqref{eq:upper_bound_opt_using_rho}, for every $\gamma^* \in [0,\frac{1}{5}]$, we have
\[
\tmst-\opt \ \geq \ \tmst \cdot \int_0^{\frac12}\rho_{\cI}\lr{\gamma}\ d\gamma
\ \geq\  \tmst \cdot \int_0^{\gamma^*}\rho_{\cI}\lr{\gamma}\ d\gamma \ \geq\ \gamma^* \cdot \rho_{\cI}(\gamma^*) \cdot \tmst,
\]
which together with \eqref{eq:lower_bound_value_using_rho} implies
\[
 \val{\cM_{\gamma^*}}\ \geq\ \frac{7-5\gamma^*}{12}\lr{\tmst - \frac{\tmst -\opt}{\gamma^*} \cdot \frac{3-7\gamma^*}{9-7\gamma^*}}\,.
\]
As explained in \Cref{sec:outline_merge_plan}, combining \Cref{thm: there is a gamma good merge plan with high value} with \Cref{cor: gamma good merge plan gives large dual} readily implies an upper bound on the integrality gap of \ref{eq:bcr} that is smaller than $2$. 
However, to prove our  upper bound of $\finalratio$ on the integrality gap, we will directly work with \Cref{thm: there is a good merge plan with rho} instead of \Cref{thm: there is a gamma good merge plan with high value}.

We will prove \Cref{thm: there is a good merge plan with rho} in \Cref{sec: bounding opt,sec: constructing large gamma good merge plan}. 
First, in \Cref{sec: bounding opt}, we use Zelikovsky's relative greedy algorithm~\cite{ImprovedRelGreedy} to  define a function $\rho_{\cI}$ satisfying \eqref{eq:upper_bound_opt_using_rho}.
Then, in \Cref{sec: constructing large gamma good merge plan}, we construct a $\gamma$-good merge plan satisfying \eqref{eq:lower_bound_value_using_rho}. 
Finally, in \Cref{sec: 1.9 in general} we use \Cref{thm: there is a good merge plan with rho} to prove our upper bound on the integrality gap of \ref{eq:bcr} (\Cref{thm:main_gap}) follows from \Cref{thm: there is a good merge plan with rho}.

\subsection{Defining the Function $\rho_{\cI}$ and Bounding \opt}\label{sec: bounding opt}

We consider a fixed instance $\cI$, for which we will define the function $\rho_{\cI}\lr{\gamma}$ and show that it satisfies~\eqref{eq:upper_bound_opt_using_rho}.
Because the instance $\cI$ is fixed, we omit the index $\cI$ in the following.

We construct a Steiner tree by an (exponential-time) relative greedy algorithm.
Starting with our initial graph $G\eqqcolon G_1$, we choose a terminal set $X_1$ and a full component $K_1$ that connects $X_1$ such that $\frac{c\lr{K_1}}{\drop[G_1]{X_1}}$ is minimized. We contract $X_1$ to obtain a new graph $G_2$. We iterate this until we obtain an MST-optimal instance. This yields sets $X_1, \dots, X_j$, components $K_1, \dots, K_j$, and graphs $G_1, \dots, G_{j+1}$.
This algorithm is essentially the same as Zelikovsky's algorithm~\cite{ImprovedRelGreedy}, 
but because our goal is an upper bound on the integrality gap rather than an algorithmic result, we do not need to restrict ourselves to components of constant size.

For $i\in \{1,\dots, j\}$ and $X\subseteq R$, we write $\dropi_i(X) \coloneqq \dropi_{G_i}(X)$.
Then we observe that the ratios $\frac{c\lr{K_i}}{\drop[G_i]{X_i}}$ are non-decreasing and always between $\frac{1}{2}$ and $1$.

\begin{observation}\label{obs:goodness_decreasing}
    $$
        \frac12\leq\frac{c\lr{K_1}}{\dropi_1(X_1)}\leq\dots\leq\frac{c\lr{K_j}}{\dropi_j(X_j)}<1
    $$
\end{observation}
\begin{proof}
Each \emph{full} component that exists in some $G_i$, already existed in $G$ and has the same cost in both graphs. Furthermore, by \Cref{obs:drop reduces under contraction}, the drop of a terminal set can only decrease by contracting components.
Thus, by the choice of our components and terminal sets, for all $i\in \{1,\dots, j-1\}$, we have
$$
\frac{c\lr{K_i}}{\dropi_{i}\lr{X_i}}
    \ \leq\ \frac{c\lr{K_{i+1}}}{\dropi_{i}\lr{X_{i+1}}}
    \ \leq\ \frac{c\lr{K_{i+1}}}{\dropi_{{i+1}}\lr{X_{i+1}}}.
$$
Moreover, the component $K_1$ connecting terminals $X_1$ cannot be cheaper than half the cost of a spanning tree on $X_1$ in the metric closure (because we obtain such a spanning tree by doubling all edges of $K_1$ and short-cutting an Eulerian walk in the resulting graph; similar to the proof of \Cref{prop:simple_lb}). 
The cost of such a spanning tree is at least $\drop[1]{X_1}$. 
Thus, the ratio $\frac{c\lr{K_1}}{\drop[1]{X_1}}$ is at least $\frac12$. Because we stop the algorithm when the remaining graph is MST-optimal, the last component $K_j$ for which we contracted the set $X_j$ was an improving component and thus $\frac{c\lr{K_j}}{\dropi_j(X_j)}<1$.
\end{proof}

\Cref{obs:goodness_decreasing} implies that for every $\gamma \in [0,\frac{1}{5}]$, there is an index $j(\gamma)$ such that the component $K_i$ is $\gamma$-improving with respect to $G_i$ if and only if $i\leq j(\gamma)$.
We define $\rho\lr{\gamma}$ to be the fraction of $\tmst$ that vanished due to contracting these $\gamma$-improving components. 
More precisely, $j(\gamma)$ is the last index such that $\frac{c\lr{K_{j({\gamma})}}}{\dropi_{j(\gamma)}(X_{j(\gamma)})} < 1-\gamma$ and we define
\begin{equation}\label{eq:definition_rho}
    \rho\lr{\gamma}\ \coloneqq\ \frac{\tmst - \tmst_{G_{j(\gamma)}}}{\tmst} \ =\ \frac1{\tmst}\sum_{1\leq i\leq j({\gamma})}\dropi_i(X_i). 
\end{equation}

The function $\rho$ is non-increasing. 
Moreover, it allows us to bound the cost of an optimal Steiner tree, as we show next.
\begin{lemma}\label{lem: opt with integral}
    The function $\rho$ satisfies \eqref{eq:upper_bound_opt_using_rho}.
\end{lemma}
\begin{proof}
For $\alpha, \gamma \in \mathbb{R}$, we write $\mathbbm{1}[\gamma < \alpha]$ to denote the indicator variable that is equal to $1$ if $\gamma < \alpha$ and equal to $0$ otherwise.
The union of a terminal MST in the final graph $G_{j+1}$ together with the contracted components $K_1,\dots, K_j$ is a Steiner tree for the original instance $\cI$.
Therefore,
\begin{align*}
    \opt&\leq\tmst_{G_{j+1}}+\sum_{1\leq i\leq j}c\lr{K_i}\\
    &=\tmst_{G_{j+1}}+\sum_{1\leq i\leq j}\Big(\dropi_i(X_i)-\big(\dropi_i(X_i)-c\lr{K_i}\big)\Big)\\
    &=\tmst_{G_{j+1}}+\sum_{1\leq i\leq j}\lr{\tmst_{G_i}-\tmst_{G_{i+1}}}-\sum_{1\leq i\leq j}\big(\dropi_i(X_i)-c\lr{K_i}\big)\\
    &=\tmst-\sum_{1\leq i\leq j}\lr{1-\frac{c\lr{K_i}}{\dropi_i(X_i)}}\cdot\dropi_i(X_i)\\
    &=\tmst-\sum_{1\leq i\leq j}\int_0^{\frac12}\mathbbm{1}\left[\gamma<1-\frac{c\lr{K_i}}{\dropi_i(X_i)}\right]\,d\gamma\cdot\dropi_i(X_i)\\
    &=\tmst-\int_0^{\frac12}\sum_{1\leq i\leq j}\mathbbm{1}\left[\gamma<1-\frac{c\lr{K_i}}{\dropi_i(X_i)}\right]\cdot\dropi_i(X_i)\ d\gamma\\
    &=\tmst-\int_0^{\frac12}\sum_{1\leq i\leq j\lr\gamma}\dropi_i(X_i)\ d\gamma\\
    &=\tmst-\int_0^{\frac12}\tmst\cdot\rho\lr\gamma\,d\gamma,
\end{align*}
where we used \eqref{eq:definition_rho} and the definition of the index $j(\gamma)$ in the last equations.
\end{proof}

To prove \Cref{thm: there is a good merge plan with rho}, it remains to show that for every $\gamma\in [0,\frac{1}{5}]$, there exists a $\gamma$-good merge plan with sufficiently large value, that is, a merge plan satisfying \eqref{eq:lower_bound_value_using_rho}.

\subsection{Existence of a $\gamma$-Good Merge Plan}\label{sec: constructing large gamma good merge plan}

In this section, we complete the proof of \Cref{thm: there is a good merge plan with rho} by showing that there exists a merge plan that is $\gamma$-good and has a large value. We do this using the tools from \Cref{sec:merge_plan}. Namely, we will construct an upper bound $u_{\gamma}$ and show that every merge plan that is upper bounded by $u_{\gamma}$ is $\gamma$-good (\Cref{lem: what we construct is actually gamma good}), and the cost of a terminal MST with respect to $u_{\gamma}$ is large~(\Cref{lem: MST wrt u is large}). By \Cref{lem: max merge plan has value mst and there are defining paths}, the latter implies the existence of a merge plan that is upper bounded by $u_{\gamma}$ and has large value.

We fix $\gamma \in [0,\frac{1}{5}]$.
By construction, the graph $G'\coloneqq G_{j(\gamma)+1}$ does not contain any $\gamma$-improving components.
This will allow us to establish a sufficient condition for a merge plan (for $G$) being $\gamma$-good.
To this end, we use the following notation: if $\cM_G$ is the canonical merge plan of $G$, we define
\begin{align*}
    \merge{G}{x,y}\coloneqq\merge{\cM_G}{x,y},\quad
    \val{G; X}\coloneqq\val{\cM_G, X},\quad
    \val{G}=\val{\cM_G},
\end{align*}
for $x,y\in R$ and $X\subseteq R$. The values $\merge{G'}{x,y}$, $\val{G', X}$, and $\val{G'}$ are defined analogously.

We define the upper bound $u_{\gamma}:\binom{R}{2}\to\R+$ by
\begin{align}
    u_{\gamma}\lr{\lrb{x, y}}=\max\lrb{\frac{7-5\gamma}{6}\,\merge{G'}{x, y},\ \frac{7-5\gamma}{9-7\gamma}\,\merge{G}{x, y}}\qquad\text{for all }x, y\in R\,.\label{ieq: definition of upper bound for gamma good merge plan}
\end{align}

This upper bound is chosen such that it implies $\gamma$-goodness of any merge plan upper bounded by $u_{\gamma}$.
To prove this, we need the following observation.

\begin{observation}\label{obs:monotonicity_value}
    Let $\cM_1$ and $\cM_2$ be two merge plans, let $X$ be a terminal set, and let $\alpha\in\R_{\geq 0}$. 
    If $\merge{\cM_1}{x,y}\leq\alpha\cdot\merge{\cM_2}{x, y}$ for all $x,y\in X$, then $\val{\cM_1, X}\leq\alpha\cdot\val{\cM_2,X}$.
\end{observation}
\begin{proof}
    Let $\cM_1=\lr{\cS_1^t}_{t\geq0}$ and $\cM_2=\lr{\cS_2^t}_{t\geq0}$.
    If two terminals $x$ and $y$ are merged in $\cM_2$ at time $t$, then they are merged  in $\cM_1$ no later than at time $\alpha\cdot t$.
    Hence, for any $t\geq0$ the partition $\cS_2^t$ is a refinement of the partition $\cS_1^{\,\alpha \cdot t}$.
    Consequently, the set $X$ cannot intersect more parts of $\cS_1^{\alpha\cdot t}$ than it intersects parts of $\cS_2^t$. By the definition of the local value (\Cref{def: local value of merge plan}), this implies
    \begin{align*}
        \val{\cM_1, X}\leq\alpha\cdot\val{\cM_2, X}\,.
    \end{align*}
\end{proof}

\begin{lemma}\label{lem: what we construct is actually gamma good}
    Every merge plan that is upper bounded by $u_{\gamma}$ is $\gamma$-good.
\end{lemma}
\begin{proof}
    Let $\cM$ be a merge plan that is upper bounded by $u_{\gamma}$.
    Let $K$ be a component that connects a set $X$ of terminals. 
    By the choice of $j(\gamma)$ and the greedy choice of components, 
    we have $c\lr{K}\geq\lr{1-\gamma}\dropGp{X}$. 
    We distinguish two cases: if $\val{\cM, X}\leq\frac{7-5\gamma}{12}\,\drop[G']{X}$, then
    $$
        c\lr{K}\geq\lr{1-\gamma}\dropGp{X}\geq\lr{1-\gamma}\frac{12}{7-5\gamma}\val{\cM, X}\,,
    $$
    and thus $K$ is $\gamma$-expensive with respect to $\cM$.\\
    Otherwise, if $\val{\cM, X}>\frac{7-5\gamma}{12}\,\drop[G']{X}=\frac{7-5\gamma}{6}\,\val{G', X}$ by \Cref{lem: drop in terms of canonical merge plan}.
    Thus, by \Cref{obs:monotonicity_value}, there are terminals $x, y\in X$ with $\merge{\cM}{x, y}>\frac{7-5\gamma}{6}\merge{G'}{x, y}$.
    Using that $\cM$ is upper bounded by $u_{\gamma}$, this implies
    $$
        \merge{\cM}{x, y}\ \leq\ \frac{7-5\gamma}{9-7\gamma}\,\merge{G}{x, y}\ \leq\ \frac{7-5\gamma}{18-14\gamma}\,\dist{x, y}\,.
    $$
    We conclude that $\cM$ is $\gamma$-good.
\end{proof}

\begin{lemma}\label{lem: MST wrt u is large}
    Let $\tmst_{u_{\gamma}}$ be the value of a terminal MST with respect to $u_{\gamma}$. Then
    $$
        \tmst_{u_{\gamma}}\geq\lr{\frac{7-5\gamma}{12}-\frac{7-5\gamma}{18-14\gamma}}\tmst_{G'}+\frac{7-5\gamma}{18-14\gamma}\tmst_G\,.\label{eq:mst_inequality_to_prove}
    $$
\end{lemma}

Before providing the proof of \Cref{lem: MST wrt u is large}, we first observe that it implies \Cref{thm: there is a good merge plan with rho}.

\begin{proof}[Proof of \Cref{thm: there is a good merge plan with rho}]
We define the function $\rho$ as in \eqref{eq:definition_rho}.
Then by \Cref{lem: opt with integral}, it satisfies \eqref{eq:upper_bound_opt_using_rho}.
Let $\cM_{\gamma}$ be the merge plan with maximum value that is upper bounded by $u_{\gamma}$.
 \Cref{lem: what we construct is actually gamma good} implies that $\cM_{\gamma}$ is $\gamma$-good. 
 Together, \Cref{lem: max merge plan has value mst and there are defining paths} and \Cref{lem: MST wrt u is large} imply
\begin{align*}
    \val{\cM_{\gamma}}=\tmst_{u_{\gamma}}\geq\lr{\frac{7-5\gamma}{12}-\frac{7-5\gamma}{18-14\gamma}}\tmst_{G'}+\frac{7-5\gamma}{18-14\gamma}\,\tmst_G\,.
\end{align*}
Moreover, by the definition~\eqref{eq:definition_rho} of the function $\rho$, we have 
\[
\tmst_{G'} = \lr{1-\rho\lr\gamma} \cdot \tmst_G,
\]
and hence the merge plan $\cM_\gamma$ satisfies
\begin{align*}
    \frac{\val{\cM_{\gamma}}}{\tmst_G}\geq\lr{\frac{7-5\gamma}{12}-\frac{7-5\gamma}{18-14\gamma}}\lr{1-\rho\lr\gamma}+\frac{7-5\gamma}{18-14\gamma}=\frac{7-5\gamma}{12}\lr{1-\rho\lr\gamma\frac{3-7\gamma}{9-7\gamma}}\,,
\end{align*}
which shows \eqref{eq:lower_bound_value_using_rho}.
\end{proof}

The intuition behind our proof of \Cref{lem: MST wrt u is large} is the following.
Let $T$ be a terminal MST with respect to $u_{\gamma}$.
Then by construction, we have $u_{\gamma}\lr{T}\geq\frac{7-5\gamma}{9-7\gamma}\,\val{G}=\frac{7-5\gamma}{18-14\gamma}\,\tmst_{G, c}$ (by \Cref{cor: canonical merge plan}). 
Whenever an edge $e=\lr{x, y}\in T$ is not contracted in $G'$, we want to improve on this simple bound using that $u_{\gamma}\lr{e}$ is not only at least $\frac{7-5\gamma}{9-7\gamma}\merge{G}{x, y}$ but at least $\frac{7-5\gamma}{6}\merge{G}{x, y}$. 
To prove that this strategy indeed yields \Cref{lem: MST wrt u is large}, we need the following key lemma.

\begin{lemma}\label{lem: existence of special MST wrt c bar}
    There is a terminal minimum spanning tree $T$ in $G$ with respect to $u_{\gamma}$ such that $T$ contains a spanning tree $T'$ for $G'$ consisting only of edges $\lrb{x, y}$ whose merge time in $G$ is the same as in $G'$, that is $\merge{G}{x, y}=\merge{G'}{x, y}$ for all $\{x,y\}\in T'$.
\end{lemma}

Given \Cref{lem: existence of special MST wrt c bar}, we can complete the proof of \Cref{lem: MST wrt u is large} as follows.

\begin{proof}[Proof of \Cref{lem: MST wrt u is large}]
Let $T$ and $T'$ be as in \Cref{lem: existence of special MST wrt c bar}. 
    Then for every edge $\{x,y\} \in T \setminus T'$, we have
    $u_{\gamma} \geq \frac{7-5\gamma}{9-7\gamma} \, \merge{G}{x, y}$ by the definition of $u_{\gamma}$.
    Moreover, every edge $\{x,y\} \in T'$ satisfies $\merge{G}{x,y}=\merge{G'}{x, y}$. Using $\frac{7-5\gamma}{6} > \frac{7-5\gamma}{9-7\gamma}$, this implies
    \[
    u_{\gamma}(\{x,y\}) \ =\ \frac{7-5\gamma}{6}\,\merge{G'}{x, y} 
    \ =\ \frac{7-5\gamma}{9-7\gamma} \, \merge{G}{x, y} + \lr{\frac{7-5\gamma}{6}-\frac{7-5\gamma}{9-7\gamma}} \, \merge{G'}{x, y}.
    \]
    Therefore, the total cost $u_{\gamma}\lr{T}$ is at least
    \begin{align}
        \tmst_{u_{\gamma}}\ \geq\ \frac{7-5\gamma}{9-7\gamma}\sum_{\lrb{x,y}\in T}\merge{G}{x, y}+\lr{\frac{7-5\gamma}{6}-\frac{7-5\gamma}{9-7\gamma}}\sum_{\lrb{x,y}\in T'}\merge{G'}{x, y}\,.\label{ieq: tmst wrt u geq weighted sums of merge times}
    \end{align}
    By \Cref{obs: MSTs have same cost wrt u and to merge times}, the two sums in \eqref{ieq: tmst wrt u geq weighted sums of merge times} can be lower bounded by $\frac12\,\tmst_G$ and $\frac12\,\tmst_G'$ respectively.
\end{proof}

It remains to prove \Cref{lem: existence of special MST wrt c bar}.

\paragraph{Proving \Cref{lem: existence of special MST wrt c bar}}

We define a cost function $c^*\colon \binom{R}{2} \to \mathbb{R}_{\geq 0}$ so that $c^*(x,y) < c^*(x',y')$ if and only if  the ordered pair $(u_{\gamma}(x,y), \merge{G}{x, y})$ is lexicographically smaller than the pair $(u_{\gamma}(x', y'), \merge{G}{x', y'})$.
For example, we can choose $c^*$ as
\[
c^*(x,y) \ \coloneqq\ M \cdot u_{\gamma}(x,y) + \merge{G}{x, y}
\]
for a sufficiently large number $M$.
Let $T$ be a spanning tree in $G$ that minimizes $c^*$.
Note that $T$ also minimizes $u_{\gamma}$.
Then, let $T' \subseteq T$ be a spanning tree for $G'$ that minimizes $c^*$ among all such trees.

To prove \Cref{lem: existence of special MST wrt c bar}, it remains to show $\merge{G'}{x, y}=\merge{G}{x, y}$ for all $\{x,y\}\in T'$.
Note that we always have $\merge{G'}{x, y} \leq \merge{G}{x, y}$
because $G'$ arises from $G$ by contractions. 
Hence, it remains to show that $T'$ does not contain any edge $\{x,y\}$ with $ \merge{G'}{x, y} < \merge{G}{x, y}$.

\begin{lemma}\label{lem: path with only cheap edges exists}
    Let $e=\lrb{x, y} \in E\lr{G'}$ be an edge with $\merge{G'}{x, y}<\merge{G}{x, y}$.
    Then, there exists an $x$-$y$ path in $G'$ consisting only of edges $f$ with $c^*\lr{f}<c^*\lr{e}$.
\end{lemma}
\begin{proof}
    By \Cref{lem: max merge plan has value mst and there are defining paths}, there is an $x$-$y$ path $P'$ in $G'$ such that
    \begin{align*}
        \merge{G'}{x, y}=\max_{\lrb{a, b}\in P'}\frac{\dist[G']{a, b}}{2}\,.
    \end{align*}
    We represent each edge of $P'$ by a a corresponding edge $\{a,b\}$ from $G$.
    By choosing suitable representatives $a, b\in R$ for each edge, we can assume that $\dist[G']{a, b}=\dist[G]{a, b}$. Thus, for every edge $f=\lrb{a, b}$ of $P'$, we have
    \begin{align*}
        \merge{G'}{a, b}\ \leq\ \merge{G}{a, b}\ \leq\ \frac{\dist[G]{a, b}}{2}\ \leq\ \merge{G'}{x, y}\ <\ \merge{G}{x, y}\,.
    \end{align*}
    This implies 
    \begin{align*}
        u_{\gamma}\lr{f} \
        &=\ \max\lrb{\frac{7-5\gamma}{6}\,\merge{G'}{a, b},\ \frac{7-5\gamma}{9-7\gamma}\,\merge{G}{a, b}}\\
        &\leq\ \max\lrb{\frac{7-5\gamma}{6}\,\merge{G'}{x, y},\ \frac{7-5\gamma}{9-7\gamma}\,\merge{G}{x, y}}\ =\ u_{\gamma}\lr{e}\,.
    \end{align*}
    Together with $\merge{G}{a, b}<\merge{G}{x, y}$ we obtain $c^*\lr{f}<c^*\lr{e}$.
\end{proof}

Next, we show that we may assume that the path we obtain from \Cref{lem: path with only cheap edges exists} consists only of edges from the tree $T$.
To this end, we use that $T$ is a minimum spanning tree in $(G, c^*)$, and hence it satisfies the following property.
\begin{observation}\label{obs: replace edge by path in T}
    Let $e=\lrb{x, y}\in E\lr{G}$ be an edge of $G$. Then, the $x$-$y$ path $P$ in the tree $T$ satisfies $c^*\lr{f}\leq c^*\lr{e}$ for all $f\in P$.
\end{observation}

\begin{corollary}\label{cor: if path with only cheap edges exists it can be chosen in T}
    Let $e=\lrb{x, y}\in E\lr{G'}$ be an edge of $G'$ such that there is an $x$-$y$ path in $G'$ consisting only of edges $f$ with $c^*\lr{f}<c^*\lr{e}$.
    Then there also exists such a path whose edges all belong to $T$.
\end{corollary}
\begin{proof}
Call an edge $f$ \emph{cheap} if it satisfies $c^*\lr{f}<c^*\lr{x,y}$.
Let $P'$ be an $x$-$y$ path in $G'$ consisting only of cheap edges.
Consider an edge $f\in P'$, and let $u,v \in V(G')$ be the endpoints of the edge~$f$.
By \Cref{obs: replace edge by path in T}, there is an $u$-$v$ path $P_f$ in $T$ whose edges are all cheap.
Concatenating the paths $P_f$ for all $f\in P'$ yields an $x$-$y$-walk in $T$ consisting only of cheap edges.
\end{proof}

Now suppose for the sake of deriving a contradiction that there is an edge $e=\lrb{x, y}\in T'$ with $\merge{G'}{x, y}<\merge{G}{x, y}$. 
Then by \Cref{lem: path with only cheap edges exists}, there is a path $P$ in $G'$ with $c^*\lr{f}<c^*\lr{e}$ for all $f\in P$.
By \Cref{cor: if path with only cheap edges exists it can be chosen in T}, we may assume $P\subseteq T$.
In other words, the graph $G'$ restricted to edges of $T$ contains an $x$-$y$ path consisting only of edges strictly cheaper than $\{x,y\}$.
Because $T'$ is a minimum spanning tree (with respect to $c^*$) in the graph $G'$ restricted to edges of $T$, this contradicts the fact that $\{x,y\}\in T'$.
This concludes the proof of \Cref{lem: existence of special MST wrt c bar}.

\subsection{Calculating the Upper Bound on the Integrality Gap}\label{sec: 1.9 in general}

Next, we prove that the integrality gap of \ref{eq:bcr} is at most $\finalratio$, that is we prove \Cref{thm:main_gap}. 
We will use the following well-known lower bound on \ref{eq:bcr}.

\begin{proposition}\label{prop:simple_lb}
    For every instance of the Steiner tree problem, we have $\bcr \geq \frac{\tmst}{2}$.
\end{proposition}
\begin{proof}
    We apply a well-known proof strategy that can also be used to show $2\cdot\opt\geq\tmst$ and $2\cdot\hyp\geq\tmst$ (see e.g.\ Lemma~5 in \cite{ln4directedcomponent}):
    Let $x$ be an optimal (primal) solution to \ref{eq:bcr}. 
    Define $x'_{\lr{v, w}}\coloneqq x_{\lr{v, w}}+x_{\lr{w, v}}$ for every undirected edge $\{v,w\}$. 
    This yields an Eulerian solution $x'$ to \ref{eq:bcr}.
    By directed splitting-off (see e.g.~Theorem 3.3 in \cite{traub_approximation_2024}), we obtain a solution $x''$ to \ref{eq:bcr} that does not use edges incident to Steiner vertices and satisfies $c(x'')\leq c(x')$. 
    Because \ref{eq:bcr} is integral on instances without Steiner vertices (Theorem~2 in \cite{EdmondsBranching}), the cost of $x''$ cannot be smaller than the cost of a terminal MST. 
    Hence,
    \begin{align*}
        \tmst\ \leq\ c\lr{x''}\ \leq\ c\lr{x'}\ =\ 2\cdot c\lr{x}\ =\ 2\cdot \bcr\,.
    \end{align*}
\end{proof}

\Cref{cor: gamma good merge plan gives large dual} and \eqref{eq:lower_bound_value_using_rho} from \Cref{thm: there is a good merge plan with rho} imply that for every $\gamma\in\lre{0,\frac15}$, we have
\begin{align}
    \frac{\bcr}{\tmst}\ \geq\ \underbrace{\frac{7-5\gamma}{12}}_{\eqqcolon a\lr{\gamma}}\,-\,\rho\lr\gamma\, \underbrace{\frac{\lr{7-5\gamma}\lr{3-7\gamma}}{12\lr{9-7\gamma}}}_{\eqqcolon b\lr{\gamma}}\,.\label{ieq: definition of a and b}
\end{align}

We choose the value $\gamma$ so that this lower bound on the value of \ref{eq:bcr} is as large as possible. 
Namely, let $\gamma^*\in\lre{0,\frac15}$ such that $a\lr{\gamma^*}-b\lr{\gamma^*}\rho\lr{\gamma^*}$ is maximal. 
Then, using \Cref{prop:simple_lb}, we conclude
$$
    \frac\bcr\tmst\geq\max\lrb{a\lr{\gamma^*}-b\lr{\gamma^*}\rho\lr{\gamma^*}, \frac12}\,.
$$
By the choice of $\gamma^*$, we have for all $\gamma\in\lre{0, \frac15}$
\begin{align*}\rho\lr\gamma
    \ \geq\ \frac{a\lr\gamma-\lr{a\lr{\gamma^*}-b\lr{\gamma^*}\rho\lr{\gamma^*}}}{b\lr\gamma}
    \ \geq\ \frac{a\lr\gamma-\max\lrb{a\lr{\gamma^*}-b\lr{\gamma^*}\rho\lr{\gamma^*}, \frac12}}{b\lr\gamma}\,.
\end{align*}
Hence, by \eqref{eq:upper_bound_opt_using_rho} from \Cref{thm: there is a good merge plan with rho}, we have
\begin{align*}
     \frac\opt\tmst
    &\leq\ 1-\int_{0}^{\frac12}\rho\lr\gamma\,d\gamma\\
    &\leq\ 1-\int_{0}^{\frac15}\rho\lr\gamma\,d\gamma\\
    &\leq\
    1-\int_{0}^{\frac15}\frac{a\lr\gamma-\max\lrb{a\lr{\gamma^*}-b\lr{\gamma^*}\rho\lr{\gamma^*}, \frac12}}{b\lr\gamma}\ d\gamma . 
\end{align*}
Therefore, we can upper-bound the integrality gap by
\begin{align*}
    \frac\opt\bcr
    &\leq\ \frac{1-\int_{0}^{\frac15}\frac{a\lr\gamma-\max\lrb{a\lr{\gamma^*}-b\lr{\gamma^*}\rho\lr{\gamma^*}, \frac12}}{b\lr\gamma}\,d\gamma}{\max\lrb{a\lr{\gamma^*}-b\lr{\gamma^*}\rho\lr{\gamma^*}, \frac12}}\\[2mm]
    &=\ \frac{1-\int_0^{\frac15}\frac{a\lr\gamma}{b\lr\gamma}\,d\gamma}{\max\lrb{a\lr{\gamma^*}-b\lr{\gamma^*}\rho\lr{\gamma^*}, \frac12}}+\int_0^{\frac15}\frac1{b\lr\gamma}\\[2mm]
    &\leq\ 2 \cdot \Big( 1-\int_0^{\frac15}\frac{a\lr\gamma}{b\lr\gamma}\,d\gamma\Big)+\int_0^{\frac15}\frac1{b\lr\gamma}\\[2mm]
    &=\ 2\lr{1-\int_0^{\frac15}\frac{a\lr\gamma-\frac12}{b\lr\gamma}\,d\gamma}\,,
\end{align*}
where we used that for our choice of the functions $a$ and $b$ we have $\int_0^{\frac15}\frac{a\lr\gamma}{b\lr\gamma}\,d\gamma = \int_0^{\frac15}  \frac{9-7\gamma}{3-7\gamma} \,d\gamma<1$, as well as $\int_0^{\frac15}\frac{1}{b\lr\gamma}\,d\gamma<\infty$.
Plugging in the definition of the functions $a$ and $b$ from \eqref{ieq: definition of a and b} and evaluating the integral yields
$$
    \frac{\opt}{\bcr}\ \leq\ 2 \cdot \lr{1-0.051}
    \ =\ 2-0.102 \ =\ 1.898\,,
$$
which proves \Cref{thm:main_gap}.
 \section{Lower Bound for Merge Plan Based Dual Solutions}\label{sec:lower_bound}

The goal of this section is to prove \Cref{thm: merge plan algos cannot be better than 7 over 12 mst}, which we restate here.
\thmMergePlanAlgoLowerBound*

\begin{definition}[merge plan based dual solution]
    Given an instance $\cI$ of the Steiner tree problem, we call a solution $y$ of the dual program \eqref{eq:dual_bcr} \emph{merge plan based}, if there is a merge plan $\cM$ such that \Cref{algo:dual_growth_with_merge_plan} outputs $y$ for the input $\lr{\cI, \cM}$.
\end{definition}

We remark that on the instance we construct to prove \Cref{thm: merge plan algos cannot be better than 7 over 12 mst}, the (relevant) behavior of \Cref{algo:dual_growth_with_merge_plan} does not change under subdivision of edges.
In particular, an analysis analogous to our proof can be applied for any subdivision of the constructed instance.
Thus, considering dual solutions arising from \Cref{algo:dual_growth_with_merge_plan} applied to any subdivision of the original instance (as we did in the proofs of \Cref{thm:main_gap} and \Cref{thm:gap_mst_optimal}; see \Cref{sec: assumptions}), will not yield strong enough lower bounds to prove an integrality gap below $\frac{12}{7}$.

The key ingredient of our construction is a Steiner tree instance that will allow us to impose an upper bound on the merge time of two terminals $s$ and $s'$ in any merge plan that leads to a dual solution with high value.
\Cref{lem: if we dont merge x and y everybody reaches v}, which we state next, captures the important properties of this instance.
To prove \Cref{thm: merge plan algos cannot be better than 7 over 12 mst}, we will then take one copy of this instance for every pair of terminals and combine these copies to a single new instance by identifying the terminal sets.

\begin{lemma}\label{lem: if we dont merge x and y everybody reaches v}
    Let $\varepsilon>0$, let $R$ be a set of terminals, and let $s, s'\in R$. 
    Then there are an MST-optimal Steiner tree instance $\cI=(G,c, R)$ and a vertex $v$ of $G$ such that
    \begin{enumerate}[label=(\roman*)]
        \item \label{item:pairwise_distance1} 
        $\dist{x, y}\geq 2$ for all $x,y\in R$,
        \item \label{item:pairwise_distance2} 
        $\dist{x, s}=2$ for all $x\in R\setminus\{s\}$,
        \item \label{item:distance_v} $\dist{x, v}\leq 4$ for all $x\in R$, and 
        \item \label{item:reachability_v} for every merge plan $\cM$ with $\merge{\cM}{s, s'}\geq\frac76+\varepsilon$, at time $\frac76+\varepsilon$ all active sets can reach~$v$.
    \end{enumerate}
    We call $v$ a \emph{central} vertex of $G$.
\end{lemma}

Before we prove \Cref{lem: if we dont merge x and y everybody reaches v}, we first show that it indeed implies \Cref{thm: merge plan algos cannot be better than 7 over 12 mst}. To this end, we need a well-known fact about full components.
\begin{observation}\label{obs: improving full comp}
    Every improving component contains a full component that is also improving.
\end{observation}
\begin{proof}
    Let $K$ be an improving component connecting terminals $X$. Decompose $K$ into maximal full components $K_1, \dots, K_k$, each connecting terminals $X_i$. 
    Define a sequence $G_1,\dots, G_{k+1}$ of graphs by $G_1 \coloneqq G$ and $G_{i+1} \coloneqq G_i / X_i$.
    Then by \Cref{obs:drop reduces under contraction}, we have
    \begin{align*}
        \drop{X} \ =&\ \tmst_G - \tmst_{G/X} 
        \ =\ \sum_{i=1}^k \big( \tmst_{G_i} - \tmst_{G_{i+1}} \big) \
        =\ \sum_{i=1}^k \dropi_{G_i}(X_i) \\[-2mm]
        \leq&\ \sum_{i=1}^k \drop{X_i}\,,
    \end{align*}
    where $\dropi_{G_i}(X_i)$ denotes the drop of $X_i$ in the graph $G_i$.
    We conclude,
    \begin{align*}
        \min_{1\leq i\leq k}\ \frac{c\lr{K_i}}{\drop{X_i}}
        \ \leq\ \frac{\sum_{1\leq i\leq k}c\lr{K_i}}{\sum_{1\leq i\leq k}\drop{X_i}}
        \ \leq\ \frac{c\lr{K}}{\drop{X}}<1\,.
    \end{align*}
\end{proof}

\begin{proof}[Proof of \Cref{thm: merge plan algos cannot be better than 7 over 12 mst}]
    Let $R$ be a set of terminals with $\lrv{R}\geq\frac4{\varepsilon}+1$. For every pair $\lr{s, s'}\in R \times R$ we consider the graph $G_{\lr{s, s'}}$ (with its edge costs) that we obtain from \Cref{lem: if we dont merge x and y everybody reaches v} and its central vertex $v_{\lr{s, s'}}$. We combine all these graphs $G_{\lr{s, s'}}$ to a single graph $G$.
    The vertex set of $G$ consists of the terminal set $R$ and the disjoint union of the Steiner vertices of the different graphs, that is
    $$
        V\lr{G}\ =\ R\,\dot\cup\ \dot\bigcup_{\lr{s, s'}\in R \times R}\ \Big(V\lr{G_{\lr{s, s'}}}\setminus R\Big)\,.
    $$
    The edge set of $G$ is the disjoint union of the edge sets of the different graphs, that is
    $$
        E\lr{G}\ =\ \dot\bigcup_{\lr{s, s'}\in R\times R}\ E\lr{G_{\lr{s, s'}}}\,.
    $$

    The instance $\lr{G, c, R}$ is MST-optimal: Suppose not, then there is an improving component. 
    By \Cref{obs: improving full comp}, there has to be a \emph{full} improving component. 
    Because, the graphs $G_{\lr{s, s'}}$ we obtained from \Cref{lem: if we dont merge x and y everybody reaches v} only overlap in the terminals, such a full component would have to be contained in a single graph $G_{\lr{s, s'}}$. 
    This cannot be the case, because by \Cref{lem: if we dont merge x and y everybody reaches v} the graph $G_{\lr{s, s'}}$ is MST-optimal.
    Hence, every Steiner Tree in $G$ costs at least $\opt=\tmst=2\lr{\lrv{R}-1}$, where the second equality follows from \ref{item:pairwise_distance1} and \ref{item:pairwise_distance2} of \Cref{lem: if we dont merge x and y everybody reaches v}.

    Fix a merge plan $\cM=\lr{\cS^t}_{t\geq0}$. The value of the dual solution we obtain by running \Cref{algo:dual_growth_with_merge_plan} with merge plan $\cM$ is
    $$
        \dual\lr{\cM}=\int_0^{\infty}\lrv{\lrb{S\in\cS^t\,:\,S\text{ cannot reach }r\text{ at time }t}} \,dt\,,
    $$
    where $r\in R$ is the (arbitrarily chosen) root vertex.
    By \ref{item:pairwise_distance1} and \ref{item:pairwise_distance2} of \Cref{lem: if we dont merge x and y everybody reaches v}, in the graph $G$ we have $\dist{x, y}=2$ for all $x, y\in R$.
    Thus, any two terminals that are not merged at time $t=2$ can reach each other at time $t=2$. 
    Hence, we can rewrite and upper bound the value of our dual solution $\dual\lr{\cM}$ as follows
    \begin{align*}
        &\,\int_0^2\lrv{\lrb{S\in\cS^t\,:\,S\text{ cannot reach }r\text{ at time }t}} \,dt
        \\
        \leq&\ \lr{\frac76+\varepsilon}\lr{\lrv{R}-1}+\int_{\frac76+\varepsilon}^2\lrv{\lrb{S\in\cS^t\,:\,S\text{ cannot reach }r\text{ at time }t}} \,dt\,.\defaulttag\label{ieq: cost bound with the integral that is at most 4}
    \end{align*}

    We will prove that the integral on the right hand-side of \eqref{ieq: cost bound with the integral that is at most 4} has value at most $4$.
    If all terminals are merged already at time $\frac76+\varepsilon$, this trivially holds. 
    Otherwise, choose $s, s'\in R$ with $\merge{M}{s, s'}\geq\frac76+\varepsilon$ and consider the subgraph $G_{\lr{s, s'}}$ and its central vertex $v\coloneqq v_{\lr{s, s'}}$. 
    By \Cref{lem: if we dont merge x and y everybody reaches v} \ref{item:distance_v}, we have $\dist{v, r}\leq4$  and from time $\frac76+\varepsilon$ onward $v$ is reachable from all active sets. 
    Hence, from time $\frac76+\varepsilon$ onward, all active sets that do not reach the root contribute to all shortest $v$-$r$-paths. 
    This shows
    \begin{align}\label{ieq: int after time 7over6 is at most 4}
        \int_{\frac76+\varepsilon}^2\lrv{\lrb{S\in\cS^t\,:\,S\text{ cannot reach }r\text{ at time }t}}\,dt\ \leq\ \dist{v, r}\leq4\,.
    \end{align}
    Combining \eqref{ieq: int after time 7over6 is at most 4} and \eqref{ieq: cost bound with the integral that is at most 4} yields  
    $$
        \frac{\dual\lr{M}}{\tmst} 
        \ \leq\ \frac{\lr{\frac76+\varepsilon}\lr{\lrv{R}-1}+4}{2\lr{\lrv{R}-1}} \ =\ \lr{\frac7{12}+\frac{\varepsilon}2}+\frac{2}{\lrv{R}-1}
        \ \leq\ \frac7{12}+\frac{\varepsilon} 2 + \frac{\varepsilon}2
        \ =\ \frac7{12}+\varepsilon\,.
    $$
\end{proof}

We are left with proving \Cref{lem: if we dont merge x and y everybody reaches v}. 
From now on fix $\varepsilon$, $s$, $s'$, and $R$ as in \Cref{lem: if we dont merge x and y everybody reaches v} and let $\cM$ be a merge plan with $\merge{\cM}{s, s'}\geq\frac76+\varepsilon$. Further we choose $k\in\N$ with $\frac{1}{6k}\leq\varepsilon$.
 
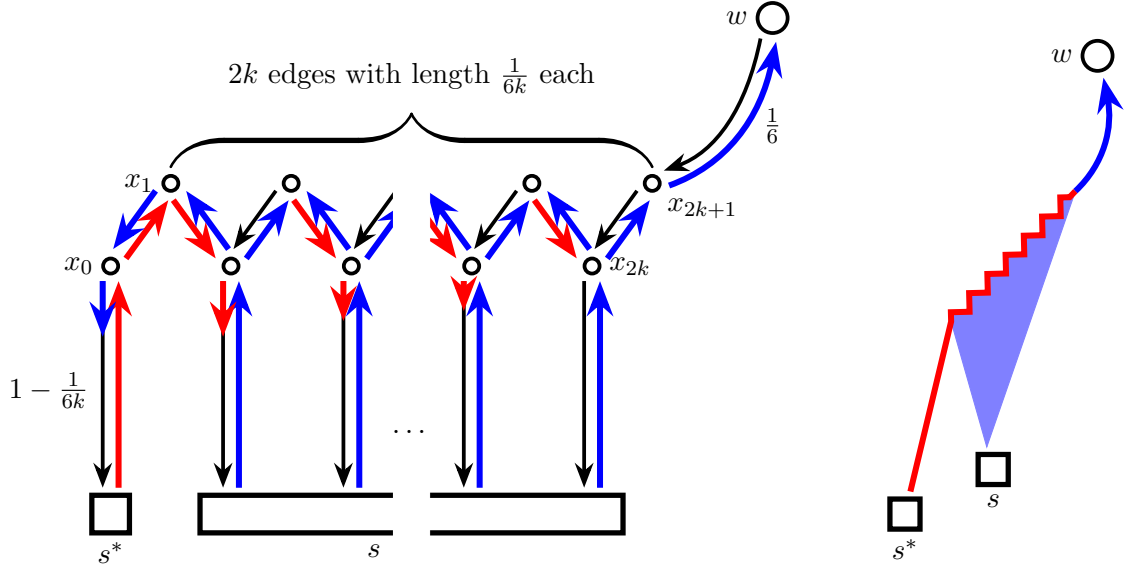
\begin{figure}
    \centering
    \clearpage{}
\begin{tikzpicture}[xscale=0.8, yscale=1.1]
    \useasboundingbox (-1.8,0.3) rectangle ++(19,6.5);
    \node[terminal, label=below:{$s^*$}] (S1) at ( 0, 1) {};
    \node[terminal, transparent] (S21) at ( 2, 1) {};
    \node[terminal, transparent] (S22) at ( 4, 1) {};
    \node[terminal, transparent] (S23) at ( 6, 1) {};
    \node[terminal, transparent] (S24) at ( 8, 1) {};
    \node[terminal,minimum width=160pt, label=below left:{$s$}] (S2) at ( 5, 1) {};

    \node[tinysteiner, label=left:{$x_0$}] (A1) at ( 0, 4) {};
    \node[tinysteiner, label=left:{$x_1$}] (A2) at ( 1, 5) {};
    \node[tinysteiner] (A3) at ( 2, 4) {};
    \node[tinysteiner] (A4) at ( 3, 5) {};
    \node[tinysteiner] (A5) at ( 4, 4) {};
    \node[tinysteiner] (A6) at ( 5, 5) {};
    \node[tinysteiner] (A7) at ( 6, 4) {};
    \node[tinysteiner] (A8) at ( 7, 5) {};
    \node[tinysteiner, label=right:{$x_{2k}$}] (A9) at ( 8, 4) {};
    \node[tinysteiner, label=below right:{$x_{2k+1}$}] (A10)at ( 9, 5) {};
    
    \node[smallsteiner, label=left:{$w$}] (C) at (11, 7) {};
    
    \draw[edgedirprogress, red] (S1) -- (A1);
    \draw[edgedirprogress] (S21) -- (A3);
    \draw[edgedirprogress] (S22) -- (A5);
    \draw[edgedirprogress] (S23) -- (A7);
    \draw[edgedirprogress] (S24) -- (A9);
    \draw[diredge] (A1) -- (S1);
    \draw[edgedirprogress] (A1) -- ($(A1)!2/7!(S1)$);
    \draw[diredge] (A3) -- (S21);
    \draw[edgedirprogress, red] (A3) -- ($(A3)!2/7!(S21)$);
    \draw[diredge] (A5) -- (S22);
    \draw[edgedirprogress, red] (A5) -- ($(A5)!2/9!(S22)$);
    \draw[diredge] (A7) -- (S23);
    \draw[edgedirprogress, red] (A7) -- ($(A7)!2/11!(S23)$);
    
    \draw[edgedirprogress, red] (A1) -- (A2);
    \draw[edgedirprogress, red] (A2) -- (A3);
    \draw[edgedirprogress] (A3) -- (A4);
    \draw[edgedirprogress, red] (A4) -- (A5);
    \draw[edgedirprogress] (A5) -- (A6);
    \draw[edgedirprogress, red] (A6) -- (A7);
    \draw[edgedirprogress] (A7) -- (A8);
    \draw[edgedirprogress, red] (A8) -- (A9);
    \draw[edgedirprogress] (A9) -- (A10);
    
    \draw[edgedirprogress] (A9) -- (A8);
    \draw[diredge] (A8) -- (A7);
    \draw[edgedirprogress] (A7) -- (A6);
    \draw[diredge] (A6) -- (A5);
    \draw[edgedirprogress] (A5) -- (A4);
    \draw[diredge] (A4) -- (A3);
    \draw[edgedirprogress] (A3) -- (A2);
    \draw[edgedirprogress] (A2) -- (A1);

    \draw[edgedirprogress] (A10) edge[bend right = 30] (C);
    \draw[diredge] (C) edge[bend left = 30] (A10);
    \draw[diredge] (A9) -- (S24);
    \draw[diredge] (A10) -- (A9);

    \draw[white, fill=white] (4.7, -1) rectangle (5.3,6);
    \node[] at (5,2) {$\dots$};
    
    \node[left] at  (-0.2,2.5) {$1-\frac{1}{6k}$};

    \draw[line width = 2pt,decoration={calligraphic  brace,amplitude=20pt},decorate] (1,5.2) -- (9,5.2);
    \node at (5,6.3) {$2k$ edges with length $\frac{1}{6k}$ each};

\node at (11,5.7) {$\frac16$};

    \begin{scope}[shift={(13.2,1)}, rotate=60, xscale=0.8]    
        \node[smallterminal, label=below:{$s^*$}] (S1) at ( 0, 0) {};
        \node[smallterminal, label=below:{$s$}] (S2) at ( 1.5,-1) {};
        \node[smallsteiner, label=left:{$w$}] (C) at (8, 0) {};

        \draw[blue!50!white, fill = blue!50!white] (S2)
        -- ($(S1)+(3,1/2)$)
        decorate[decoration=zigzag]{--++(3,-1)}
        -- cycle;

        \draw[edgeprogress, red, offset=0pt] (S1)
         -- (3,0.5) 
        decorate[decoration=zigzag]  { -- ++(3,-1) }
        edge[edgedirprogress,offset=0pt,bend right = 30] (C);
    \end{scope}

\end{tikzpicture}
\clearpage{}
    \caption{\label{fig:3x_gadget}
    The left part of the figure illustrates the $3\times$-gadget from \Cref{lem: 3x gadget}.
    The colors represent the contributions of the terminal sets containing $s^*$ (red) and $s$ (blue) at time $t=\frac{7}{6}$ (assuming that $s$ and $s^*$ are not already merged before this time $t$).
    The right part shows an abstract drawing of the same gadget that we will use in all further figures to represent the $3\times$-gadget.
    }
\end{figure}

To obtain a corresponding graph, we will first construct some smaller gadgets.
Our first gadget is depicted in \Cref{fig:3x_gadget}.
The next lemma states the properties of this gadget that we will later need to prove \Cref{lem: if we dont merge x and y everybody reaches v}.
We call the graph from \Cref{lem: 3x gadget} a \emph{$3\times$-gadget} from $s^*$ to $w$.
The reason for the name is that the construction is similar to the example from \Cref{fig:zig_zag_example}, where the distance of reachable vertices can increase effectively at speed $3$.

\begin{lemma}[$3\times$-gadget]\label{lem: 3x gadget}
    For every $s^*\in R\setminus\lrb{s}$, there is a connected graph, called the $3\times$-gadget, with the following properties:
    \begin{itemize}
        \item The vertex set of the gadget consists of the terminals $s$, $s^*$, a distinguished Steiner vertex $w$, and possibly further Steiner vertices.
        \item Suppose we run \Cref{algo:dual_growth_with_merge_plan} with merge plan $\cM$ on a graph containing the $3\times$-gadget as a subgraph. 
        If before time $t=\frac{7}{6}$, every tight path from some terminal to a vertex of the gadget is completely contained in the gadget itself, then
        $\atf{s}{w}=\frac76$ and $\min\lrb{\atf{s^*}{w}, \merge{\cM}{s^*, s}}\leq\frac76$.
        \item We have $\dist{s, w}=\frac76$, $\dist{s^*,w}=\frac32$, and $\dist{s, s^*}=2$.
        \item Every component connecting $s$, $s^*$ and $w$ has cost at least $\frac52-\frac{1}{6k}$.
    \end{itemize}
\end{lemma}
\begin{proof}
    We construct a gadget as illustrated in \Cref{fig:3x_gadget}.
    The graph we construct has vertices $s$, $s^*$, $w$, as well as further Steiner vertices $x_0, \dots, x_{2k+1}$.
    The edges are the following:
    \begin{itemize}
        \item $\lrb{x_j,x_{j+1}}$ for $j=0, \dots, 2k$ each with length $\frac{1}{6k}$,
        \item $\lrb{s^*,x_0}$ and $\lrb{s,x_{2j}}$ for $j=1,\dots, k$ each with length $1-\frac{1}{6k}$, and
        \item $\lrb{x_{2k+1},w}$ with length $\frac16$.
    \end{itemize}

    We now show that this gadget has the claimed properties.
    The first property regarding the vertex set of the gadget is satisfied by construction.

    Next, we consider a run of \Cref{algo:dual_growth_with_merge_plan} on a graph containing the $3\times$-gadget as a subgraph.
    First suppose that the merge time of $s$ and $s^*$ is at least $\frac{7}{6}$.
    At time $1-\frac1{6k}$, all edges leaving the terminals are tight. At time $1$ both terminals meet in $x_1$. At this time, each edge $\lrb{x_j,x_{j+1}}$ has one tight orientation. 
    Afterwards, every $\frac{1}{6k}$ time units, the contribution of the set $S^*\in \cS^t$ containing $s^*$ makes an edge $\lr{x_{2j-1},x_{2j}}$ tight; the edge $\lr{x_{2j},x_{2j+1}}$ is already tight (since time $1$). This pattern continues for $k$ steps until time $\frac76$. 
    Meanwhile the contribution of the set $S\in \cS^t$ containing $s$ makes the edge $\lrb{x_{2k+1}, w}$ tight.
    Thus, we have $\atf{s}{w}=\frac{7}{6}$ and $\atf{s^*}{w}=\frac{7}{6}$.
    If the merge time of $s$ and $s^*$ is at most $\frac{7}{6}$, then it suffices to prove $\atf{s}{w}=\frac{7}{6}$. 
    This is indeed the case because until time $t=\frac{7}{6}$, the set $S\in \cS^t$ containing $s$ contributes to the $s$-$w$ path of length $\frac{7}{6}$ and no other set does so.
    
    Now, we observe $\dist{s, w}=\frac76$, $\dist{w, s^*}=\frac32$, and $\dist{s, s^*}=2$.
    Indeed, the shortest $s$-$w$ path has vertices $s_1$, $x_{2k}$, $x_{2k+1}$, and $w$. Its length is $\lr{1-\frac{1}{6k}}+\frac{1}{6k}+\frac16=\frac76$.
    The shortest $s^*$-$w$ path goes through all Steiner vertices $x_0$ to $x_{2k+1}$ and has length $\lr{1-\frac{1}{6k}}+\frac{2k+1}{6k}+\frac16=\frac32$.
    The shortest $s^*$-$s$ path has vertices $s^*$, $x_0$, $x_1$, $x_2$, and $s$ and length $\lr{1-\frac1{6k}}+\frac1{6k}+\frac1{6k}+\lr{1-\frac1{6k}}=2$.

    Finally, we determine the length of a shortest component connecting $s$, $s^*$, and $w$.
    If the $s$-$w$ path in the component contains $s$, then the component has length (at least) $\dist{s^*,s} +\dist{s,w} = 3 +\frac{1}{7}$.
    Otherwise, because the shortest $s^*$-$w$-path is also the only $s^*$-$w$-path that does not contain~$s$, the component consists of the shortest $s^*$-$w$-path together with one of the edges incident to $s$. 
    This component has cost $\frac32+\lr{1-\frac{1}{6k}}=\frac52-\frac{1}{6k}$.
\end{proof}

\begin{figure}
    \centering
    \clearpage{}
\newcommand{\ThreeXGadget}[2]{
    \draw[blue!50!white, fill = blue!50!white] ($(D#1)+(0.5,0.35)$) 
        -- ($(C#1)+(1/3, 2)$)
        decorate[decoration=zigzag]{--++(-2/3, 2)}
        -- cycle;
    \draw[edgeprogress, #2, offset=0pt] (C#1)
        -- ++(1/3, 2) 
        decorate[decoration=zigzag]  { -- ++(-2/3, 2) }
        edge[edgedirprogress,offset=0pt,bend right = 20, draw=blue] (A#1);
}

\begin{tikzpicture}[xscale=0.9]
\useasboundingbox (-2,0) rectangle ++(16,9);
    \node[smallsteiner, label=left:{$x_0$}] (A0) at (-1, 7) {};
    \node[smallsteiner] (A1) at ( 1, 7) {};
    \node[smallsteiner] (A2) at ( 3, 7) {};
    \node[smallsteiner] (A3) at ( 5, 7) {};
    \node[smallsteiner] (A4) at ( 7, 7) {};
    \node[smallsteiner, label=right:{$x_{30k}$}] (A5) at ( 9, 7) {};
    \node[smallsteiner] (B1) at ( 0, 6) {};
    \node[smallsteiner] (B2) at ( 2, 6) {};
    \node[smallsteiner] (B3) at ( 4, 6) {};
    \node[smallsteiner] (B4) at ( 6, 6) {};
    \node[smallsteiner] (B5) at ( 8, 6) {};
    \node[smallsteiner, label=right:{$w$}] (B6) at (10, 6) {};
    \node[terminal, transparent] (C1) at ( 1, 1) {};
    \node[terminal, transparent] (C2) at ( 3, 1) {};
    \node[terminal, transparent] (C3) at ( 5, 1) {};
    \node[terminal, transparent] (C4) at ( 7, 1) {};
    \node[terminal, transparent] (C5) at ( 9, 1) {};
    \node[terminal,minimum width=240pt, label=right:{$s'$}] (C) at ( 5, 1) {};
    \node[terminal, transparent] (D1) at ( 0, 2) {};
    \node[terminal, transparent] (D2) at ( 2, 2) {};
    \node[terminal, transparent] (D3) at ( 4, 2) {};
    \node[terminal, transparent] (D4) at ( 6, 2) {};
    \node[terminal, transparent] (D5) at ( 8, 2) {};
    \node[terminal,minimum width=240pt, label=right:{$s$}] (D) at ( 4, 2) {};
    \node[terminal, label=below:{$s^*$}] (E) at (-1, 1) {};

    \ThreeXGadget{1}{red}
    \ThreeXGadget{2}{red}
    \ThreeXGadget{3}{red}
    \ThreeXGadget{4}{red}

    \draw[blue!50!white, fill = blue!50!white] ($(D1)+(-0.5,0.35)$) 
        -- ($(E)+(-1/3, 2)$)
        decorate[decoration=zigzag]{--++(2/3, 2)}
        -- cycle;
    \draw[edgeprogress, green!70!black, offset=0pt] (E)
        -- ++(-1/3, 2) 
        decorate[decoration=zigzag]  { -- ++(2/3, 2) }
        edge[edgedirprogress,offset=0pt,bend left = 20, draw=blue] (A0);
    \draw[edgedirprogress,offset=0pt] (D1) -- (B1);
    \draw[edgedirprogress,offset=0pt] (D2) -- (B2);
    \draw[edgedirprogress,offset=0pt] (D3) -- (B3);
    \draw[edgedirprogress,offset=0pt] (D4) -- (B4);
    \draw[edgedirprogress,offset=0pt] (D5) -- (B5);
    
    \draw[edgedirprogress] (B1) -- (A0);
    \draw[edgedirprogress, green!70!black] (A0) -- (B1);
    \draw[edgedirprogress] (B1) -- (A1);
    \draw[edgedirprogress, red] (A1) -- (B1);
    \draw[edgedirprogress] (B2) -- (A1);
    \draw[edgedirprogress, red] (A1) -- (B2);
    \draw[edgedirprogress] (B2) -- (A2);
    \draw[edgedirprogress, red] (A2) -- (B2);
    \draw[edgedirprogress] (B3) -- (A2);
    \draw[edgedirprogress, red] (A2) -- (B3);
    \draw[edgedirprogress] (B3) -- (A3);
    \draw[edgedirprogress, red] (A3) -- (B3);
    \draw[edgedirprogress] (B4) -- (A3);
    \draw[edgedirprogress, red] (A3) -- (B4);
    \draw[edgedirprogress] (B4) -- (A4);
    \draw[edgedirprogress, red] (A4) -- (B4);
    \draw[edgedirprogress] (B5) -- (A4);
    \draw[edgedirprogress, red] (A4) -- (B5);
    \draw[edgedirprogress, offset=0pt] (B5) -- (A5);
    \draw[edgedirprogress, offset=0pt] (A5) -- (B6);

    \draw[white, fill=white] (3.75, -1) rectangle (5.39,8);
    \node[] at (4.6,4) {$\dots$};

\draw[line width = 2pt,decoration={calligraphic  brace,amplitude=25pt},decorate] (-1,7.4) -- (9,7.4);
    \node at (5,8.6) {$30k$ edges with length $\frac{1}{6k}$ each};

    \begin{scope}[shift={(13.5,1)}, rotate=90, xscale=0.7]    
        \node[terminal, label=below:{$s^*$}] (S1) at ( 0, 0) {};
        \node[smallsteiner, label=above:{$w$}] (V) at (9, 0) {};

        \draw[edgeprogress, green!70!black, offset=0pt] (S1)
         -- ++(2,1/3) 
        decorate[decoration=zigzag]  { -- ++(2,-2/3) }
        edge[edgeprogress,offset=0pt,bend left = 30] ++(2,1/3);
        \draw[edgedirprogress,offset=0pt,decorate,decoration={zigzag,post length=5mm}, green!70!black] (6,0) -- (V);
        \draw[edgeprogress,offset=2pt,decorate,decoration={zigzag,post length=5mm},arrows = {-Stealth[harpoon]},shorten >=2.5pt] (6,0) -- (V);
        \draw[edgeprogress,red,offset=-2pt,decorate,decoration={zigzag,post length=5mm},arrows = {-Stealth[harpoon,swap]},shorten >=2.5pt] (6,0) -- (V);
    \end{scope}
\end{tikzpicture}\clearpage{}
    \caption{\label{fig:jump_gadget}
    The left part of the figure illustrates the jump-gadget from \Cref{lem: jump gadget}.
    The colors represent the contributions of the terminal sets containing $s^*$ (green), $s$ (blue) and $s'$ (red) at time $t=\frac{7}{6} + \frac{1}{6k}$ (assuming that none of the terminals are merged before this time $t$).
    The right part shows an abstract drawing of the same gadget that we will use in the next figure to represent the jump-gadget.
    }
\end{figure}
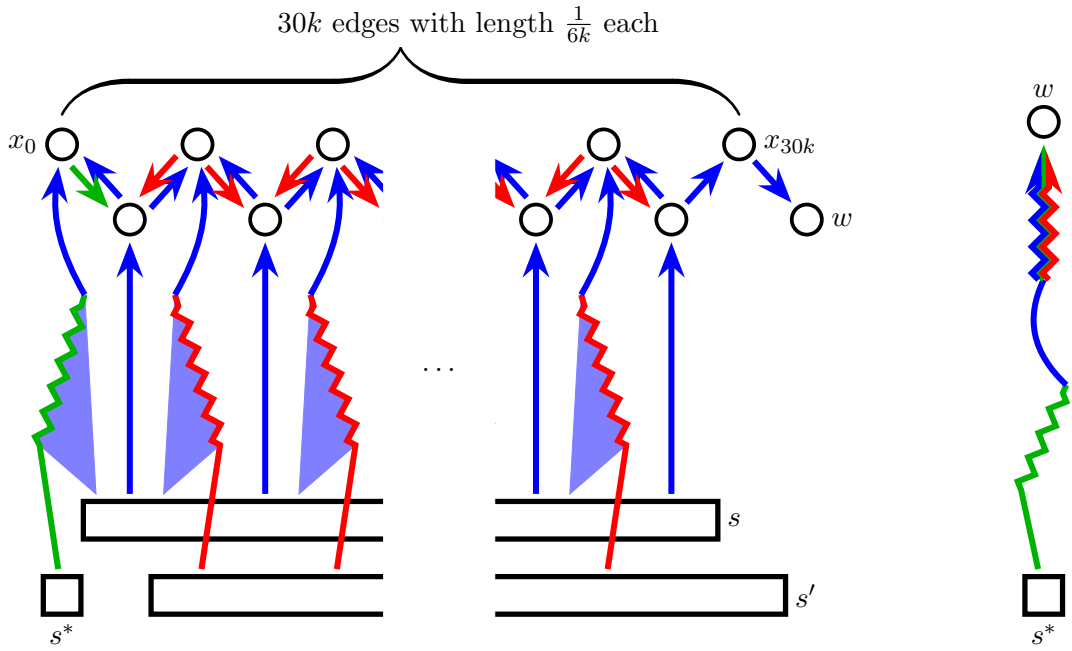

Next, we use the $3\times$-gadget from \Cref{lem: 3x gadget} to construct another gadget, which is illustrated in \Cref{fig:jump_gadget}.
\Cref{lem: jump gadget} captures the properties of this gadget that we will need in the proof of \Cref{lem: if we dont merge x and y everybody reaches v}.
We will call the gadget from \Cref{lem: jump gadget} a \emph{jump-gadget} from $s^*$ to $v$.
Recall that the terminals $s$ and $s'$ are fixed throughout and the merge plan $\cM$ is such that $\merge{\cM}{s,s'} \geq \frac{7}{6}+\epsilon$.

\begin{lemma}[jump-gadget]\label{lem: jump gadget}
    For every $s^*\in R\setminus\lrb{s, s'}$, there is a connected graph, called the jump-gadget with the following properties:
    \begin{itemize}
        \item The vertex set of the gadget consists of the terminals $s$, $s'$, $s^*$, a distinguished Steiner vertex $w$, and possibly further Steiner vertices.
        \item Suppose we run \Cref{algo:dual_growth_with_merge_plan} with merge plan $\cM$ on a graph containing the jump-gadget as a subgraph. 
        If before time $t=\frac{7}{6}+\frac{1}{6k}$, every tight path from a terminal to a vertex of the gadget is completely contained in the gadget itself, then we have
        $\atf{s}{w}=\atf{s'}{w}=\frac76+\frac{1}{6k}$ and $\min\lrb{\atf{s^*}{w}, \merge{M}{s^*, s}}\leq\frac76+\frac{1}{6k}$. 
        \item The pairwise distances of any two terminals in the gadget is at least $2$ and $\dist{s^*,s}= \dist{s',s}=2$.
        \item We have $\dist{s, w}=\frac76+\frac{1}{6k}$, $\dist{s', w}=\frac32+\frac{3}{6k}$, and $\dist{s^*, w}=2+\frac76+\frac{1}{6k}$.
        \item The only $s^*$-$w$ path that does not go through another terminal costs at least $6$.
        \item Any component connecting $s$, $s'$, and $s^*$ has cost at least $4$.
    \end{itemize}
\end{lemma}
\begin{proof}
    We construct the gadget as shown in \Cref{fig:jump_gadget}. 
    To construct the graph we include the vertices $s$, $s'$, $s^*$, and $w$, as well as additional Steiner vertices $x_0, \dots, x_{30k}$.
Moreover, the graph contains the following edges:
    \begin{itemize}
        \item $\lrb{x_j,x_{j+1}}$ for $j=0, \dots, 30k-1$ and $\lrb{x_{30k}, w}$ each with length $\frac{1}{6k}$, and
        \item $\lrb{s,x_{2j-1}}$ for $j=1,\dots, 15k$ each with length $\frac76-\frac{1}{6k}$.
    \end{itemize}
    Additionally, we add a $3\times$-gadget from $s^*$ to $x_0$ and $3\times$-gadgets from $s'$ to $x_{2j}$ for $j=1, \dots, 15k-1$.

    We now show that this gadget has the desired properties.
    The first property regarding the vertex set of the gadget is satisfied by construction.

    Next, we consider a run of \Cref{algo:dual_growth_with_merge_plan} with merge plan $\cM$ on a graph containing the $3\times$-gadget as a subgraph.
    Recall that we chose the merge plan $\cM$ such that $\merge{\cM}{s, s'}\geq\frac76+\varepsilon \geq \frac76+\frac{1}{6k}$.
    
    First suppose that $\merge{M}{s^*, s} \geq \frac{7}{6} +\frac{1}{6k}$.
    Then at time $t=\frac76$, the terminal $s^*$ reaches $x_0$ and $s'$ reaches $x_{2j}$ for $j=1, \dots, 3k$ via paths of tight edges (by \Cref{lem: 3x gadget}).
    The terminal $s$ reaches all vertices $x_j$ and the set $S\in \cS^t$ containing $s$ has already made one orientation of each $\lrb{x_j, x_{j+1}}$-edge tight, namely the direction from odd index to even index (the short blue edges in \Cref{fig:jump_gadget}, except for the edge to $w$).
    During the next $\frac{1}{6k}$ time units, the terminal set from $\cS^t$ containing $s^*$ makes the edge $\lr{x_0, x_1}$ tight (short green edge in \Cref{fig:jump_gadget}) and the terminal set from $\cS^t$ containing $s'$ makes all remaining orientations of edges $\lrb{x_j, x_{j+1}}$ tight (short red edges in \Cref{fig:jump_gadget}). 
    Meanwhile, the terminal set from $\cS^t$ containing $s$ makes the edge $\lr{x_{30k}, w}$ tight. 
    Hence, we have $\atf{s}{w}=\atf{s'}{w}=\atf{s^*}{w}=\frac{7}{6}+\frac{1}{6k}$, as desired.

    If $\merge{M}{s^*, s}\leq \frac{7}{6} +\frac{1}{6k}$, then it suffices to prove $\atf{s}{w}=\atf{s'}{w}=\frac{7}{6}+\frac{1}{6k}$.
    This is indeed the case because the edges of the shortest $s$-$w$ path and the shortest $s'$-$w$ path become tight at exactly the same time as in the first case.
    (The contributions of sets containing $s^*$ but neither $s$ nor $s'$, which are shown in green in \Cref{fig:jump_gadget}, do not affect this part of the gadget before time $\frac{7}{6}+\frac{1}{6k}$.)

    Using \Cref{lem: 3x gadget}, we conclude that the pairwise distance between any two terminals is at least $2$ and $\dist{s^*,s}= \dist{s',s}=2$.
    Moreover, we next observe $\dist{s, w}=\frac76+\frac{1}{6k}$, $\dist{s', w}=\frac32+\frac{3}{6k}$, and $\dist{s^*, w}=2+\frac76+\frac{1}{6k}$.
    Indeed, a shortest $s$-$w$ path has vertices $s$, $x_{30k-1}$, $x_{30k}$, $w$ and has length 
    $\lr{\frac76-\frac{1}{6k}}+\frac{1}{6k}+\frac{1}{6k}=\frac76+\frac{1}{6k}$.
    The shortest $s'$-$w$ path uses the $3\times$-gadget from $s'$ to $x_{30k-2}$ and edges $\lrb{x_{30k-2},x_{30k-1}}$, $\lrb{x_{30k-1},x_{30k}}$, and $\lrb{x_{30k},w}$. By \Cref{lem: 3x gadget}, its length is $\frac32+3\cdot\frac{1}{6k}$.
    The only $s^*$-$w$ path that does not go through another terminal has cost $\frac32+\frac{30k}{6k}+\frac{1}{6k}>6$.
    Thus, a shortest $s^*$-$w$ path consists of a shortest $s^*$-$s$ path inside a $3\times$-gadget, which has length $2$ by \Cref{lem: 3x gadget}, and a shortest $s$-$w$ path, which has length $\frac{7}{6}+\frac{1}{6k}$.

    It remains to prove that every component connecting $s$, $s'$, and $s^*$ has cost at least $4$.
    Because the three terminals have pairwise at least $2$ from each other, every component consisting of two disjoint paths between terminals has cost at least $4$.
    The only other possibility to construct a component connecting the three terminals is to connect $s$ to some $x_i$ and either  $s'$ or $s^*$ inside some $3\times$-gadget and then connect the remaining terminal via a path to $x_i$.
    By \Cref{lem: 3x gadget}, connecting $s$ to $x_i$ and $s'$ or $s^*$ inside a $3\times$-gadget has cost at least $\frac{5}{2} - \frac{1}{6k}$. 
    However, then the shortest path from the remaining terminal ($s^*$ or $s'$) to the vertex $x_i$ has length at least $\frac{3}{2}+\frac{2}{6k}$ (where we again used \Cref{lem: 3x gadget}).
    Hence, also in this case the component connecting the three terminals has length at least $4$.
\end{proof}

\begin{figure}
    \centering
    \clearpage{}

\begin{tikzpicture}[scale =0.4]

\foreach \x in {1,...,8}
{
\begin{scope}[rotate=\x*40+180]    
        \node[smallterminal, label=left:{$s_{\x}$}] (S1) at (-9, 0) {};
        \node[smallsteiner] (V) at (0, 0) {$v$};

        \draw[edgeprogress, green!70!black, offset=0pt] (S1)
         -- ++(2,1/3) 
        decorate[decoration=zigzag]  { -- ++(2,-2/3) }
        edge[edgeprogress,offset=0pt,bend left = 30] ++(2,1/3);
        \draw[edgedirprogress,offset=0pt,decorate,decoration={zigzag,post length=5mm}, green!70!black] (-3,0) -- (V);
        \draw[edgeprogress,offset=2pt,decorate,decoration={zigzag,post length=5mm},arrows = {-Stealth[harpoon]},shorten >=2.5pt] (-3,0) -- (V);
        \draw[edgeprogress,red,offset=-2pt,decorate,decoration={zigzag,post length=5mm},arrows = {-Stealth[harpoon,swap]},shorten >=2.5pt] (-3,0) -- (V);
    \end{scope}
}
\begin{scope}[rotate=190, draw opacity=0.3]    
        \node[smallterminal, fill=black, opacity=1, label=right:{$r$}] (R) at (-9, 0) {};
        \node[smallsteiner] (V) at (0, 0) {$v$};

        \draw[edgeprogress, green!70!black, offset=0pt] (R)
         -- ++(2,1/3) 
        decorate[decoration=zigzag]  { -- ++(2,-2/3) }
        edge[edgeprogress,offset=0pt,bend left = 30] ++(2,1/3);
        \draw[edgedirprogress,offset=0pt,decorate,decoration={zigzag,post length=5mm}, green!70!black] (-3,0) -- (V);
        \draw[edgeprogress,offset=2pt,decorate,decoration={zigzag,post length=5mm},arrows = {-Stealth[harpoon]},shorten >=2.5pt] (-3,0) -- (V);
        \draw[edgeprogress,red,offset=-2pt,decorate,decoration={zigzag,post length=5mm},arrows = {-Stealth[harpoon,swap]},shorten >=2.5pt] (-3,0) -- (V);
\end{scope}
    \node[smallterminal, label=below:{$s$}] (S1) at ( 8.5,-2.5) {};
    \node[smallterminal, label=below:{$s'$}] (S2) at ( 12,-2) {};
    \draw[line width = 2pt,decoration={calligraphic  brace,amplitude=20pt},decorate] (R) -- node[below=28pt, right=-5pt] {$\leq 4$} (V);
\end{tikzpicture}
\clearpage{}
    \caption{\label{fig:tight_example}
    An illustration of the graph $G$ that we construct in the proof of \Cref{lem: if we dont merge x and y everybody reaches v}, where the jump-gadgets are drawn as explained in \Cref{fig:jump_gadget}.
    Note that each of the jump-gadgets also has edges incident to $s$ and $s'$, although this is not explicitly visible in this abstract picture.
    The graph contains a jump-gadget from $r$ to $v$, but for our proof it will only be relevant that the distance from $v$ to $r$ is at most $4$. 
    }
\end{figure}
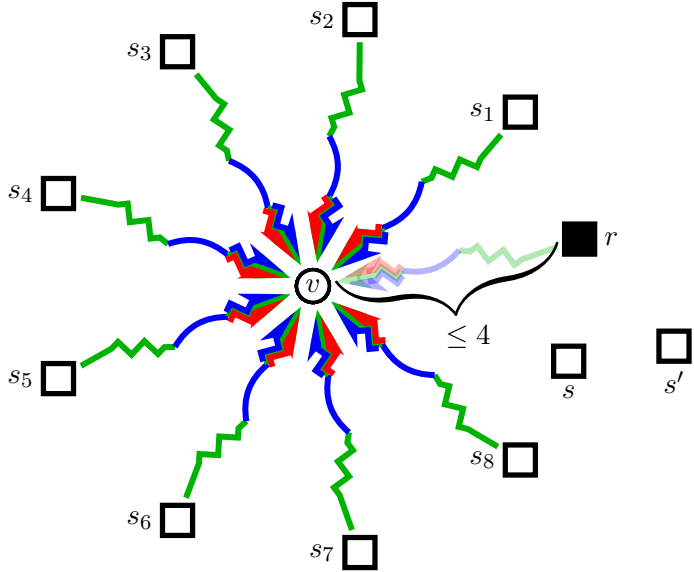

Finally, we use the jump-gadget (\Cref{lem: jump gadget}) to prove \Cref{lem: if we dont merge x and y everybody reaches v}.

\begin{proof}[Proof of \Cref{lem: if we dont merge x and y everybody reaches v}]
    To obtain $G$, we introduce a Steiner vertex $v$ and insert jump-gadgets $G_{s^*}$ from each $s^*\in R\setminus\lrb{s, s'}$ to $v$. 
    See \Cref{fig:tight_example} for an illustration.
    Note that each of the jump-gadgets also introduces edges incident to $s$ and $s'$.
    It remains to show that $G$ has the properties listed in \Cref{lem: if we dont merge x and y everybody reaches v}.

    Property~\ref{item:pairwise_distance1} follows from the bounds on pairwise distances of vertices in a jump segment from \Cref{lem: jump gadget}.
    Similarly, also the other distance restrictions \ref{item:pairwise_distance2} and \ref{item:distance_v}, namely, $\dist{x, s}=2$ for all $x\in R\setminus\{s\}$ and $\dist{x, v}\leq 4$ for all $x\in R$ follow directly from \Cref{lem: jump gadget}.

    Next, we prove that \ref{item:reachability_v} of \Cref{lem: if we dont merge x and y everybody reaches v} is satisfied.
    Our goal is to prove that at time $t=\frac{7}{6}+\varepsilon$ all sets in $\cS^t$ can reach the central vertex $v$.
    By \Cref{lem: jump gadget}, both $s$ and $s'$ can reach $v$ at time $\frac76+\frac{1}{6k}$.
    Now, consider a terminal $s^*\in R\setminus\lrb{s, s'}$ that was merged neither with $s$ nor with $s'$ at time $\frac76+\frac{1}{6k}$. 
    Then, by \Cref{lem: jump gadget}, $s^*$ can reach $v$ at time $\frac76+\frac{1}{6k}$. 
    Hence, from time $\frac76+\frac{1}{6k}\leq\frac76+\varepsilon$ onward, all sets in $\cS^t$ can reach the vertex $v$.

    Finally, we have to show that the Steiner tree instance we constructed is indeed MST-optimal.
    Connecting all other terminals directly to $s$ yields a terminal MST with cost $2\lr{\lrv{R}-1}$. 
    Then, for each $X\subseteq R$, we have $\drop{X} = 2\cdot (|X|-1)$ and thus a component connecting $\ell$ terminals is improving if and only if its cost is strictly less than $2(\ell-1)$.
    Suppose for the sake of deriving a contradiction that the Steiner tree instance is not MST-optimal. 
    Then let $K$ be an improving component with minimum cost, and let $X$ be the set of terminals connected by $K$.
    Then $c(K) < 2\cdot (|X| -1)$.
    Because $K$ is a \emph{minimal} improving component (and thus by \Cref{obs: improving full comp} a full component), every terminal that is connected by $K$ has degree $1$ in $K$.
    Every improving component connects at least three terminals and thus $X$ contains some terminal $s^*\neq s, s'$. 
    By \Cref{lem: jump gadget}, no jump-gadget contains an improving component.
    Thus, the component $K$ has to contain an $s^*$-$v$ path that does not go through another terminal, and by \Cref{lem: jump gadget} there is only one such path. 
    Removing this path from the component $K$ leaves us with a component $K'$ that connects at least all terminals in $X\setminus\lrb{s, s', s^*}$ and has cost at least $6$ less than $K$. 
    Moreover, we have
    $c(K') \leq c(K) - 6 < 2 \cdot (|X|-1) - 6 \leq 2\cdot (|X|-3 -1)$.
    Recall that every component connecting at least $\ell$ terminals with cost strictly less than $2(\ell-1)$ is improving.
    For $\ell = |X|-3$, this implies  that the component $K'$ is improving, which contradicts the minimality of $c(K)$.
\end{proof}

 \section{BCR and the Hypergraphic Relaxation}\label{sec:hypergraphic}

In this section, we show that the maximum possible ratio between the Hypergraphic Relaxation and \ref{eq:bcr} is equal to the integrality gap of \ref{eq:bcr} restricted to MST-optimal instances.

The Hypergraphic Relaxation has a variable for every nonempty set $X$ of terminals and every vertex $v\in X$, where we think of $X$ as the vertex set of a directed component with root $v$.
We write
\[
\vec{\mathcal{X}} \coloneqq \left\{ (X,v) \colon X\subseteq R, v\in X \right\}.
\]
For a terminal set $S\subseteq R$, we denote by 
\[
 \delta^+_{\vec{\cX}}(S) \coloneqq \left\{ (X,v)\in \vec{\cX} \colon X\cap S \ne \emptyset,\  v\notin S \right\}
\]
the elements of $\vec{\cX}$ that leave the set $S$.
Then the Hypergraphic Relaxation can be stated as follows:
\begin{equation}\label{eq:hypergraphic_lp}
\begin{aligned}
\min \sum_{(X,v)\in\vec{\cX}} \cost(X) \cdot x_{(X,v)} \\
\sum_{(X,v)\in  \delta^+_{\vec{\cX}}(S)} x_{(X,v)} \ \geq&\ 1 &\forall\; \emptyset \ne S\subseteq R\setminus\{r\}\\
x_{(X,v)} \ \ge&\ 0 &\forall\;(X,v)\in\vec{\cX}.
\end{aligned}
\end{equation}
We write $\hyp(\cI)$ to denote the value of \eqref{eq:hypergraphic_lp} for a Steiner Tree instance $\cI$.
The goal of this section is to prove \Cref{thm:ratios_same}, which we restate here.
\HypRatiosSame*

We first show that the left hand-side of \eqref{eq:ratios_same} is at least the right hand-side of \eqref{eq:ratios_same}.
To this end, we use the ``bridge lemma'' from \cite{ln4directedcomponent}. Recall that $\dropG{X}=\mst\lr{G}-\mst\lr{G/X}$.

\begin{lemma}[Lemma~11 in \cite{ln4directedcomponent}]\label{lem:bridge}
Let $(G,c,R)$ be an instance of Steiner Tree, and let $T$ be a terminal minimum spanning tree.
Then for every solution $x$ to the Hypergraphic Relaxation, we have
\begin{equation*}
c(T) \le \sum_{(X,v)\in\vec{\cX}} x_{(X,v)} \cdot \dropG{X}.
 \end{equation*}
 
\end{lemma}

Using the bridge lemma, we now show that for MST-optimal instances $\cI$, we have $\hyp(\cI) = \opt(\cI)$, which implies one of the two directions of \Cref{thm:ratios_same}.

\begin{lemma}\label{lem:ineq_bridge_lemma}
We have
    \begin{equation*}
       \begin{aligned}
    &\ \sup \left\{ \frac{\hyp(\cI)}{\bcr(\cI)} \colon \cI \text{ instance of Steiner Tree} \right\} \\[2mm]
    \ \geq &\
    \sup \left\{ \frac{\opt(\cI)}{\bcr(\cI)} \colon \cI \text{ MST-optimal instance of Steiner Tree} \right\}.
    \end{aligned} 
    \end{equation*}  
\end{lemma}
\begin{proof}
Let $\cI= (G,c,R)$ be an MST-optimal instance of Steiner Tree and let $T$ be a terminal minimum spanning tree.
It suffices to show $\hyp(\cI) = \opt(\cI)$.

To prove this, let $x$ be an optimal solution to the Hypergraphic Relaxation for $\cI$ and suppose that $\hyp(\cI) < \opt(\cI)$.
Then by \Cref{lem:bridge}, we have
\begin{equation*}
\sum_{(X,v)\in\vec{\cX}} x_{(X,v)} \cdot \cost(X) \ =\ \hyp(\cI) 
\ < \ \opt(\cI) \ =\ c(T) \ \leq\ \sum_{(X,v)\in\vec{\cX}} x_{(X,v)} \cdot \dropG{X}.
\end{equation*}
Hence, there exists some $X\subseteq R$ such that $\cost(X) < \dropG{X}$, contradicting the MST-optimality of $\cI$.
We conclude that indeed $\hyp(\cI) = \opt(\cI)$.
\end{proof}

To complete the proof of \Cref{thm:ratios_same}, we consider the dual of the Hypergraphic Relaxation:
\begin{equation}\label{eq:hypergraphic_dual}
\begin{aligned}
\max \sum_{\emptyset \ne S\subseteq R\setminus\{r\}} y_U \\
\sum_{S\subseteq R: (X,v)\in  \delta^+_{\vec{\cX}}(S)} y_S \ \leq&\ \cost(X) &\forall\; (X,v)\in\vec{\cX} \\
y_S \ \ge&\ 0 &\forall\; \emptyset \ne S\subseteq R\setminus\{r\}
\end{aligned}
\end{equation}

Using standard uncrossing techniques, we can show that \eqref{eq:hypergraphic_dual} always has an optimum solution with laminar support.
That is, there exists an optimum dual solution $y$ such that for any $U,W\subseteq R\setminus\{r\}$ with $y_U >0$ and $y_W > 0$, we have $U\subseteq W$, or $W\subseteq U$, or $U\cap W=\emptyset$.

\begin{lemma}\label{lem:hypergraphic_dual_laminar}
For every instance of the Steiner tree problem, the dual \eqref{eq:hypergraphic_dual} of the Hypergraphic Relaxation has an optimum solution with laminar support.
\end{lemma}
\begin{proof}
    Among all optimal dual solutions let $y$ be one that maximizes
    \begin{align*}
        \sum_{S\subseteq R\setminus\lrb{r}}y_S\cdot\lrv{S}^2\,.
    \end{align*}
    We claim that $y$ has laminar support. 
    Suppose not. Then there are sets $S_1, S_2\subseteq R\setminus\lrb{r}$ with $y_{S_1},y_{S_2}>0$ that are neither disjoint, nor is one of $S_1$ and $S_2$ a subset of the other one. 
    In particular, none of the sets $S_1$, $S_2$, $S_1\cap S_2$, and $S_1\cup S_2$ is empty.
    Moreover,
\begin{align}
        \lrv{S_1\cup S_2}^2+\lrv{S_1\cap S_2}^2\ >\ \lrv{S_1}^2+\lrv{S_2}^2\,.
        \label{ieq:squares increase during uncrossing}
    \end{align}
    Let $0<\varepsilon\leq y_{S_1}, y_{S_2}$, and define a dual solution $y'$ by $y'_{S_1\cup S_2}\coloneqq y_{S_1\cup S_2}+\varepsilon$, $y'_{S_1\cap S_2}\coloneqq y_{S_1\cap S_2}+\varepsilon$, $y'_{S_1}\coloneqq y_{S_1}-\varepsilon$, $y'_{S_2}\coloneqq y_{S_2}-\varepsilon$, and $y'_S\coloneqq y_S$ otherwise.
    
    By construction, $y$ and $y'$ have the same objective value 
    \[
    \sum_{\emptyset \ne S\subseteq R\setminus\{r\}} y'_U\ =\ 2\varepsilon-2\varepsilon+\sum_{\emptyset \ne S\subseteq R\setminus\{r\}} y_U.
    \]
Furthermore, by \eqref{ieq:squares increase during uncrossing}, we have
    \begin{align*}
        \sum_{S\subseteq R\setminus\lrb{r}}y'_S\cdot\lrv{S}^2
        \ >\ 
        \sum_{S\subseteq R\setminus\lrb{r}}y_S\cdot\lrv{S}^2\,.
    \end{align*}
    Thus, to derive a contradiction to our choice of $y$, it remains to show that $y'$ is a \emph{feasible} dual solution.
    Indeed, by our choice of $\varepsilon$, the dual solution $y'$ is non-negative.
    If a directed component $\lr{X, v}$ is contained in $\delta^+_{\vec{\cX}}(S_1\cap S_2)$ or in $\delta^+_{\vec{\cX}}(S_1\cup S_2)$, it is also contained in at least one of $\delta^+_{\vec{\cX}}(S_1)$ and $\delta^+_{\vec{\cX}}(S_2)$.
    If a directed component $\lr{X, v}$ is contained in both $\delta^+_{\vec{\cX}}(S_1\cap S_2)$ and $\delta^+_{\vec{\cX}}(S_1\cup S_2)$, it is also contained in both $\delta^+_{\vec{\cX}}(S_1)$ and $\delta^+_{\vec{\cX}}(S_2)$.
    Thus, for every directed component $\lr{X, v}$:
    \begin{align*}
        \sum_{S\subseteq R: (X,v)\in  \delta^+_{\vec{\cX}}(S)} y'_S \ 
        \leq\ \sum_{S\subseteq R: (X,v)\in  \delta^+_{\vec{\cX}}(S)} y_S \ 
        \leq\ \cost(X)\,,
    \end{align*}
    which implies that $y'$ is feasible, contradicting our choice of $y$.
\end{proof}

Next, we use \Cref{lem:hypergraphic_dual_laminar} to complete the proof of \Cref{thm:ratios_same}, which  follows directly from \Cref{lem:ineq_bridge_lemma,lem:complete_same_ratio}.

\begin{lemma}\label{lem:complete_same_ratio}
We have
    \begin{equation*}
       \begin{aligned}
    &\ \sup \left\{ \frac{\hyp(\cI)}{\bcr(\cI)} \colon \cI \text{ instance of Steiner Tree} \right\} \\[2mm]
    \ \leq &\
    \sup \left\{ \frac{\opt(\cI)}{\bcr(\cI)} \colon \cI \text{ MST-optimal instance of Steiner Tree} \right\}.
    \end{aligned} 
    \end{equation*}  
\end{lemma}
\begin{proof}
Let $\cI$ be an instance of Steiner Tree and let $y$ be an optimum solution to \eqref{eq:hypergraphic_dual} with laminar support, which exists by \Cref{lem:hypergraphic_dual_laminar}.
We will construct a terminal spanning tree $T\subseteq \binom{R}{2}$ and a length function
$\ell \colon T \to \mathbb{R}_{\geq 0}$ such that
\begin{enumerate}[label = (\roman*)]
 \item\label{item:cost_equal_hyp} $\ell(T) = \hyp(\cI)$, and
 \item\label{item:maintain_dual_feasible} $y$ remains a feasible dual solution when decreasing the cost of every edge $e\in T$ to $\ell(e)$ (or adding the edge $e$ with cost $\ell(e)$ to the instance $\cI$).
\end{enumerate}
Then let $\cI'=(G',c', R)$ be the instance arising from $\cI$ by decreasing the cost of every edge $e\in T$ to $\ell(e)$, where we do not change the cost if it is already at most $\ell(e)$ and we add the edge $e$ (with cost $\ell(e)$) if it is not yet part of $G$.
Because $y$ is a feasible solution to \eqref{eq:hypergraphic_dual} for $\cI'$ by \ref{item:maintain_dual_feasible}, we have $\hyp(\cI') \geq \hyp(\cI)$.
Hence, we have $c'(T) \leq \ell(T) = \hyp(\cI) \leq \hyp(\cI') \leq \opt(\cI')$, which implies that the instance $\cI'$ is MST-optimal.
Moreover, $\bcr(\cI') \leq \bcr(\cI)$ because we only decreased the cost of edges in the construction of $\cI'$, and thus
\[
\frac{\hyp(\cI)}{\bcr(\cI)} \ =\ \frac{c'(T)}{\bcr(\cI)} \ \leq\ \frac{c'(T)}{\bcr(\cI')}
\ =\ \frac{\opt(\cI')}{\bcr(\cI')}.
\]
It remains to prove the existence of the terminal spanning tree $T$ and the function $\ell$ satisfying \ref{item:cost_equal_hyp} and \ref{item:maintain_dual_feasible}.
Let 
\[
\cL \coloneqq \big\{ S\subseteq R\setminus\lrb{r} : y_S > 0\big\} \cup \big\{R\big\} \cup \big\{ \{v\} \mid v\in R \big\}
\]
be the support of $y$, together with all singleton sets and the terminal set $R$.
Then $\cL$ is a laminar family and every non-singleton element of $\cL$ is the disjoint union of its children (its maximal strict subsets in $\cL$).
For a $v\in R$ and a set $L\in \cL$, we define the depth of $v$ inside $L$ as 
\[
\depth^L(v) \ \coloneqq\ \sum_{S\subsetneq L: v\in S} y_S,
\]
where we define $y_S \coloneqq 0$ for sets $S$ containing $r$ (which do not have a corresponding variable in the dual LP).

We construct a directed spanning tree $\vec{T} \subseteq R\times R$ oriented towards the root $r$ as follows.
For every set $L\in \cL$ we choose a root vertex $\rt(L) \in L$ as a vertex $v\in L$ of minimal $\depth^L(v)$.
In case of equal depth, we use consistent tie-breaking and give preference to the root $r$. In other words, we number the vertices as $v_1=r, v_2, \dots, v_n$ and in case of several vertices with the same depth, we choose $\rt(L)$ among those vertices as the vertex $v_i$ with smallest index $i$.

To construct the oriented tree $\vec{T}$, we construct a tree $\vec{T}_L \subseteq L\times L$ spanning $L$ and oriented towards $\rt(L)$ for each set $L\in \cL$.
When constructing these trees, we proceed from smaller to larger sets $L\in\cL$ and in the end $\vec{T}$ will be the tree $\vec{T}_R$, which is oriented towards $\rt(R)=r$.

For the sets $L\in \cL$ with $|L|=1$, the tree $\vec{T}_L=\emptyset$ has no edges.
Now consider a set $L\in \cL$ with $|L|\geq 2$.
Because of the order, in which we consider the sets from $\cL$, we have already constructed trees $\vec{T}_1, \dots, \vec{T}_k$ for the children $L_1,\dots, L_k$ of $L$ in $\cL$, where we may assume $\rt(L) \in L_1$.
Then we define
\[
\vec{T}_{L} \ \coloneqq\ \bigcup_{i=1}^k  \vec{T}_i \cup \bigcup_{i=2}^k \big\{ \big(\rt(L_i), \rt(L)\big) \big\}.
\]
Because $L$ is the disjoint union of its children in $\cL$, this is indeed a tree spanning $L$ and oriented towards $\rt(L) = \rt(L_1)$, where we used the consistent tie breaking in the choice of the roots.

For a directed edge $\vec{e}=(v,w)\in \vec{T} = \vec{T}_R$, we define 
\[
\ell(\vec{e}) \ \coloneqq\ \sum_{\substack{L\in \cL:\\ \vec{e} \in \delta^+(L)}} y_L.
\]
We let $T$ be the terminal spanning tree obtained from $\vec{T}$ by omitting the orientation, and for an edge $e$ of $T$, we define $\ell(e) \coloneqq \ell(\vec{e})$ for the directed edge $\vec{e} \in \vec{T}$ corresponding to $e$.
\bigskip

It remains to prove \ref{item:cost_equal_hyp} and \ref{item:maintain_dual_feasible}.
To prove \ref{item:cost_equal_hyp}, we observe that the oriented terminal minimum spanning tree $\vec{T}$ is constructed such that every set $L\in \cL\setminus \{R\}$ has exactly one outgoing edge in $\vec{T}$.
In particular, every set $S$ in the support of $y$ has exactly one outgoing edge in $\vec{T}$, implying
\[
\ell(T) \ =\ \ell(\vec{T}) \ =\ \sum_{L\in \cL} y_L \cdot \big|\vec{T} \cap \delta^+(L)\big| \ =\ \sum_{L\in \cL} y_L \ =\ \hyp(\cI).
\]

Finally, we prove \ref{item:maintain_dual_feasible}.
We observe that it suffices to prove that the dual constraints are satisfied for any $(X,v) \in \vec{\cX}$ for which the cost is determined by a full Steiner tree, that is there exists a tree of length $\cost(X)$ whose set of leaves is exactly $X$.
Indeed, otherwise, there is a tree of length $\cost(X)$ connecting $X$ and we consider the maximal full components $\vec{K}_1, \dots \vec{K}_m$ in this tree oriented towards $v$.
Let $(X_1,v_1),\dots (X_m, v_m)$ be the elements of $\vec{\cX}$ corresponding to these full components. 
Then
\[
\sum_{S\subseteq R: (X,v)\in  \delta^+_{\vec{\cX}}(S)} y_S \ \leq\ \sum_{i=1}^m \sum_{S\subseteq R: (X_i,v_i)\in  \delta^+_{\vec{\cX}}(S)} y_S  \ \leq\ \sum_{i=1}^m \cost(X_i) \ =\ \cost(X).
\]

Thus, it suffices to show that $y$ satisfies the constraints of \eqref{eq:hypergraphic_dual} for the instance $\cI'$ for every $(X,v)$ whose cost is determined by a full component.
The only full components whose cost differs between the instances $\cI$ and $\cI'$ are the cost of full components consisting of a single edge $e\in T$.
Thus, it suffices to show that for each edge $e=\{v,w\}\in T$, we have
\begin{equation}\label{eq:dual_feasibility_hyp1}
\sum_{S\subseteq R: (v,w) \in \delta^+(S)} y_S \ \leq\ c'(\{v,w\})
\end{equation}
and
\begin{equation}\label{eq:dual_feasibility_hyp2}
\sum_{S\subseteq R: (w,v) \in \delta^+(S)} y_S \ \leq\ c'(\{v,w\}),
\end{equation}
where $c'(\{v,w\})$ denotes the cost of the full component connecting $v$ and $w$ by a single edge.
By definition of $c'$, we either have $c'(\{v,w\}) = c(\{v,w\})$ or $c'(\{v,w\}) = \ell(\{v,w\})$.
If $c'(\{v,w\}) = c(\{v,w\})$, \eqref{eq:dual_feasibility_hyp1} and \eqref{eq:dual_feasibility_hyp2} follow from dual feasibility of $y$ for the original instance $(G,c,R)$.
Thus, we may assume $c'(\{v,w\}) = \ell(\{v,w\})$.
Consider an edge $e\in T$, where without loss of generality $\vec{e} =(v,w)\in \vec{T}$.
By definition of the function $\ell$, we have
\begin{equation*}
\sum_{S\subseteq R: (v,w) \in \delta^+(S)} y_S \ =\  \ell(\vec{e}) \ = c'(e).
\end{equation*}
Now let $L\in \cL$ be such that $\vec{e}\in \vec{T}_L$, but $\vec{e}$ is not contained in $\vec{T}_{L_i}$ for any child $L_i$ of $L$ in $\cL$.
Then $w= \rt(L)$.
By the choice of $\rt(L)$ as a vertex $u\in L$ with minimal $\depth^L(u)$, we have $\depth^L(w) \leq \depth^L(v)$, implying
\begin{equation*}
\sum_{S\subseteq R: (w,v) \in \delta^+(S)} y_S \ =\ \depth^L(w) \ =\ \depth^L(v) \ =\ \ell(\vec{e}) \ \leq\ c'(e).
\end{equation*}
\end{proof}
 
\newpage

\begin{sloppypar}\tolerance 900
\printbibliography
\end{sloppypar}

\clearpage

\appendix
\section{Invariance under Subdivision}\label{sec:invariance_under_subdivision}

\InvarianceSubdivision*

\begin{proof}
To construct the desired subdivision of $G=(V,E)$, we consider a variant of \Cref{algo:dual_growth_with_merge_plan}.
In this variant, we will have explicit variables $C(S,\vec{e})$ representing the contribution of a set $S\subseteq R$ to a directed edge $\vec{e}$.
Whenever we increase a dual variable $y_{U_S}$, the set $S$ will not only contribute to all edges in $\delta^+(U_S)$.
It will additionally contribute to all edges $(u,w)$, for which an edge $(u,v)$ arising from a suitable subdivision of $\{u,w\}$ would be contained in $\delta^+(U_S)$.\\
Given a solution $y$ to \eqref{eq:dual_bcr}, we say that a directed edge $\vec{e}=(v,w) \in \vec{E}$ is \emph{directed-tight} if 
\[
\sum_{S\subseteq R} C(S,(v,w)) \ =\ c(\vec{e}).
\]
Moreover, we say that an undirected edge $e=\{v,w\} \in E$ is \emph{undirected-tight} if 
\[
\sum_{S\subseteq R} C(S,(v,w)) + \sum_{S\subseteq R} C(S,(w,v))\ \geq\ c(e).
\]
Then the continuous dual-growth algorithm is given by \Cref{algo:continuous_dual_growth}.

\begin{algorithm}[h]
\caption{Continuous Growth with Merge Plan \label{algo:continuous_dual_growth}}
\KwIn{an instance $(G,c,R)$ of Steiner Tree and a merge plan $\cM = \lr{\cS^t}_{t\geq0}$ }
\KwOut{a solution $y$ to \eqref{eq:dual_bcr}}
\vspace*{2mm}
Initialize $t= 0$ and $y_U=0$ for all $U\subseteq V$.\;
Let $(V, \Vec{E})$ result from $G$ by bidirecting all edges.\;
\While{ $\lrv{\cS^t}>1$}{
For $S\in \cS^t$, define
$U_S \ =\ \{v\in V : v\text{ is reachable from $S$ by a path of directed-tight edges in $(V, \Vec{E})$}\}$.\;
For $S\in \cS^t$, let 
$C_S \ \coloneqq \delta^+(U_S) \cup \{ (v,w) \in \vec{E} \colon v\in U_S \text{ and $\{v,w\}$ is not undirected-tight} \}$.\;
Increase the time $t$, the dual variables $y_{U_S}$ for $S\in \cS^t$, and the the variables $C(S,\vec{e})$ with $S\in \cS^t$ and $\vec{e} \in C_S$ at unit rate until the partition $\cS^t$ changes, some edge $e\in E$ becomes undirected tight, or some edge $\vec{e}\in \vec{E}$ becomes directed tight.
}
\Return{y}
\end{algorithm}

We observe that for every edge $\vec{e}\in \vec{E}$, we have  
\[
\sum_{U: \vec{e}\in \delta^+(U)} y_U \ \leq\ \sum_{S\subseteq R} C(S,\vec{e}) \ \leq\ c(\vec{e}),
\]
implying that $y$ is indeed a feasible dual solution.
Next, we observe that \Cref{algo:continuous_dual_growth} is invariant under subdivision. 

\paragraph{Invariance of \Cref{algo:continuous_dual_growth} under subdivision.}
Suppose that we subdivide an edge $\{u,w\}$ of $G$ by a vertex $v$ resulting in a graph $G'$. 
Define $c_u \coloneqq c(\{u,v\})$ and $c_w \coloneqq c(\{v,w\})$, where $c_u +c_w = c(\{u,w\})$.
Then we claim that at any time of the algorithm, the contributions $C$ in \Cref{algo:continuous_dual_growth} applied to $G$ and the contributions $C'$ in \Cref{algo:continuous_dual_growth} applied to $G'$ satisfy for every set $S\subseteq R$
\begin{equation}\label{eq:subdivision_1}
\begin{aligned}
    C'\big(S, (u,v)\big) \ =&\ \min \Big\{C\big(S,(u,w)\big),\ c_u\Big\} \\
    C'\big(S, (v,w)\big) \ =&\ \max \Big\{C\big(S,(u,w)\big) - c_u,\ 0\Big\}
\end{aligned}
\end{equation}
and symmetrically
\begin{equation}\label{eq:subdivision_2}
\begin{aligned}
    C'\big(S, (w,v)\big) \ =&\ \min \Big\{C\big(S,(w,u)\big),\ c_w\Big\} \\
    C'\big(S, (v,u)\big) \ =&\ \max \Big\{C\big(S,(w,u)\big) - c_w,\ 0\Big\},
\end{aligned}
\end{equation}
where for all other edges $\vec{e}$ and sets $S$ contributions $C(S,\vec{e)}$ and $C'(S,\vec{e})$ are identical.
In particular, $C(S,(u,w)) = C'(S,(u,v)) + C'(S, (v,w))$ and
$C(S,(w,u)) = C'(S,(w,v)) + C'(S, (v,u))$.

We prove that \eqref{eq:subdivision_1} and \eqref{eq:subdivision_2} are maintained throughout the algorithm.
So suppose they are true at the current time $t$.
Then we show that for every $S\in \cS^t$, the sets of edges for which the contributions $C(S,\vec{e})$ or $C'(S, \vec{e})$ are increased are such that
\eqref{eq:subdivision_1} and \eqref{eq:subdivision_2} are maintained.

Note that to show \eqref{eq:subdivision_1} and \eqref{eq:subdivision_2} are preserved, it suffices to consider the contributions on $\lrb{u, v}$, $\lrb{v, w}$, $\lrb{u, w}$. Each of these edges is contained in exactly one of $G$ and $G'$.
Thus, the notion of (directed and undirected) tightness will in the following always refer to the run of \Cref{algo:continuous_dual_growth} for the graph ($G$ or $G'$) containing this edge.

We distinguish the following cases:

\textbf{Case 1:} $v\notin U_S$ \\
If none of $u$ and $w$ are in $U_S$, then because $v$ is not in $U_S$, none of the contributions appearing in \eqref{eq:subdivision_1} and  \eqref{eq:subdivision_2} changes.

If both $u$ and $w$ are in $U_S$, none of the edges $(u,v)$ and $(w,v)$ is directed-tight, which by \eqref{eq:subdivision_1} and \eqref{eq:subdivision_2} implies that the edge $\{v,w\}$ is not undirected-tight.
Then $C'(S,(u,v))$ and $C'(S,(w,v))$, as well as $C(S,(u,w))$ and $C(S,(w,u))$ are increased, which maintains \eqref{eq:subdivision_1} and \eqref{eq:subdivision_2}.

Otherwise only one of $u$ and $w$ is contained in $U_S$, say $u\in U_S$ and $w\notin U_S$.
Because also $v\notin U_S$, none of the contributions appearing in \eqref{eq:subdivision_2} changes, implying that \eqref{eq:subdivision_2} is maintained.
Moreover, because $v,w\notin U_S$, the edges $(u,v)$ and $(u,w)$ are not directed-tight and the $C'(S,(u,v))$ and $C(S,(u,w))$ increase, but $C(S,(v,w))$ does not change. This implies that \eqref{eq:subdivision_1} is maintained.

\textbf{Case 2:} $v\in U_S$ \\
Then we may assume without loss of generality $(u,v)$ is directed-tight (and thus also undirected tight) and $u\in U_S$.
If $w\notin U_S$, then $C(S, (u,w))$ and $C'(S,(v,w))$ increase, but all other contributions appearing in \eqref{eq:subdivision_1} or \eqref{eq:subdivision_2} remains unchanged, implying that \eqref{eq:subdivision_1} and \eqref{eq:subdivision_2} are maintained.
Otherwise $w\in U_S$. Then by \eqref{eq:subdivision_1} and \eqref{eq:subdivision_2}, the edge $\{u,w\}$ is undirected tight if and only if $\{v,w\}$ is undirected tight.
If both of these edges are undirected tight, none of the contributions appearing in \eqref{eq:subdivision_1} or \eqref{eq:subdivision_2} change.
If none of these edges is undirected-tight then $C(S,(u,w))$ and $C(S,(w,u))$, as well as $C'(S,(v,w))$ and $C'(S,(w,v))$ are increased, whicle the other contributions in \eqref{eq:subdivision_1} and \eqref{eq:subdivision_2} remain unchanged.
Thus, in all cases \eqref{eq:subdivision_1} and \eqref{eq:subdivision_2} are maintained.

\paragraph{Constructing the desired subdivision of $G$.}

We call an instance \emph{nice}, if \Cref{algo:continuous_dual_growth} on an instance that satisfies \ref{itm: reachability of both endpoints} of \Cref{def:well-subdivided}.
That is, whenever both endpoints of an edge are reachable (for \Cref{algo:continuous_dual_growth}), one of the orientations of the edge is already directed tight.

Observe that for any nice instance, \Cref{algo:continuous_dual_growth} and \Cref{algo:dual_growth_with_merge_plan} produce exactly the same dual solution, and moreover, an edge is directed-tight in \Cref{algo:continuous_dual_growth} if and only if it is also tight in \Cref{algo:dual_growth_with_merge_plan}.
This is because, for nice instances, we have throughout \Cref{algo:continuous_dual_growth} that $C_S = \delta^+(U_S)$ for all $S\in \cS^t$.

Because of \eqref{eq:subdivision_1} and \eqref{eq:subdivision_2},
every instance arising from a nice instance by subdivision of edges, is again a nice instance. 
Thus, because \Cref{algo:continuous_dual_growth} is invariant under subdivision of edges, \Cref{algo:dual_growth_with_merge_plan} is also invariant under subdivision of edges on nice instances.

Hence, it suffices to prove that for every instance $(G,c,R)$ and a fixed merge plan $\cM$, there is a subdivision $G'$ of $G$ such that $(G',c,R)$ (with merge plan $\cM$) is nice.
To this end, for every edge $\{u,w\}$ of $G$, we consider the contributions $C$ in \Cref{algo:continuous_dual_growth} at the time when $\{u,w\}$ became undirected tight.
Then we subdivide the edge by a vertex $v$ and define
\begin{align*}
    c(\{u,v\}) \ &\coloneqq\ \sum_{S\subseteq R} C(S, (u,w)) \\
    c(\{v,w\}) \ &\coloneqq\ \sum_{S\subseteq R} C(S, (w,u))\,.
\end{align*}
(If this would create an edge of length $0$, we do not subdivide the edge. In such cases no subdivision is needed anyways.)
Because we consider the contributions at the time where $\{u,w\}$ became undirected tight, these two costs indeed sum up to $c(\{u,w\})$.
Edges $\lrb{u, w}$ that never become undirected tight are subdivided by a vertex $v$ such that 
\begin{align*}
    c(\{u,v\}) \ &>\ \sum_{S\subseteq R} C(S, (u,w)) \\
    c(\{v,w\}) \ &<\ \sum_{S\subseteq R} C(S, (w,u))\,.
\end{align*}
and $c\lr{\lrb{u, w}}=c\lr{\lrb{u, v}}+c\lr{\lrb{v, w}}$.
Using the fact that the contributions after subdivision change according to \eqref{eq:subdivision_1} and \eqref{eq:subdivision_2} (and they do not change when subdividing further edges), we conclude that the resulting instance is indeed nice:

In the new instance, each edge has one original vertex as an endpoint and one newly introduced vertex as the other endpoint. Each newly introduced vertex $v$ has two incident edges. When $v$ is first reachable, both its incident edges have a tight orientation (the one pointing towards $v$). Note that the new vertices we introduced on edges that never got undirected tight will never be reached, meaning that \ref{itm: reachability of both endpoints} of \Cref{def:well-subdivided} trivially holds for their adjacent edges.
\end{proof}
 
\end{document}